\documentclass[onefignum,onetabnum]{siamonline220329}
\usepackage{algorithmic}
\usepackage{amssymb}
\usepackage{amsfonts}
\usepackage{amsmath}
\usepackage{amsopn,appendix}
\usepackage{array}
\usepackage{bm}
\usepackage{booktabs}
\usepackage[sort]{cite} 
\usepackage{color}
\usepackage{hyperref}
\usepackage{cleveref}
\usepackage{mathrsfs}
\usepackage{enumitem}
\usepackage{epsfig}
\usepackage{epstopdf}
\usepackage{flafter}
\usepackage{float}
\usepackage[T1]{fontenc}
\usepackage{graphicx}
\usepackage{ifthen}
\usepackage{lipsum}
\usepackage{makecell}
\usepackage{mdframed}
\usepackage{multirow}
\usepackage{pifont}
\usepackage{times}
\usepackage{threeparttable}
\usepackage{setspace}
\usepackage{subfigure}
\usepackage{url}

\usepackage{titlesec}
\usepackage{titletoc}
\usepackage{listings}
\usepackage{verbatim}
\usepackage{bbm}
\usepackage{mathtools}
\usepackage{xcolor, tikz}

\ifpdf
  \DeclareGraphicsExtensions{.eps,.pdf,.png,.jpg}
\else
  \DeclareGraphicsExtensions{.eps}
\fi

\setlist[enumerate]{leftmargin=.5in}
\setlist[itemize]{leftmargin=.5in}

\newsiamremark{remark}{Remark}
\newsiamremark{hypothesis}{Hypothesis}
\crefname{hypothesis}{Hypothesis}{Hypotheses}
\newsiamthm{claim}{Claim}

\definecolor{orange}{RGB}{255,107,0}

\newcommand{\C}{\boldsymbol{C}}

\newcommand{\A}{\boldsymbol{A}}

\newcommand{\B}{\boldsymbol{B}}

\renewcommand{\c}{\boldsymbol{c}}

\newcommand{\T}{{\!\top\!}}

\newcommand{\IH}{I_{\rm H}}
\newcommand{\JH}{J_{\rm H}}
\newcommand{\KH}{K_{\rm H}}
\newcommand{\IM}{I_{\rm M}}
\newcommand{\JM}{J_{\rm M}}
\newcommand{\KM}{K_{\rm M}}
\newcommand{\rH}{{\rm H}}
\newcommand{\rM}{{\rm M}}
\newcommand{\rS}{{\rm S}}
\newcommand{\tX}{\underline{\bm X}}

\newcommand{\tY}{\underline{\bm Y}}

\newcommand{\tT}{\underline{\bm T}}
\newcommand{\tD}{\underline{\bm D}}

\DeclareMathOperator*{\minimize}{\textrm{minimize}}

\makeatletter
\renewcommand{\maketag@@@}[1]{\hbox{\m@th\normalsize\normalfont#1}}%
\makeatother

\headers{Rethinking CTD for Hyperspectral Super-Resolution}
{M. Ding and X. Fu}

\title{Rethinking Coupled Tensor Analysis for Hyperspectral Super-Resolution: Recoverable Modeling Under Endmember Variability
\thanks{Submitted to the editors on Sep 9, 2025; revision submitted on Nov 2, 2025; accepted on Dec. 16.
\funding{The work of M. Ding was supported in part by the National Natural Science Foundation of China (12201522) and Natural Science Foundation of Sichuan Province of China (2024NSFSC1389). The work of X. Fu was supported in part by the National Science Foundation (NSF) under Project ECCS 2024058. 
}}
}

\author{Meng Ding\thanks{School of Mathematics, Southwest Jiaotong University, Chengdu, China
  (\email{dingmeng56@swjtu.edu.cn}).}
  \and Xiao Fu\thanks{Corresponding author. School of Electrical Engineering and Computer Science, Oregon State University, Corvallis, OR 97331, United States
  (\email{xiao.fu@oregonstate.edu}).}
}

\begin{document}

\maketitle
\begin{abstract}
This work revisits the hyperspectral super-resolution (HSR) problem, i.e., fusing a pair of spatially co-registered hyperspectral (HSI) and multispectral (MSI) images to recover a super-resolution image (SRI) that enhances the spatial resolution of the HSI. Coupled tensor decomposition (CTD)-based methods have gained traction in this domain, offering recoverability guarantees under various assumptions. Existing models such as canonical polyadic decomposition (CPD) and Tucker decomposition provide strong expressive power but lack physical interpretability. The block-term decomposition model with rank-$(L_r,L_r,1)$ terms (the LL1 model) yields interpretable factors under the linear mixture model (LMM) of spectral images, but LMM assumptions are often violated in practice---primarily due to nonlinear effects such as \textit{endmember variability} (EV).
To address this issue, we propose representing spectral images using a more flexible block-term tensor model with rank-$(L_r,M_r,N_r)$ terms (the LMN model). This modeling choice retains interpretability, subsumes CPD, Tucker, and LL1 as special cases, and robustly accounts for non-ideal effects such as EV, offering a balanced tradeoff between expressiveness and interpretability for HSR. Importantly, under the LMN model for HSI and MSI, recoverability of the SRI can still be established under proper conditions---providing strong theoretical support. Extensive experiments on synthetic and real datasets demonstrate the effectiveness and robustness of the proposed method.
\end{abstract}

\begin{keywords}
hyperspectral super-resolution,
endmember variability,
block-term tensor decomposition,
recoverability.
\end{keywords}


\section{Introduction}
The tradeoff between spatial and spectral resolution in spectral imaging is well known, primarily due to limitations in the image acquisition process \cite{Yokoya2017HSRoverview}. Hyperspectral images (HSIs) typically offer high spectral resolution but suffer from low spatial resolution, while multispectral images (MSIs) exhibit the opposite characteristics. To recover a super-resolution image (SRI) with high resolution in both spatial and spectral domains, hyperspectral super-resolution (HSR) has been widely studied. The goal of HSR is to integrate a pair of spatially co-registered HSI and MSI so that the resulting SRI combines the spectral richness of the HSI with the spatial detail of the MSI; see \cite{Yokoya2017HSRoverview, Prevost2022Coupled}.

Early approaches to HSR relied on heuristic methods such as component substitution \cite{Aiazzi2007Component} and multiresolution analysis \cite{Aiazzi2022Context}. Later on, more advanced algebraic and optimization techniques---such as low-rank and nonnegative matrix factorization \cite{Lin2018Maximum,Lee1999NMF,Teboulle2020Novel}---were introduced to better address the fusion challenge; see, e.g., \cite{Yokoya2012HSR,Wu2019Hyperspectral,Wei2015HSRSylvester}. These methods typically leverage the \textit{linear mixture model} (LMM) for spectral images \cite{Gillis2014Successive,Ma2014HUoverview,Gillis2015Enhancing}, along with carefully designed degradation models that reflect the data acquisition process, to cast the HSR task as various coupled matrix factorization (or low-rank matrix recovery) problems---forming a class of structured inverse problems.
The LMM describes the spectral image pixels as nonnegative combinations of endmembers (spectral signature of materials within the pixels), and is widely used in spectral image analysis (see the survey \cite{Ma2014HUoverview}).
Interestingly, as noted in \cite{Li2018HSRmatrix,liu2019there}, the recoverability of the SRI using low-rank matrix factorization can be theoretically guaranteed under certain conditions---for example, when each spectral pixel consists of only a small number of constituent materials.

More recently, a series of promising tensor factorization-based approaches have been proposed for HSR \cite{Kanatsoulis2018HSR,Prevost2020HSR,Ding2021HSR}. In contrast to matrix factorization methods, these techniques model spectral images as third-order tensors, thereby preserving the inherent multi-dimensional structure of the SRI. The \textit{coupled tensor decomposition} (CTD) framework assumes that the spectral images follow a tensor factorization model, such as the canonical polyadic (CP) decomposition model \cite{Hitchcock1927CPD}, the Tucker model \cite{Tucker1966Tucker}, or the block-term decomposition with multilinear rank-$(L_r,L_r,1)$ terms (the LL1 model) \cite{Lathauwer2008BTD2}.
By leveraging the algebraic properties of the CP decomposition, the work \cite{Kanatsoulis2018HSR} established theoretical guarantees for recoverability of the SRI under mild conditions. Similar recoverability results have also been demonstrated under the Tucker and LL1-based models, based on their respective decomposition structures; see \cite{Prevost2020HSR,Ding2021HSR}.
Other structural tensor models, such as tensor ring decomposition \cite{Zhao2016TR}, tensor singular value decomposition \cite{Kilmer2011Factorization,Zheng2024Scale}, and related methods \cite{He2022HSR,Xu2022HSR,Li2018CSTF}, have also been explored for HSR, but these efforts primarily focus on algorithmic design and empirical performance, without providing formal recoverability analysis, to our best knowledge.

\subsection{Challenges in Recoverability-Guaranteed HSR via CTD} While CPD-, Tucker-, and LL1-based CTD formulations all offer recoverability guarantees, they involve important tradeoffs. The LL1 model, grounded in the LMM, assigns clear physical meaning to its latent factors---enabling effective incorporation of prior information for structural regularization, which is especially beneficial under noise. However, the LMM assumption has its own limitations, as it fails to capture effects like \textit{endmember variability} (EV) \cite{zare2013endmember,Kervazo2021Provably,Borsoi2021Variability}. 
When EV is present, the endmembers vary across pixels and this violates the LMM grossly, making models relying on the LMM suffer from substantial performance degradation.
In contrast, CPD and Tucker models \cite{Kanatsoulis2018HSR,Prevost2020HSR} serve as universal tensor approximators \cite{sidiropoulos2017tensor} that can flexibly model scenarios containing EV and other nonlinear effects (via adjusting the tensor rank). Yet, their latent factors lack physical interpretability, making it difficult to apply meaningful regularization in practice. This oftentimes leads to performance deterioration under noisy and complex scenarios \cite{Ding2021HSR}.

In recent years, efforts have been made to address EV issues in spectral image modeling and CTD-based HSR, with particular emphasis on the recoverability of the SRI. In \cite{Borsoi2021Coupled}, a CTD approach based on the Tucker decomposition was proposed to account for EV, with exact SRI recovery established under certain conditions. Subsequently, \cite{Prevost2022Coupled} modeled the HSR problem with EV using a coupled LL1 formulation and likewise provided recoverability guarantees. However, both works \cite{Prevost2022Coupled,Borsoi2021Coupled} adopted an inter-image EV model, assuming that variability only occurs across HSI and MSI, whereas EV across spectral pixels \cite{Somers2011Variability} was not considered.

\subsection{Contributions}
This work revisits CTD-based HSR methods that offer recoverability guarantees, and proposes a new CTD model that strikes a balance among robustness to across-pixel EV, physical interpretability, and SRI recoverability. We should mention that while EV was repeatedly considered in various prior works for HSR (see, e.g., \cite{Borsoi2020Variability,Ye2022Bayesian,Camacho2022variability}), these approaches do not rely on tensor models and lack recoverability guarantees; hence, they are beyond the scope of this study. Our detailed contributions are as follows:

\begin{itemize}
    \item {\bf Using LMN Model to Incorporate EV.}
Our key idea is to adopt the tensor block-term decomposition (BTD) in rank-$(L_r, M_r, N_r)$ terms---referred to as the ``LMN'' model~\cite{Lathauwer2008BTD2}---to represent spectral images under EV. The LMN model extends beyond the classical LMM, while retaining strong physical interpretability in the presence of EV, thereby facilitating the incorporation of prior information to improve HSR performance.

\item {\bf Recoverability Supports.}
Building on this model, we provide a theoretical analysis showing that the CTD formulation with LMN-structured spectral images guarantees recoverability of the SRI under reasonable conditions. Our proof begins with a new essential uniqueness condition for the LMN decomposition, which could be of broader interest beyond HSR. Leveraging this result, we generalize existing analyses of CTD approaches based on CPD, Tucker, and LL1 models~\cite{Kanatsoulis2018HSR,Prevost2020HSR,Ding2021HSR} to establish recoverability under LMN. 
Unlike \cite{Prevost2022Coupled,Borsoi2021Coupled} which focuses on inter-image EV across HSI and MSI, our result demonstrates recoverability in scenarios where unknown EV is present across spectral pixels.

\item {\bf Regularization and Algorithm Design.}
Moreover, we leverage the physical interpretation of the LMN latent factors to design model-based constraints and regularization. We propose a first-order inexact block coordinate descent algorithm to solve the resulting constrained CTD problem. Extensive simulations and real-data experiments validate the effectiveness of the proposed model and algorithm.

\end{itemize}

\medskip
\noindent
{\bf Notation.} The scalars, vectors, matrices, and tensors are represented by the lowercase ($x$) or uppercase ($X$), boldface lowercase ($\bm{x}$), boldface uppercase ($\bm{X}$), and underlined boldface uppercase ($\underline{\bm{X}}$), respectively.
The symbols $[\bm{x}]_i$, $[\bm{X}]_{i,j}$, and $[\underline{\bm{X}}]_{i,j,k}$ respectively denote the $i$-th, $(i,j)$-th, and $(i,j,k)$-th element of $\bm{x}\in \mathbb{R}^{I}$, $\bm{X}\in \mathbb{R}^{I\times J}$, and $\underline{\bm{X}}\in \mathbb{R}^{I\times J\times K}$. 
$\|\bm{X}\|_{F}=\sqrt{\sum_{i,j}[\bm{X}]_{i,j}^{2}}$ and $\|\underline{\bm{X}}\|_{F}
=\sqrt{\sum_{i,j,k}[\underline{\bm{X}}]_{i,j,k}^{2}}$ respectively denote the Frobenius norms of the matrix $\bm{X}\in \mathbb{R}^{I\times J}$ and the tensor $\underline{\bm{X}}\in \mathbb{R}^{I\times J\times K}$. 
The mode-$n$ product of a tensor $\underline{\bm{X}}\in \mathbb{R}^{I_1\times I_2\times \cdots \times I_N}$ and a matrix $\bm{Y}\in \mathbb{R}^{I_n\times J}$ yields a tensor $\underline{\bm{X}}\times_n \bm{Y}\in \mathbb{R}^{I_1\times \cdots I_{n-1}\times J \times I_{n+1} \cdots \times I_N}$, which each element is given by
$[\underline{\bm{X}}_{\times n}\bm{Y}]_{i_1\ldots i_{n-1} j i_{n+1} \ldots i_N}=\sum_{i_n=1}^{I_n}[\underline{\bm{X}}]_{i_1,\ldots,i_N}[\bm{U}]_{j,i_n}$.

\section{Background}
\label{sec:Background}
\subsection{Problem Setting of HSR}
We consider the classical setting \cite{Kanatsoulis2018HSR,Ding2021HSR,Wu2019Hyperspectral,Yokoya2012HSR,Wei2016Fusion,Borsoi2021Coupled,Borsoi2021Variability,Prevost2020HSR,Prevost2022Coupled} where a pair of spatially co-registered HSI and MSI are available.
The term ``co-registration'' means that the HSI and MSI cover the exactly identical spatial region with a shared underlying continuous spatial coordinate system.
The HSI and MSI images are denoted as $$\underline{\bm{Y}}_{\rm H}\in \mathbb{R}^{I_{\rm H}\times J_{\rm H}\times K_{\rm H}}\quad\&\quad \underline{\bm{Y}}_{\rm M}\in \mathbb{R}^{I_{\rm M}\times J_{\rm M}\times K_{\rm M}},$$
respectively,
where $I_{\rm H}$ and $J_{\rm H}$ are the vertical and horizontal spatial dimensions, respectively, $\KH$ denotes the number of wavelengths, and $\IM, \JM$ and $\KM$ are defined in the same way for the MSI.
Note that under the HSR setting, the following inequalities hold \cite{Kanatsoulis2018HSR,Ding2021HSR,Prevost2020HSR}: $$\KH\gg \KM,\quad \IM\JM\gg \IH\JH;$$ i.e., the HSI has a higher spectral resolution yet the MSI's spatial resolution is much finer.
The goal of HSR is to fuse $\tY_{\rH}$ and $\tY_{\rM}$ to produce an SRI, denoted as $\underline{\bm{Y}}_{\rm S}\in \mathbb{R}^{\IM\times \JM\times \KH}$---which takes the spatial resolution of the MSI and the spectral resolution of the HSI.  

\begin{figure}[!t]
\centering
\includegraphics[width=0.8\linewidth]{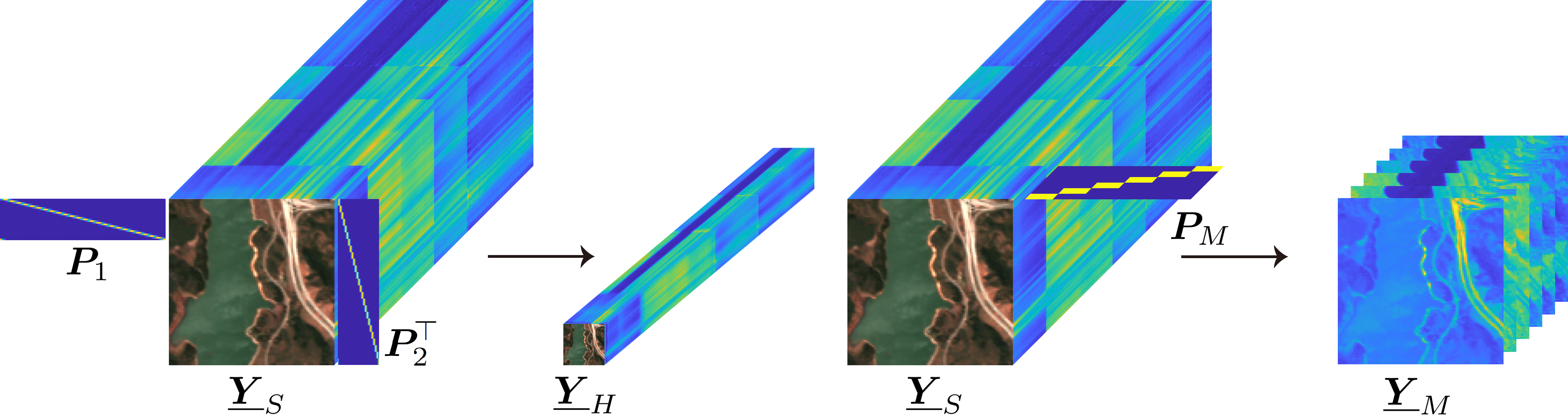}
\caption{Illustration of spatial and spectral degradations from SRI ($\underline{\bm{Y}}_{\rm S}$) to HSI ($\underline{\bm{Y}}_{\rm H}$) and MSI ($\underline{\bm{Y}}_{\rm M}$), respectively; figure adapted from \cite{Kanatsoulis2018HSR,Ding2021HSR}.}
  \label{fig:SRI_degradation}
\end{figure}

In the literature, the HSI and MSI are assumed to be downsampled and degraded from $\underline{\bm{Y}}_{\rm S}$; see, e.g., \cite{Kanatsoulis2018HSR,Ding2021HSR,Prevost2022Coupled,Borsoi2021Coupled}. To be specific, the spatial degradation model to obtain the HSI is as follows \cite{Kanatsoulis2018HSR,Ding2021HSR,Prevost2022Coupled,Borsoi2021Coupled}:
\begin{align}\label{eq:SRI2HSI}
  \underline{\bm{Y}}_{\rm H}(:,:,k) = \bm P_1 \underline{\bm{Y}}_{\rm S}(:,:,k) \bm P_2^\top,~ k=1,\ldots,K_{\rm H},
\end{align}
where $\underline{\bm{Y}}_{\rm S}(:,:,k)$ denotes the $k$th band of $\underline{\bm{Y}}_{\rm S}$, $\bm P_1\in\mathbb{R}^{I_{\rm H}\times I_{\rm M}}$ and $\bm P_2\in\mathbb{R}^{J_{\rm H}\times J_{\rm M}}$ denote two spatial degradation matrices applied to the vertical and horizontal dimensions. 
These operators are assumed to consists of blurring kernels and downsampling operations \cite{Kanatsoulis2018HSR,Ding2021HSR,Prevost2020HSR}.
The operations in \eqref{eq:SRI2HSI} can be expressed using the tensor mode-product:
\begin{align}\label{eq:modeproduct12}
\underline{\bm{Y}}_{\rm H} =\underline{\bm{Y}}_{\rm S}\times_1 \bm{P}_1\times_2 \bm{P}_2,
\end{align}
where $\times_n$ denotes the mode-$n$ product.
Similarly, the MSI is obtained via spectral degradation from the SRI:
\begin{align*}
  \underline{\bm{Y}}_{\rm H}(i,j,:) = \bm P_{\rm M} \underline{\bm{Y}}_{\rm S}(i,j,:),~\forall (i,j)\in [I_{\rm M}]\times [J_{\rm M}],
\end{align*}
where $\underline{\bm{Y}}_{\rm S}(i,j,:)$ denotes the $(i,j)$th spectral pixel of $\underline{\bm{Y}}_{\rm S}$, and $\bm P_{\rm M}\in\mathbb{R}^{K_{\rm M}\times K_{\rm H}}$ is a spectral degradation matrix \cite{Kanatsoulis2018HSR,Ding2021HSR,Prevost2022Coupled,Borsoi2021Coupled}, which is equivalently expressed as follows:
\begin{align}\label{eq:modeproduct3}
    \underline{\bm{Y}}_{\rm M} =\underline{\bm{Y}}_{\rm S} \times_3 \bm P_{\rm M}.
\end{align}
The spectral degradation is often modeled as relatively simple operations such as weighting and aggregation \cite{Kanatsoulis2018HSR,Ding2021HSR,Prevost2020HSR}.
The degradation model in Eqs~\eqref{eq:modeproduct12}-\eqref{eq:modeproduct3} is illustrated in Fig. \ref{fig:SRI_degradation}.
Note that all the matrices $\bm P_1$, $\bm P_2$, and $\bm P_{\rm M}$ are ``compressing'' matrices (i.e., fat matrices). Therefore, recovering $\tY_{\rm S}$ from the compressed measurements $\tY_{\rm H}$ and $\tY_{\rm M}$ is an ill-posed inverse problem. That is, even if  $\bm P_1$, $\bm P_2$, and $\bm P_{\rm M}$ are known, there might be an infinite number of solutions of $\tY_{\rm S}$ that satisfy Eqs~\eqref{eq:modeproduct12}-\eqref{eq:modeproduct3}, making establishing recoverability of SRI nontrivial.

\subsection{CTD-based HSR via Various Tensor Models} \label{sec:ctd}
Starting from \cite{Kanatsoulis2018HSR},
a series of CTD based HSR approaches, e.g., \cite{Kanatsoulis2018HSR,Ding2021HSR,Prevost2020HSR,Prevost2022Coupled,Borsoi2021Coupled}, showed promising capability in establishing recoverability of $\tY_{\rm S}$ under reasonable assumptions. The CTD formulations can be summarized as follows:
\begin{align}\label{eq:coupledctd}
\min_{\bm{\theta}}~ &\left\|\underline{\bm{Y}}_{\rm H} -\underline{\bm{Y}}_{\rm S}(\bm \theta)\times_1 \bm{P}_1\times_2 \bm{P}_2 \right\|^2_{F}  + \left\|\underline{\bm{Y}}_{\rm M} -\underline{\bm{Y}}_{\rm S}(\bm \theta) \times_3 \bm{P}_{\rm M} \right\|^2_{F},  \nonumber
\end{align}
where $\underline{\bm{Y}}_{\rm S}(\bm \theta)$ represents a tensor model parameterized by $\bm \theta$. 
In the above, the two tensors $\tY_{\rm H}$ and $\tY_{\rm M}$ are decomposed simultaneously with their latent factors shared across $\bm \theta$---making the two tensors ``coupled'' through the latent factors and leading to the terminology CTD. 
As shown in \cite{Kanatsoulis2018HSR,Ding2021HSR,Prevost2020HSR,Prevost2022Coupled,Borsoi2021Coupled}, $\underline{\bm{Y}}_{\rm S}(\bm \theta)$, by picking a proper tensor model $\underline{\bm{Y}}_{\rm S}(\bm \theta)$, it is possible to show that any optimal solution of \eqref{eq:coupledctd} recovers the ground truth SRI, i.e., $\tY_{\rm S}(\bm \theta^\star) =\tY_{\rm S}$.
The work in \cite{Kanatsoulis2018HSR} also first established that, under certain conditions, such recoverability holds even without knowing $\bm P_1$ and $\bm P_2$ if the CPD model is used to represent $\tY_{\rm S}(\bm \theta)$. This property was then shown to hold for other tensor models, e.g., the LL1 model \cite{Ding2021HSR}.
A number of representative models for $\tY_{\rm S}(\bm \theta)$ that admit provable recoverability are discussed in the Sec.~\ref{sec:cpd}-Sec. \ref{sec:ll1}.

\subsubsection{CPD-based CTD for HSR}\label{sec:cpd}
In \cite{Kanatsoulis2018HSR}, the idea is to model $\underline{\bm{Y}}_{\rm S}$ using a CPD model, i.e.,
\begin{align}\label{eq:CP}
\underline{\bm{Y}}_{\rm S}(\bm \theta) =\sum_{f=1}^F \bm{A}(:,f)\circ \bm{B}(:,f)\circ \bm{C}(:,f),
\end{align}
where $\bm \theta=(\bm A,\bm B,\bm C)$,
$\bm{A}\in \mathbb{R}^{I_{\rm M}\times F}$, $\bm{B}\in \mathbb{R}^{J_{\rm M}\times F}$, and $\bm{C}\in \mathbb{R}^{K_{\rm H}\times F}$ are the latent factors, 
and $F$ is the tensor rank.  Note that any tensor can be expressed as a CPD model, with a sufficiently large $F$ \cite{sidiropoulos2017tensor}.
More interestingly, the work \cite{Kanatsoulis2018HSR} empirically showed that real spectral images can be approximated by \eqref{eq:CP} with a relatively low rank $F$, which is a result of the fact that spectral images exhibit correlations among all three modes.
The work \cite{Kanatsoulis2018HSR} showed that the optimal solution of \eqref{eq:coupledctd} ensures recovering the SRI.
To be specific, assume that $\tY_{\rm S}= \sum_{f=1}^F \bm{A}^\natural(:,f)\circ \bm{B}^\natural(:,f)\circ \bm{C}^\natural(:,f)$ is the ground truth SRI.
Denote the optimal solution of \eqref{eq:coupledctd} as $\bm A^\star$, $\bm B^\star$, and $\bm C^\star$ under the parameterization in \eqref{eq:CP}. Then, it was shown in \cite{Kanatsoulis2018HSR} that the following holds under reasonable conditions (e.g., $F$ is sufficiently low in \eqref{eq:CP}):
\begin{align}\label{eq:recover}
        \tY_{\rm S} = \sum_{f=1}^F \bm{A}^\star(:,f)\circ \bm{B}^\star(:,f)\circ \bm{C}^\star(:,f)=\sum_{f=1}^F \bm{A}^\natural(:,f)\circ \bm{B}^\natural(:,f)\circ \bm{C}^\natural(:,f);
\end{align}
that is, any optimal solution of the CTD formulation under the CPD model provably recovers the SRI.
The CTD perspective provided in \cite{Kanatsoulis2018HSR} established recoverability of the SRI for the first time.

\subsubsection{Tucker-based CTD for HSR}\label{sec:tucker}
Later on, the work \cite{Prevost2020HSR} replaced the CPD model in \eqref{eq:CP} with a Tucker model, i.e.,
\begin{align}\label{eq:Tucker}
\underline{\bm{Y}}_{\rm S}(\bm \theta) ={\underline{\bm{D}}}\times_1 {\bm{A}}\times_2 {\bm{B}}\times_3  {\bm{C}},
\end{align}
where $\bm \theta=(\underline{\bm{D}},\bm A,\bm B,\bm C)$,
${\bm{A}}\in \mathbb{R}^{I_{\rm M}\times L}$, ${\bm{B}}\in \mathbb{R}^{J_{\rm M}\times M}$, and ${\bm{C}}\in \mathbb{R}^{K_{\rm H}\times N}$ are factor matrices, and $\underline{\bm{D}}\in \mathbb{R}^{L\times M\times N}$ is the core tensor, and the tuple $(L,M,N)$ is often referred to the multilinear Tucker rank.
Note that $L\leq I_{\rm M}$, $M\leq J_{\rm M}$, and $N\leq K_{\rm H}$ and that every tensor admits a Tucker representation with sufficiently large $L$, $M$, and $N$.
Using this coupled Tucker model, the work \cite{Prevost2020HSR} showed that recoverability of the SRI can be retained if $(L,M,N)$ are sufficiently small.
They also proposed a fast HSR algorithm leveraging the HOSVD algorithm for the Tucker decomposition.

\begin{figure}[!t]
\centering
\includegraphics[width=0.69\linewidth]{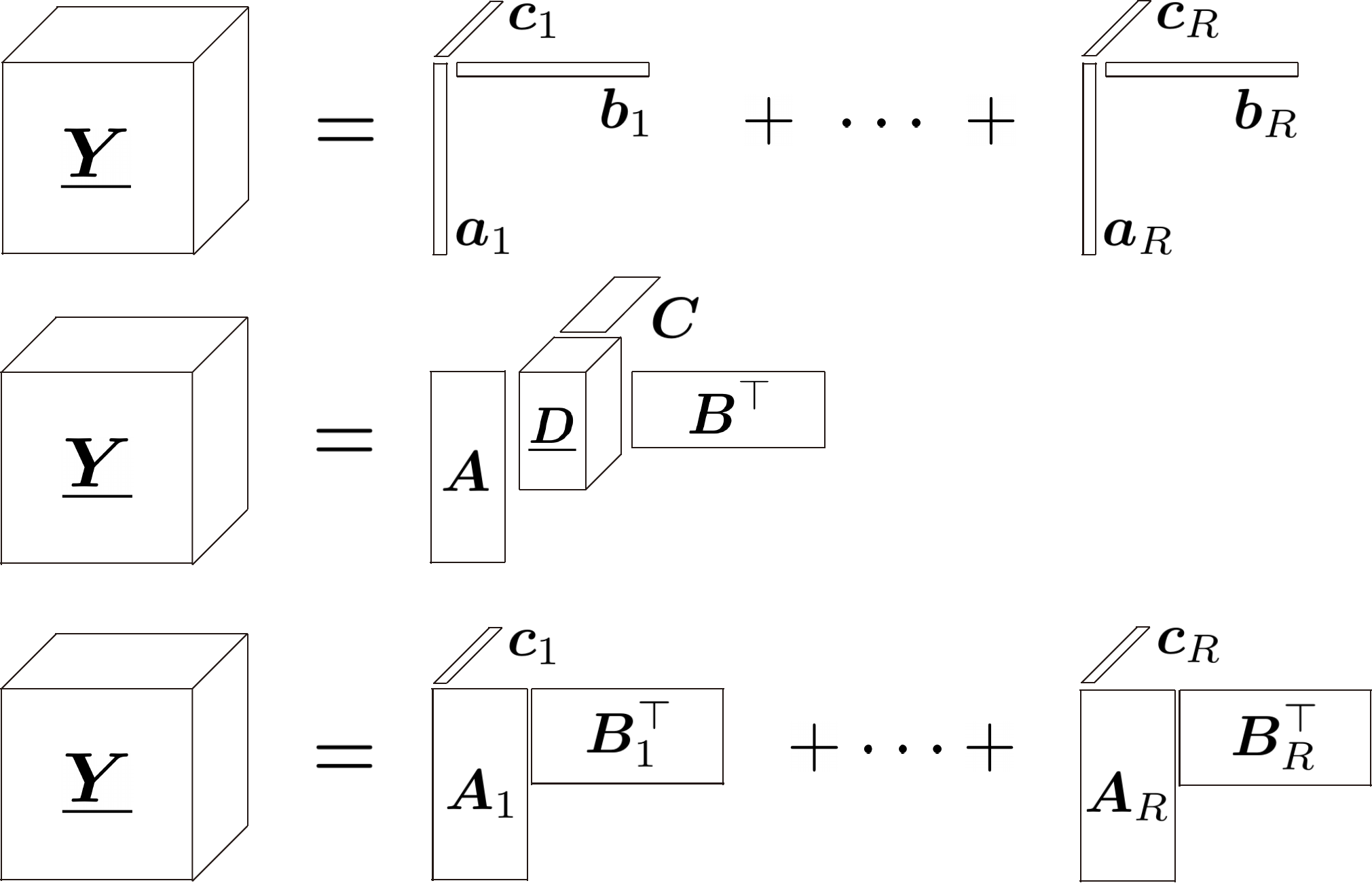}
\caption{Illustration of three tensor decompositions. Top to bottom: CPD, Tucker, and LL1, respectively; illustration adapted from \cite{Ding2021HSR}. }
  \label{fig:tensor_comparison}
\end{figure}

\subsubsection{LL1-based CTD for HSR}\label{sec:ll1}
Following the idea in \cite{Qian2017MVNTF}, the work \cite{Ding2021HSR} proposed to use the LL1 model for the SRI, i.e.,
\begin{align}\label{eq:LL1}
\underline{\bm{Y}}_{\rm S} =\sum_{r=1}^R  \left(\bm{A}_r \bm{B}_r^\top\right)\circ \bm{C}(:,r),
\end{align}
where $\bm{A}_r\in \mathbb{R}^{I_{\rm M}\times L_r}$, $\bm{B}_r\in \mathbb{R}^{J_{\rm M}\times L_r}$, and $\bm{C}\in \mathbb{R}^{K_{\rm H}\times R}$ are factor matrices; see Fig.~\ref{fig:tensor_comparison}.
There is an interesting link between the LL1 model and the LMM of spectral images.
Denote $\bm y_{\rm S}^{(i,j)}=\tY_{\rm S}(i,j,:)$ 
as the pixel located at $(i,j)$.
Assume that there are $R$ materials contained in the SRI, whose spectral signatures are $\bm c_1,\ldots,\bm c_R$. Then, under the LMM, the pixel can be expressed as
$  \bm y_{\rm S}^{(i,j)} = \sum_{r=1}^R \bm S_r(i,j)\bm c_r$;
i.e., the pixels $\bm y_{\rm S}^{(i,j)}$ is a weighted sum of $R$ spectral signatures (or endmembers) represented by $\bm c_1,\ldots,\bm c_R$. The term $\bm S_r(i,j)$ represents the ``abundance'' of material $r$ at the location $(i,j)$, and $\bm S_r\in\mathbb{R}^{I_{\rm M}\times J_{\rm M}}$ is the {\it abundance map} indicating the amount of material $r$ over the region.
Note that if one puts all pixels $\bm y_{\rm S}^{(i,j)}$ together, we have
\begin{align}
    \tY_{\rm S} =\sum_{r=1}^R \bm S_r \circ {\bm C}(:,r),
\end{align}
where $\bm C=[\bm c_1,\ldots,\bm c_R]$. If ${\rm rank}(\bm S_r) \leq L_r$, one can re-express it as
$\bm S_r =\A_r\B_r^\T$, which recovers the LL1 model in \eqref{eq:LL1}.
It was noticed that assuming the abundance maps to be low rank matrices is reasonable, as the spread of a material in the spatial domain often changes gracefully \cite{Ding2023Fast,Qian2017MVNTF,Ding2021HSR}.

\subsection{The Expressiveness and Interpretability Tradeoff}
Fig.~\ref{fig:tensor_comparison} illustrates the tensor models that ensure recoverability of $\tY_{\rm S}$ under the CTD formulation in~\eqref{eq:coupledctd}. Several key observations are outlined below:

First, a salient feature of the CTD framework in~\eqref{eq:coupledctd} is its ability to characterize the recoverability of the SRI in relation to the complexity of the underlying tensor model---often determined by notions of tensor rank; see~\cite{Prevost2020HSR,Ding2021HSR,Kanatsoulis2018HSR}. Recoverability guarantees are crucial in inverse problems, where infinitely many spurious solutions may exist. They help ensure that the recovered solution is meaningful.

\begin{figure}[!t]
\centering
\includegraphics[width=0.5\linewidth]{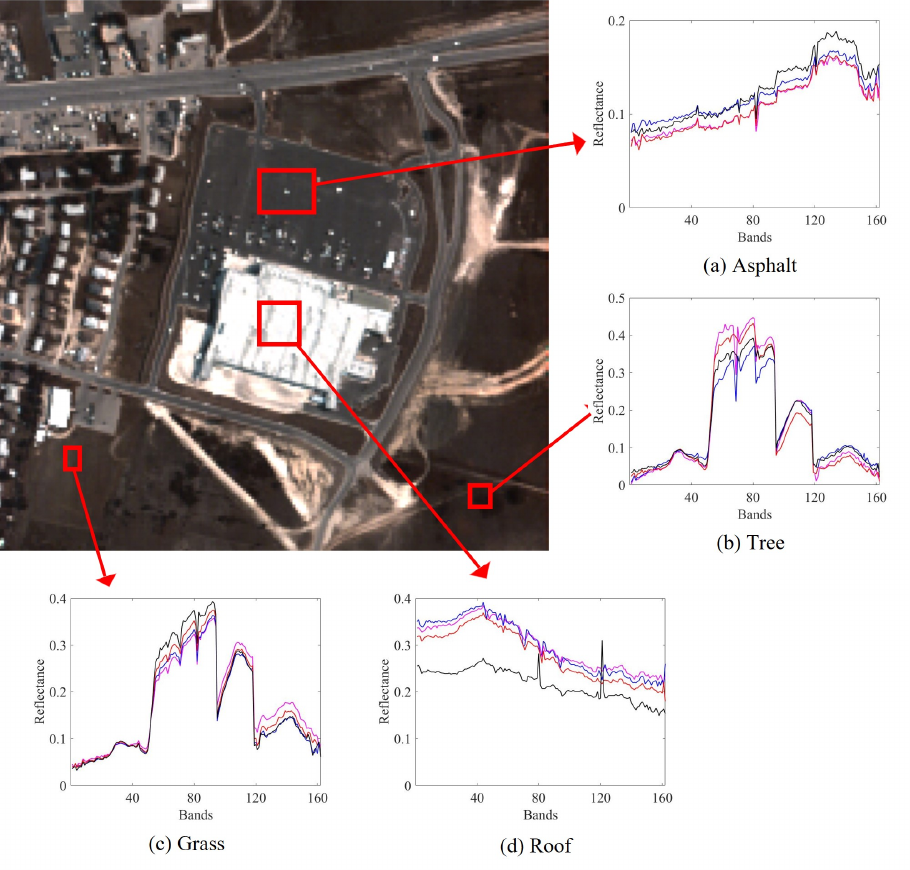}
\caption{Illustration of endmember variability of the Urban dataset, i.e., the spectrum of a single material (e.g., asphalt, tree, grass, and roof) can vary from pixel to pixel. (The source image is a subimage of the Urban data, which is collected by the HYDICE sensor over an urban area at Copperas Cove, TX, U.S.)}
  \label{fig:variability1}
\end{figure}

Second, the CPD and Tucker models are known for their expressive power: with sufficiently high rank (CPD) or multilinear rank (Tucker), they can approximate {\it any} spectral image. However, a limitation of these models is that their latent factors generally lack physical interpretability in the context of HSR. This makes it difficult to impose physical meaning-based regularization or constraints on the latent factors---which are often essential for improving robustness in challenging or noisy conditions.

Third, the LL1 model enjoys a meaningful connection to the classical LMM (cf. Sec. \ref{sec:ll1}). This connection enables the application of structured constraints on the latent factors (e.g., nonnegativity or spatial smoothness on $\bm S_r$), providing interpretability and control. Nonetheless, the LMM---and by extension, LL1---struggles to capture more complex, realistic phenomena such as EV. While CPD and Tucker can, in principle, model such complexities by increasing model rank, doing so often results in higher computational cost and numerical instability. Moreover, without structural regularization, increased rank may not yield actual performance gains.


Beyond the three aforementioned tensor models, various other tensor decomposition approaches have also been explored for representing $\tY_{\rm S}$ within the CTD framework---such as tensor ring~\cite{He2022HSR} and tensor singular value decomposition~\cite{Xu2022HSR}. 
However, these models are not the focus of this work, as recoverability guarantees were not primary considerations under these formulations. Moreover, they face similar challenges as previously discussed: namely, the trade-off between model expressiveness and physical interpretability, particularly in complex scenarios.

\section{Rethinking CTD: Balancing Interpretability and Expressiveness}

\subsection{Modeling Spectral Images with EV}
The LL1 model-based HSR in \cite{Ding2021HSR} showed state-of-the-art performance when the spectral images closely follow the LMM (in their experiments, the images were generated following the LMM with small noise perturbations). This demonstrates the importance of using prior information and physical meaning in HSR. However, as we mentioned, the LMM is not always accurate.

\begin{figure}[!t]
\centering
\includegraphics[width=0.6\linewidth]{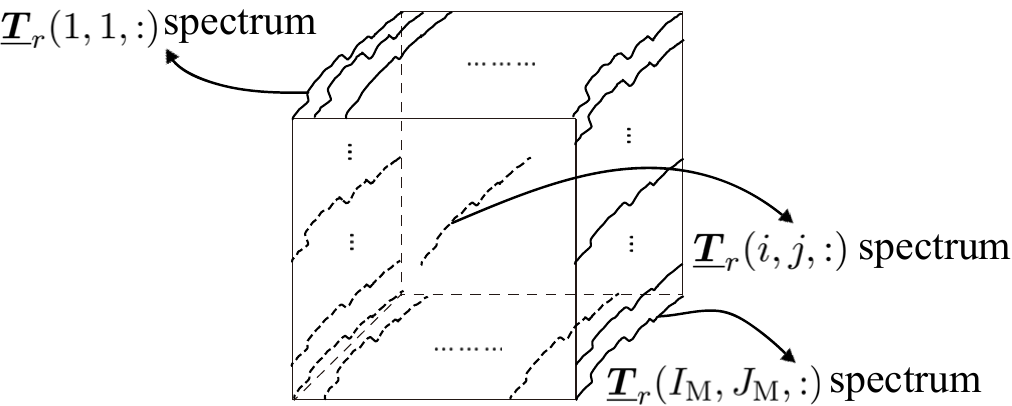}
\caption{Illustration of the mode-3 fibers of $\underline{\bm{T}}_r$, which correspond to the spectral signatures of the same materials across the $I_M\times J_M$ space. The key postulate is that, as all of the spectra are variations of the same endmember, all the fibers are similar to each other. Therefore, each $\underline{\bm{T}}_r$ has a low multi-linear rank.}
  \label{fig:LMN_Tr}
\end{figure}

Fig.~\ref{fig:variability1} shows a case where the LMM is clearly violated. Under the LMM, the endmembers $\bm c_1,\ldots,\bm c_R$ do not change across all pixels. However, one can see that the spectral signatures of \texttt{asphalt}, \texttt{tree}, \texttt{grass}, and \texttt{roof} all change, sometimes drastically---showing visible EV effects.
In the presence of EV, relying on the LMM may not always obtain satisfactory results for spectral imaging tasks; see discussions in \cite{Borsoi2021Coupled,Borsoi2021Variability,Kervazo2021Provably,Prevost2022Coupled,zare2013endmember,Somers2011Variability}.

\begin{figure}[!t]
\centering
\includegraphics[width=0.75\linewidth]{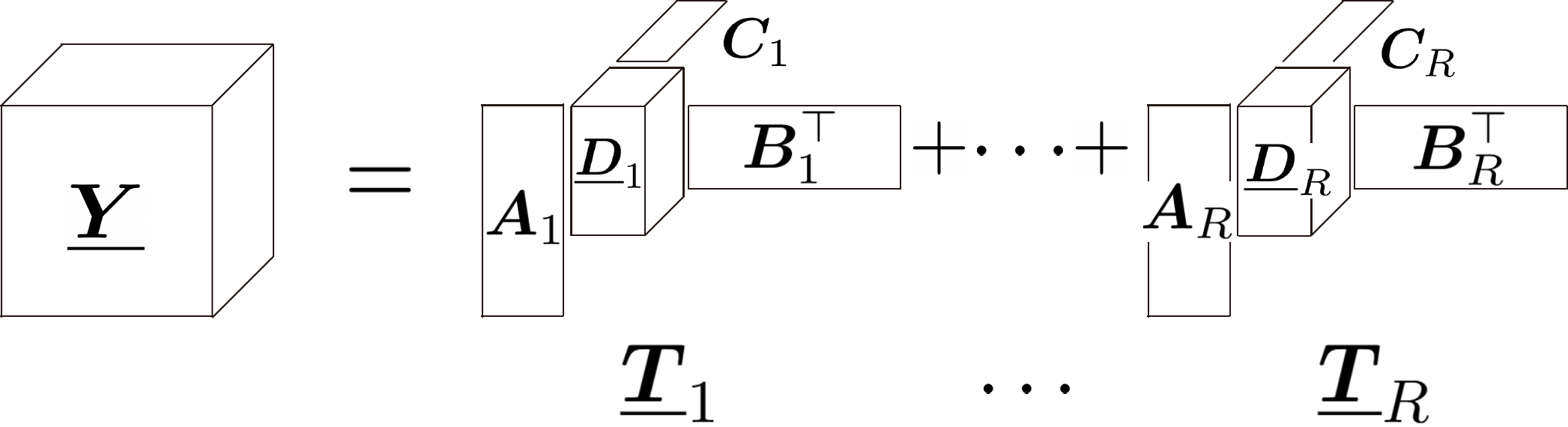}
\caption{Illustration of the LMN model defined in \cite{Lathauwer2008BTD1,Lathauwer2008BTD2}. The LMN model subsumes the CPD, Tucker, and LL1 models in Fig.~\ref{fig:tensor_comparison} as its special cases.}
  \label{fig:LMN}
\end{figure}

\begin{table*}[!t]
\renewcommand\arraystretch{1}
\setlength{\tabcolsep}{5pt}
\renewcommand\arraystretch{1.1}
  \centering
  \caption{The fitting errors of using CPD/Tucker/LL1/LMN models to approximate the Urban data of size $200\times 200\times 162$ (The numbers in the lower left subscript indicate the corresponding rank settings).}
  \resizebox{\linewidth}{!}{
  \begin{tabular}{l|l|l|l|l|l|l|l}\hline

    \hline
    num. of paras. & \multicolumn{1}{c|}{$\sim$30K} 
    & \multicolumn{1}{c|}{$\sim$45K}  & \multicolumn{1}{c|}{$\sim$60K} & \multicolumn{1}{c|}{$\sim$75K}  & \multicolumn{1}{c|}{$\sim$90K} & \multicolumn{1}{c|}{$\sim$105K} & \multicolumn{1}{c}{$\sim$120K}\\ \hline
    $\textrm{CPD}_{(F)}$    
    & $0.1394_{(52)}$ & $0.0992_{(80)}$ & $0.0765_{(106)}$ 
    & $0.0605_{(133)}$ & $0.0495_{(160)}$ & $0.0415_{(186)}$ & $0.0355_{(212)}$\\
    $\textrm{Tucker}_{(L,M,N)}$ 
    & $0.1336_{(51,54,3)}$ & $0.1073_{(66,66,4)}$ & $0.0869_{(90,88,3)}$ 
    & $0.0753_{(96,96,4)}$ & $0.0670_{(108,104,4)}$ & 0.$0591_{(120,118,4)}$ & $0.0557_{(120,118,5)}$\\
    $\textrm{LL1}_{(L)}$    
    & $0.1470_{(18)}$ & $0.1122_{(27)}$ & $0.0892_{(37)}$ 
    & $0.0772_{(46)}$ & $0.0725_{(50)}$ & $0.0724_{(50)}$ & $0.0724_{(50)}$\\
    $\textrm{LMN}_{(L,M,N)}$  
    & $\textbf{0.1273}_{(14,14,5)}$ & $\textbf{0.0884}_{(22,22,3)}$ 
    & $\textbf{0.0671}_{(28,28,4)}$ & $\textbf{0.0563}_{(32,32,5)}$ 
    & $\textbf{0.0456}_{(39,39,4)}$ & $\textbf{0.0392}_{(44,44,4)}$ 
    & $\textbf{0.0341}_{(49,49,4)}$ \\
    \hline
    \end{tabular}}%
  \label{table:error_Urban}%
\end{table*}%

\begin{table*}[!t]
\renewcommand\arraystretch{1}
\setlength{\tabcolsep}{5pt}%
\renewcommand\arraystretch{1.1}%
  \centering
  \caption{The fitting errors of using CPD/Tucker/LL1/LMN models to approximate the Pavia University data of size $200\times 200\times 203$ (The numbers in the lower left subscript indicate the corresponding rank settings).}
  \resizebox{\linewidth}{!}{
  \begin{tabular}{l|l|l|l|l|l|l|l}\hline

    \hline
    num. of paras. & \multicolumn{1}{c|}{$\sim$20K} 
    & \multicolumn{1}{c|}{$\sim$30K} & \multicolumn{1}{c|}{$\sim$40K} & \multicolumn{1}{c|}{$\sim$50K} & \multicolumn{1}{c|}{$\sim$60K}  & \multicolumn{1}{c|}{$\sim$70K} & \multicolumn{1}{c}{$\sim$80K}\\ \hline
    $\textrm{CPD}_{(F)}$
    & $0.1707_{(39)}$ & $0.1320_{(59)}$ & $0.1059_{(79)}$ 
    & $0.0877_{(99)}$ & $0.0746_{(119)}$ & $0.0646_{(139)}$ & $0.0582_{(155)}$\\
    $\textrm{Tucker}_{(L,M,N)}$ 
    & $0.1640_{(39,37,3)}$ & $0.1523_{(42,42,7)}$ & $0.1165_{(62,60,4)}$ 
    & $0.0964_{(78,78,3)}$ & $0.0875_{(80,98,3)}$ & 0.$0799_{(91,91,4)}$ & $0.0728_{(98,101,4)}$\\
    $\textrm{LL1}_{(L)}$    
    & $0.1909_{(12)}$ & $0.1546_{(18)}$ & $0.1291_{(24)}$ 
    & $0.1107_{(30)}$ & $0.0949_{(37)}$ & $0.0853_{(43)}$ & $0.0777_{(49)}$\\
    $\textrm{LMN}_{(L,M,N)}$  
    & $\textbf{0.1636}_{(10,4,4)}$ & $\textbf{0.1199}_{(16,3,4)}$ 
    & $\textbf{0.0999}_{(20,3,4)}$ & $\textbf{0.0819}_{(25,3,4)}$ 
    & $\textbf{0.0710}_{(29,3,4)}$ & $\textbf{0.0609}_{(34,3,4)}$ 
    & $\textbf{0.0548}_{(38,3,4)}$ \\
    \hline
    \end{tabular}}%
  \label{table:error_Pavia}%
\end{table*}%

\subsubsection{Using LMN Model to Incorporate EV}
Under the LMM, the spectral pixel of the SRI at location $(i,j)$ is expressed as follows 
\begin{align}\label{eq:lmm_pixel}
\bm y_{\rm S}^{(i,j)}=\sum_{r=1}^R \bm S_r(i,j)\bm c_r,~\forall (i,j).
\end{align}
However, when EV exists, the spectral signature $\bm c_r$ changes across $(i,j)$. Denote $\bm c_r^{(i,j)}$ as the spectral signature variation of material $r$ at the pixel location $(i,j)$. Then, the following pixel model is more plausible:
\begin{align}\label{eq:ev_lmm}
    \bm y_{\rm S}^{(i,j)} = \sum_{r=1}^R \bm S_r(i,j)\bm c_r^{(i,j)},~\forall (i,j).
\end{align}
Following \eqref{eq:ev_lmm}, the SRI can be expressed as
\begin{align}\label{eq:ev_T}
    \underline{\bm{Y}}_{\rm S} = \sum_{r=1}^R \underline{\bm{T}}_r,\quad
    \underline{\bm{T}}_r(i,j,:)  =  \bm S_r(i,j)\bm c_r^{(i,j)}.
\end{align}
Note that $\underline{\bm{T}}_r(i,j,:)$ is often called a mode-3 fiber of a third-order tensor $\underline{\bm{T}}_r$, which represents the spectrum of material $r$ in pixel $(i,j)$ in our context.

In \eqref{eq:ev_T}, the EV effect exists across all pixels, i.e., the $\bm c_r^{(i,j)}$'s for all $(i,j)$ are different. However, as endmember variation is graceful over space, the $\bm c_r^{(i,j)}$'s in the entire image are still non-trivially correlated. Therefore, it is reasonable to model $\underline{\bm{T}}_r$ using a low multi-linear rank Tucker model---this is the key of our approach (also see Fig.~\ref{fig:LMN_Tr} for illustration of $\underline{\bm T}_r$).
Hence, each $\underline{\bm{T}}_r$ in \eqref{eq:ev_T} can be written as
\begin{align}\label{eq:tTr}
        \underline{\bm{T}}_r = \underline{\bm{D}}_r \times_1 \bm{A}_r \times_2 \bm{B}_r \times_3 \bm{C}_r,~\forall r,
\end{align}
where $\underline{\bm{D}}_r\in\mathbb{R}^{L_r\times M_r\times N_r}$, $\bm{A}_r\in \mathbb{R}^{I_{\rm M}\times L_r}$, $\bm{B}_r\in \mathbb{R}^{J_{\rm M}\times M_r}$, and $\bm{C}_r\in \mathbb{R}^{K_{\rm H}\times N_r}$---with reasonable $L_r,M_r$ and $N_r$.
Under \eqref{eq:tTr}, the SRI can be expressed as
\begin{align}\label{eq:BTD}
    \underline{\bm{Y}}_{\rm S} = \sum_{r=1}^R \underline{\bm{D}}_r \times_1 \bm{A}_r \times_2 \bm{B}_r \times_3 \bm{C}_r,~\forall r.
\end{align}
This way, the SRI is modeled as a block-term decomposition with multi-linear rank-$(L_r,M_r,N_r)$ terms (LMN) \cite{Lathauwer2008BTD2}; see the illustration in Fig.~\ref{fig:LMN}. To further support this hypothesis, we extract $\underline{\bm T}_r$ from the Jasper Ridge and Washington DC datasets and plot the singular values mode-1, 2, and 3 unfoldings (see Figs.~\ref{fig:Ridge_LR}-\ref{fig:WDC_LR} in Appendix \ref{app:low-rankness}). One can see that all matrix unfoldings of $\underline{\bm T}_r$ are low rank matrices. This attest to the effectiveness of modeling $\underline{\bm T}_r$ using a low multilinear rank Tucker tensor. Hence, by employing the LMN model to represent the spectral images naturally accounts for endmember variability.

Note that unlike CPD or Tucker whose latent factors do not have physical interpretations, the terms in \eqref{eq:ev_T} still exhibits strong physical connections. That is, each $\tT_r$ is contributed by a specific material. 
More over, the tensor rank $F$ in CPD does not have physical meaning, yet the $R$ in \eqref{eq:ev_lmm} still corresponds to the number of materials contained in the image.
Hence, the spatial and spectral characteristics of a single material can still be used to serve as structural regularization/constraints---we will leverage this for algorithm design.

\subsubsection{Numerical Validation for LMN-based Modeling}
To validate our postulate, we fit real-world hyperspectral images using various tensor models.
We use the normalized reconstruction error (\texttt{NRE}): 
$$\texttt{NRE}=\|\tY(\bm \theta)-\underline{\bm{Y}}\|_{ F}/\|\underline{\bm{Y}}\|_{ F}$$ (where $\tY(\bm \theta)$ is the approximated spectral image by the used tensor model, and $\underline{\bm{Y}}$ is the spectral image being fitted to) under different number of tensor model parameters.
We use real hyperspectral images, namely, the Urban data \cite{Pan2023Separable} and the Pavia University data \cite{Kunkel1988ROSIS}, as $\tY$. Note that these images are well-known to exhibit endmember variability \cite{Zhou2016EV,Borsoi2021Variability}.

Tables~\ref{table:error_Urban}-\ref{table:error_Pavia} show the results when $\tY(\bm \theta)$ is modeled using CPD, Tucker, LL1, and LMN, respectively, where all models use approximately the same amount of parameters.
For two datasets, we set the number of materials $R=4$ and choose ranks of four tensor models so that the corresponding number of parameters is approximately equal to a given value. 
One can see that when using similar amounts of parameters, the LMN model outputs the smallest fitting errors relative to the other models. In particular, the LMN model often obtains an NRE that is 50\% lower than that of the LL1 model. This attests to the effectiveness of relaxing the model of $\tT_r$ from an LL1 structure (i.e., a matrix outer-product with a vector) to an LMN Tucker structure in the context of HSI approximation.
In addition, note that both the CPD and Tucker models can also attain competitive \texttt{NRE}s. As mentioned, both CPD and Tucker are universal tensor approximators, with sufficiently high rank. Nonetheless, the lose interpretability of their latent components.

\subsection{CTD Formulation under LMN Model}
Under the LMN model in \eqref{eq:ev_T} and \eqref{eq:ev_lmm}, and
according to the degradation model in \eqref{eq:modeproduct3} and \eqref{eq:modeproduct12}, we have
\begin{align}\label{eq:HSI_degradation}
  \underline{\bm{Y}}_{\rH} =\sum_{r=1}^R\underline{\bm{D}}_r\times_1 (\bm{P}_1\bm{A}_r)\times_2 (\bm{P}_2\bm{B}_r)\times_3  \bm{C}_r,
\end{align}
and
\begin{align}\label{eq:MSI_degradation}
\underline{\bm{Y}}_{\rM}=\sum_{r=1}^R\underline{\bm{D}}_r\times_1 \bm{A}_r\times_2 \bm{B}_r\times_3  (\bm{P}_{\rm M}\bm{C}_r).
\end{align}
Using the observation models in \eqref{eq:HSI_degradation} and \eqref{eq:MSI_degradation}, we rewrite the CTD formulation in \eqref{eq:coupledctd} as follows:
\begin{align}\label{Non_blind_BTD_model}
    &\min_{\bm \theta}~\frac{1}{2}\left\| \tY_{\rm H} - \sum_{r=1}^R\underline{\bm{D}}_r\times_1 (\bm{P}_1\bm{A}_r)\times_2 (\bm{P}_2\bm{B}_r)\times_3  \bm{C}_r \right\|^2_{F} \nonumber \\
    & +  \frac{1}{2}\left\| \tY_{\rm M} - \sum_{r=1}^R\underline{\bm{D}}_r\times_1 \bm{A}_r\times_2 \bm{B}_r\times_3  (\bm{P}_{\rm M}\bm{C}_r)\right\|^2_{F},
\end{align}
where $\bm \theta =(\{\underline{\bm{D}}_r, \bm{A}_{r},\bm{B}_{r}, \bm{C}_r \}_{r=1}^{R} )$.

The works in \cite{Kanatsoulis2018HSR,Prevost2020HSR,Ding2021HSR} also considered cases where the spatial degradation operator $\bm P_1$ and $\bm P_2$ are unknown. This setup is meaningful, as the spatial operators are nontrivial to estimate (yet the spectral operator is relatively easy to acquire, e.g., by inspecting sensor specifications \cite{Kanatsoulis2018HSR,Ding2021HSR}). Our formulation can also be amended to cover this case:
\begin{align}\label{blind_BTD_model}
    &\min_{\bm \theta'}~\frac{1}{2}\left\| \tY_{\rm H} - \sum_{r=1}^R\underline{\bm{D}}_r\times_1 \widetilde{\bm{A}}_r\times_2 \widetilde{\bm{B}}_r\times_3  \bm{C}_r \right\|^2_{F} \nonumber \\
    & +  \frac{1}{2}\left\| \tY_{\rm M} - \sum_{r=1}^R\underline{\bm{D}}_r\times_1 \bm{A}_r\times_2 \bm{B}_r\times_3  (\bm{P}_{\rm M}\bm{C}_r)\right\|^2_{F},
\end{align}
where $\bm \theta' =(\{\underline{\bm{D}}_r, \widetilde{\bm{A}}_{r},\widetilde{\bm{B}}_{r}, \bm{A}_{r},\bm{B}_{r},\bm{C}_r \}_{r=1}^R)$, and  $\widetilde{\bm{A}}_{r}\in \mathbb{R}^{I_{H} \times L_{r}}$ and $\widetilde{\bm{B}}_{r}\in \mathbb{R}^{J_{H} \times M_{r}}$ are used to model $\bm P_1\bm A_r$ and $\bm P_2\bm B_r$, respectively.

\subsection{Recoverability Guarantees}
One of the most appealing aspects of the CTD framework proposed by \cite{Kanatsoulis2018HSR} is that it provided a systematic way of establishing recoverability of $\tY_{\rm S}$. This allows a suite of CTD framework to have recoverability guarantees of $\tY_{\rS}$; see the CPD-based \cite{Kanatsoulis2018HSR}, the LL1-based \cite{Ding2021HSR}, and the Tucker-based \cite{Prevost2020HSR} versions, respectively.

In this work, we show that the LMN model is not an exception---i.e., using the LMN model under the CTD model still retains similar recoverability guarantees.
The proof framework resembles those in \cite{Kanatsoulis2018HSR,Ding2021HSR}.
Nonetheless, as we will leverage the essential uniqueness of the LMN decomposition to approach the recoverability problem, careful custom analyses are provided to carry out the proof.

\begin{definition}[Essential Uniqueness]
The tensor $\tY=\sum_{r=1}^R \tD_r\times_1 \A_r\times_2 \B_r \times_3 \C_r$ has an essnetially unique LMN decomposition if the following is met: for any representation $\underline{\bm{Y}} = \sum_{r=1}^R\underline{\bar{\bm{D}}}_r\times_1 \bar{\bm{A}}_r\times_2 \bar{\bm{B}}_r\times_3 \bar{\bm{C}}_r$, we have 
\begin{align}
&\bar{\bm{A}}_{\pi(r)} = \bm{A}_r\bm \Theta_{a,r},~
     \bar{\bm{B}}_{\pi(r)} = \bm{B}_r \bm \Theta_{b,r},~
     \bar{\bm{C}}_{\pi(r)} = \bm{C}_r \bm \Theta_{c,r},\\
&    \underline{\bar{\bm{D}}}_{\pi(r)}\times_1 \bm \Theta_{a,r}\times_2 \bm \Theta_{b,r}\times_3 \bm \Theta_{c,r}=\underline{\bm{D}}_r, ~ \forall r,    
\end{align}
     where $\pi$ is a permutation of $\{1,\ldots,R\}$, and $\bm \Theta_{a,r}$, $\bm \Theta_{b,r}$, and $\bm \Theta_{c,r}$ are nonsingular matrices.  
\end{definition}

We note that existing LMN decomposition conditions cover a range of cases in terms of the relationship between $L,M,N$ and $I,J,K$. However, the these conditions are not realistic in the context of HSR. For example, the seminal paper \cite{Lathauwer2008BTD2} proved that, under reasonable conditions, LMN uniqueness holds if
\begin{align}
    &L=M, I\geq LR, J\geq MR, N\geq 3, \bm{C}_r~ \textrm{is full column rank}, 1\leq r\leq R;\label{eq:L=M}\\
\text{or}~&N> L+M -2, \min \left(\left\lfloor\frac{I}{L} \right\rfloor, R \right)+
\min \left(\left\lfloor\frac{J}{M} \right\rfloor, R \right)+
\min \left(\left\lfloor\frac{K}{N} \right\rfloor, R \right)\geq 2R+2;\label{eq:Nlarger}
\end{align}
see \cite[Theorem 5.1, Theorem 5.5]{Lathauwer2008BTD2}, respectively.
We note that the $L=M$ part in Eq.~\eqref{eq:L=M} is too special to meet by spectral image tensors. 
In addition, the first inequality in \eqref{eq:Nlarger} essentially means that the spectral mode of the spectral images have less correlations relative to the spatial dimensions, which is also hard to justify in the context of HSR.

To overcome these unrealistic conditions for spectral images, we first extend the results in \cite{Lathauwer2008BTD2} and
show a new identifiability result of LMN decomposition:

\begin{theorem}[Identifiability of LMN]
\label{the:identifiability_LMN}
Let $(\{\bm{A}_{r}\in \mathbb{R}^{I \times L_{r}}, \bm{B}_{r}\in \mathbb{R}^{J \times M_{r}}, \bm{C}_r\in \mathbb{R}^{K \times N_r}, \underline{\bm{D}}_r\in \mathbb{R}^{L_{r}\times M_{r} \times N_r}\}_{r=1}^{R})$ be the latent factors of the LMN tensor $\underline{\bm{Y}} = \sum_{r=1}^R\underline{\bm{D}}_r\times_1 \bm{A}_r\times_2 \bm{B}_r\times_3 \bm{C}_r$ where $L_r=L$, $M_r=M$, and $N_r=N$. Assume that $\bm{A}_{r}$, $\bm{B}_{r}$, $\bm{C}_r$, and $\underline{\bm{D}}_r$ are drawn from any absolutely continuous distributions. Then, the LMN decomposition of $\underline{\bm Y}$ is essentially unique almost surely, if 
\begin{align}
 I\geq LR, ~~J\geq MR, ~~LM\geq N \geq \max\left\{\left\lceil \frac{L}{M} \right\rceil + \left\lceil \frac{M}{L} \right\rceil, 3\right\},~ 1\leq r\leq R. \nonumber
\end{align}
\end{theorem}

\begin{proof}
 The proof is a combination of the new development in \cite{Domanov2024BTD} and the proof pipeline in  \cite[Theorem 5.1]{Lathauwer2008BTD2}. The details are in Appendix \ref{app:identifiability_LMN}.
\end{proof}

Note that under the above new theorem, the hard-to-justify conditions in the context of spectral imaging, i.e., $L=M$ and $N> L+M -2$, no longer exist. The conditions in Theorem~\ref{the:identifiability_LMN} only requires that $R$ and $L,M,N$ are sufficiently small. The physical meaning is clear: when the number of endmembers in the spectral image is not too large and the correlations of the underlying component associated with each endmember (i.e.,
$\tT_r$) across three modes are sufficiently high, the LMN decomposition is essentially unique. Based on this uniqueness theorem, we show the following:

\begin{theorem} \label{the:recoverability_nonblind}
Assume that $\underline{\bm{Y}}_{\rS}$ follows the LMN model in \eqref{eq:ev_lmm}-\eqref{eq:ev_T} with $L_r=L, M_r=M, N_r=N$ for $r=1,\ldots,R$, and that $\underline{\bm{Y}}_{\rH}$ and $\underline{\bm{Y}}_{\rM}$ follow the degradation models in \eqref{eq:HSI_degradation} and \eqref{eq:MSI_degradation}, respectively. Also assume that each of $\underline{\bm{D}}_r\in \mathbb{R}^{L\times M\times N}$, $\bm{A}_{r}\in\mathbb{R}^{I_{\rm M}\times L}$, $\bm{B}_{r}\in\mathbb{R}^{J_{\rm M}\times M}$, and $\bm{C}_r\in\mathbb{R}^{K_{\rm H}\times N}$ is drawn from any absolutely continuous distribution, and that $\bm{P}_{1}$, $\bm{P}_{2}$, and $\bm{P}_{M}$ have full row rank. Let $\{\underline{\bm{D}}_r^{\star},\bm{A}_{r}^{\star},\bm{B}_{r}^{\star},\bm{C}_r^{\star}\}_{r=1}^{R}$ represent any feasible solution of Problem \eqref{Non_blind_BTD_model}. Then, the ground-truth $\underline{\bm{Y}}_{\rS}$ is uniquely recovered almost surely by
\[
\underline{\bm{Y}}_{\rS} = \sum_{r=1}^R\underline{\bm{D}}_r^{\star}\times_1 \bm{A}_r^{\star}\times_2 \bm{B}_r^{\star}\times_3  \bm{C}_r^{\star},
\]
if the following conditions are met: 
\begin{subequations}
\begin{align}
 &I_{\rm H}J_{\rm H}\geq LMR, \\
 &I_{\rm M}\geq LR, ~J_{\rm M}\geq MR, \\
 &LM\geq N \geq \max\{\lceil L/M \rceil + \lceil M/L \rceil, 3\}
\end{align}
\end{subequations}
for all $r\in [R]$. 
\end{theorem}
The theorem asserts that, as long as the spatial region of interest does not have an excessive amount of endmembers and each endmember's variations exhibit sufficient correlations across the region, recoverability of the SRI is guaranteed. Note that the result is naturally robust to the cases where EV is present---as discussed, even in the presence of EV, $\tT_r(i,j,:)$ for different $(i,j)$'s (i.e., endmember $r$ appearing at different locations) are still correlated as shown in Fig.~\ref{fig:LMN_Tr}. The result also does not rely on the somewhat stringent conditions in existing EV-aware CTD frameworks (e.g., that EV only occurs across HSI and MSI \cite{Prevost2022Coupled,Borsoi2021Coupled}).

Similarly, we show recoverability of the semi-blind recovery formulation \eqref{blind_BTD_model} where the spatial degradation operators are unknown:
\begin{theorem}\label{the:recoverability_blind}
Under the same assumptions as in Theorem~\ref{the:recoverability_nonblind}, assume that 
$\{\widetilde{\A}_{r}^{\star},\widetilde{\B}_{r}^{\star},\underline{\bm{D}}_r^\star,$ $\bm{A}_{r}^{\star},$ $\bm{B}_{r}^{\star},\bm{C}_{r}^{\star}\}_{r=1}^{R}$ is any feasible solution of Problem \eqref{blind_BTD_model}. Then, 
if $K_M \geq 2N$, $I_{\rm H}\geq LR$, $J_{\rm H}\geq MR$, and $LM \geq N\geq \max\{\lceil L/M \rceil + \lceil M/L \rceil, 3\}$, $1\leq r\leq R$,
the ground-truth $\underline{\bm{Y}}_{S}$ is uniquely recovered by
$
\underline{\bm{Y}}_{\rS} = \sum_{r=1}^R\underline{\bm{D}}_r^\star\times_1 \bm{A}_r^\star\times_2 \bm{B}_r^\star\times_3  \bm{C}_r^\star
$ with probability one.
\end{theorem}

\begin{remark}
Theorems~\ref{the:recoverability_nonblind}--\ref{the:recoverability_blind} confirm that adopting the LMN model in the CTD framework retains similar recoverability guarantees previously established for other tensor models. To some extent, this outcome is expected, since the LMN model subsumes CPD, LL1, and Tucker as special cases. Nevertheless, its applicability under EV and the resulting practical implications for HSR are significant. Moreover, while our proof strategy shares similarities with those in \cite{Kanatsoulis2018HSR,Ding2021HSR} for CPD and LL1, the technical underpinning is far from straightforward. Two aspects deserve particular emphasis: (i) unlike \cite{Kanatsoulis2018HSR,Ding2021HSR}, which directly invoke existing uniqueness conditions, we derive new conditions for LMN uniqueness tailored to the HSR setting (see Theorem~\ref{the:identifiability_LMN}); and (ii) because the definition of essential uniqueness for LMN differs from those of CPD and LL1, the proof requires careful handling of the associated mathematical nuances, making the overall recoverability analysis nontrivial to complete.
\end{remark}

\section{Regularization Design and Algorithm}
In this section, we design practical algorithms to tackle \eqref{Non_blind_BTD_model} and \eqref{blind_BTD_model}. In particular, we will design regularization terms that incorporate physical meaning of $\tT_r$ for performance enhancement.

\subsection{Regularization Design for Problems~\eqref{Non_blind_BTD_model} and \eqref{blind_BTD_model}}
One of the characteristics of the endmembers is that they spread smoothly over space \cite{Luo2025NeurTV,Iordache2012TV,Ma2025Multispectral}.
In the spectral domain, the endmember's variations are also smooth.
This means that the mode-1, mode-2 and mode-3 ``fibers'' of $\tT_{r}$ (i.e., $\tT_{r}(:,j,k)$, $\tT_{r}(i,:,k)$, and $\tT_{r}(i,j,:)$) all exhibit a certain degree of smoothness; see \cite{Kolda2009Tensor} for illustration of fibers.
Note that the mode-1, mode-2 and mode-3 fibers of $\tT_{r}$ reside in ${\rm range}(\A_r)$, ${\rm range}(\B_r)$ and ${\rm range}(\C_r)$, respectively.
Therefore, to regularize the smoothness of these fibers, we propose to regularize the bases of these range spaces:
\begin{align}\label{eq:TV}
\phi_r(\bm{A}_r,\bm{B}_r,\bm{C}_r)
= \phi_{p,\varepsilon}(\bm{H}_1\bm{A}_r)+\phi_{p,\varepsilon}(\bm{H}_2\bm{B}_r)+\phi_{\rm Tik}(\bm{H}_3\bm{C}_r),
\end{align}
where $\phi_{p,\varepsilon}(\bm{X})=\sum_i\sum_j(x_{i,j}^2+\varepsilon)^{\frac{p}{2}}$ with $\varepsilon>0$ \cite{Ding2021HSR,Fu2015Joint}.
 We employ this nonconvex total variation regularization $\phi_{p,\varepsilon}(\bm{H}_1\bm{A}_r)+\phi_{p,\varepsilon}(\bm{H}_2\bm{B}_r)$ to promote spatial smoothness. Here, $\bm{H}_1 \in \mathbb{R}^{(I_{\rM}-1)\times I_{\rM}}$ with $\bm{H}_1(i,i)=1, \bm{H}_1(i,i+1)=-1,i=1,\ldots, I_{\rM}-1$, and $\bm{H}_2 \in \mathbb{R}^{(J_{\rM}-1)\times J_{\rM}}$ is defined in the same way. When one takes $0<p\leq 1$, the regularization promotes the small total variation; see \cite{Ding2021HSR}.
We use the Tikhonov regularization, i.e. $\phi_{\rm Tik}(\bm{H}_3\bm{C}_r)=\|\bm H_3\bm C_r\|_{F}^2$ in \eqref{eq:TV}, to represent the spectral smoothness, where the Tikhonov matrix is expressed as follows \cite{Boyd2004Convex,Song2019Online}:
\[
\bm{H}_3 =
\left[
\begin{array}{ccccccc}
1 & -2 & 1  & \cdots   & 0 & 0 & 0\\
0 & 1  & -2 & \cdots  & 0 & 0 & 0\\
\vdots & \vdots & \vdots  &  \vdots   & \vdots & \vdots & \vdots\\
0 & 0 & 0 &  \cdots & -2 & 1 & 0\\
0 & 0 & 0 & \cdots & 1 & -2 & 1\\
\end{array}
\right]
\in \mathbb{R}^{(K_{H}-2) \times K_{\rm H}}.
\]
This regularization promotes second order smoothness---reflecting the characteristics of the spectral signatures (see Fig.~\ref{fig:variability1}).
For the tensor $\underline{\bm{D}}_r$, we apply the regularization $\varphi_r(\tD_r)=\frac{1}{2}\|\tD_r\|^2_{F}$ to combat the scaling/counter-scaling effect in factorization models \cite{Fu2015Joint}.

\begin{figure*}[!t]
\scriptsize\setlength{\tabcolsep}{0.3pt}
\begin{center}
\begin{tabular}{ccccc}
\includegraphics[width=0.189\textwidth]{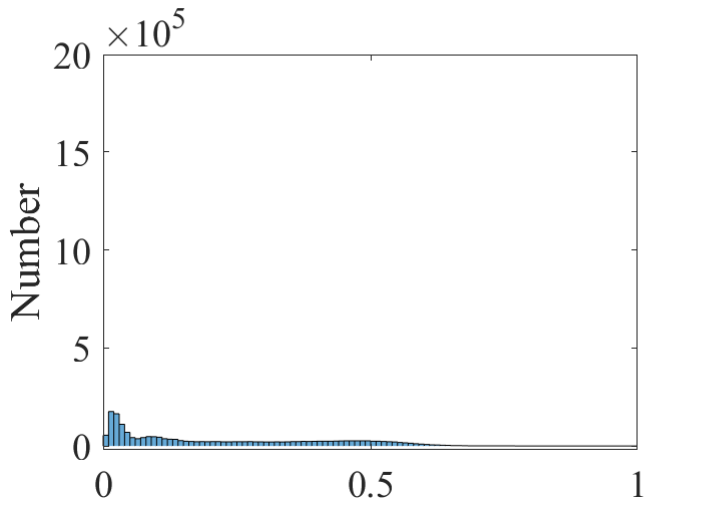}&
\includegraphics[width=0.189\textwidth]{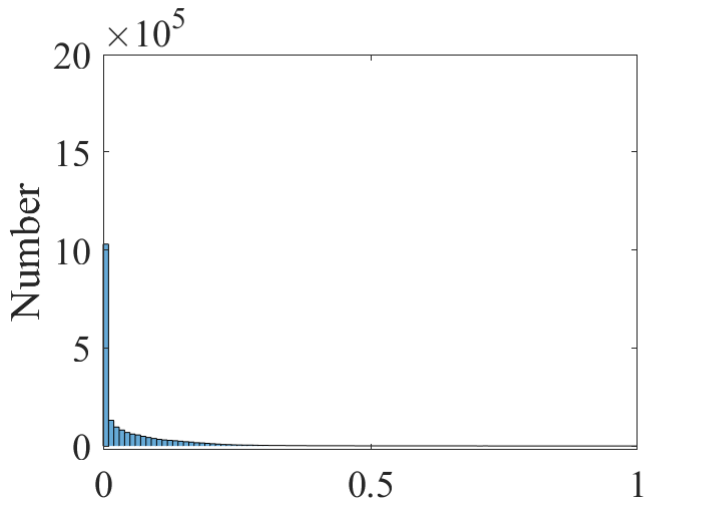}&
\includegraphics[width=0.189\textwidth]{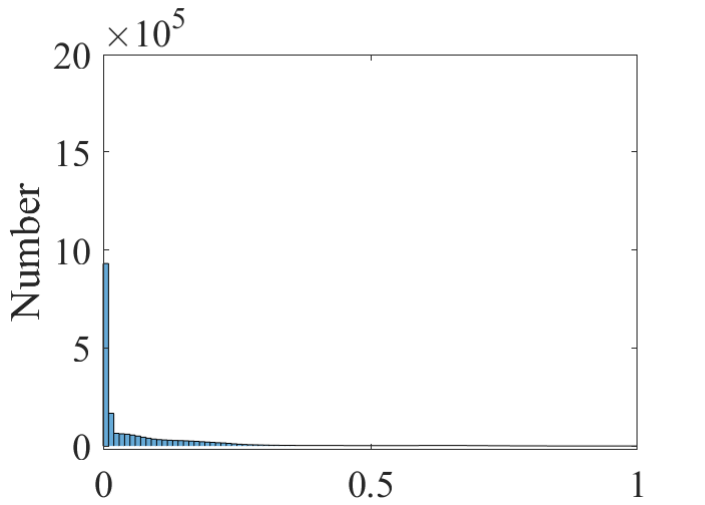}&
\includegraphics[width=0.189\textwidth]{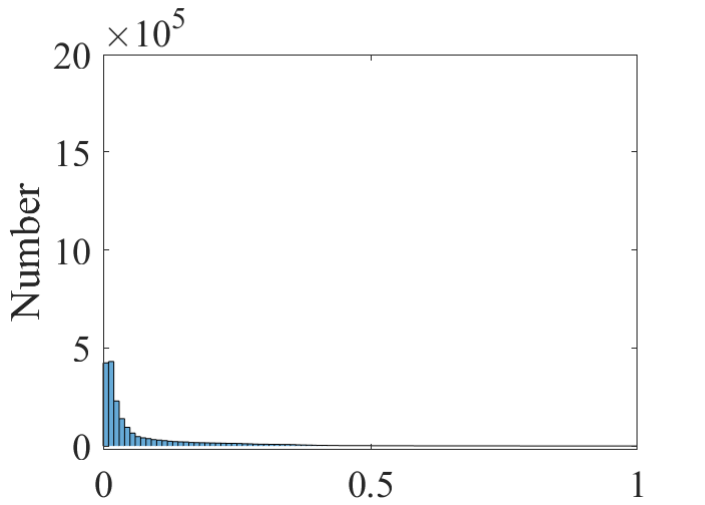}&
\includegraphics[width=0.189\textwidth]{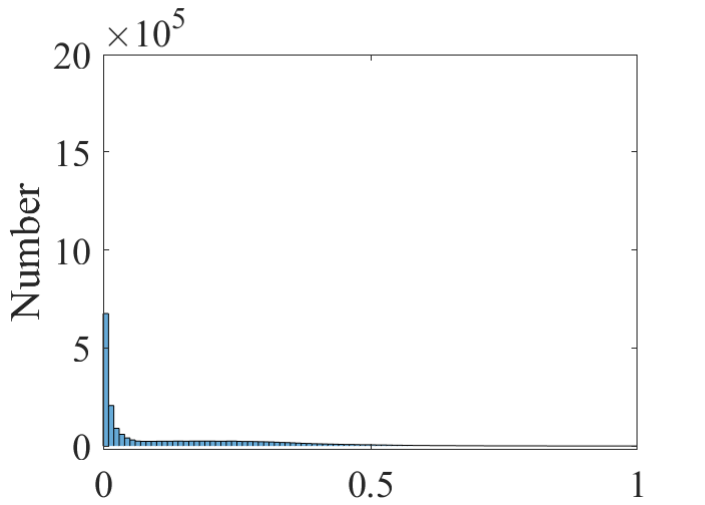}\\
\includegraphics[width=0.189\textwidth]{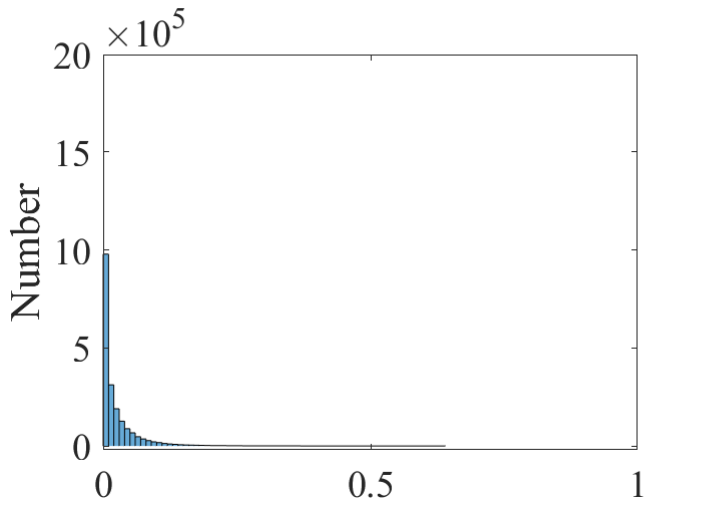}&
\includegraphics[width=0.189\textwidth]{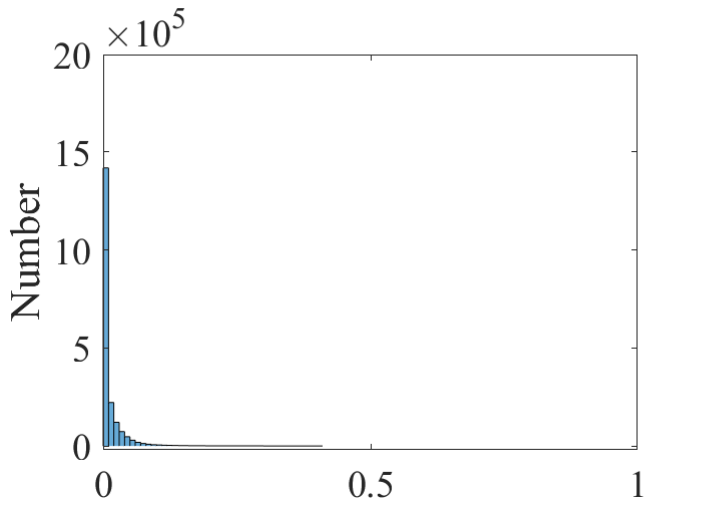}&
\includegraphics[width=0.189\textwidth]{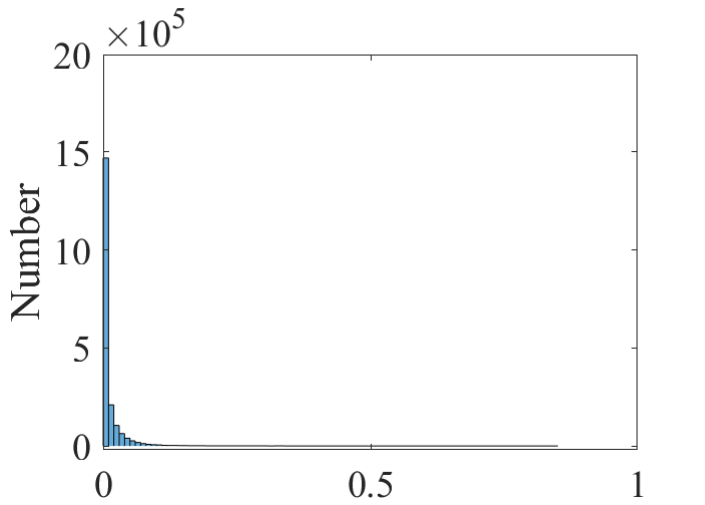}&
\includegraphics[width=0.189\textwidth]{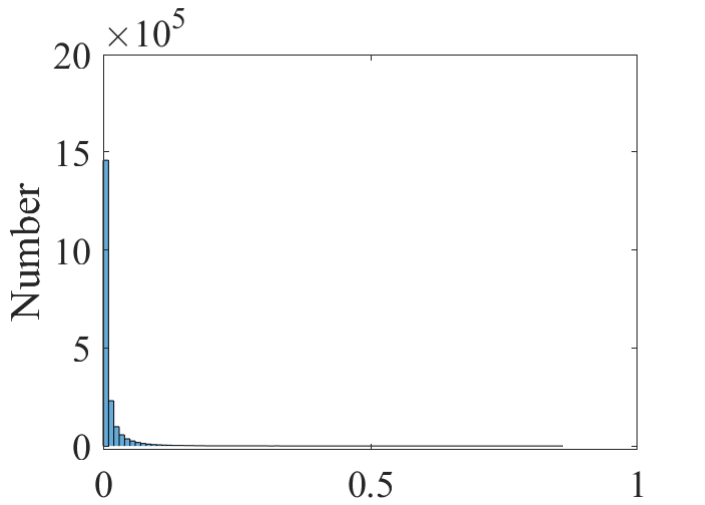}&
\includegraphics[width=0.189\textwidth]{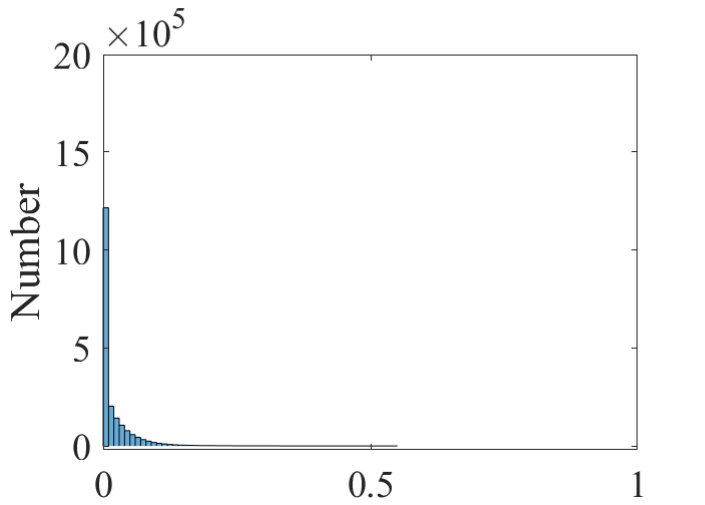}\\
\includegraphics[width=0.189\textwidth]{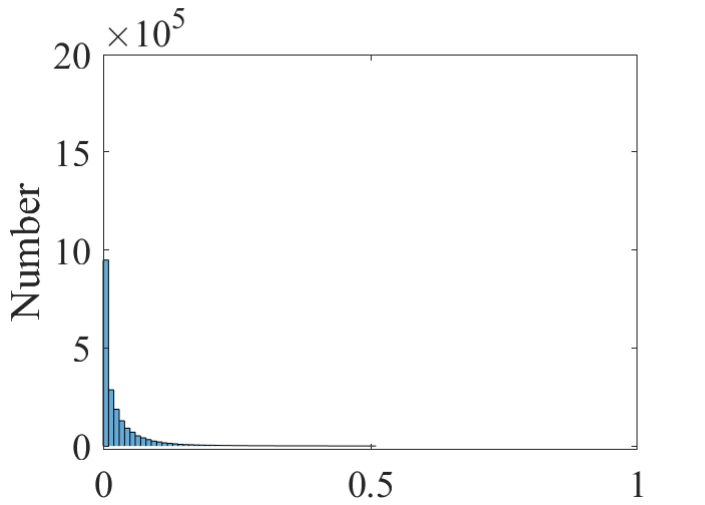}&
\includegraphics[width=0.189\textwidth]{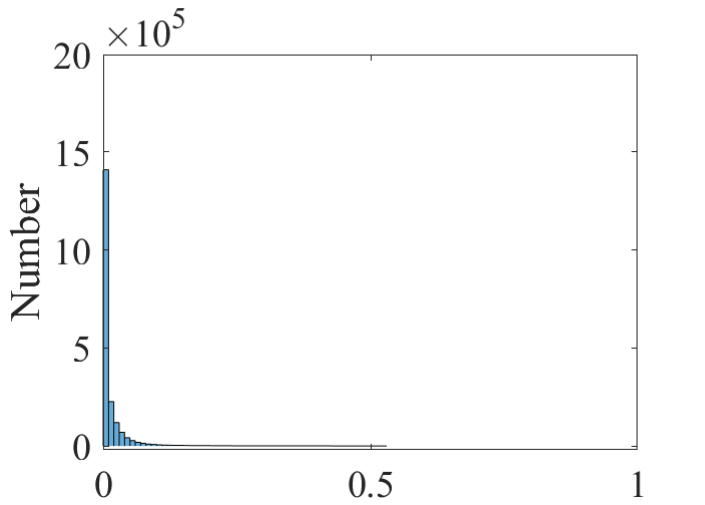}&
\includegraphics[width=0.189\textwidth]{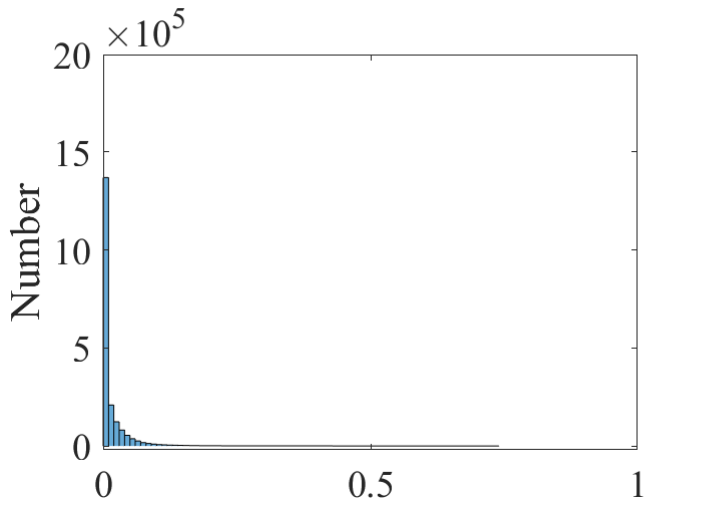}&
\includegraphics[width=0.189\textwidth]{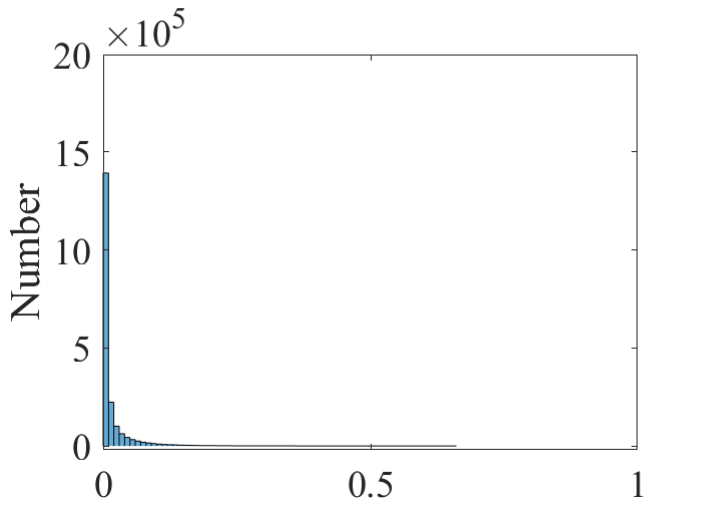}&
\includegraphics[width=0.189\textwidth]{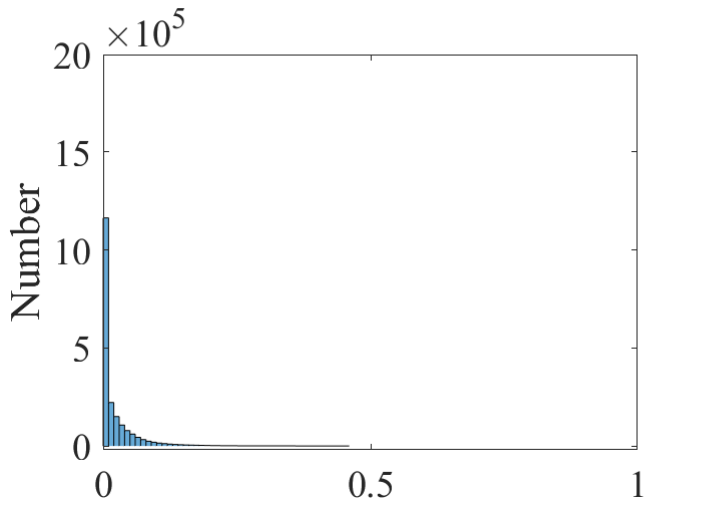}\\
\includegraphics[width=0.189\textwidth]{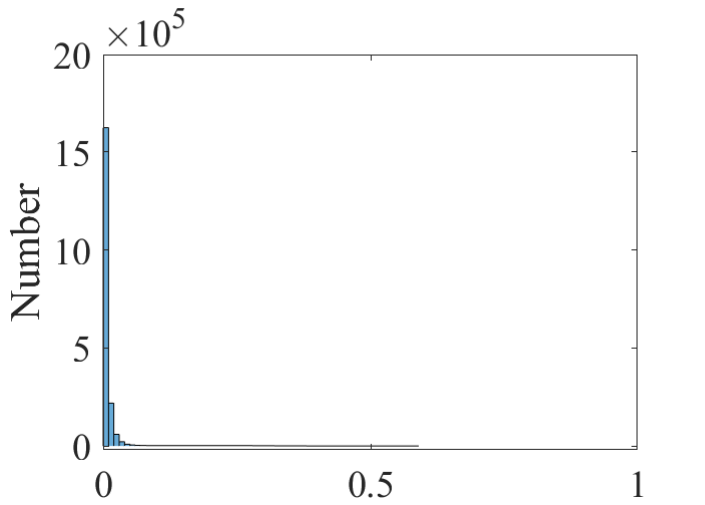}&
\includegraphics[width=0.189\textwidth]{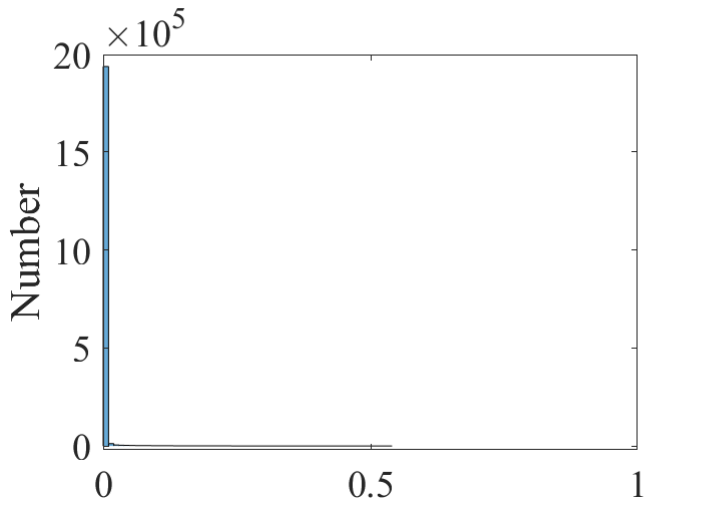}&
\includegraphics[width=0.189\textwidth]{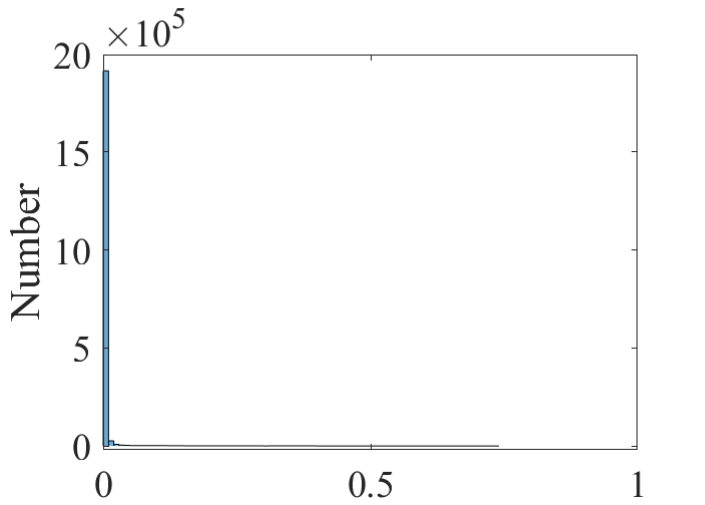}&
\includegraphics[width=0.189\textwidth]{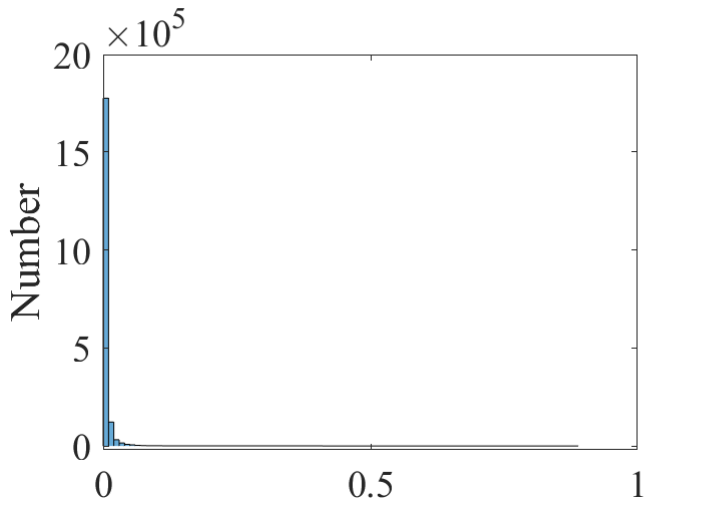}&
\includegraphics[width=0.189\textwidth]{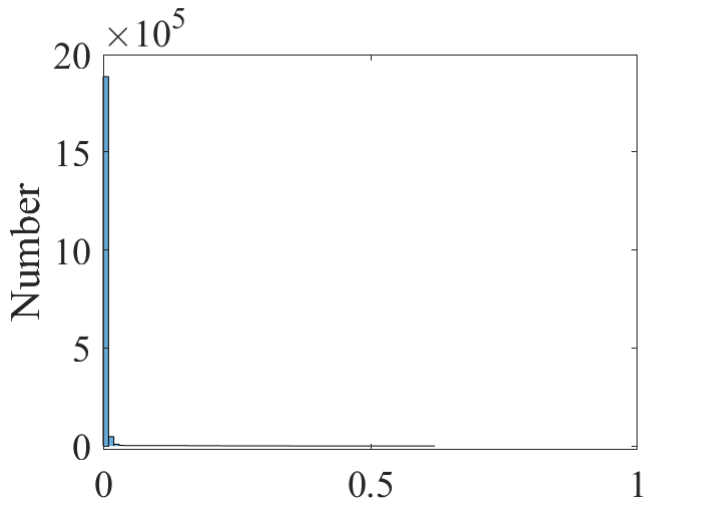}\\
(a)~~$\tY_{\rm S}$ &  (b)~~$\tT_1$ &  (c)~~$\tT_2$ & (d)~~$\tT_3$ &  (e)~~$\tT_4$\\
\end{tabular}
\caption{Comparison of sparsity (first row) and spatial smoothness (second row to fourth row: mode-1, mode-2, and mode-3, respectively) exhibited by original $\tY_{\rm S}$ and each $\tT_r$ for Jasper Ridge dataset with four materials. } 
  \label{fig:Ridge_smoothness}
  \vspace{-2mm}
\end{center}
\end{figure*}

Using the above, our optimization objective when spatial degradation is known amounts to the following:
\begin{align}\label{eq:model-non-blind}
    \minimize_{\bm \theta}~{\cal L} + \lambda \sum_{r=1}^R \phi_r(\A_r,\B_r,\C_r) + \eta \sum_{r=1}^R \varphi_r(\tD_r),
\end{align}
where $\bm \theta$ collects all the optimization variables, ${\cal L}$ represents the loss in \eqref{Non_blind_BTD_model}, and $\lambda$ and $\eta$ are two non-negative regularization parameters. We should mention that the formulation is continuously differentiable, which allows us to design efficient first-order algorithms for handling the problem.

Similarly, \eqref{blind_BTD_model} is recast as
\begin{align}\label{eq:model-blind}
    \minimize_{\bm \theta'}~\widetilde{{\cal L}} + \lambda \sum_{r=1}^R \phi_r(\A_r,\B_r,\C_r) + \eta \sum_{r=1}^R \varphi_r(\tD_r),
\end{align}
where $\widetilde{\cal L}$ represents the loss in \eqref{blind_BTD_model}.

\begin{figure*}[!ht]
\scriptsize\setlength{\tabcolsep}{0.3pt}
\begin{center}
\begin{tabular}{ccccc}
\includegraphics[width=0.189\textwidth]{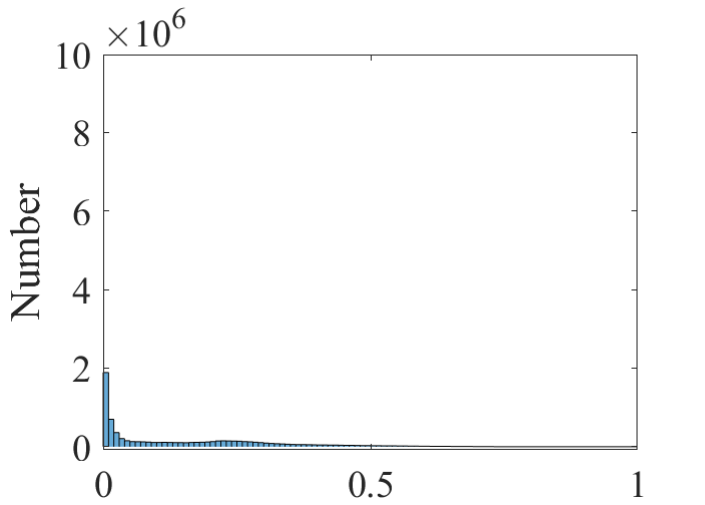}&
\includegraphics[width=0.189\textwidth]{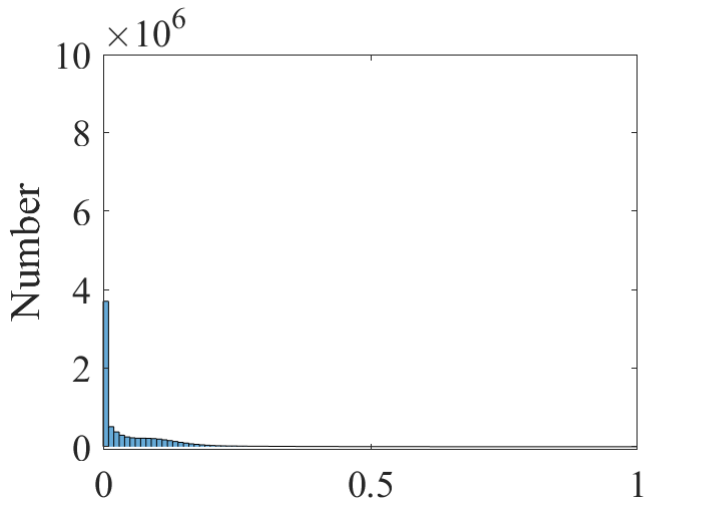}&
\includegraphics[width=0.189\textwidth]{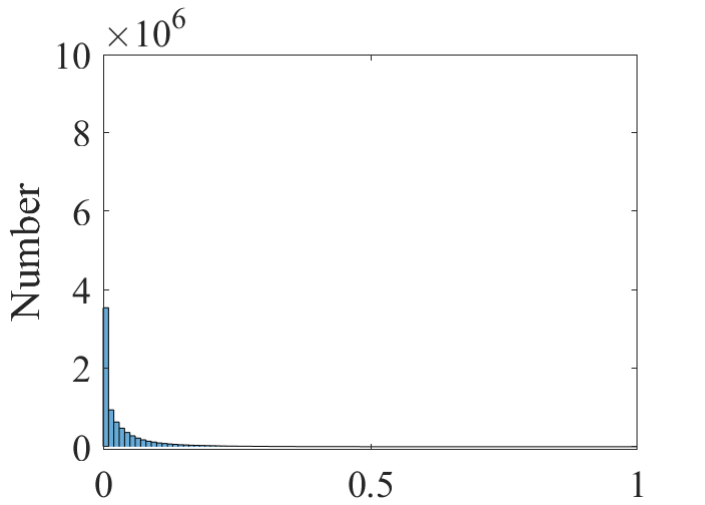}&
\includegraphics[width=0.189\textwidth]{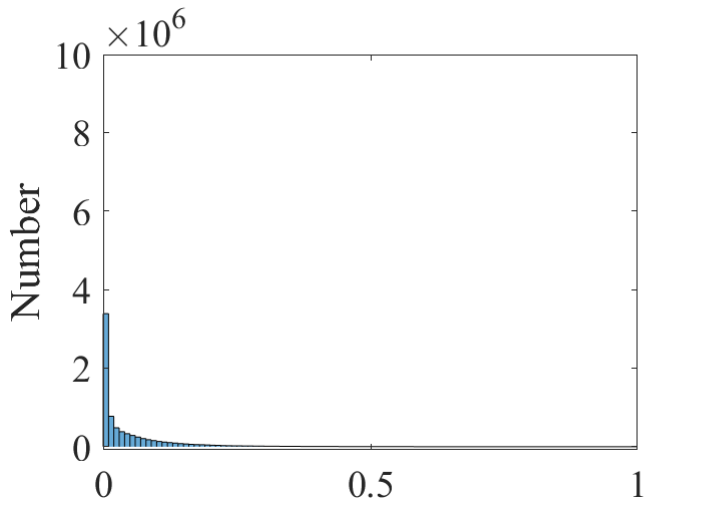}&
\includegraphics[width=0.189\textwidth]{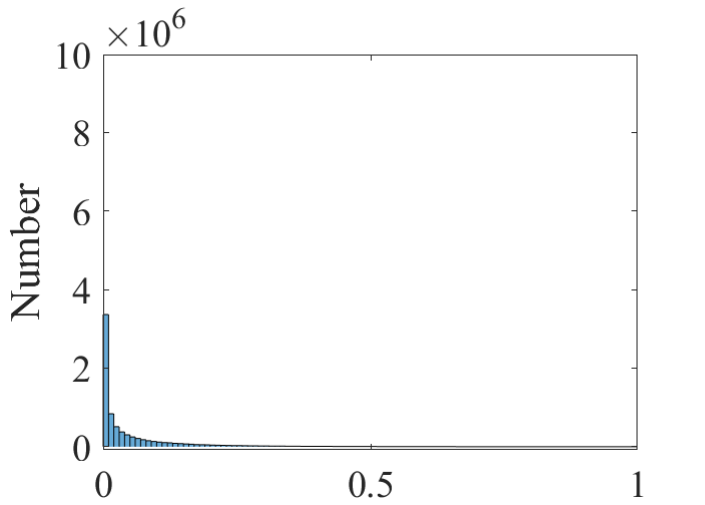}\\
\includegraphics[width=0.189\textwidth]{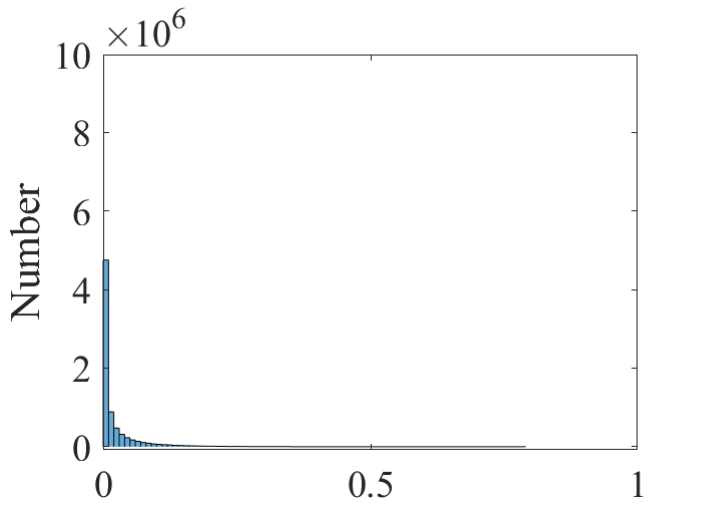}&
\includegraphics[width=0.189\textwidth]{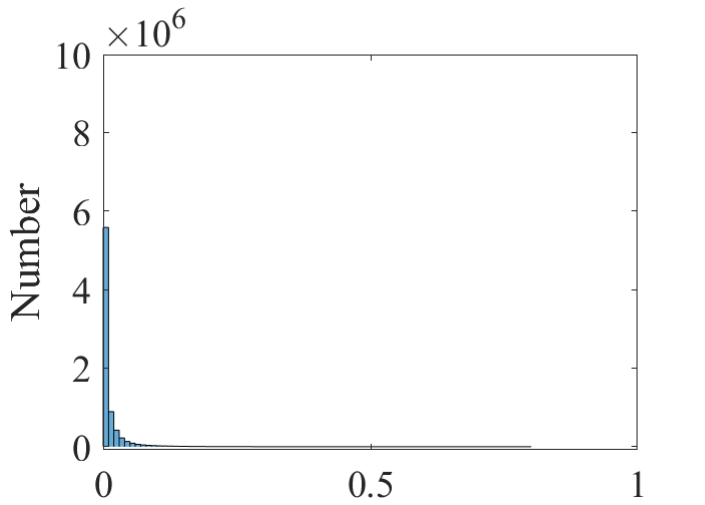}&
\includegraphics[width=0.189\textwidth]{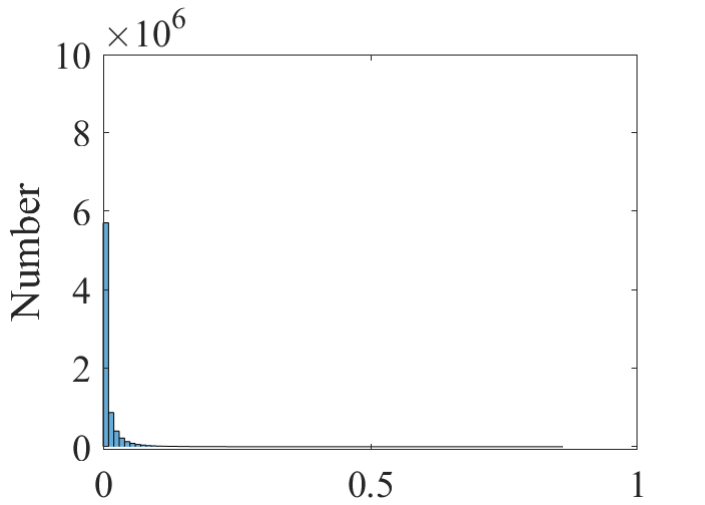}&
\includegraphics[width=0.189\textwidth]{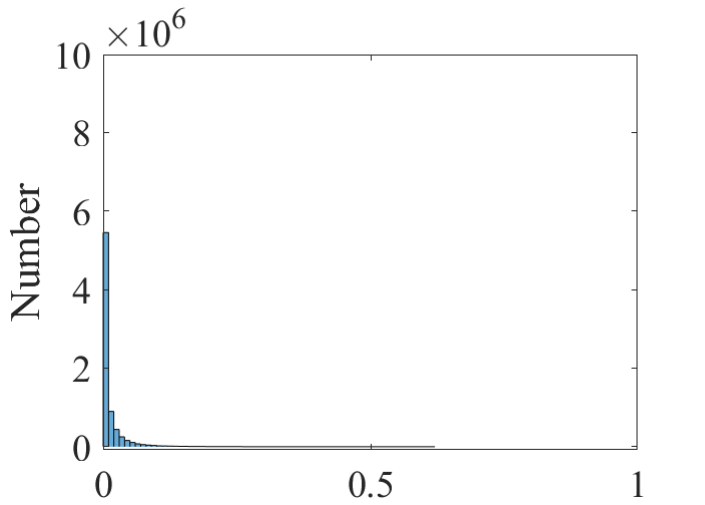}&
\includegraphics[width=0.189\textwidth]{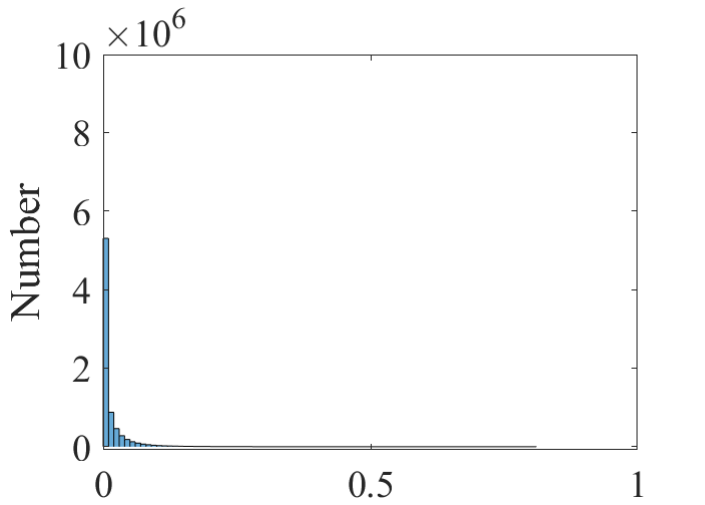}\\
\includegraphics[width=0.189\textwidth]{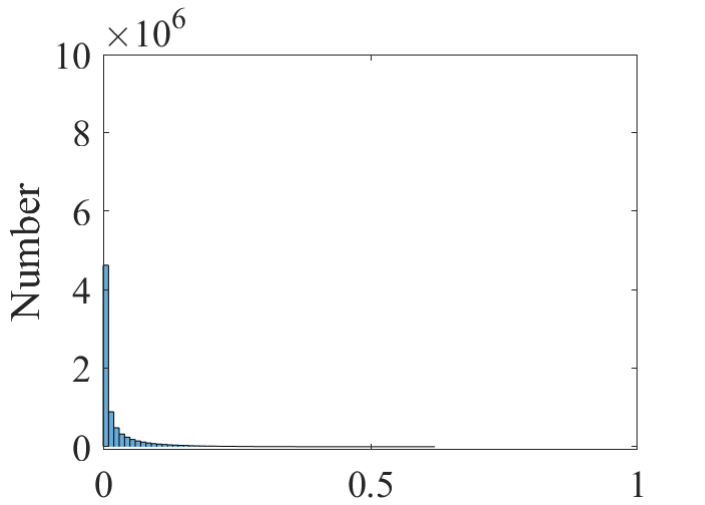}&
\includegraphics[width=0.189\textwidth]{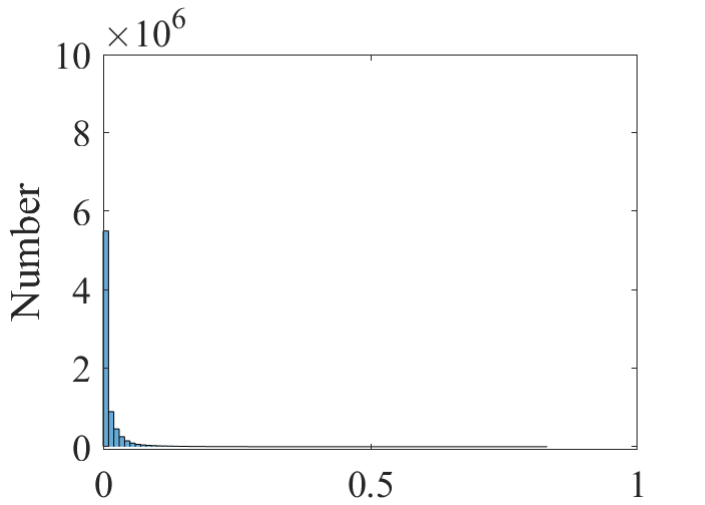}&
\includegraphics[width=0.189\textwidth]{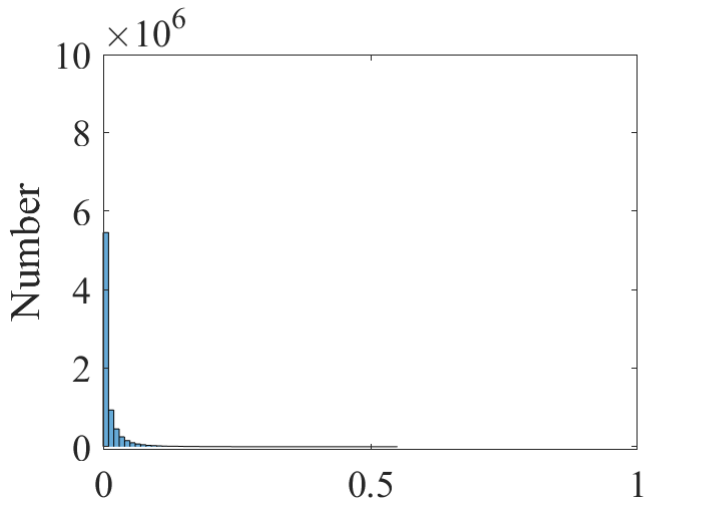}&
\includegraphics[width=0.189\textwidth]{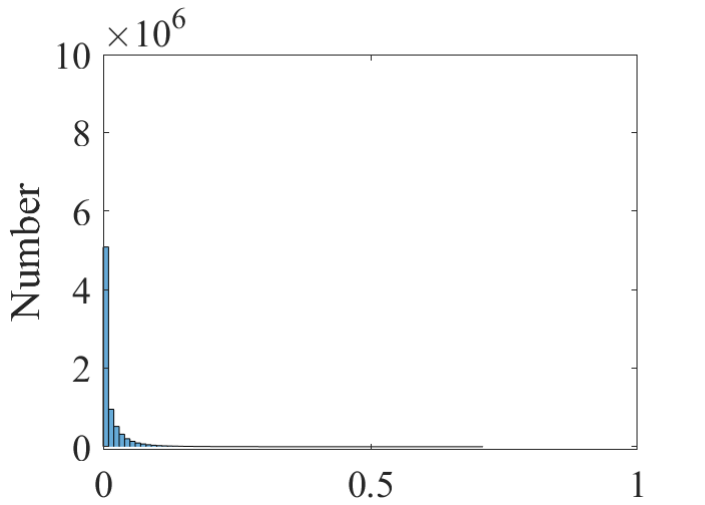}&
\includegraphics[width=0.189\textwidth]{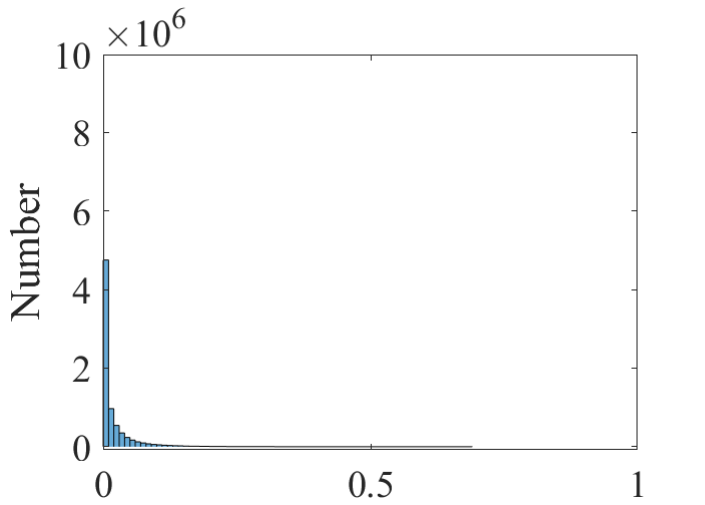}\\
\includegraphics[width=0.189\textwidth]{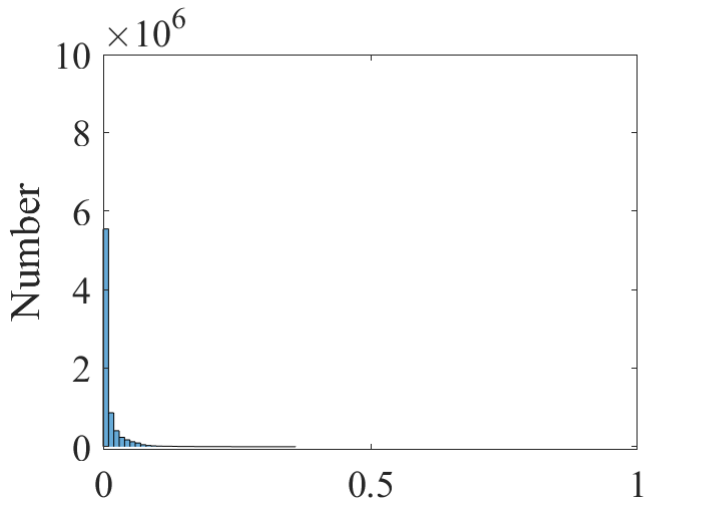}&
\includegraphics[width=0.189\textwidth]{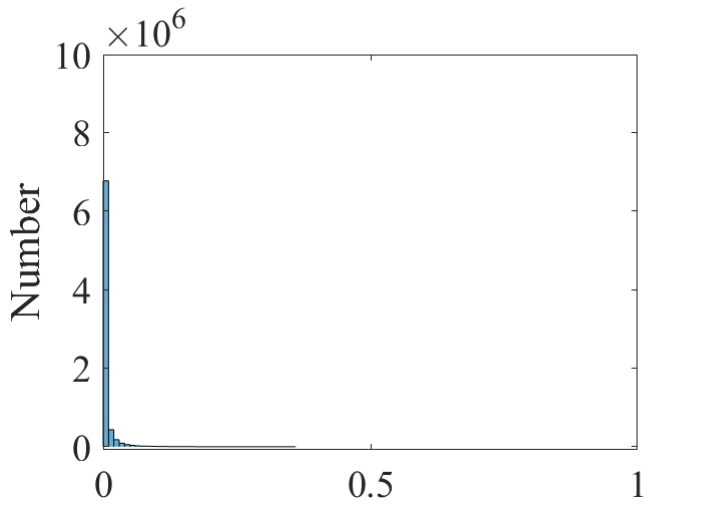}&
\includegraphics[width=0.189\textwidth]{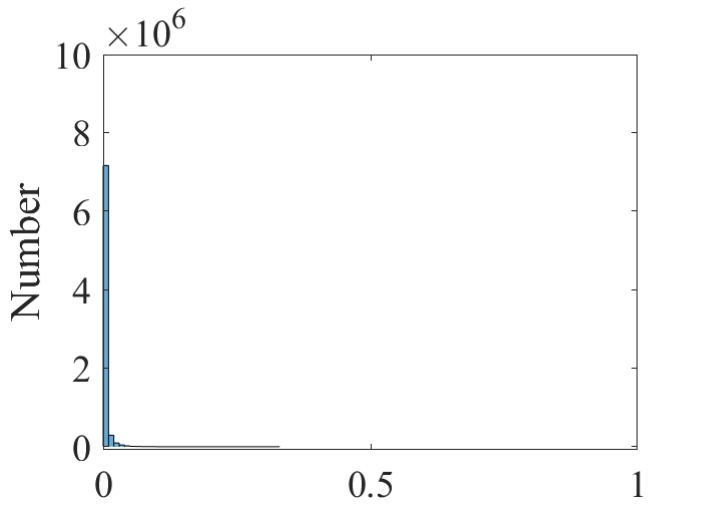}&
\includegraphics[width=0.189\textwidth]{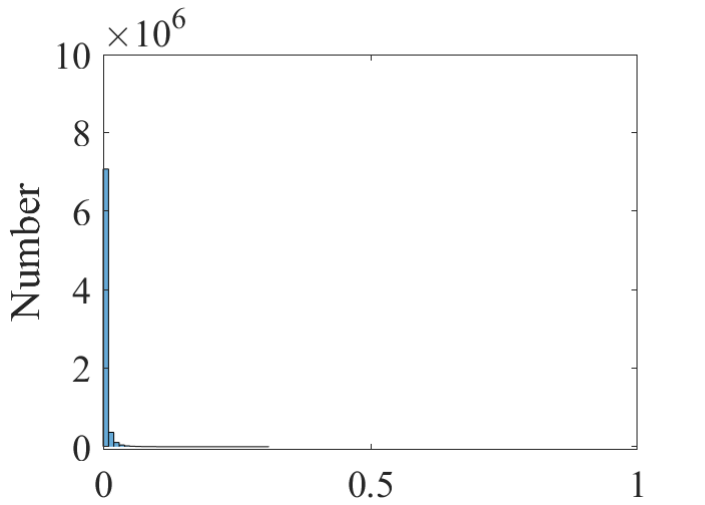}&
\includegraphics[width=0.189\textwidth]{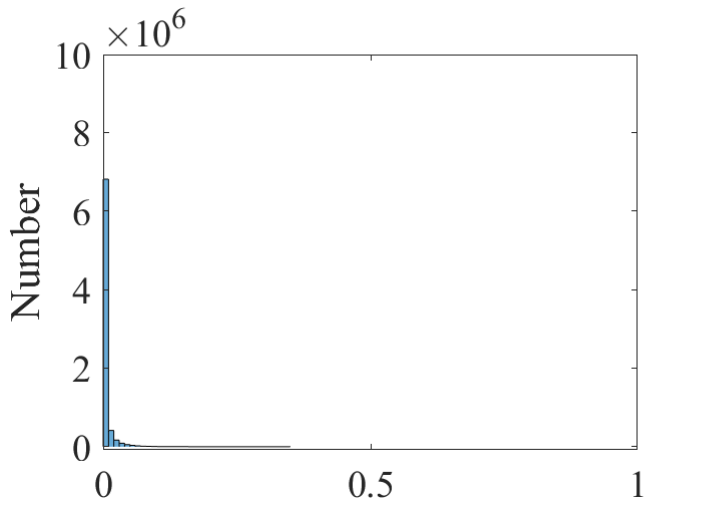}\\
(a)~~$\tY_{\rm S}$ &  (b)~~$\tT_1$ &  (c)~~$\tT_2$ & (d)~~$\tT_3$ &  (e)~~$\tT_4$\\
\end{tabular}
\caption{Comparison of sparsity (first row) and spatial smoothness (second row to fourth row: mode-1, mode-2, and mode-3, respectively) exhibited by original $\tY_{\rm S}$ and each $\tT_r$ for Washington DC dataset with four materials. } 
  \label{fig:WDC_smoothness}
\end{center}
\end{figure*}

\begin{remark}
    Besides the above-mentioned smoothness regularization, one could incorporate a variety of regularization terms or constraints that reflect the prior information of $\tT_r$. Some examples include:
    \begin{itemize}
        \item {\it Sparsity}: It is reasonable to assume $\tT_r$ is sparse in space as a single material does not appear in all pixels \cite{Qian2011HUNMF}. 
        \item {\it Nonlocal self-similarity}: One can impose the nonlocal self-similarity regularization on each material $\underline{\bm{T}}_r$, which is motivated by the fact that a spectral image often contains many repetitive local patterns, thus a local patch often has many similar patches across the spatial domain \cite{Jin2017Nonlocal,Zhuang2021Hyperspectral,Xu2019NonlocalCP,Guo2022Gaussian}.
        \item {\it Nonnegativity}: Per the physical meaning of $\tT_r(i,j,:)$, namely, the spectral signature of material $r$ at pixel $(i,j)$, adding nonnegativity on $\tT_r$ makes intuitive sense.
    \end{itemize}
     Nonetheless, in this work, our focus is not exhausting potential regularization/constraints. Instead, we use the spatial/spectral smoothness regularization in \eqref{eq:TV} to showcase the effectiveness of imposing structural constraints on the latent factors of the LMN model for CTD-based HSR.  
\end{remark}

\begin{remark}
One may argue that the CPD and Tucker models could also incorporate similar ideas by imposing constraints or regularization on their latent components. However, there are subtle yet important distinctions. In CPD and Tucker, constraints on latent factors can only capture aggregated characteristics arising from multiple endmembers at the spectral-pixel level. By contrast, in the LMN model, constraints are imposed directly on each \textit{individual} endmember, whose physical meaning and characteristics are arguably more interpretable.
To illustrate this phenomenon, we examine the spatial sparsity, spatial smoothness, and spectral smoothness of $\tY_{\rm S}$ and $\tT_r$ using the Jasper Ridge and Washington DC datasets (Figs.~\ref{fig:Ridge_smoothness} and \ref{fig:WDC_smoothness}). The $\tT_r$'s are obtained using the nonnegative BTD implementation in Tensorlab \cite{Vervliet2016Tensorlab}. Spatial sparsity is quantified by the number of zero entries. For a tensor $\tX$, mode-1 and mode-2 spatial smoothness are visualized through the absolute values of the elements of $\tX \times_1 {\bm{H}_1}$ and $\tX \times_2 {\bm{H}_2}$, respectively, while spectral smoothness is shown via $\tX \times_3 {\bm{H}_3}$.
The results (with all the values normalized by the largest) show that both sparsity and smoothness are more pronounced in the $\tT_r$ components than in the aggregated $\tY_{\rm S}$. This is likely because the spatial structure of $\tY_{\rm S}$ is influenced by mixtures of multiple materials (e.g., the Urban image in Fig.~\ref{fig:variability1} does not exhibit strong spatial sparsity or small total variation).
\end{remark}

\subsection{Accelerated Alternating Gradient Descent}
Both the formulations in \eqref{eq:model-non-blind} and \eqref{eq:model-blind} admit continuously differentiable objective functions. Hence, we employ the alternating accelerated gradient approach from \cite{Xu2013BCD,Ding2021HSR,Ding2023Fast,xu2017globally} to tackle the problem. The same design principle was used for the LL1-based HSR approach in \cite{Ding2021HSR}.

For \eqref{eq:model-non-blind},
the algorithm updates the block variables $\{\bm{A}_r
\}_{r=1}^R$, $\{\bm{B}_r\}_{r=1}^R$, $\{\bm{C}_r\}_{r=1}^R$, and $\{\underline{\bm{D}}_r\}_{r=1}^R$ in an alternating manner.
For each variable, we take a one-step gradient with an extrapolation based acceleration, e.g.,
\begin{align}
    \bm{A}_r^{(t+1)}&\leftarrow \check{\bm{A}}_r^{(t)}-\alpha^{(t)}\bm{G}_{\check{\bm{A}}_r}^{(t)},
\end{align}
where $\alpha^{(t)}$ is the step size and $\bm{G}_{\check{\bm{A}}_r}^{(t)}$ is the gradient of the objective taken at $\check{\bm{A}}_r^{(t)}$, in which
\begin{align}\label{eq:Aextra}
\check{\bm{A}}_r^{(t)}\leftarrow \bm{A}_r^{(t)}+\mu^{(t)}(\bm{A}_r^{(t)}-\bm{A}_r^{(t-1)}),
\end{align}
using a pre-defined combination parameter $\mu^{(t)}$.
After updating $\A_r$ for $r=1,\ldots,R$ in parallel, the algorithm moves to the next block (i.e., $\bm B_r$ for $r=1,\ldots,R$) and updates in the same manner.

The same algorithm is applied to handling \eqref{eq:model-blind}, with the blocks iteratively changed to $\{\widetilde{\bm{A}}_r
\}_{r=1}^R$, $\{\widetilde{ \bm{B}}_r\}_{r=1}^R, \{\bm{A}_r
\}_{r=1}^R$, $\{\bm{B}_r\}_{r=1}^R$, $\{\bm{C}_r\}_{r=1}^R$, and $\{\underline{\bm{D}}_r\}_{r=1}^R$.

Convergence of the algorithm has also been well supported in the literature \cite{Xu2013BCD,Shao2019OneBit,Ding2023Fast,xu2017globally}. The early work \cite{Xu2013BCD,xu2017globally} showed asymptotic convergence of such algorithms, with careful selection of $\alpha$ and $\mu$. The work of \cite{Shao2019OneBit} established finite-step convergence (or, iteration complexity) of this algorithm, showing that $O(1/T)$-stationary point (see definitions in \cite{Shao2019OneBit,Ding2023Fast}) is attained after $T$ iterations. The work of \cite{Ding2023Fast} extended the results of \cite{Shao2019OneBit} by considering nonconvex constraints. Invoking these results, we have the following convergence result:
\begin{proposition}
  Using the alternating accelerated gradient algorithm to tackle \eqref{eq:model-non-blind} and \eqref{eq:model-blind} produces solution sequences that converge to their respective stationary points under proper choices of $\alpha^{(t)}$ and $\mu^{(t)}$ for each block. In particular, there exist $\alpha^{(t)}$ and $\mu^{(t)}$ for each block such that the iteration complexity is $O(1/T)$; i.e., the solution reaches an $O(1/T)$-stationary points of \eqref{eq:model-non-blind} and \eqref{eq:model-blind}.
\end{proposition}
The proof is by invoking the results in \cite{Shao2019OneBit,Ding2023Fast} and thus is omitted.
Note that although the proof is conceptually the same as previous works, underpinning the exact choice of the parameters (e.g., step size), involves (somewhat tedious) problem-specific derivations.
To facilitate reproducible research,
the detailed algorithm description, the ways of obtaining the step size $\alpha^{(t)}$ and the extrapolation parameter $\mu^{(t)}$, and a per-iteration complexity analysis, are presented in the supplementary materials. We refer to the proposed approach as \textit{\textbf{C}TD with \textbf{L}MN for \textbf{I}mage fusion and \textbf{M}ultispectral-hyperspectral \textbf{B}oosting} (\texttt{CLIMB}).

\section{Experiments}
\label{sec:Experiments}
In this section, we conduct various experiments to demonstrate the effectiveness of the proposed algorithm.

\subsection{Experiment Settings}
\subsubsection{Baselines}
We compare our \texttt{CLIMB} algorithm with a number of baselines, including \texttt{CNMF} \cite{Yokoya2012HSR},  
\texttt{FUSE} \cite{Wei2015HSRSylvester}, 
\texttt{SCOTT} \cite{Prevost2020HSR}, \texttt{STEREO} \cite{Kanatsoulis2018HSR}, \texttt{SCLL1} \cite{Ding2021HSR}, \texttt{CBSTAR}\cite{Borsoi2021Coupled}, \texttt{BTDvar} \cite{Prevost2022Coupled}, \texttt{NPTSR} \cite{Xu2019Nonlocal}, and \texttt{NLSTF} \cite{Wan2020Nonnegative}.
Note that the first two methods are coupled low-rank matrix factorization-based HSR approaches; \texttt{SCOTT}, \texttt{STEREO}, and \texttt{SCLL1} are CTD methods based on CP, Tucker, and LL1 models, respectively; \texttt{CBSTAR} is Tucker-based method accounting for endmember variability; \texttt{BTDvar} is LL1-based tensor approach accounting for endmember variability;
\texttt{NPTSR} and \texttt{NLSTF} consider the nonlocal self-similarity prior and combine the tensor sparse representation.
All experiments are performed using MATLAB 2023b on a desktop with 3.4 GHz i7 CPU and 64GB RAM.

\begin{figure*}[!t]
\scriptsize\setlength{\tabcolsep}{0.3pt}
\begin{center}
\begin{tabular}{cccc}
\includegraphics[width=0.3\textwidth]{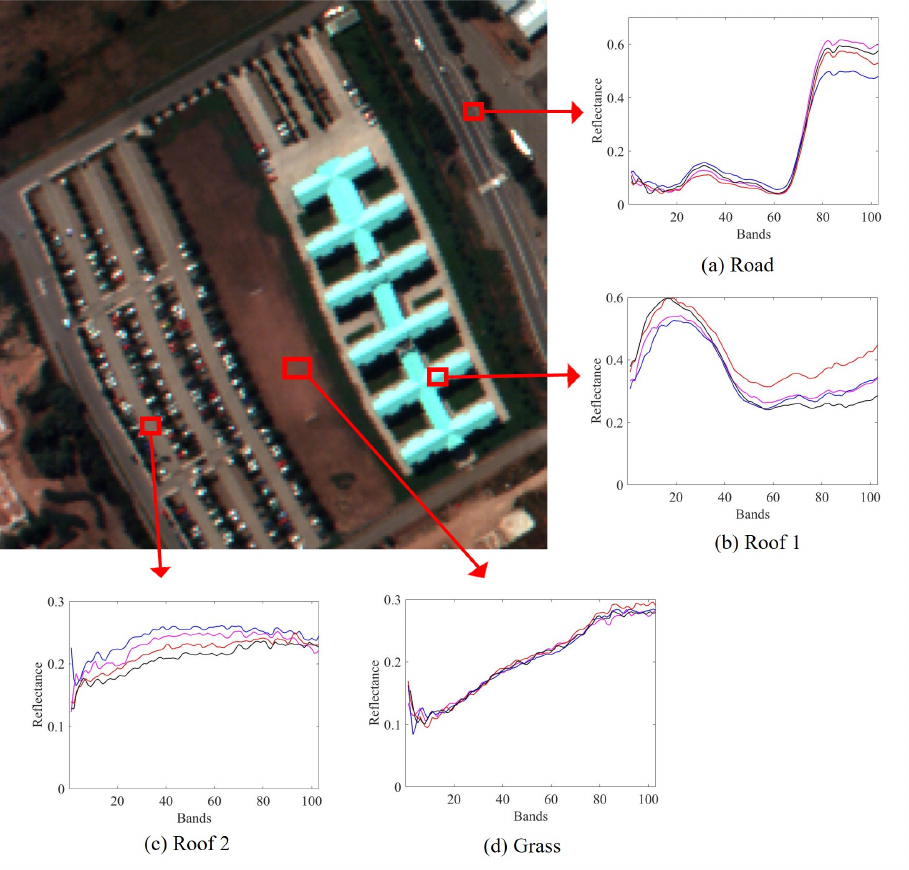}&
\hspace{3mm}
\includegraphics[width=0.3\textwidth]{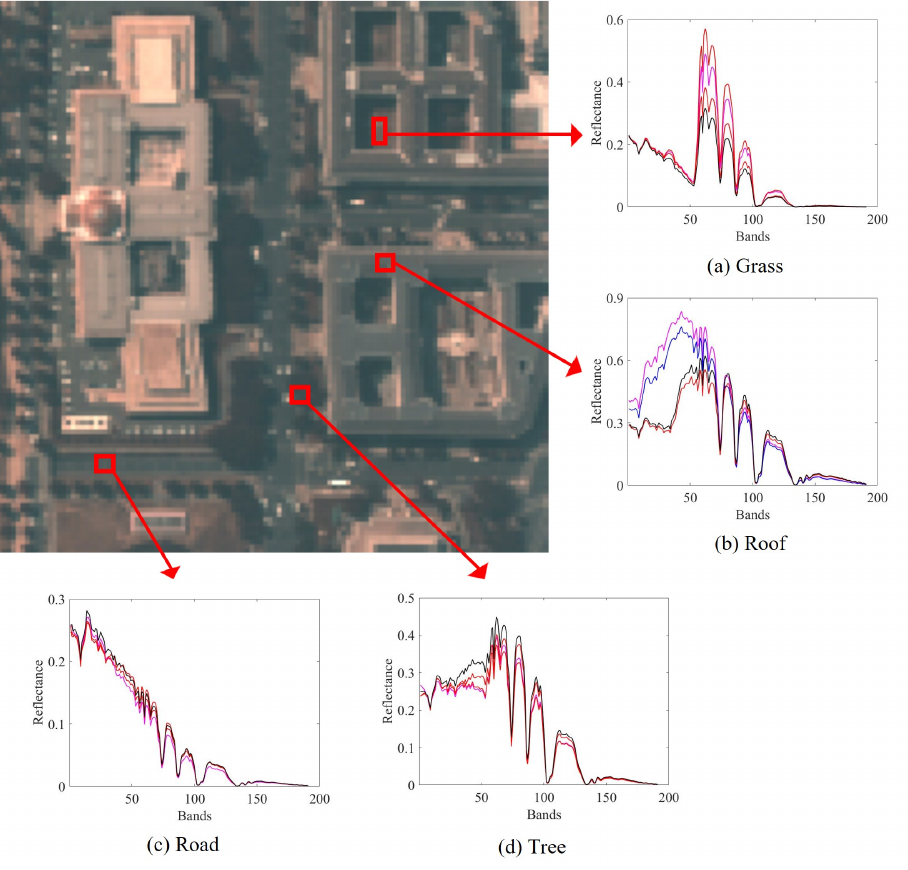}&
\hspace{3mm}
\includegraphics[width=0.3\textwidth]{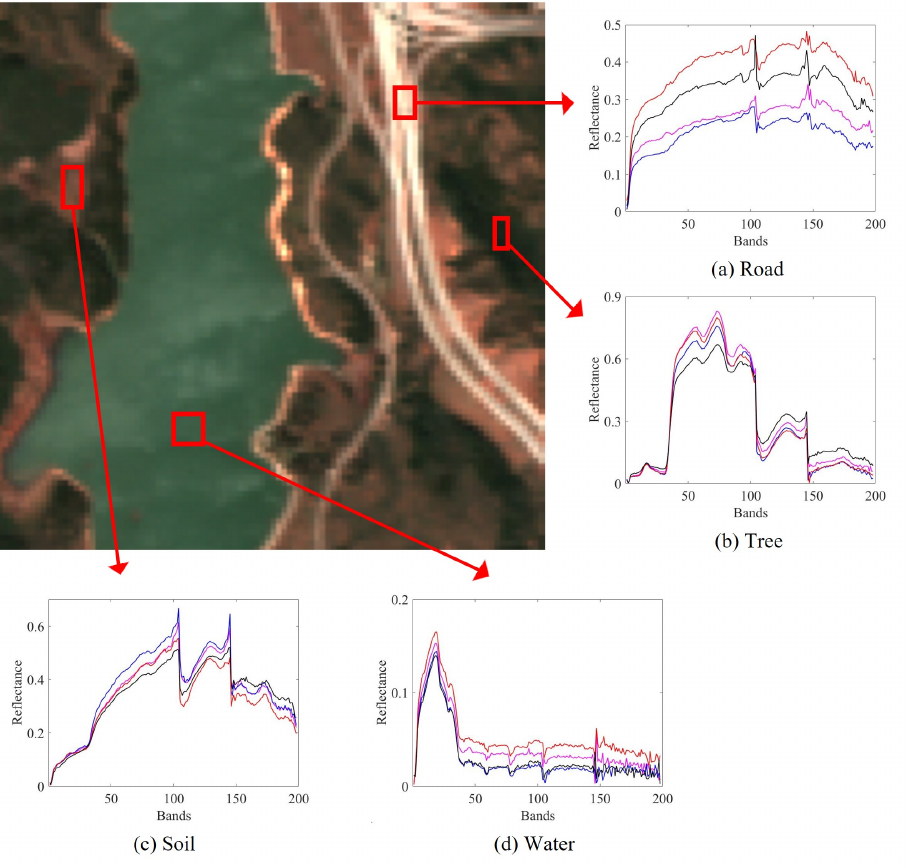}\\
 (a) Pavia University & (b) Washington DC & (c) Jasper Ridge\\
\end{tabular}
\caption{Illustration of endmember variability of the subimages of Pavia University, Washington DC, and Jasper Ridge datasets. (The spectral signatures of materials of each dataset are manually selected pure pixels.)}
  \label{fig:variability2}
\end{center}
\end{figure*}

\subsubsection{Degradation Model}
In the semi-real data experiments, we adopt the simulation strategy from \cite{Kanatsoulis2018HSR,Ding2021HSR}. Specifically, a clean hyperspectral image serves as the ``ground-truth'' SRI, from which the simulated HSI and MSI are generated following Wald's protocol \cite{Wald1997Fusion}. The availability of ground-truth allows for quantitative evaluation of the recovered SRI. To generate the HSI, the SRI is first blurred band-by-band using a $9\times 9$ Gaussian kernel and then downsampled by a factor of 8 along each spatial dimension. To generate the MSI, the spectral degradation matrix $\bm{P}_{M}$ performs band selection and aggregation based on the characteristics of the multispectral sensor, as detailed in \cite{Kanatsoulis2018HSR}. In this work, we consider two sensors: LANDSAT\footnote{https://landsat.gsfc.nasa.gov/} and QuickBird\footnote{https://www.satimagingcorp.com/satellite-sensors/quickbird/}. Unless stated otherwise, zero-mean white Gaussian noise is added to both HSI and MSI with an SNR of 35dB. All results are averaged over 10 random trials with varying noise realizations.

Note that as we use real hyperspectral images as SRIs, due to variable illumination and atmospheric conditions, EV naturally occurs in the obtained spectral images (see Figs. \ref{fig:variability1} and \ref{fig:variability2}).

\subsubsection{Quantitative Metrics}
For the semi-real data experiments, we again follow the established convention \cite{Kanatsoulis2018HSR,Ding2021HSR} to evaluate the quality of the recovered SRIs using six standard metrics \cite{Loncan2015HSRoverview,Yokoya2017HSRoverview}. Specifically, we employ reconstruction signal-to-noise ratio (RSNR) ($\uparrow \hspace{-0.5mm}\infty$), structural similarity index (SSIM) ($\uparrow \hspace{-0.5mm}1$), cross-correlation (CC) ($\uparrow \hspace{-0.5mm}1$), relative global dimensional error (ERGAS) ($\downarrow \hspace{-0.5mm}0$), root mean square error (RMSE) ($\downarrow \hspace{-0.5mm}0$), and spectral angle mapper (SAM) ($\downarrow \hspace{-0.5mm}0$). Definitions of these metrics can be found in \cite{Loncan2015HSRoverview,Wei2016Fusion}. In general, higher values of SSIM, RSNR, and CC, and lower values of ERGAS, RMSE, and SAM indicate better reconstruction performance.

\begin{table}[!t]
\renewcommand\arraystretch{1}
\setlength{\tabcolsep}{5pt}
\renewcommand\arraystretch{1.1}
  \centering
  \caption{Performance for Urban data with the degradation known. (The highest and second-highest values are highlighted in bold and underlined, respectively.)}
  \resizebox{0.6\linewidth}{!}{
  \begin{tabular}{c|c|c|c|c|c|c}\hline

    \hline
    Methods & \multicolumn{1}{c|}{RSNR} & \multicolumn{1}{c|}{SSIM} & \multicolumn{1}{c|}{CC} 
    & \multicolumn{1}{c|}{ERGAS} & \multicolumn{1}{c|}{RMSE} 
    & \multicolumn{1}{c}{SAM}\\ \hline
    \texttt{CNMF}   & 22.02 & 0.9616 & 0.9864 & 0.4085 & 0.0207 & 0.0738\\
    \texttt{FUSE}   & 25.13 & 0.9778 & 0.9930 & 0.2813 & 0.0144 & 0.0504\\
    \texttt{SCOTT}  & 23.74 & 0.9612 & 0.9900 & 0.3356 & 0.0169 & 0.0696\\
    \texttt{STEREO} & 24.52 & 0.9604 & 0.9915 & 0.3038 & 0.0154 & 0.0663\\
    \texttt{SCLL1}  & 26.66 & \textbf{0.9837} & 0.9948 & 0.2521 & 0.0121 & 0.0504\\
    \texttt{CBSTAR} & 25.19 & 0.9798 & 0.9929 & 0.2862 & 0.0143 & 0.0612\\
    \texttt{BTDvar} & 23.62 & 0.9519 & 0.9893 & 0.3532 & 0.0171 & 0.0723\\
    \texttt{NPTSR}  & \underline{27.98} & \underline{0.9815} & \underline{0.9962} & \underline{0.2058} & \underline{0.0104} & \underline{0.0436}\\
    \texttt{NLSTF}  & 25.93 & 0.9778 & 0.9944 & 0.2444 & 0.0131 & 0.0529\\
    \texttt{CLIMB}  & \textbf{28.28} & 0.9809 & \textbf{0.9965} & \textbf{0.1970} & \textbf{0.0100} & \textbf{0.0434}\\
    \hline
    \end{tabular}}%
  \label{table:Urban}%
\end{table}%

\subsubsection{Parameter Settings}
In our algorithm, the ranks are set as $L_{r} = M_{r} = L$, and $N_{r} = N$ for all $r = 1, \ldots, R$. The parameters $\lambda$, $\eta$, $L$, and $N$ are selected following the heuristic strategy proposed in \cite{Ding2021HSR}, which involves tuning based on metrics such as the SAM computed between the observed HSI and the reconstructed HSI (after applying spatial degradation to the estimated SRI). This approach has been found effective in practice. The influence of these parameters on performance is further discussed in Sec.~\ref{sec:para}.
For the nonconvex total variation regularization, we set $p = 0.5$ and $\varepsilon = 0.01$. 
The initializations of ${\underline{\bm{D}}_r^{(0)}, \bm{A}_{r}^{(0)}, \bm{B}_{r}^{(0)}, \bm{C}_r^{(0)}}$ are as follows: $\underline{\bm{D}}_r^{(0)}$ is generated following the method in \cite{Li2018CSTF}; $\bm{A}_{r}^{(0)}$ and $\bm{B}_{r}^{(0)}$ are obtained from the right singular matrices of the skinny SVD of the mode-$1$ and mode-$2$ unfoldings $\bm{Y}_{\rm M}^{(1)}$ and $\bm{Y}_{\rm M}^{(2)}$ of the observed MSI $\underline{\bm{Y}}_{\rm M}$, where each unfolding $\bm{Y}_{\rm M}^{(n)}$ contains the mode-$n$ fibers as columns; $\bm{C}_r^{(0)}$ is initialized by applying the vertex component analysis \cite{Nascimento2005VCA} to the mode-$3$ unfolding $\bm{Y}_{\rm H}^{(3)}$ of the observed HSI $\underline{\bm{Y}}_{\rm H}$.
The algorithm terminates when the relative change in the objective function falls below $10^{-8}$ or the number of iterations exceeds 1000. For baselines, parameters are carefully tuned following the recommendations in corresponding papers.

\subsection{Semi-real Experiments with Known Degradation Operators}
In this subsection, we test the performance of our method using semi-real data under the
assumption that all $\bm{P}_1$, $\bm{P}_2$, and $\bm{P}_M$ are known.

\begin{figure}[!t]
\scriptsize\setlength{\tabcolsep}{0.3pt}
\begin{center}
\begin{tabular}{cccccccccccc}
\includegraphics[width=0.085\textwidth]{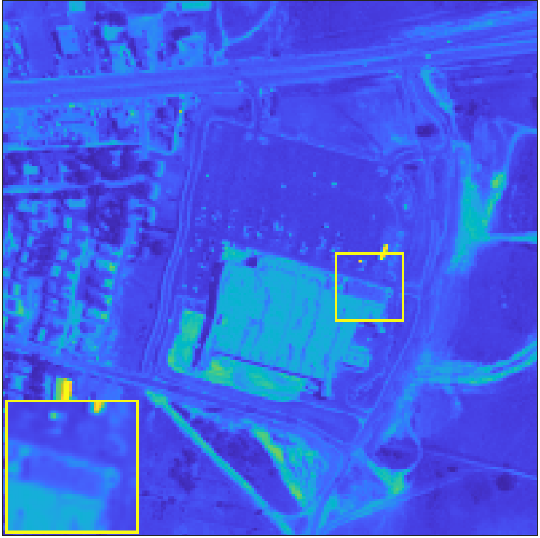}&
\includegraphics[width=0.085\textwidth]{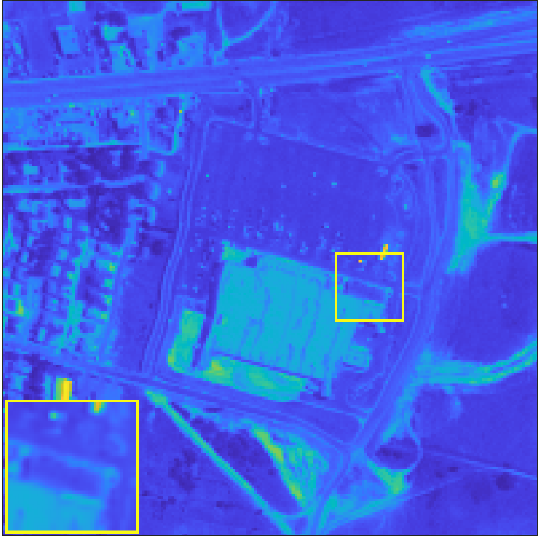}&
\includegraphics[width=0.085\textwidth]{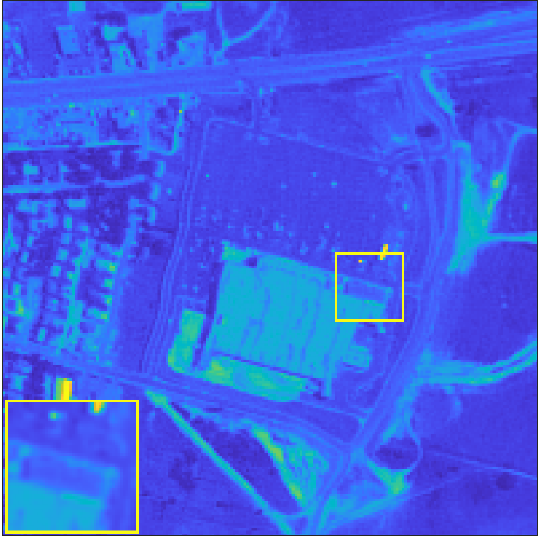}&
\includegraphics[width=0.085\textwidth]{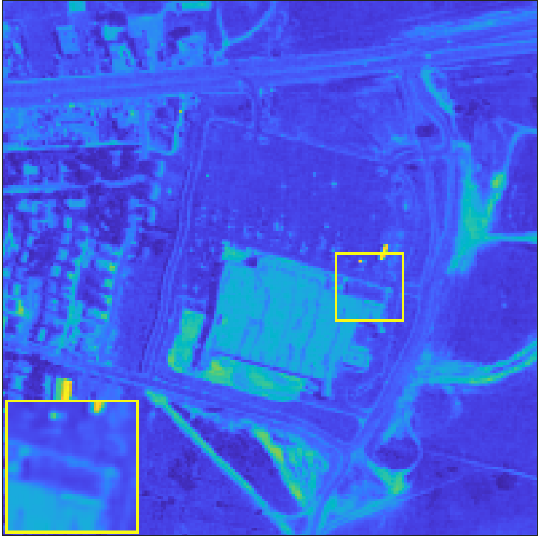}&
\includegraphics[width=0.085\textwidth]{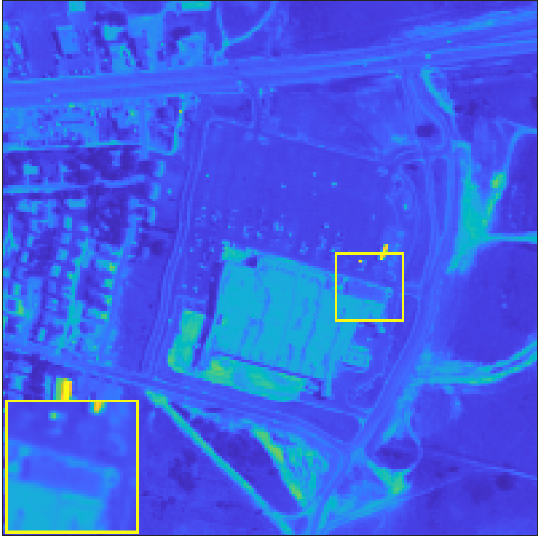}&
\includegraphics[width=0.085\textwidth]{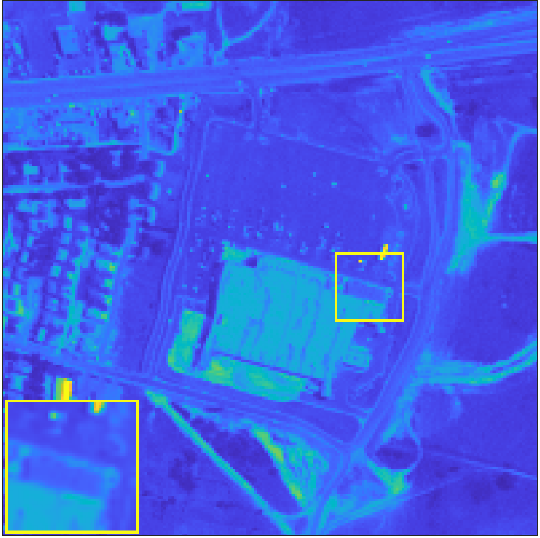}&
\includegraphics[width=0.085\textwidth]{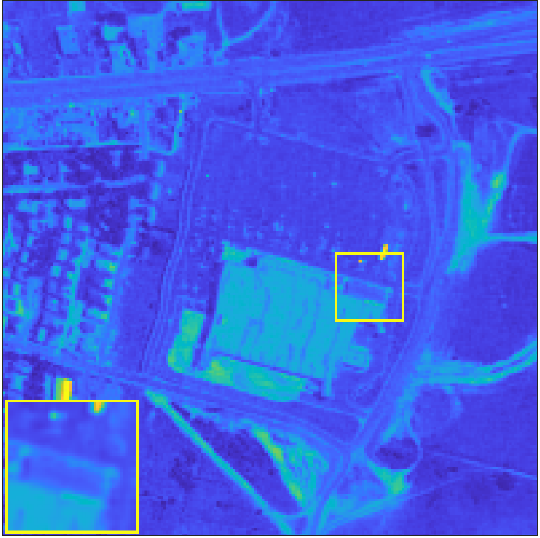}&
\includegraphics[width=0.085\textwidth]{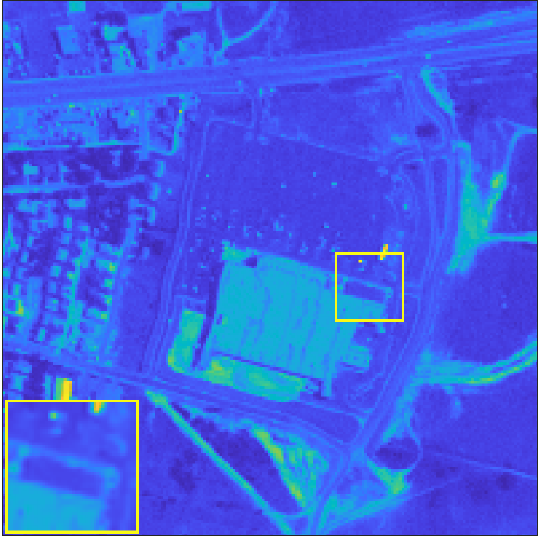}&
\includegraphics[width=0.085\textwidth]{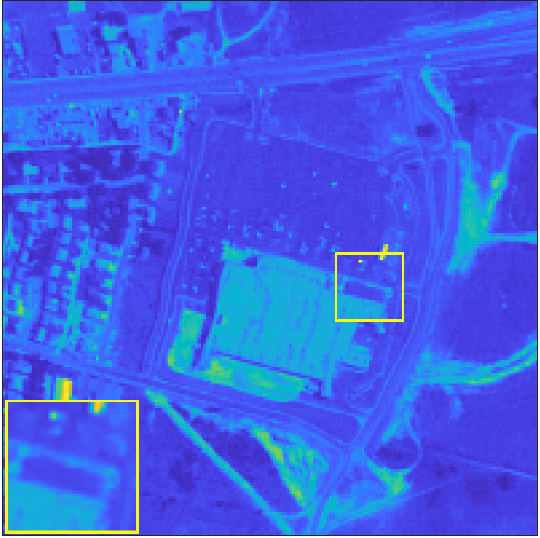}&
\includegraphics[width=0.085\textwidth]{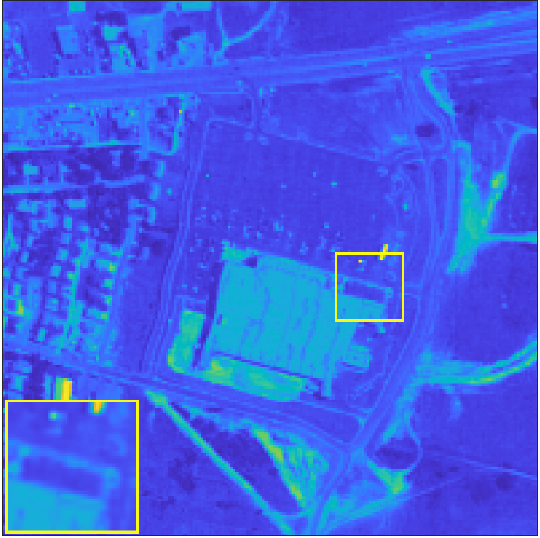}&
\includegraphics[width=0.101\textwidth]{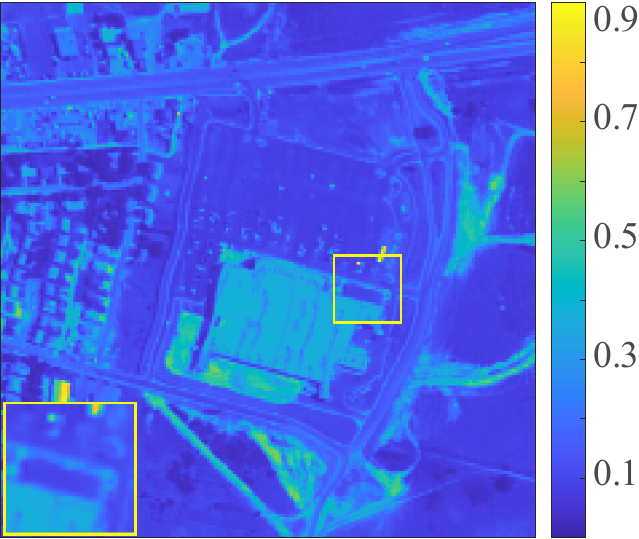}\\
\includegraphics[width=0.085\textwidth]{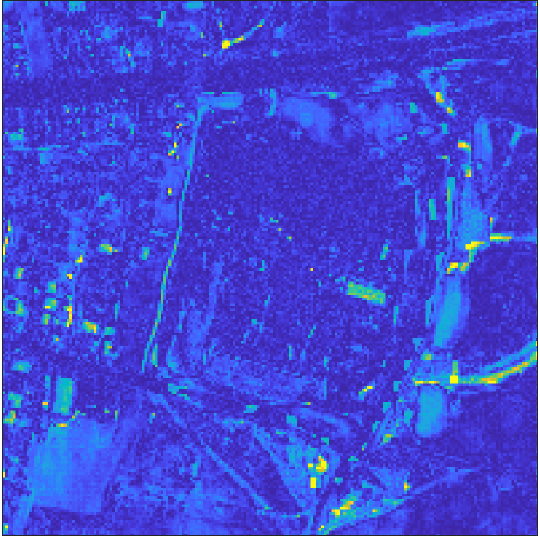}&
\includegraphics[width=0.085\textwidth]{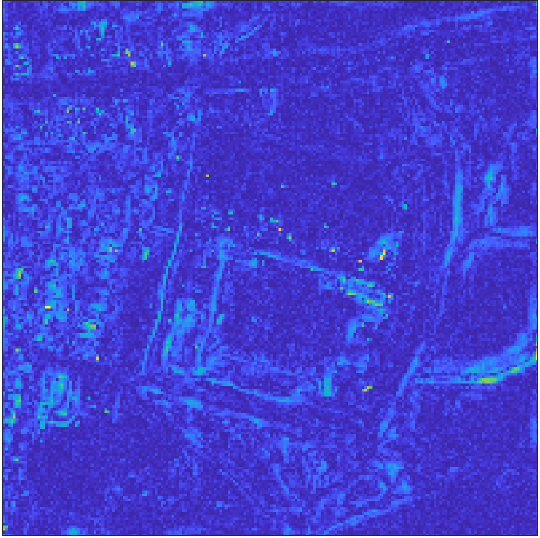}&
\includegraphics[width=0.085\textwidth]{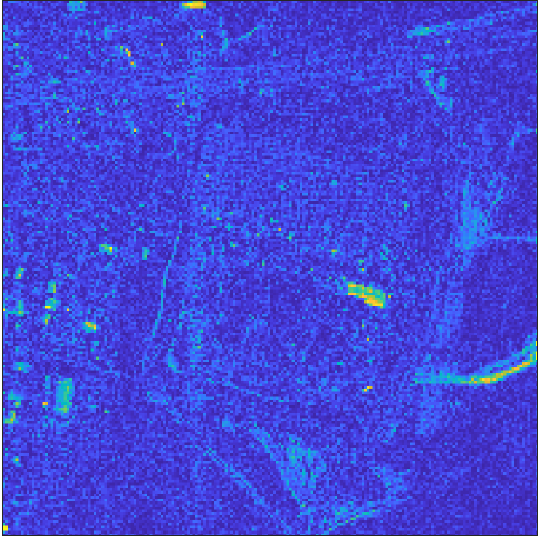}&
\includegraphics[width=0.085\textwidth]{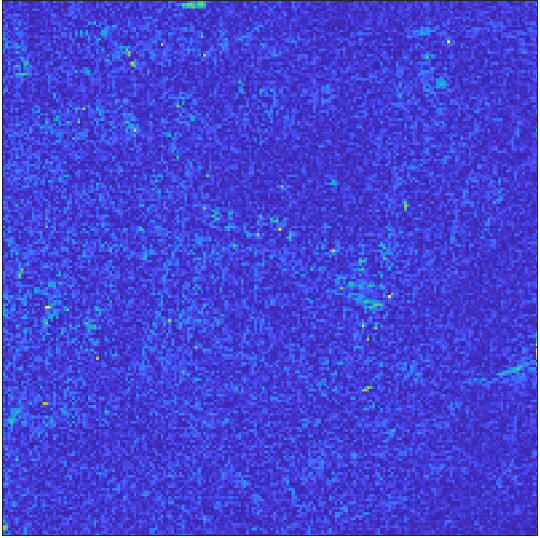}&
\includegraphics[width=0.085\textwidth]{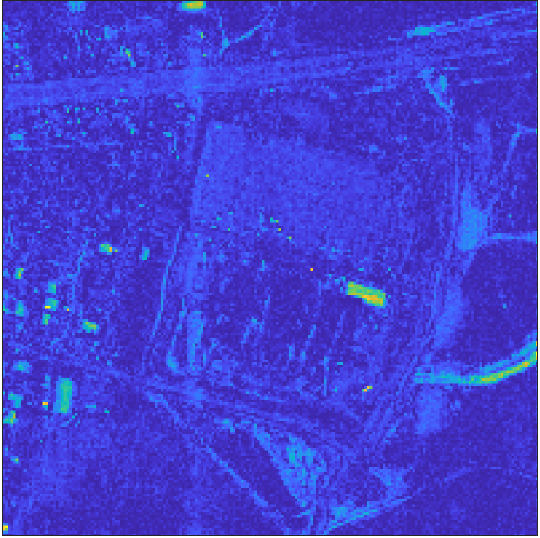}&
\includegraphics[width=0.085\textwidth]{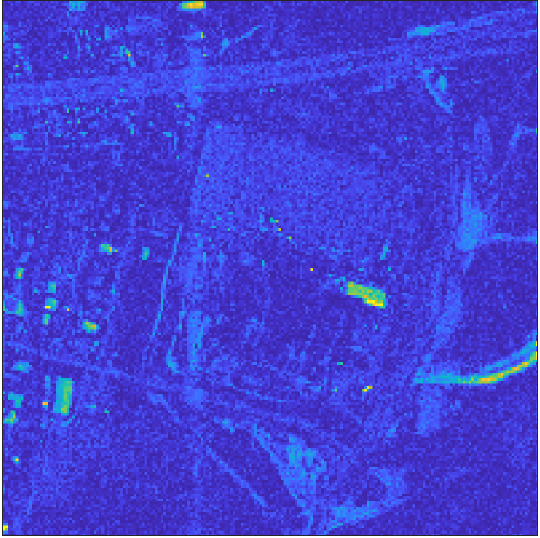}&
\includegraphics[width=0.085\textwidth]{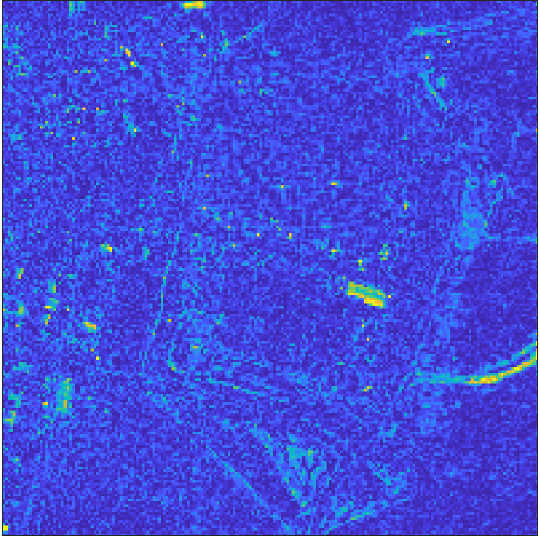}&
\includegraphics[width=0.085\textwidth]{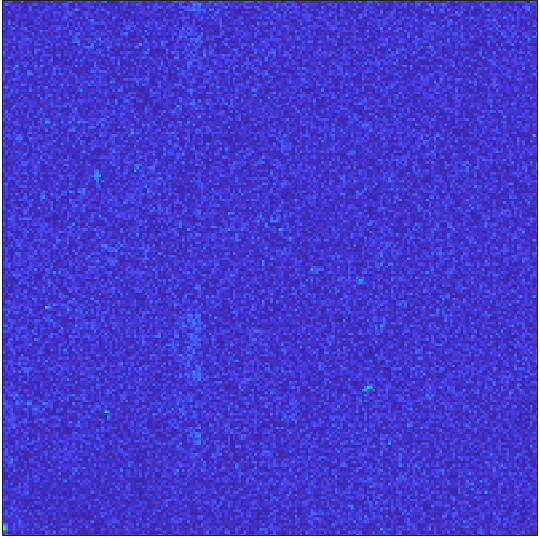}&
\includegraphics[width=0.085\textwidth]{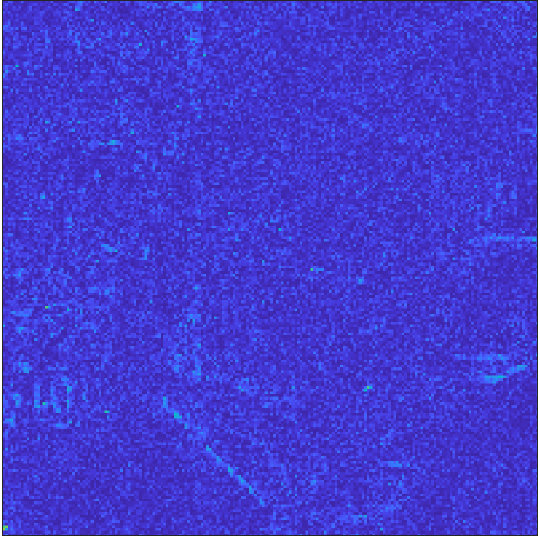}&
\includegraphics[width=0.085\textwidth]{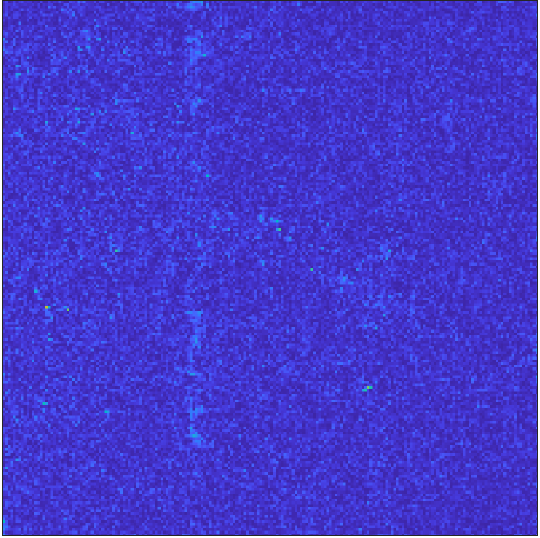}&
\includegraphics[width=0.102\textwidth]{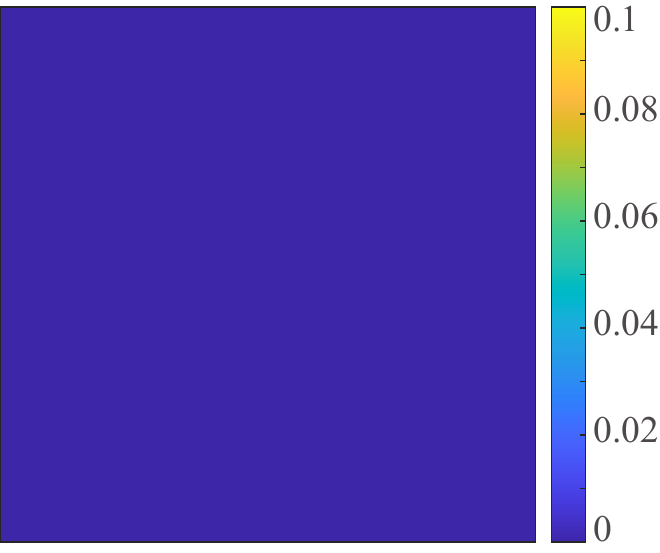}\\
\includegraphics[width=0.085\textwidth]{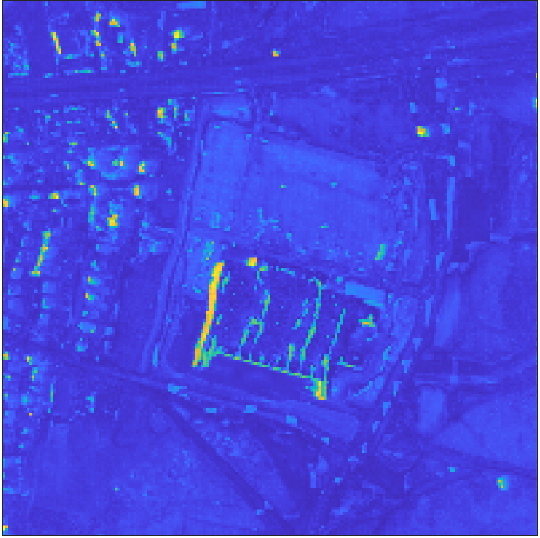}&
\includegraphics[width=0.085\textwidth]{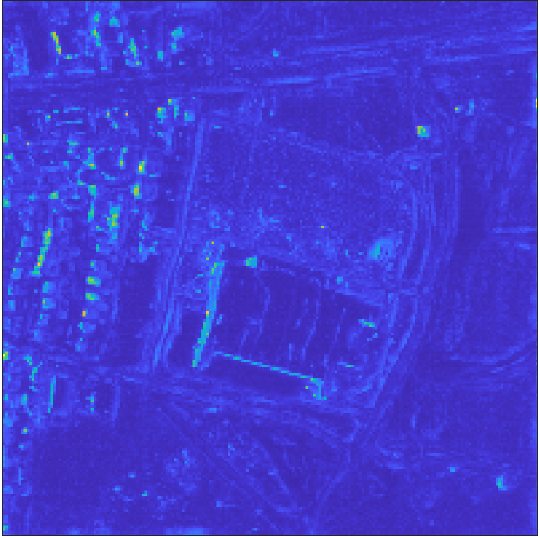}&
\includegraphics[width=0.085\textwidth]{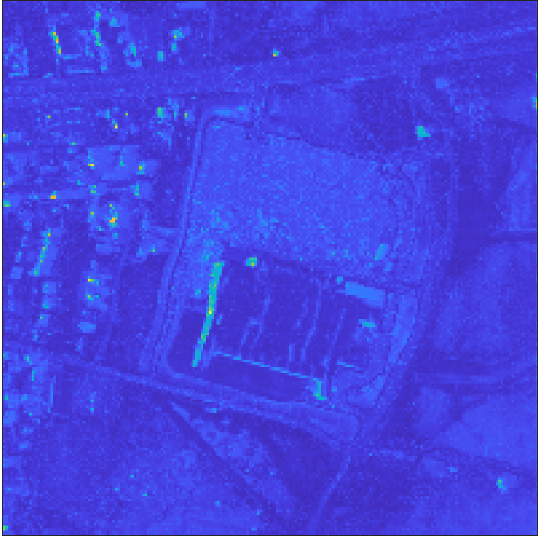}&
\includegraphics[width=0.085\textwidth]{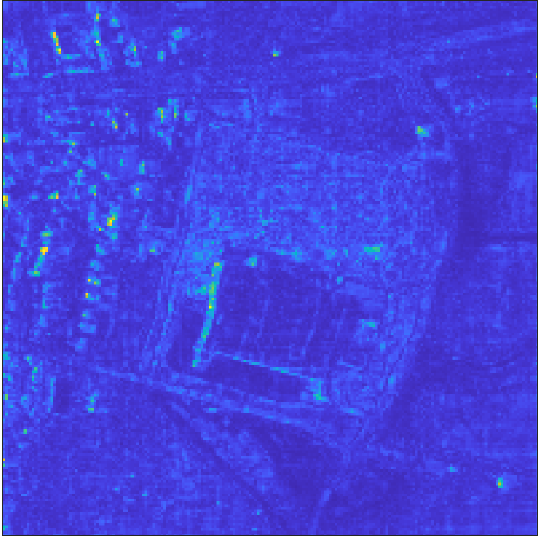}&
\includegraphics[width=0.085\textwidth]{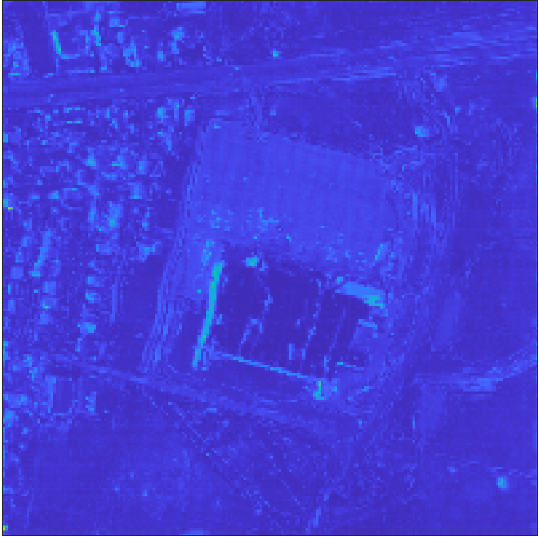}&
\includegraphics[width=0.085\textwidth]{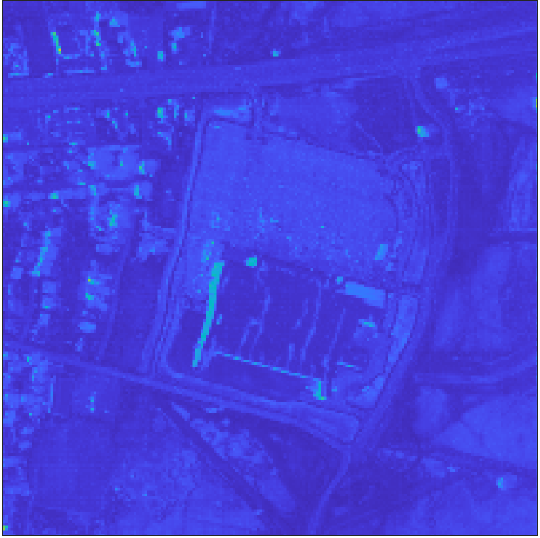}&
\includegraphics[width=0.085\textwidth]{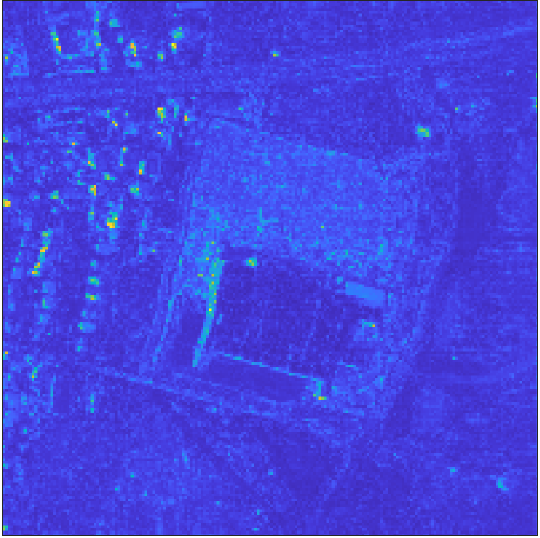}&
\includegraphics[width=0.085\textwidth]{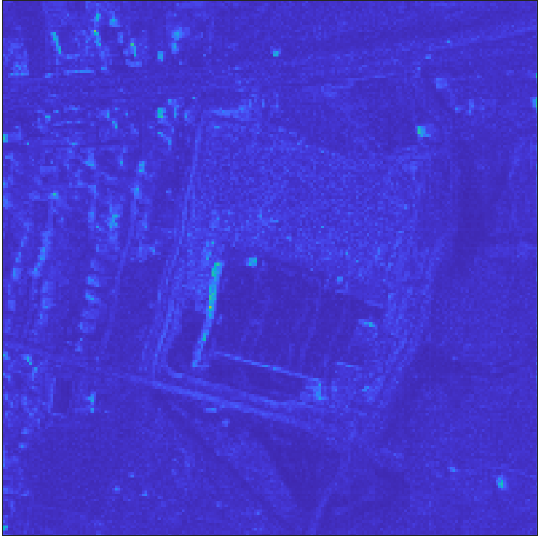}&
\includegraphics[width=0.085\textwidth]{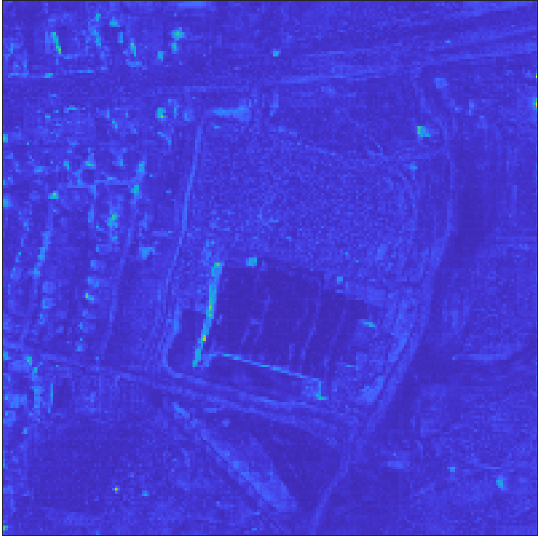}&
\includegraphics[width=0.085\textwidth]{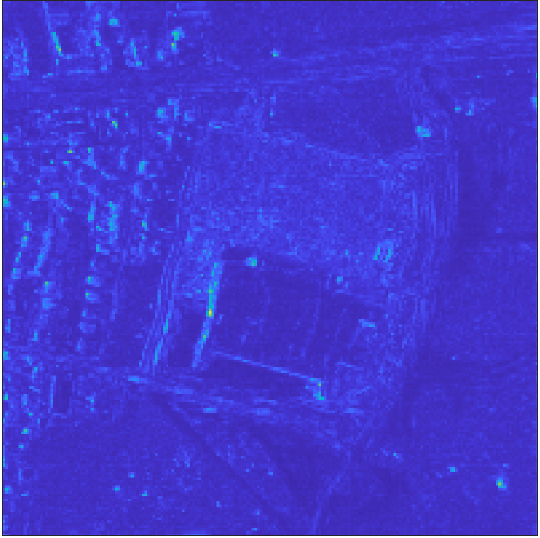}&
\includegraphics[width=0.101\textwidth]{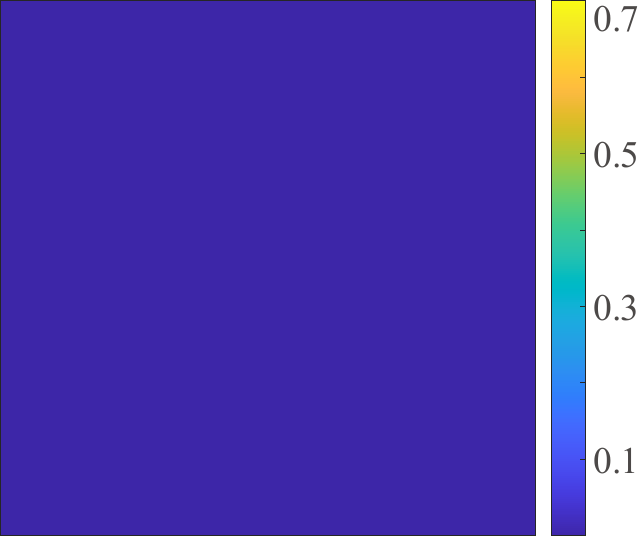}\\
(a) \texttt{CNMF} & (b) \texttt{FUSE} & (c) \texttt{SCOTT} & (d) \texttt{STEREO} & (e) \texttt{SCLL1} & (f) \texttt{CBSTAR}& (g) \texttt{BTDvar} & (h) \texttt{NPTSR} & (i) \texttt{NLSTF} & (j) \texttt{CLIMB} & (k) SRI\\
\end{tabular}
\caption{The recovered results of Urban with the degradation known.
First row: the recovered SRIs of the 43rd band; 
Second row: the corresponding residual images of the 43rd band; 
Third row: the SAM maps.}
  \label{fig:Urban}
\end{center}
\end{figure}

\begin{figure}[!t]
\scriptsize\setlength{\tabcolsep}{0.3pt}
\begin{center}
\begin{tabular}{cccccc}
\includegraphics[width=0.3\textwidth]{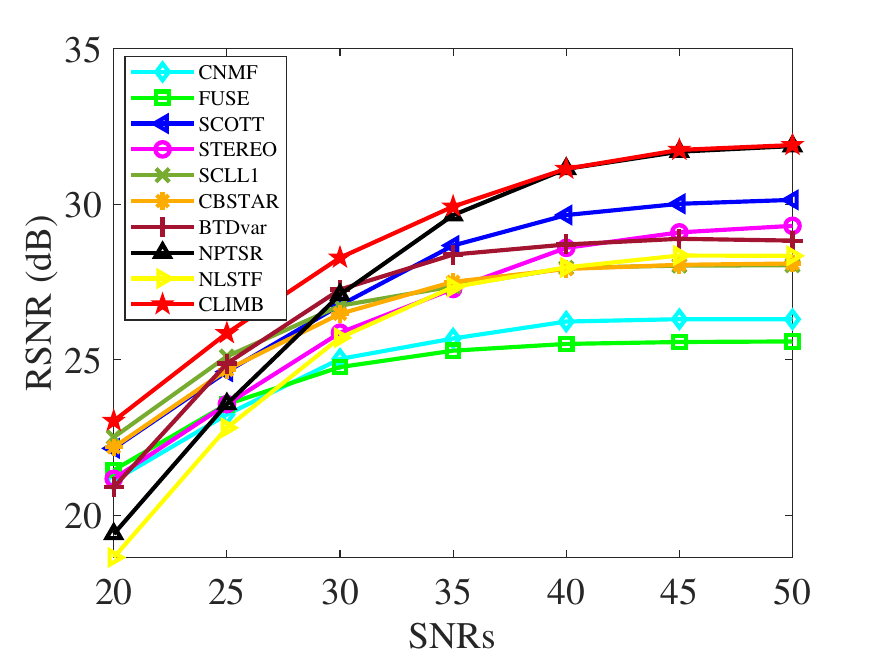}&
\includegraphics[width=0.3\textwidth]{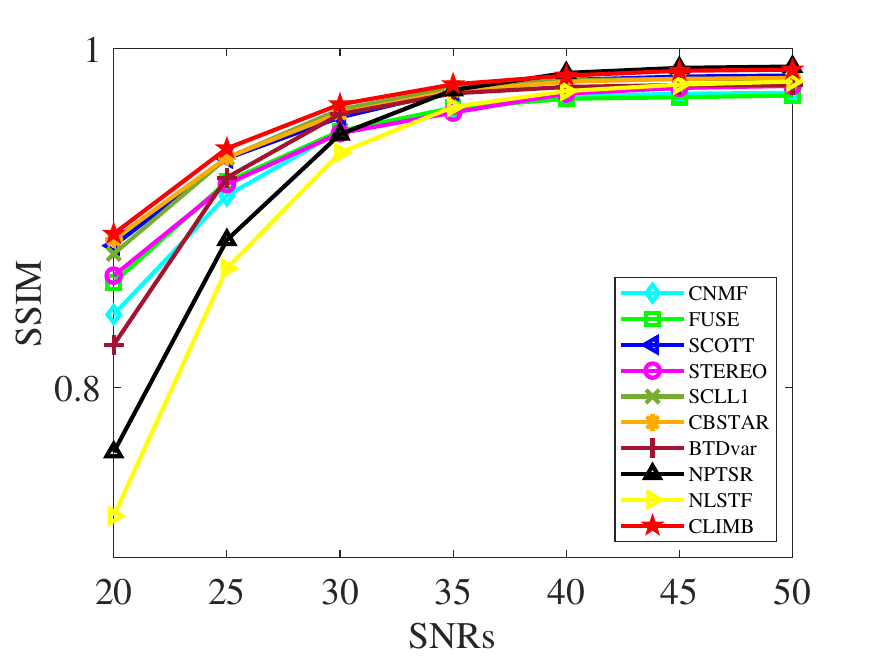}&
\includegraphics[width=0.3\textwidth]{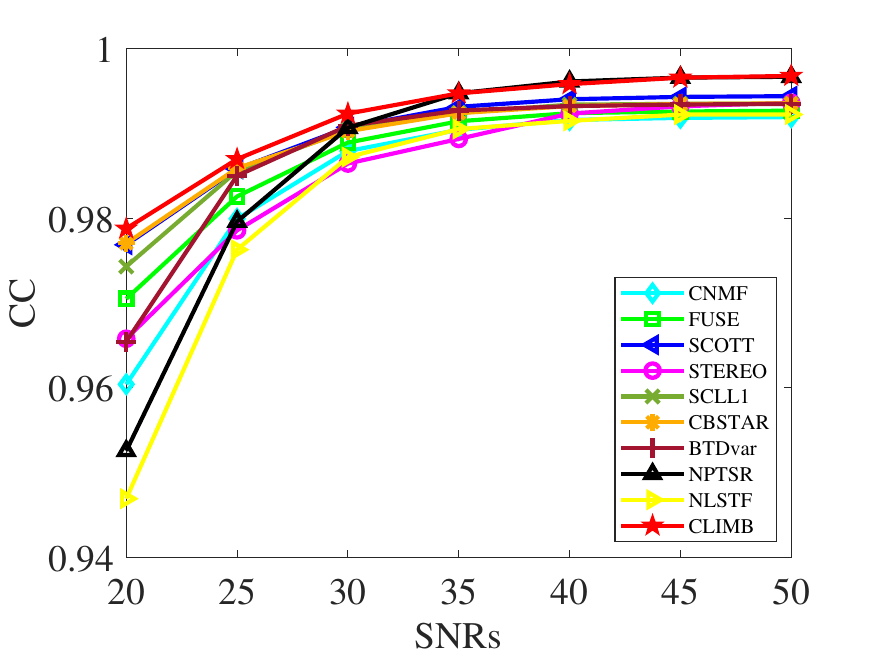}\\
 (a) RSNR & (b) SSIM & (c) CC \\
\includegraphics[width=0.3\textwidth]{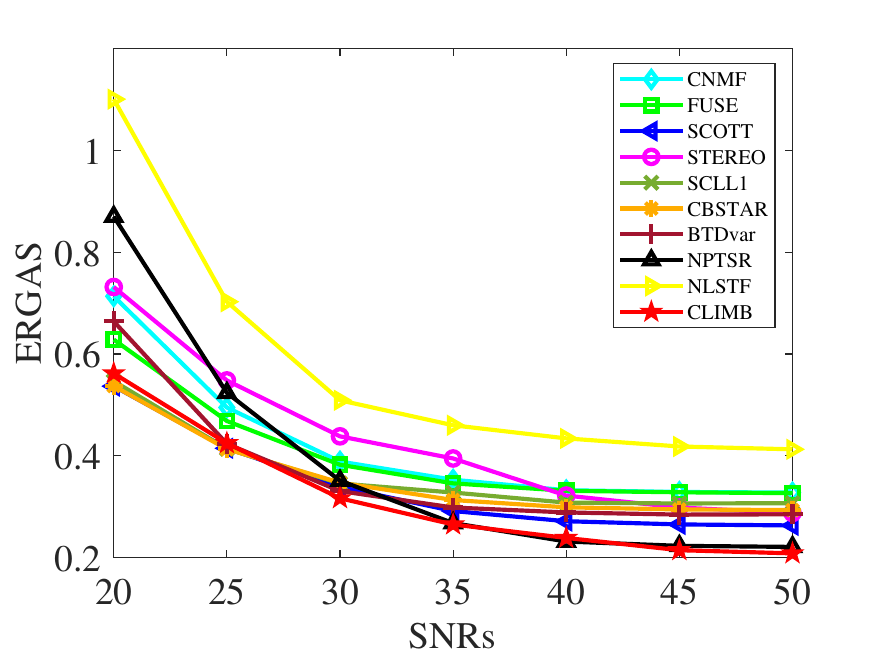}&
\includegraphics[width=0.3\textwidth]{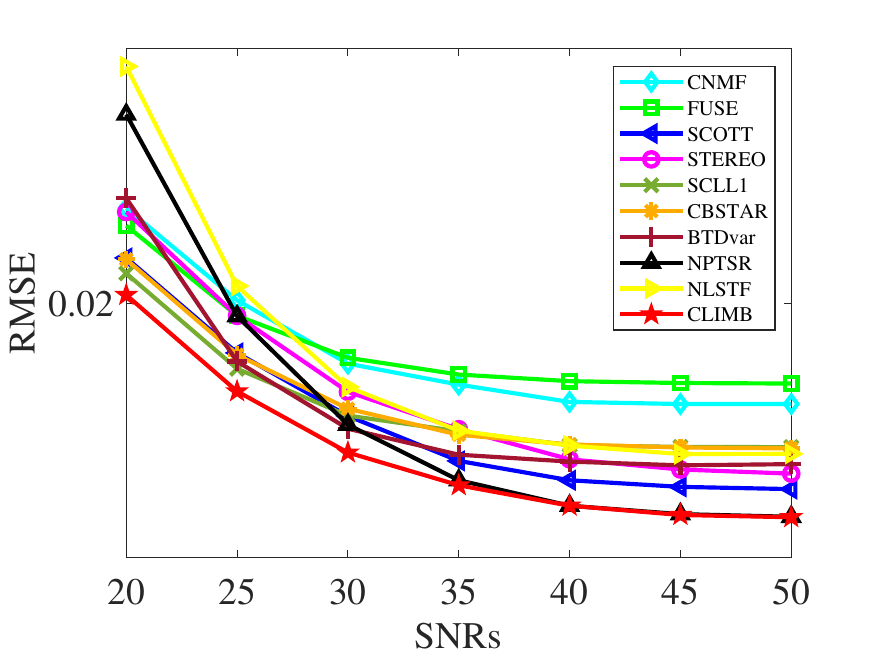}&
\includegraphics[width=0.3\textwidth]{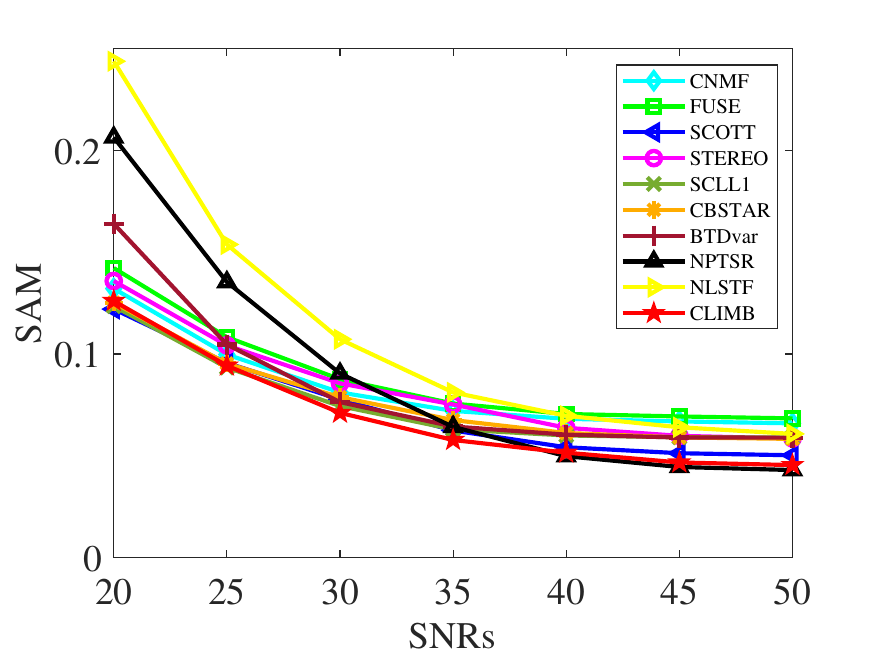}\\
 (d) ERGAS & (e) RMSE & (f) SAM\\
\end{tabular}
\caption{Reconstruction metrics for Jasper Ridge under different noises.}
  \label{fig:Ridge_noise}
\end{center}
\end{figure}

\subsubsection{Urban Dataset}
In the first experiment, we use the Urban dataset with $307 \times 307$ pixels and 210 wavelengths, captured by the HYDICE sensor over Copperas Cove, TX, U.S. After removing bands severely affected by dense water vapor and atmospheric distortions, the used subimage has a size of $200\times 200\times 162$. By applying the aforementioned spatial degradation and LANDSAT-based spectral degradation, we generate the HSI $\underline{\bm{Y}}_{\rm H} \in \mathbb{R}^{50\times 50\times 162}$ and the MSI $\underline{\bm{Y}}_{\rm M} \in \mathbb{R}^{200\times 200\times 6}$. The tensor rank is set to $R=4$ following \cite{Jia2007Spectral}. 

Table~\ref{table:Urban} summarizes the reconstruction performance under $\mathrm{SNR} = 35\,\mathrm{dB}$. As shown, the proposed method achieves the best performance across most evaluation metrics.

Fig.~\ref{fig:Urban} shows the 43rd spectral bands of the recovered SRIs obtained by different methods, the corresponding absolute residual images (i.e., $\textrm{abs}(\underline{\bm{Y}}_{\rm S}(:,:,k) - \underline{\widehat{\bm{Y}}}_{\rm S}(:,:,k))$ for $k=43$), and the SAM maps. From Fig.~\ref{fig:Urban}, we observe the following: (1) the proposed \texttt{CLIMB} method achieves the best visual quality compared to the baselines, as illustrated in the enlarged region; (2) our method yields smaller residual values across most pixels than those of the baselines; and (3) the SAM map produced by our method is more consistent with the ground truth shown in the rightmost column.

\subsubsection{Jasper Ridge Dataset} 
For the second experiment, we use the Jasper Ridge dataset, which has the full size of  $512\times 614\times 224$. This dataset was captured by the AVIRIS sensor \cite{Green1998AVIRIS} over Jasper Ridge, California \cite{Sabol1993Mapping}.
After removing the water absorption bands, we use a subscene of size $100\times 100\times 198$.
The HSI is generated with a spatial resolution of $25\times 25$ and 198 spectral bands, while the MSI is of size $100\times 100\times 6$, obtained via LANDSAT-based spectral degradation. Following the setting in \cite{Zhu2014Spectra}, the tensor rank is set to $R=4$.

To evaluate the performance of all methods under varying noise levels, Fig.~\ref{fig:Ridge_noise} presents the reconstruction metrics across different SNR values ranging from 20~dB to 50~dB. As shown, the proposed \texttt{CLIMB} consistently achieves the best performance across most metrics. A more detailed numerical comparison at $\mathrm{SNR} = 35$~dB is provided in Table~\ref{table:Ridge}.

\begin{table}[!t]
\renewcommand\arraystretch{1}
\setlength{\tabcolsep}{5pt}%
\renewcommand\arraystretch{1.1}%
  \centering
  \caption{Performance for Jasper Ridge data with the degradation known. (The highest and second-highest values are highlighted in bold and underlined, respectively.)}
  \resizebox{0.6\linewidth}{!}{
  \begin{tabular}{c|c|c|c|c|c|c}\hline

    \hline
    Methods & \multicolumn{1}{c|}{RSNR} & \multicolumn{1}{c|}{SSIM} & \multicolumn{1}{c|}{CC} 
    & \multicolumn{1}{c|}{ERGAS} & \multicolumn{1}{c|}{RMSE} 
    & \multicolumn{1}{c}{SAM}\\ \hline
    \texttt{CNMF}   & 25.66    & 0.9649    & 0.9904    & 0.3538    & 0.0152    & 0.0720 \\
    \texttt{FUSE}   & 25.29    & 0.9652    & 0.9915    & 0.3455    & 0.0158    & 0.0756\\
    \texttt{SCOTT}  & 28.67    & 0.9753    & \underline{0.9932}    & 0.2912    & 0.0107    & \underline{0.0624}\\
    \texttt{STEREO} & 27.28    & 0.9624    & 0.9893    & 0.3950    & 0.0126    & 0.0752\\
    \texttt{SCLL1}  & 27.36    & \underline{0.9769}    & 0.9924   & 0.3276    & 0.0124    & 0.0630\\
    \texttt{CBSTAR} & 27.50    & 0.9751    & 0.9924    & 0.3129    & 0.0122    & 0.0675\\
    \texttt{BTDvar} & 28.35    & 0.9738    & 0.9927    & 0.2983    & 0.0111    & 0.0644\\
    \texttt{NPTSR}  & \underline{29.65}    & 0.9757    & \textbf{0.9948}    & \underline{0.2676}    & \underline{0.0096}   & 0.0642\\
    \texttt{NLSTF}  & 27.34    & 0.9657    & 0.9905    & 0.4599    & 0.0125    & 0.0811\\
    \texttt{CLIMB}  & \textbf{29.95}    & \textbf{0.9791}    & \textbf{0.9948}    & \textbf{0.2645}    & \textbf{0.0092}    & \textbf{0.0577}\\
    \hline
    \end{tabular}}%
  \label{table:Ridge}%
\end{table}%

\begin{table}[!ht]
\renewcommand\arraystretch{1}
\setlength{\tabcolsep}{5pt}
\renewcommand\arraystretch{1.1}
  \centering
  \caption{Performance for Washington DC data with spatial degradation known. (The highest and second-highest values are highlighted in bold and underlined, respectively.)}
  \resizebox{0.6\linewidth}{!}{
  \begin{tabular}{c|c|c|c|c|c|c}\hline

    \hline
    Methods & \multicolumn{1}{c|}{RSNR} & \multicolumn{1}{c|}{SSIM} & \multicolumn{1}{c|}{CC} 
    & \multicolumn{1}{c|}{ERGAS} & \multicolumn{1}{c|}{RMSE} 
    & \multicolumn{1}{c}{SAM}\\ \hline
    \texttt{CNMF}   & 22.54    & 0.9760    & 0.9581    & 0.9821    & 0.0152    & 0.0533\\
    \texttt{FUSE}   & 18.55    & 0.9402    & 0.9453    & 0.8340    & 0.0240    & 0.1098\\
    \texttt{SCOTT}  & 25.08    & 0.9827    & 0.9511    & 0.8691    & 0.0113    & 0.0477\\
    \texttt{STEREO} & 26.61    & 0.9802    & 0.9222    & 1.7985    & 0.0095    & 0.0504\\
    \texttt{SCLL1}  & 27.51    & 0.9892    & \textbf{0.9759}    & \textbf{0.6336}    & 0.0086    & \underline{0.0376}\\
    \texttt{CBSTAR} & 25.46    & 0.9805    & 0.9555    & 0.9540    & 0.0108    & 0.0490\\
    \texttt{BTDvar} & 25.14    & 0.9797    & 0.9649    & 0.7284    & 0.0113    & 0.0554\\
    \texttt{NPTSR}  & \underline{29.80}    & \underline{0.9918}    & \underline{0.9704}    & 1.0179    & \underline{0.0066}    & \textbf{0.0326}\\
    \texttt{NLSTF}  & 25.21    & 0.9795    & 0.9646    & 2.4879    & 0.0112    & 0.0555\\
    \texttt{CLIMB}  & \textbf{30.37} & \textbf{0.9921} & 0.9680 & \underline{0.7011} & \textbf{0.0062} & \textbf{0.0326}\\
    \hline
    \end{tabular}}%
  \label{table:WDC}%
\end{table}%

\subsubsection{Washington DC Dataset}
For this experiment, we use a subimage from the Washington DC dataset collected by the HYDICE sensor\footnote{\url{https://engineering.purdue.edu/~biehl/MultiSpec/hyperspectral.html}},
containing 210 bands in the spectral range of 400-2500\,nm, with the spatial size of $1208\times 307$. 
After removing bands corrupted by atmospheric effects, the resulting SRI $\underline{\bm{Y}}_{\rm S}$ has a size of $120\times 120\times 191$. The HSI $\underline{\bm{Y}}_{\rm H} \in \mathbb{R}^{30\times 30\times 191}$ is generated via spatial degradation, and the MSI $\underline{\bm{Y}}_{\rm M} \in \mathbb{R}^{120\times 120\times 6}$ is obtained using LANDSAT spectral degradation. We set the tensor rank to $R=4$.

Table~\ref{table:WDC} summarizes the performance of all methods. Fig.~\ref{fig:WDC} displays the recovered SRIs, absolute residual images, and SAM maps on the Indian Pines scene. The proposed method better preserves spatial smoothness and sharp edges compared to the baselines. Both quantitative and qualitative results highlight its superior performance.

\begin{figure*}[!t]
\scriptsize\setlength{\tabcolsep}{0.3pt}
\begin{center}
\begin{tabular}{cccccccccccc}
\includegraphics[width=0.085\textwidth]{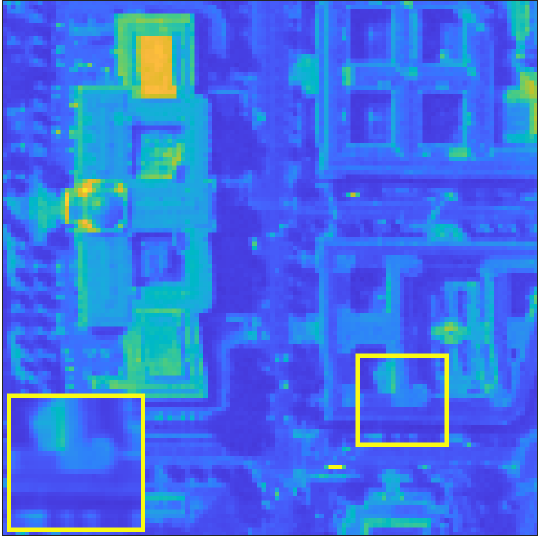}&
\includegraphics[width=0.085\textwidth]{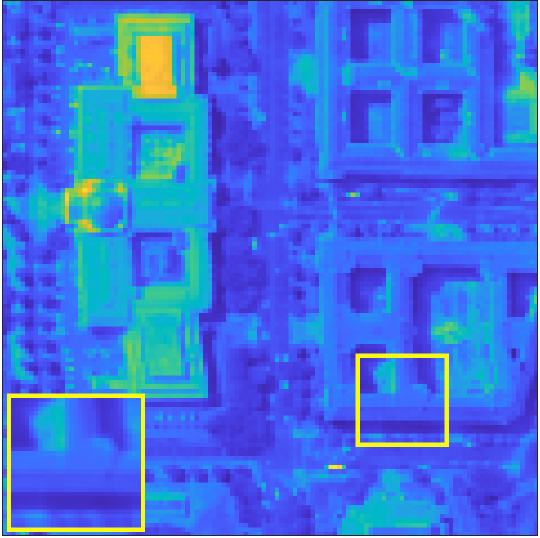}&
\includegraphics[width=0.085\textwidth]{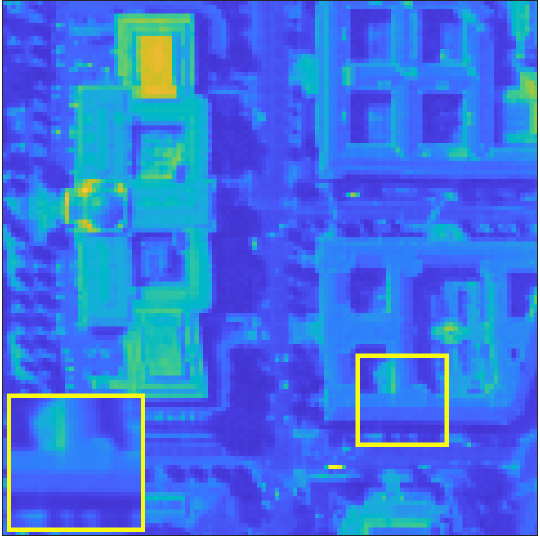}&
\includegraphics[width=0.085\textwidth]{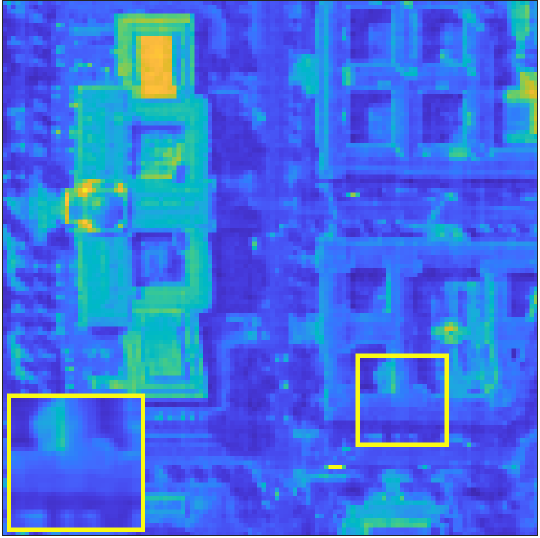}&
\includegraphics[width=0.085\textwidth]{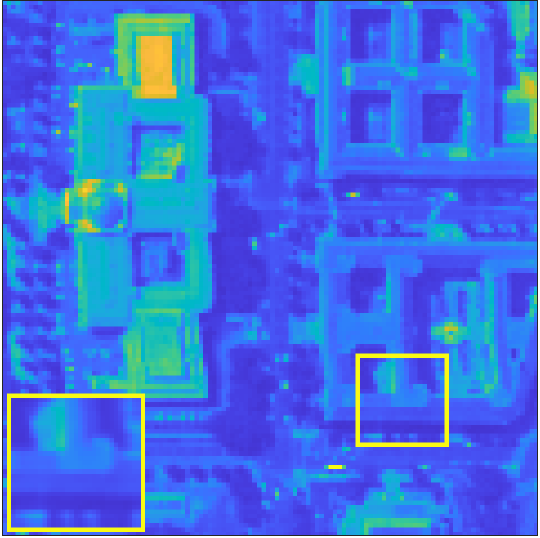}&
\includegraphics[width=0.085\textwidth]{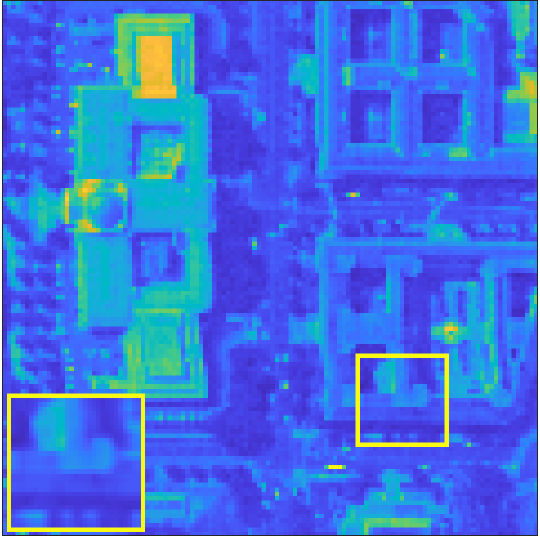}&
\includegraphics[width=0.085\textwidth]{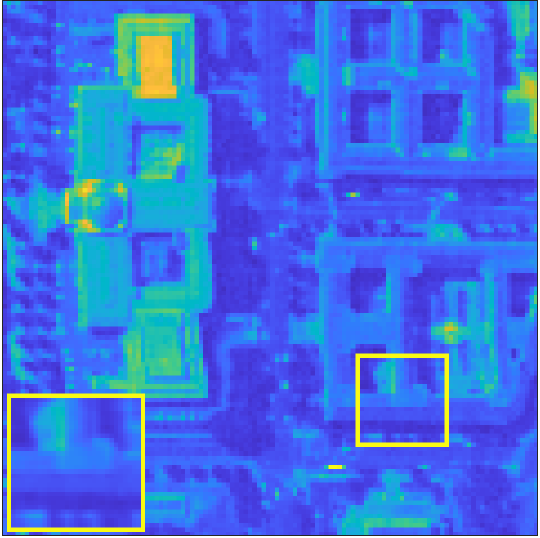}&
\includegraphics[width=0.085\textwidth]{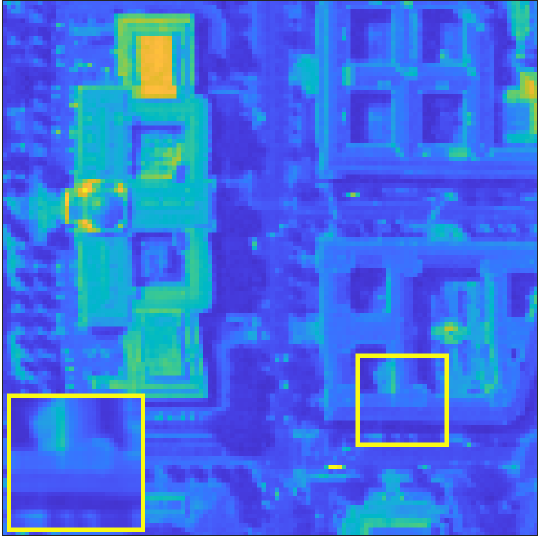}&
\includegraphics[width=0.085\textwidth]{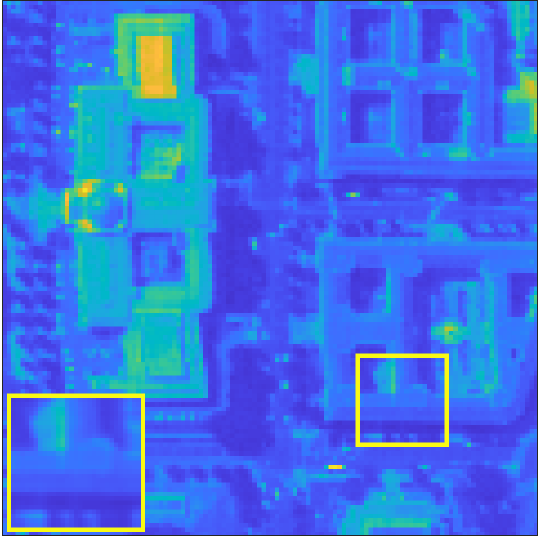}&
\includegraphics[width=0.085\textwidth]{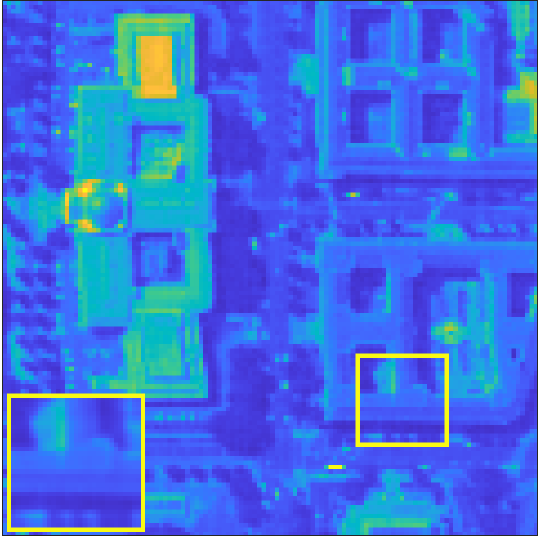}&
\includegraphics[width=0.099\textwidth]{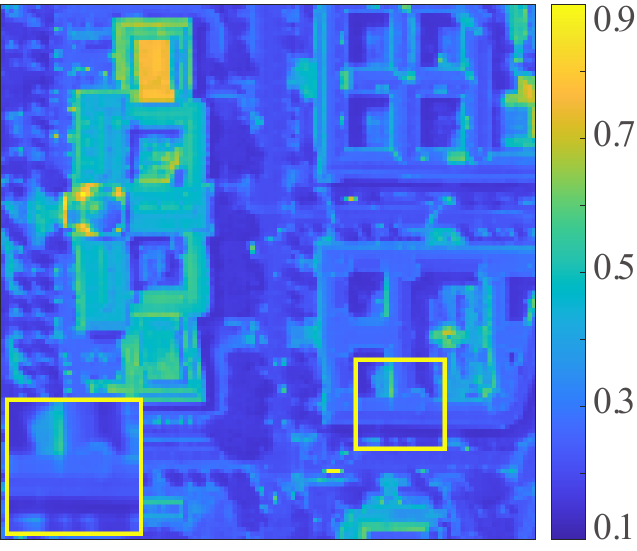}\\
\includegraphics[width=0.085\textwidth]{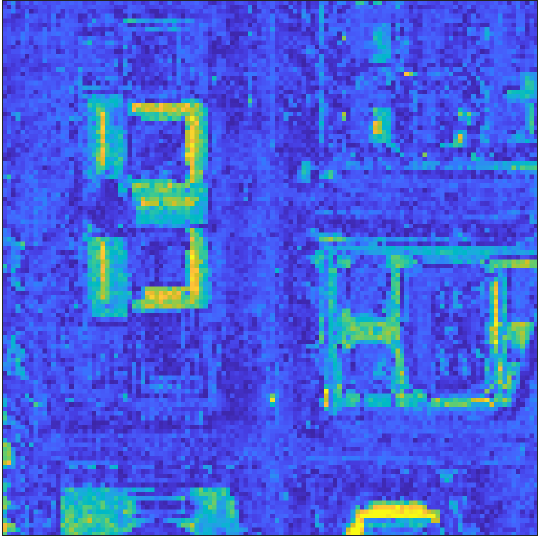}&
\includegraphics[width=0.085\textwidth]{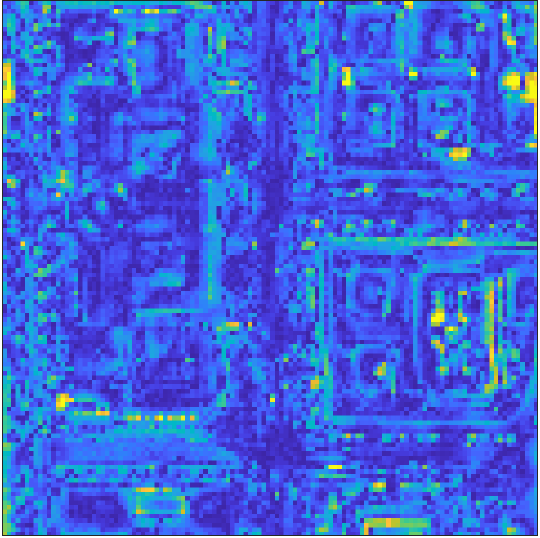}&
\includegraphics[width=0.085\textwidth]{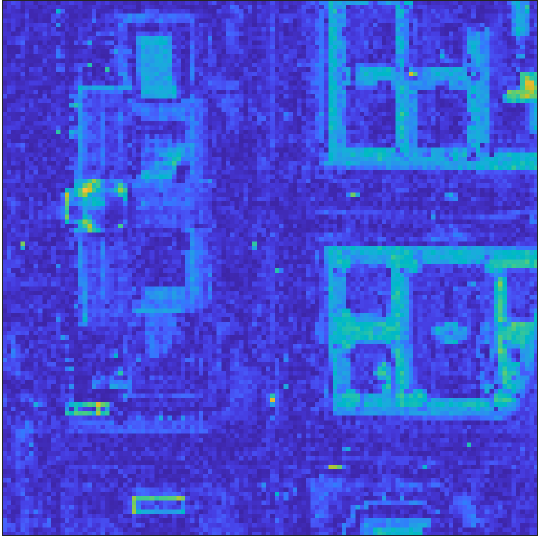}&
\includegraphics[width=0.085\textwidth]{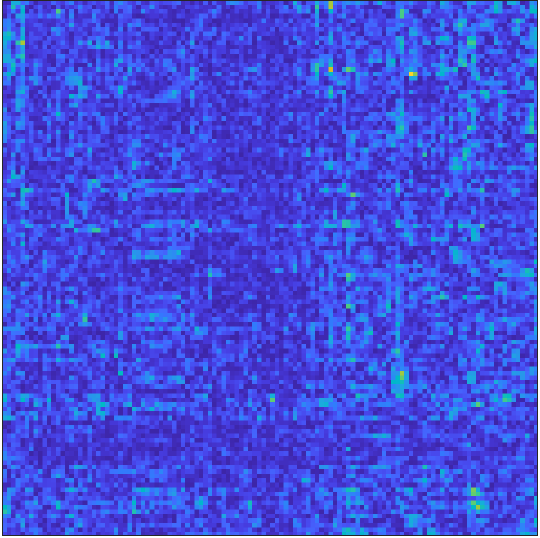}&
\includegraphics[width=0.085\textwidth]{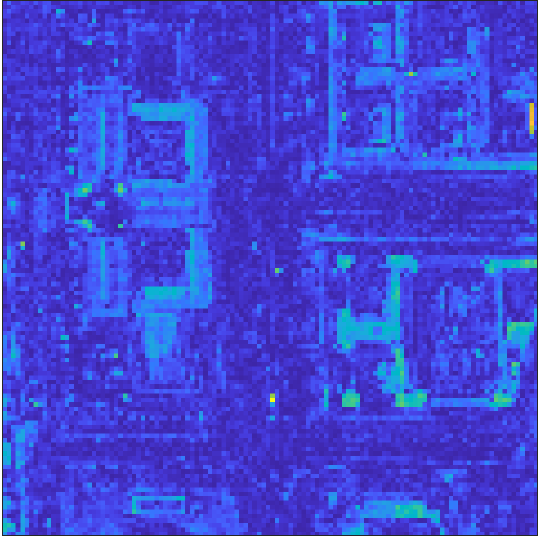}&
\includegraphics[width=0.085\textwidth]{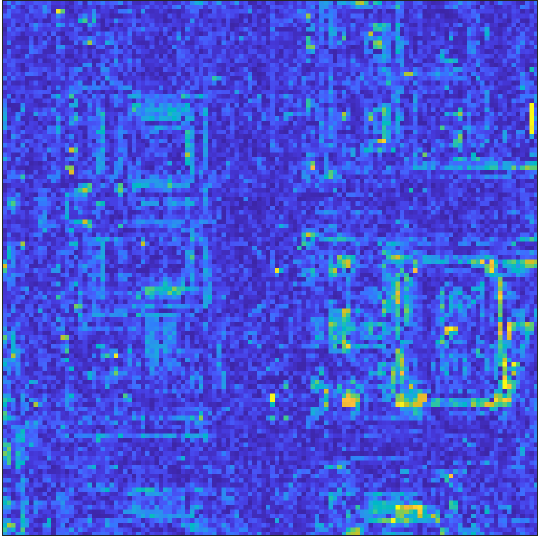}&
\includegraphics[width=0.085\textwidth]{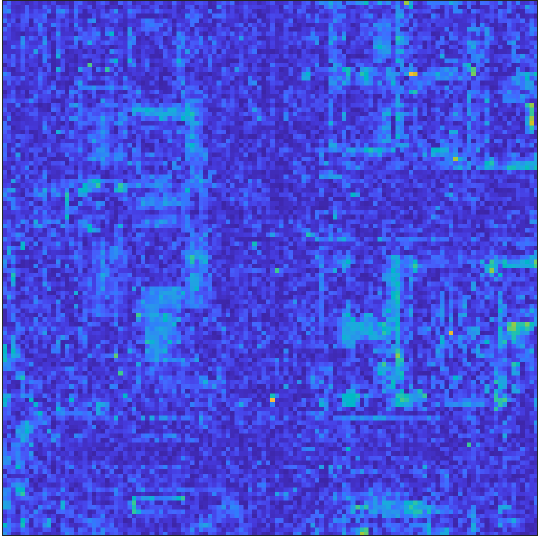}&
\includegraphics[width=0.085\textwidth]{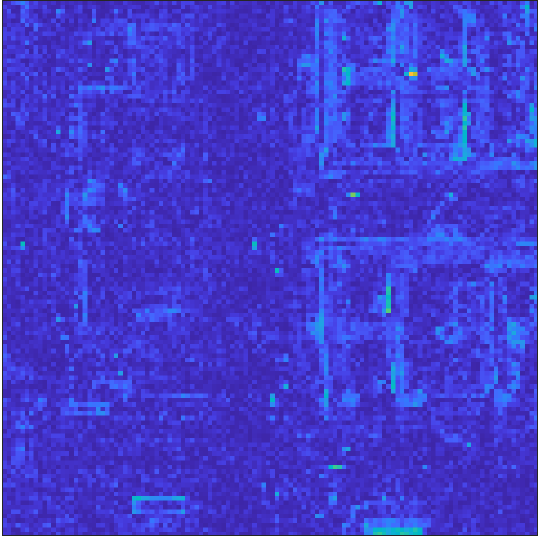}&
\includegraphics[width=0.085\textwidth]{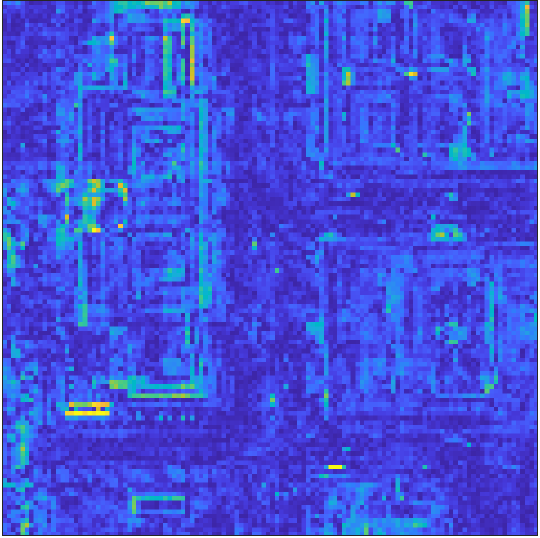}&
\includegraphics[width=0.085\textwidth]{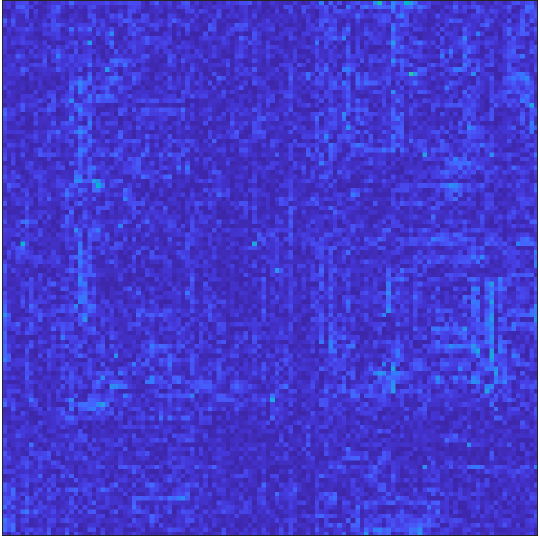}&
\includegraphics[width=0.102\textwidth]{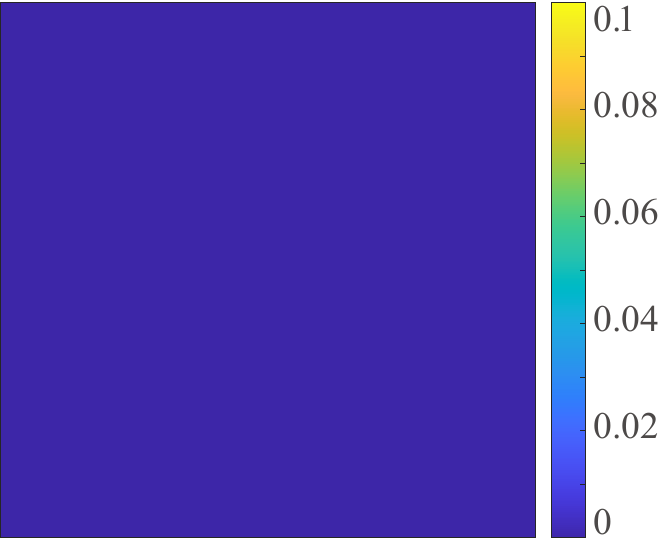}\\
\includegraphics[width=0.085\textwidth]{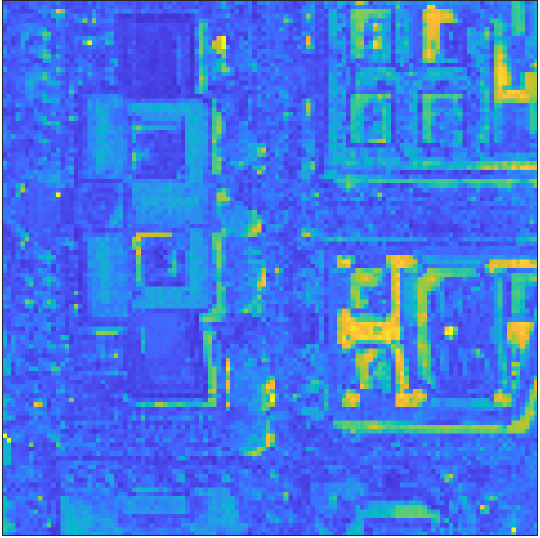}&
\includegraphics[width=0.085\textwidth]{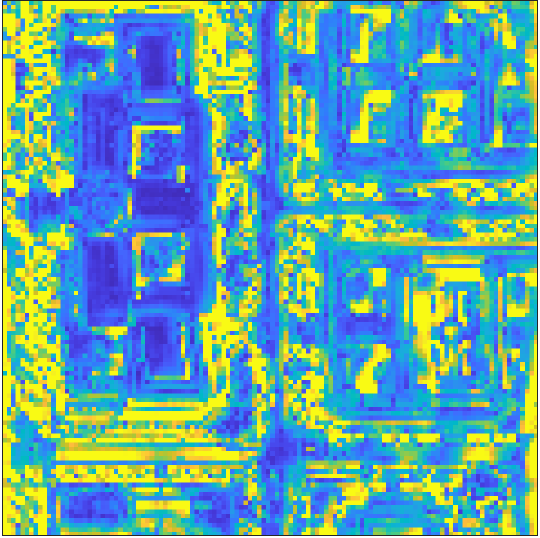}&
\includegraphics[width=0.085\textwidth]{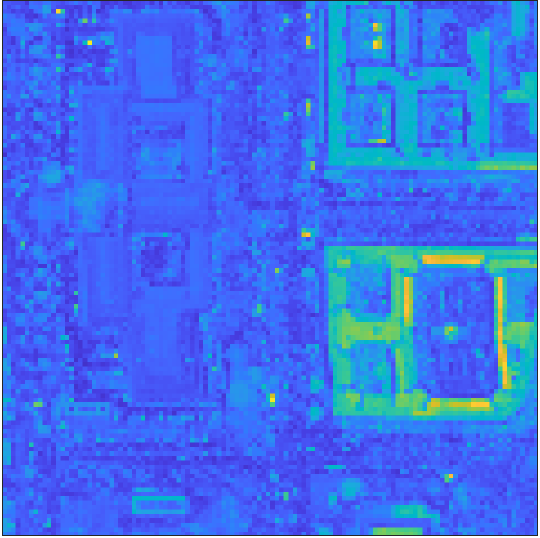}&
\includegraphics[width=0.085\textwidth]{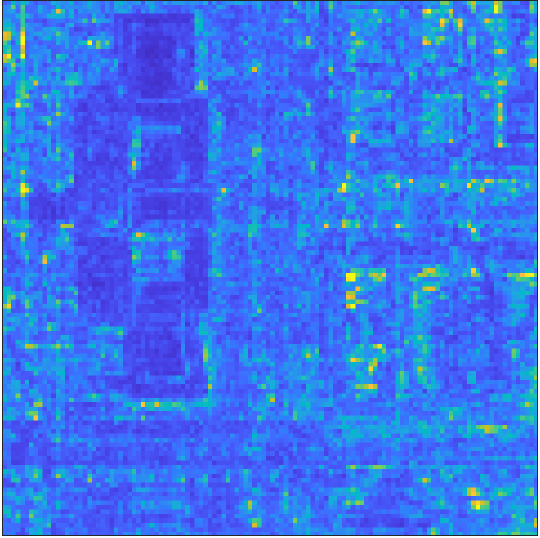}&
\includegraphics[width=0.085\textwidth]{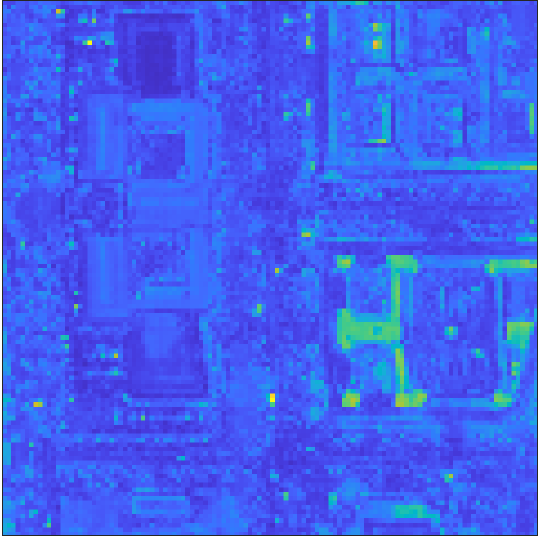}&
\includegraphics[width=0.085\textwidth]{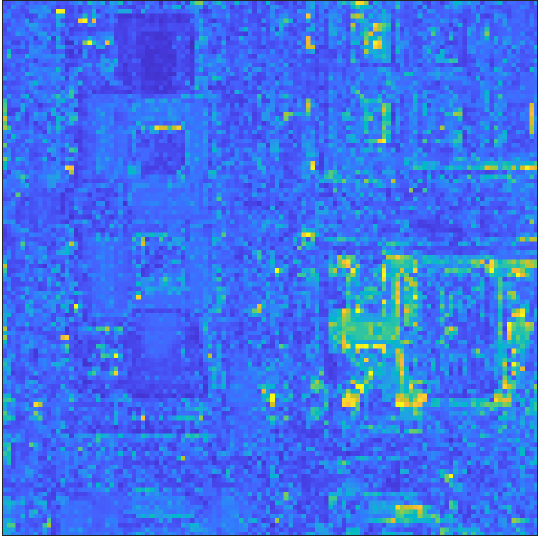}&
\includegraphics[width=0.085\textwidth]{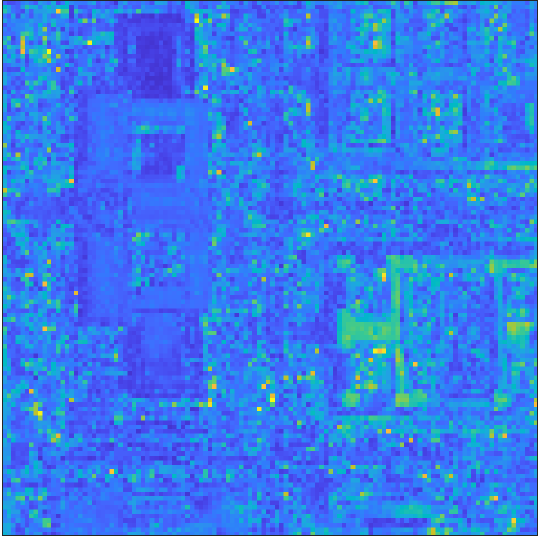}&
\includegraphics[width=0.085\textwidth]{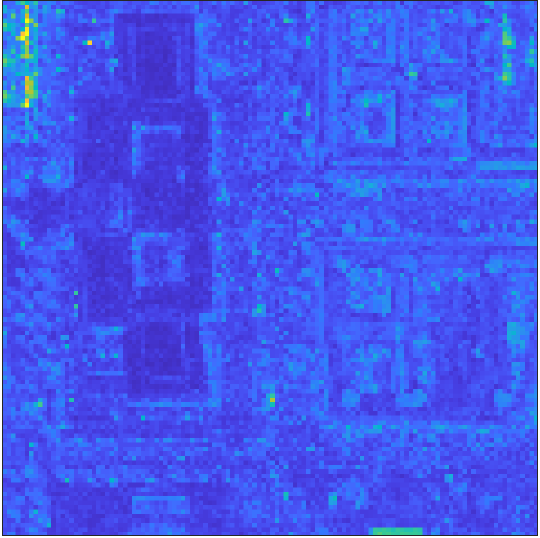}&
\includegraphics[width=0.085\textwidth]{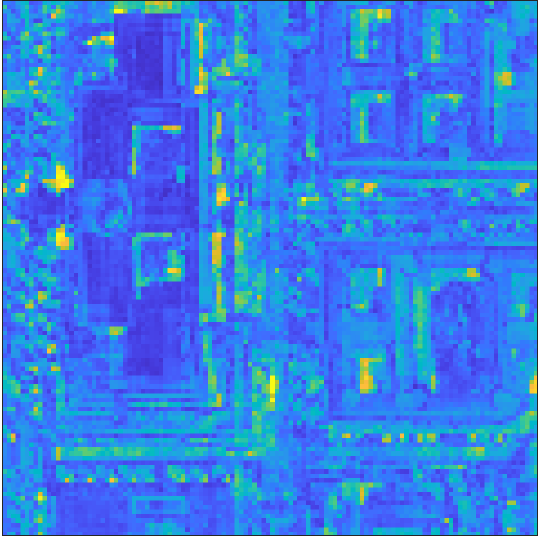}&
\includegraphics[width=0.085\textwidth]{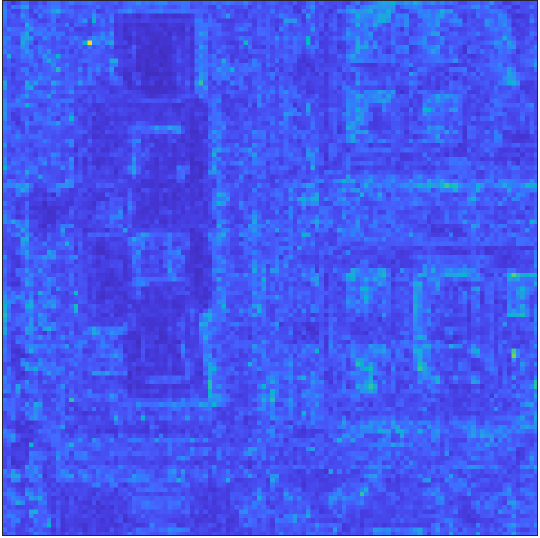}&
\includegraphics[width=0.099\textwidth]{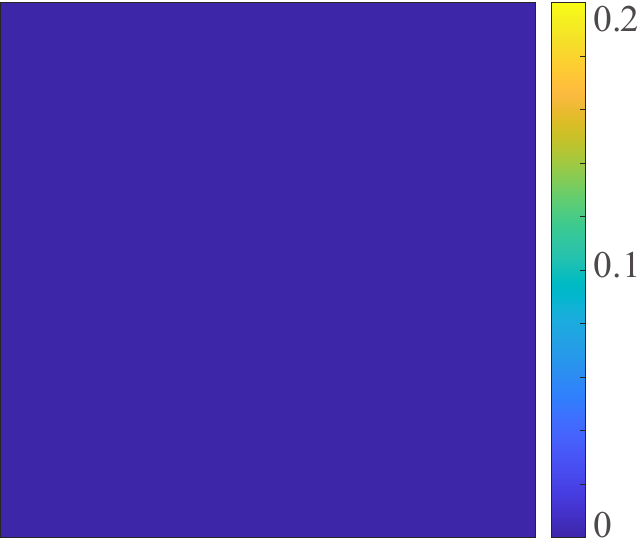}\\
(a) \texttt{CNMF} & (b) \texttt{FUSE} & (c) \texttt{SCOTT} & (d) \texttt{STEREO} & (e) \texttt{SCLL1} & (f) \texttt{CBSTAR}& (g) \texttt{BTDvar} & (h) \texttt{NPTSR} & (i) \texttt{NLSTF} & (j) \texttt{CLIMB} & (k) SRI\\
\end{tabular}
\caption{The recovered results of Washington DC with the degradation known.
First row: the recovered SRIs of the 28th band; 
Second row: the corresponding residual images of the 28th band; 
Third row: the SAM maps.}
  \label{fig:WDC}
\end{center}
\end{figure*}

\subsubsection{Pavia University Dataset}
Here, we use a subimage from the Pavia University dataset, captured by the ROSIS sensor \cite{Kunkel1988ROSIS}. 
The original dataset contains $610 \times 340$ pixels and 115 bands.
After removing bands affected by water absorption, the used SRI and HSI have sizes of $200\times 200\times 103$ and $50\times 50\times 103$, respectively. The MSI $\underline{\bm{Y}}_{\rm M} \in \mathbb{R}^{200\times 200\times 4}$ is generated using QuickBird spectral degradation. We set $R=4$.

\begin{table}[!t]
\renewcommand\arraystretch{1}
\setlength{\tabcolsep}{5pt}
\renewcommand\arraystretch{1.1}
  \centering
  \caption{Performance for Pavia University data with spatial degradation known. (The highest and second-highest values are highlighted in bold and underlined, respectively.)}
  \resizebox{0.6\linewidth}{!}{
  \begin{tabular}{c|c|c|c|c|c|c}\hline

    \hline
    Methods & \multicolumn{1}{c|}{RSNR} & \multicolumn{1}{c|}{SSIM} & \multicolumn{1}{c|}{CC} 
    & \multicolumn{1}{c|}{ERGAS} & \multicolumn{1}{c|}{RMSE} 
    & \multicolumn{1}{c}{SAM}\\ \hline
    \texttt{CNMF}   & 20.40    & 0.9489    & 0.9799    & 0.5438    & 0.0216    & 0.0785\\
    \texttt{FUSE}   & 20.32    & 0.9389    & 0.9777    & 0.5321    & 0.0217    & 0.0811\\
    \texttt{SCOTT}  & 24.60    & 0.9575    & 0.9906    & 0.3160    & 0.0132    & 0.0612\\
    \texttt{STEREO} & 25.40    & 0.9586    & 0.9919    & 0.2919    & 0.0121    & 0.0623\\
    \texttt{SCLL1}  & 26.97    & \underline{0.9776}    & 0.9942    & 0.2456    & 0.0101   & 0.0502\\
    \texttt{CBSTAR} & 18.37    & 0.8499    & 0.9639    & 0.6096    & 0.0272    & 0.0941\\
    \texttt{BTDvar} & 25.20    & 0.9623    & 0.9917    & 0.3046    & 0.0124    & 0.0611\\
    \texttt{NPTSR}  & \underline{28.10}    & \textbf{0.9788}    & \underline{0.9956}    & \underline{0.2171}    & \underline{0.0089}    & \textbf{0.0448}\\
    \texttt{NLSTF}  & 26.85    & 0.9766    & 0.9942    & 0.2500    & 0.0102    & \underline{0.0491}\\
    \texttt{CLIMB}  & \textbf{28.19}    & 0.9773    & \textbf{0.9957}    & \textbf{0.2152}    & \textbf{0.0088}    & \textbf{0.0448}\\
    \hline
    \end{tabular}}%
  \label{table:Pavia}%
\end{table}%

\begin{figure}[!t]
\scriptsize\setlength{\tabcolsep}{0.3pt}
\begin{center}
\begin{tabular}{cccc}
\includegraphics[width=0.248\textwidth]{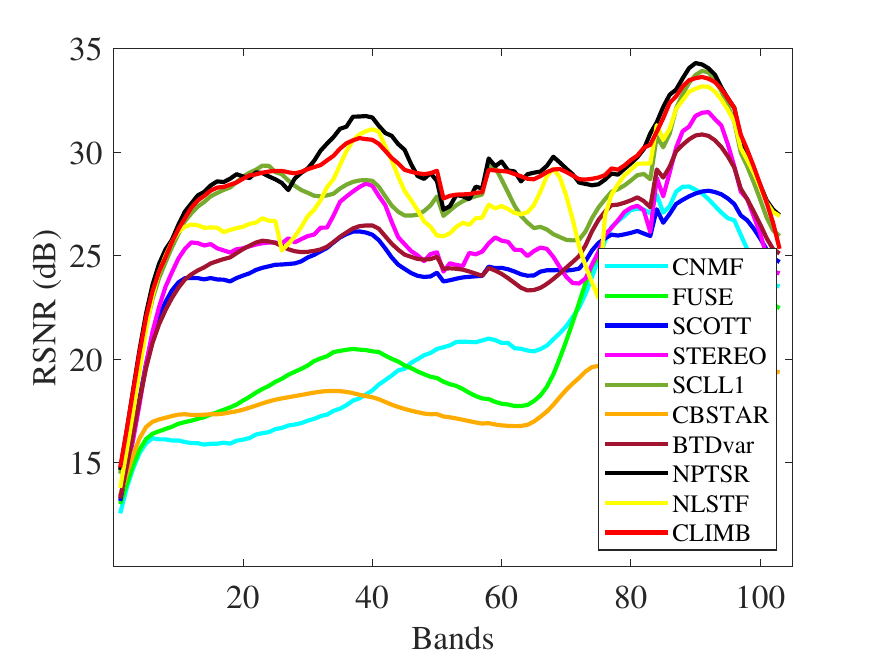}&
\includegraphics[width=0.248\textwidth]{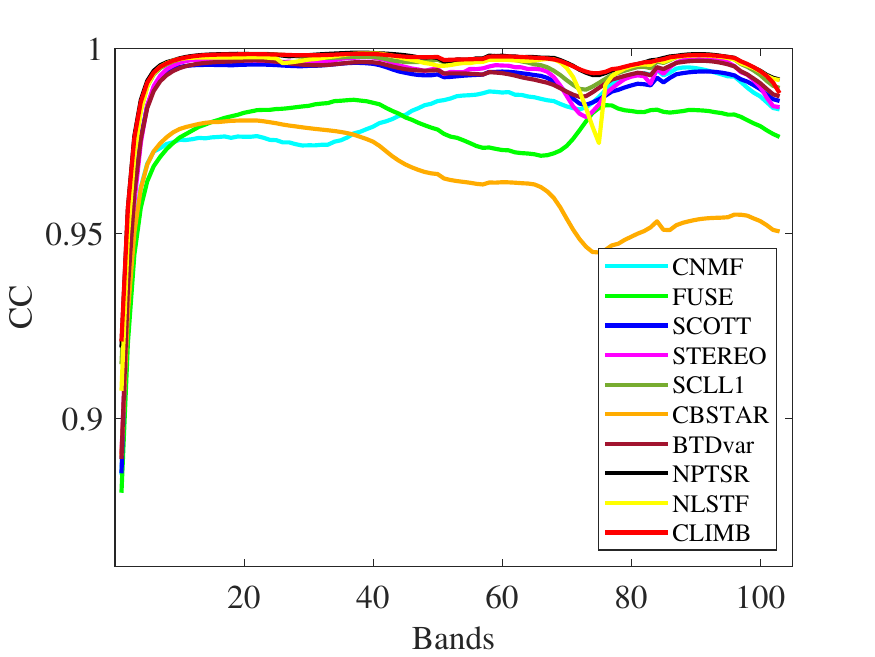}&
\includegraphics[width=0.248\textwidth]{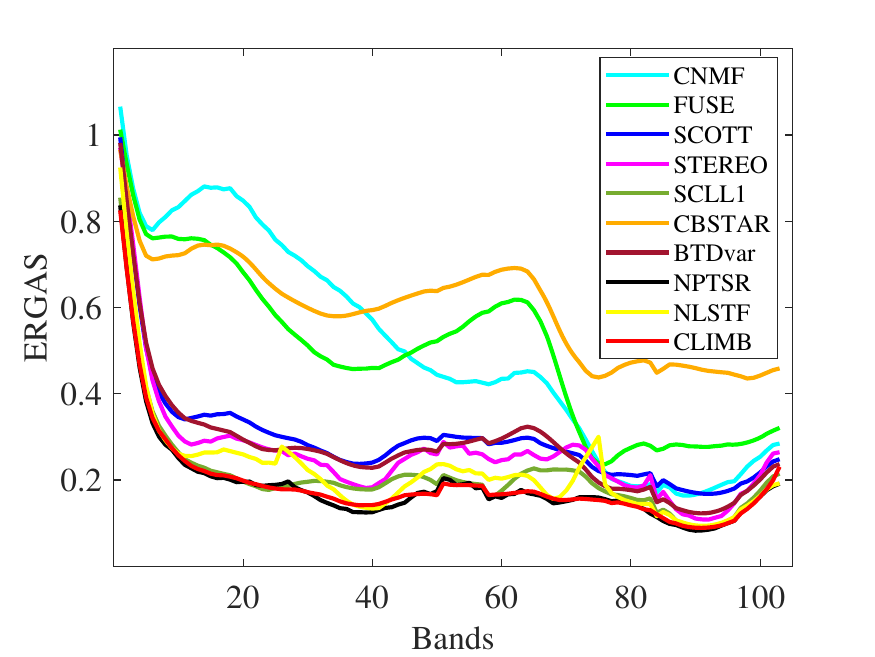}&
\includegraphics[width=0.248\textwidth]{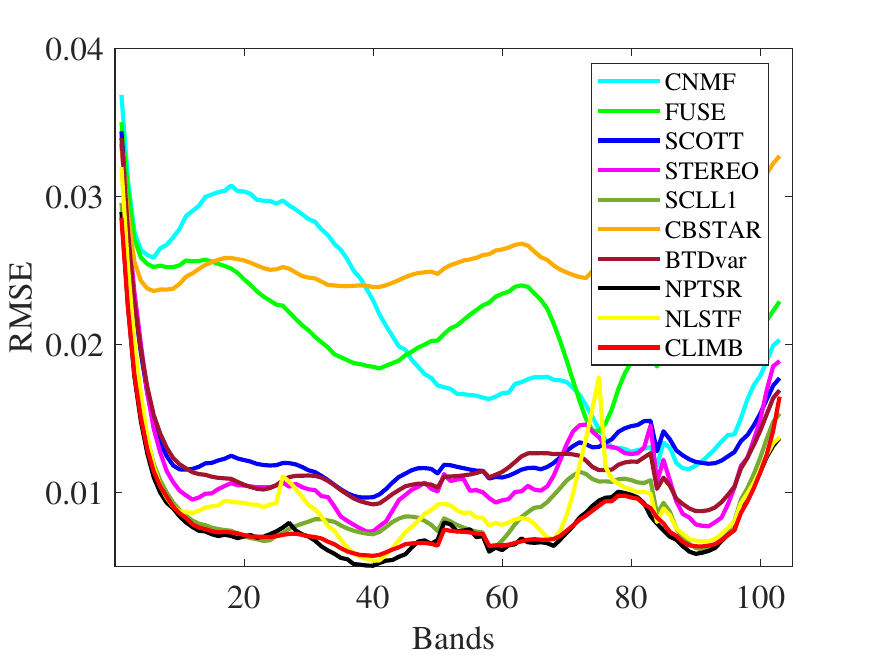}\\
 (a) RSNR & (b) CC & (c) ERGAS & (d) RMSE\\
\end{tabular}
\caption{RSNR, CC, ERGAS, and RMSE values of each band of Pavia University image.}
\vspace{-3mm}
  \label{fig:Pavia_band}
\end{center}
\end{figure}

Table~\ref{table:Pavia} reports the reconstruction performance under $\mathrm{SNR}=35\,\mathrm{dB}$, where the proposed \texttt{CLIMB} again outperforms baselines across most metrics. Fig.~\ref{fig:Pavia_band} plots RSNR, CC, ERGAS, and RMSE across spectral bands.

We note that the nonlocal prior-based method NPTSR and our method achieve comparable performance under the setting where spatial degradation operators are known. Nonetheless, our performance margin becomes more articulated in more challenging scenarios, as will be shown shortly in the next subsection.

\begin{figure*}[!t]
\scriptsize\setlength{\tabcolsep}{0.3pt}
\begin{center}
\begin{tabular}{cccccccccccc}
\includegraphics[width=0.084\textwidth]{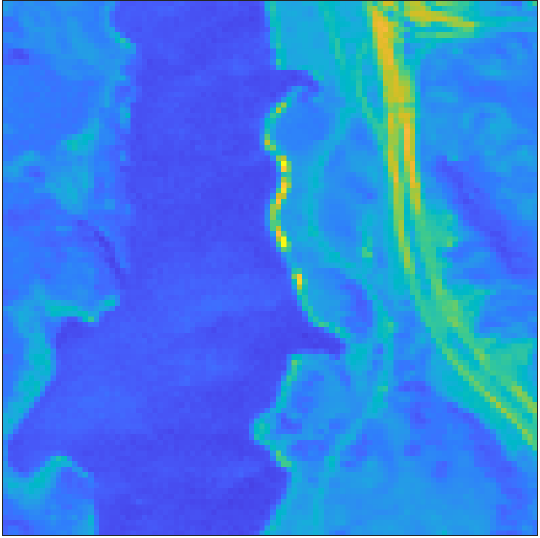}&
\includegraphics[width=0.084\textwidth]{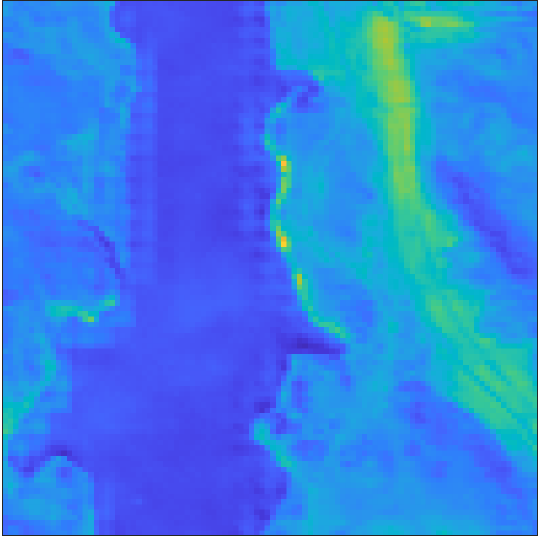}&
\includegraphics[width=0.084\textwidth]{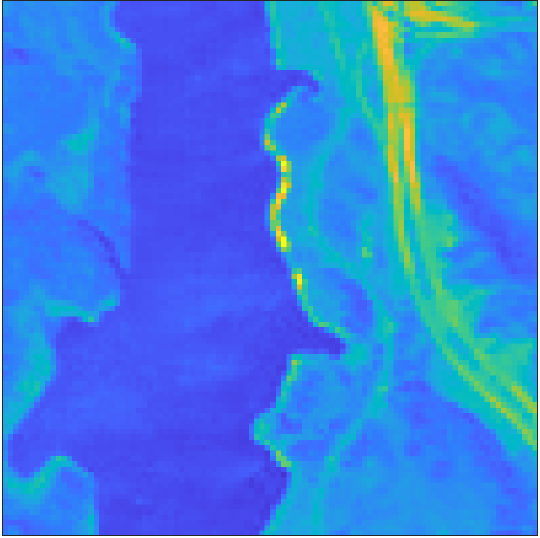}&
\includegraphics[width=0.084\textwidth]{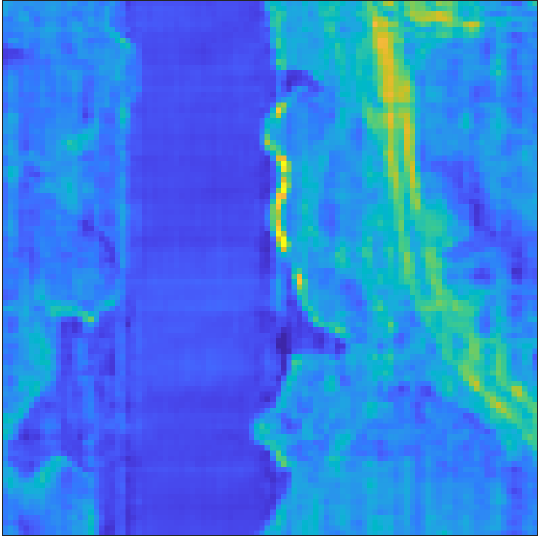}&
\includegraphics[width=0.084\textwidth]{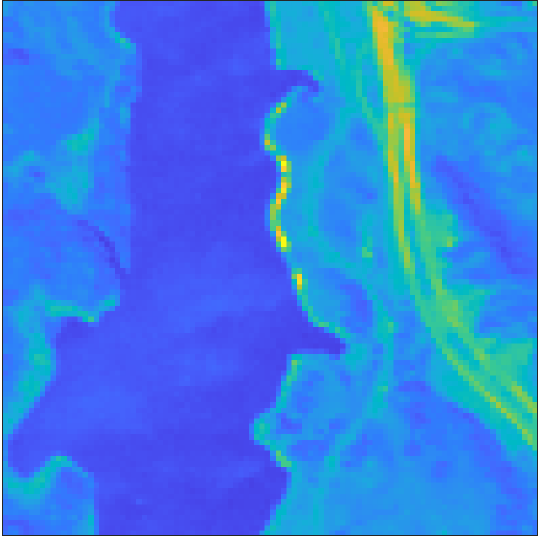}&
\includegraphics[width=0.084\textwidth]{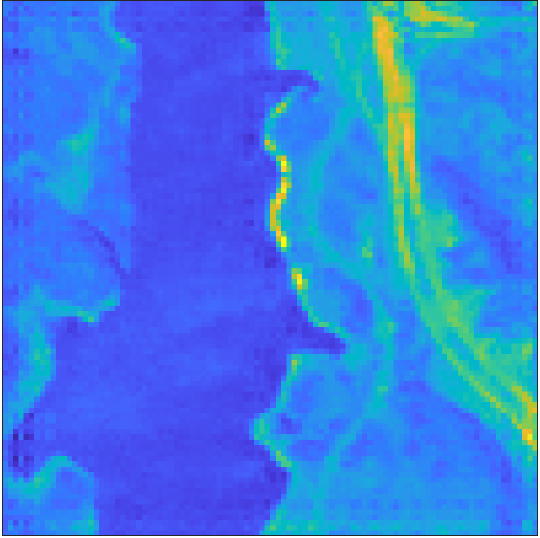}&
\includegraphics[width=0.084\textwidth]{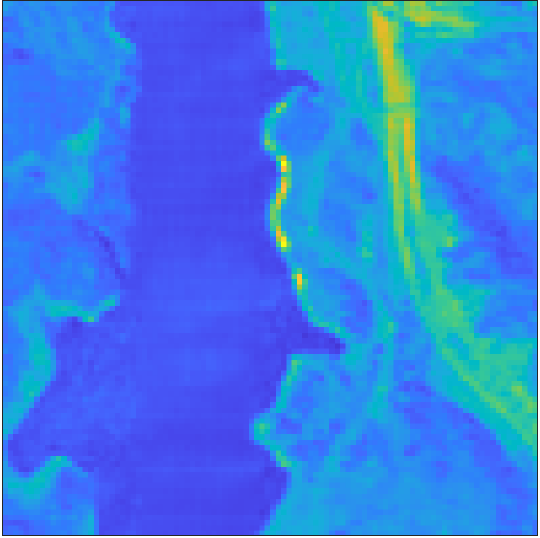}&
\includegraphics[width=0.084\textwidth]{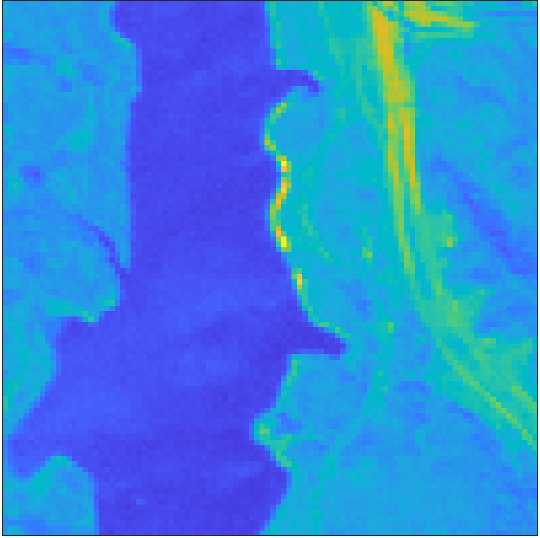}&
\includegraphics[width=0.084\textwidth]{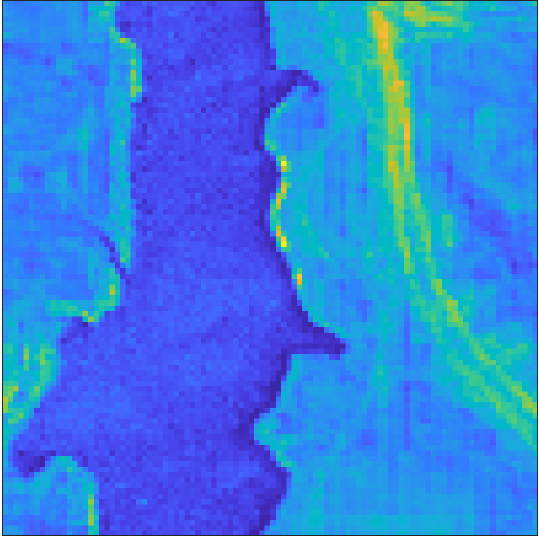}&
\includegraphics[width=0.084\textwidth]{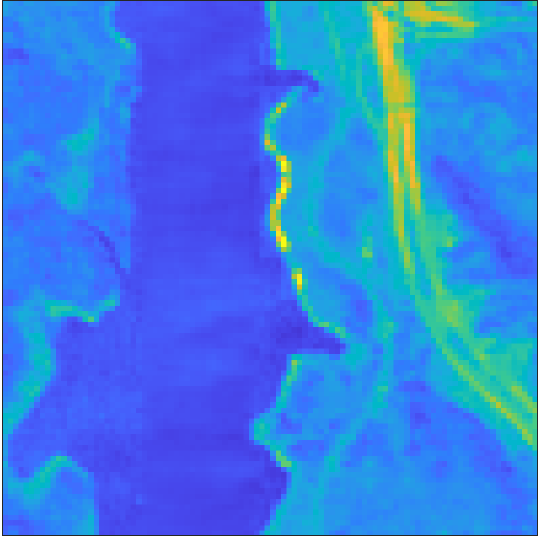}&
\includegraphics[width=0.101\textwidth]{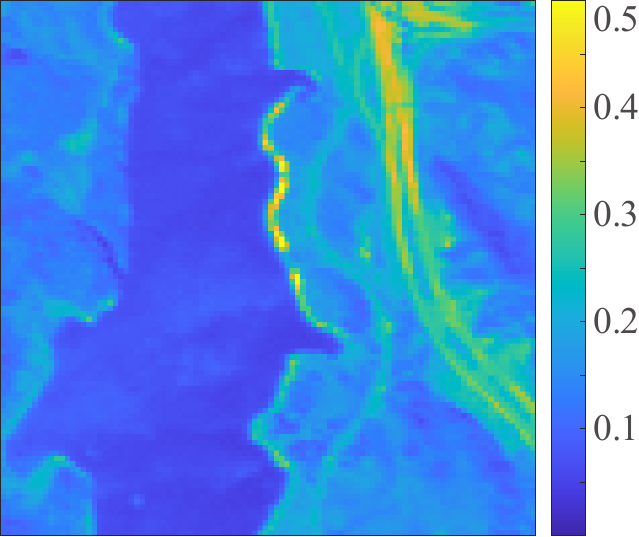}\\
\includegraphics[width=0.084\textwidth]{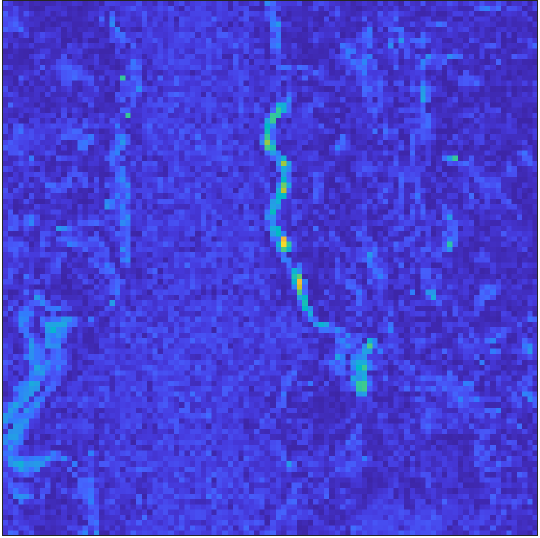}&
\includegraphics[width=0.084\textwidth]{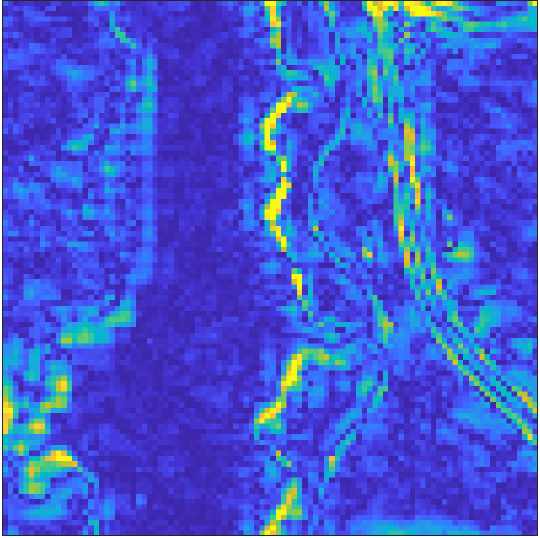}&
\includegraphics[width=0.084\textwidth]{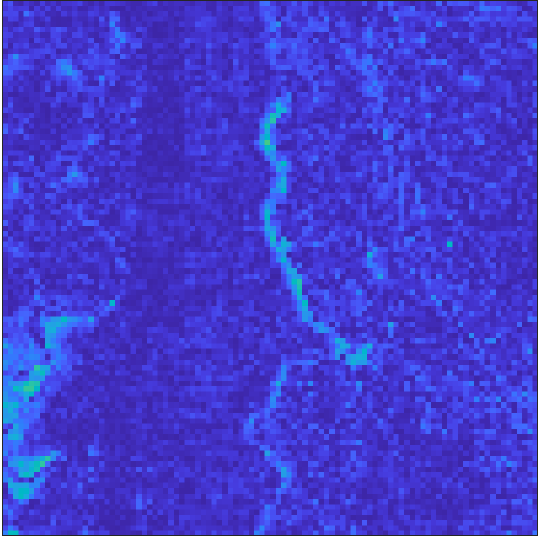}&
\includegraphics[width=0.084\textwidth]{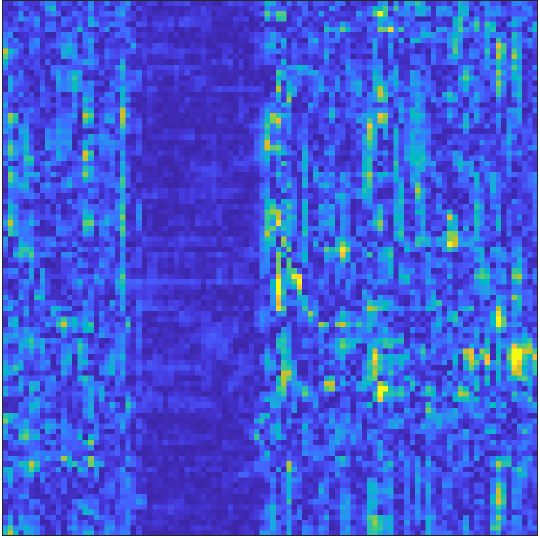}&
\includegraphics[width=0.084\textwidth]{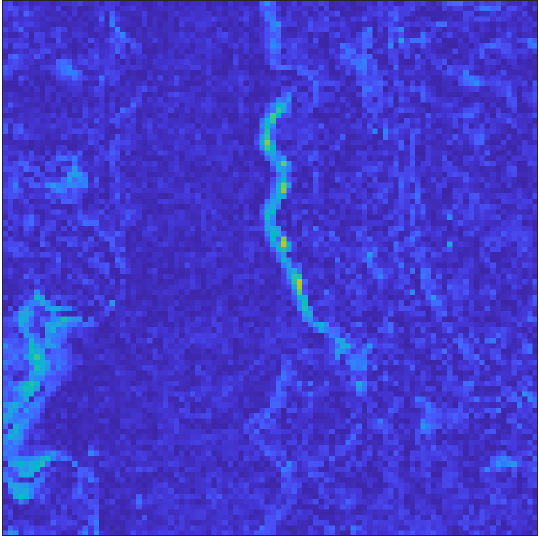}&
\includegraphics[width=0.084\textwidth]{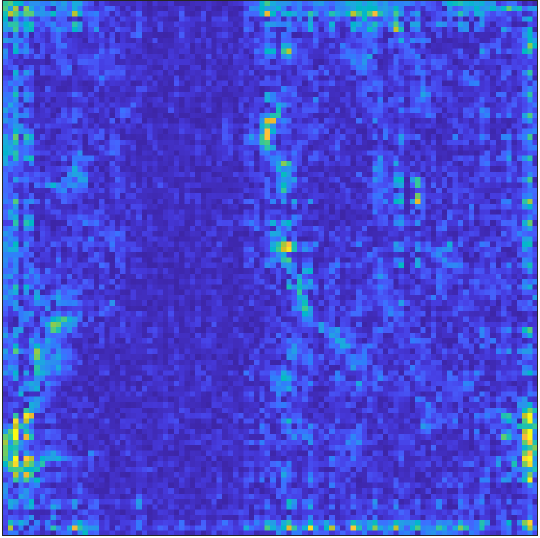}&
\includegraphics[width=0.084\textwidth]{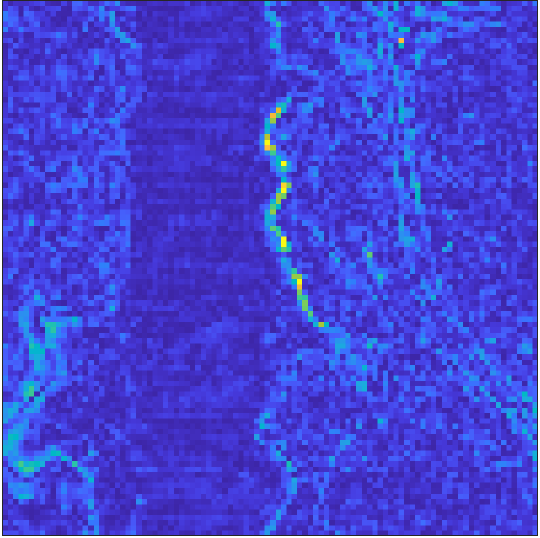}&
\includegraphics[width=0.084\textwidth]{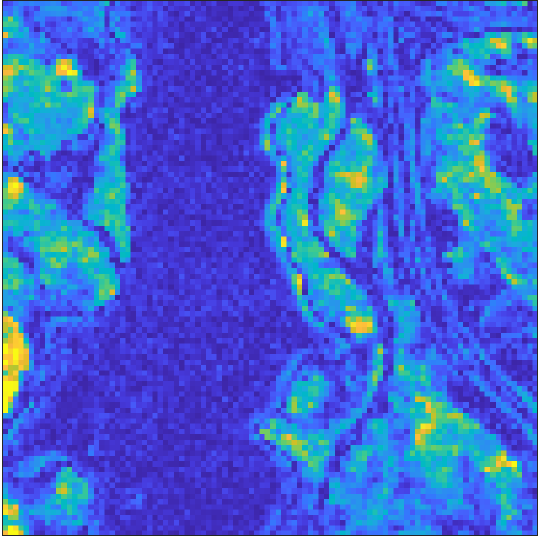}&
\includegraphics[width=0.084\textwidth]{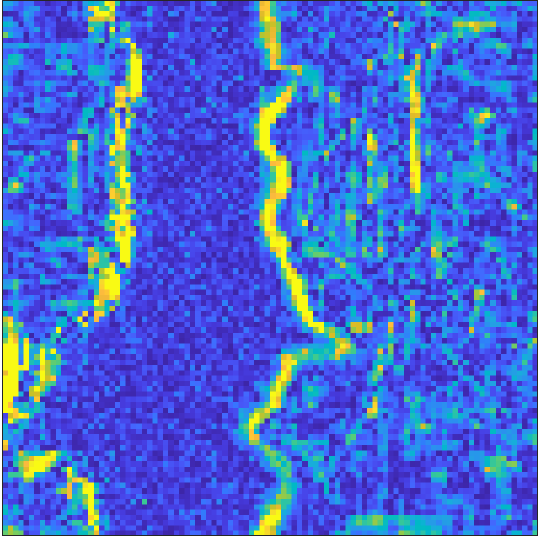}&
\includegraphics[width=0.084\textwidth]{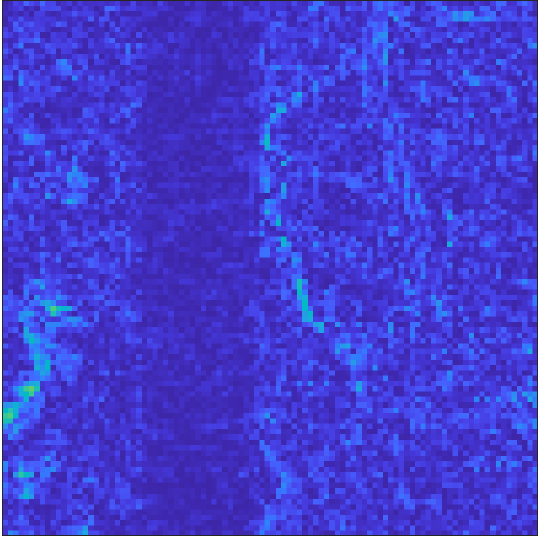}&
\includegraphics[width=0.103\textwidth]{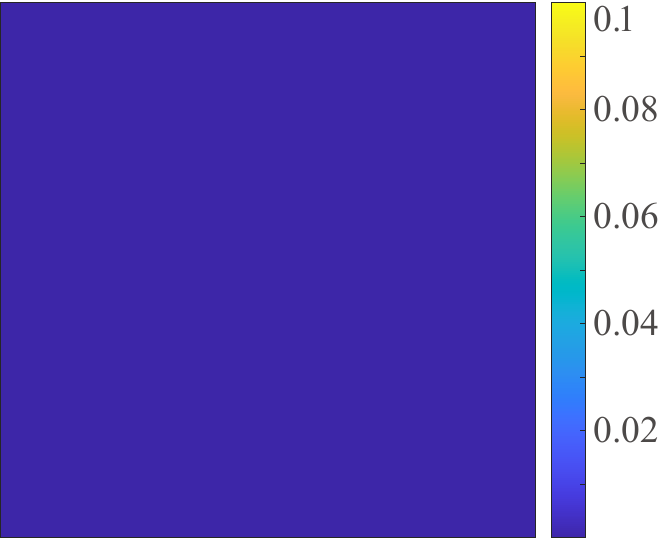}\\
\includegraphics[width=0.084\textwidth]{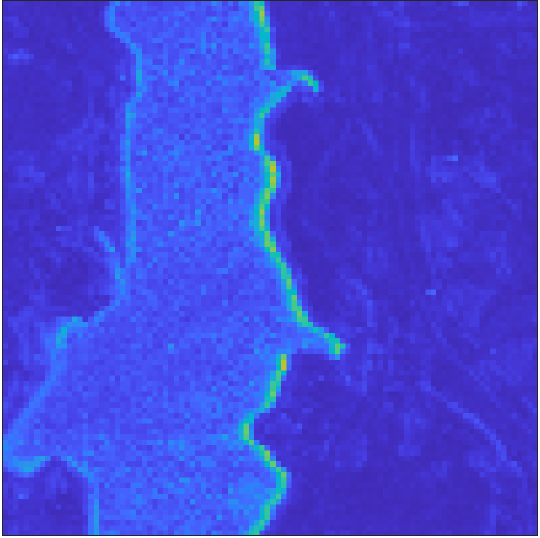}&
\includegraphics[width=0.084\textwidth]{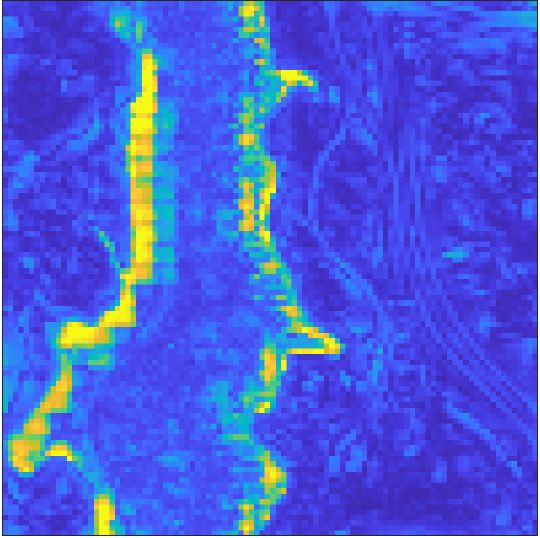}&
\includegraphics[width=0.084\textwidth]{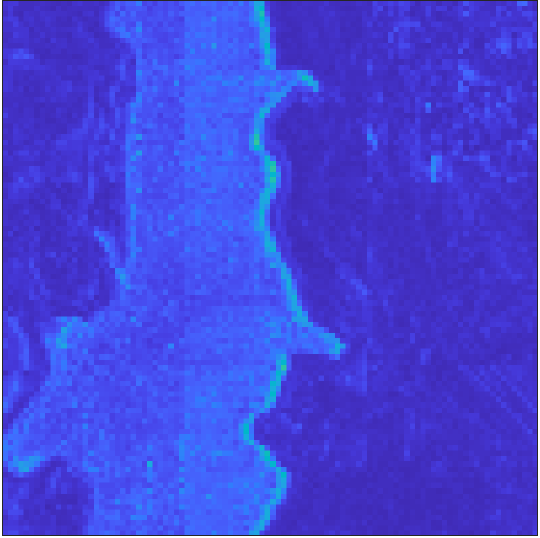}&
\includegraphics[width=0.084\textwidth]{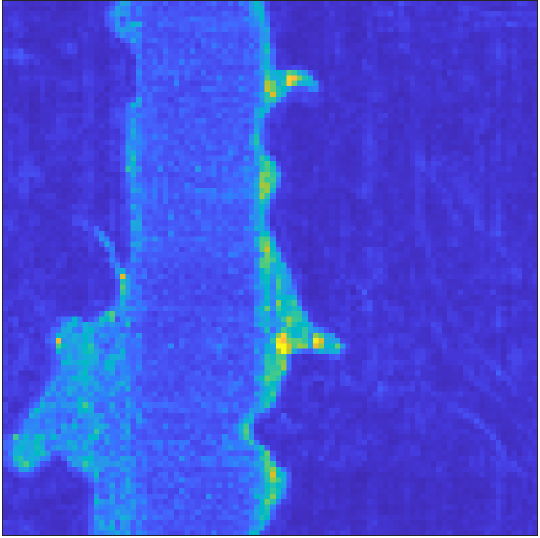}&
\includegraphics[width=0.084\textwidth]{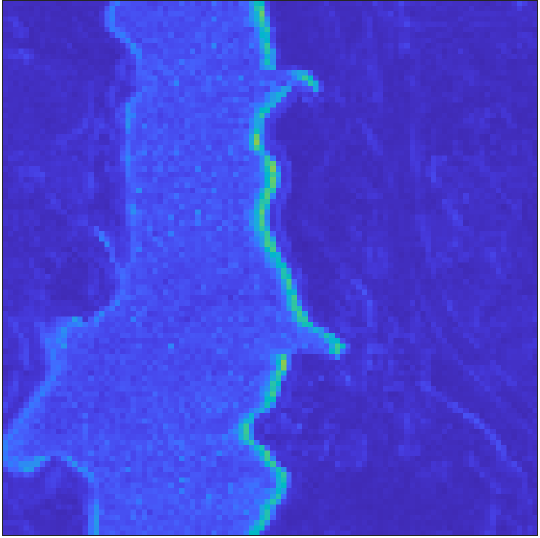}&
\includegraphics[width=0.084\textwidth]{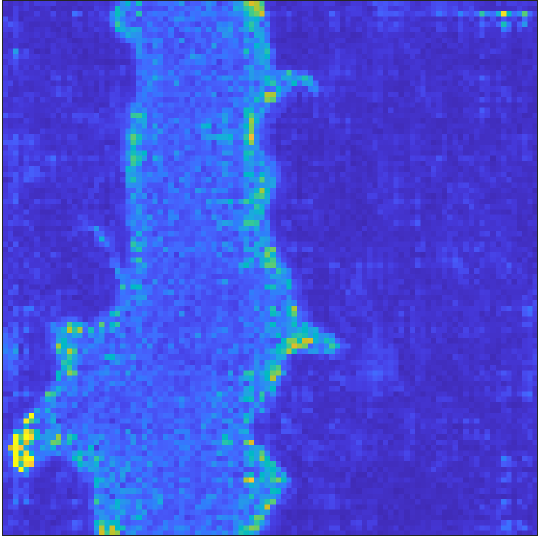}&
\includegraphics[width=0.084\textwidth]{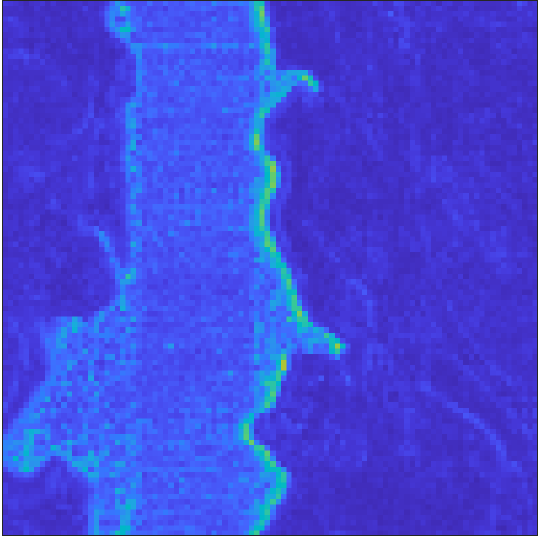}&
\includegraphics[width=0.084\textwidth]{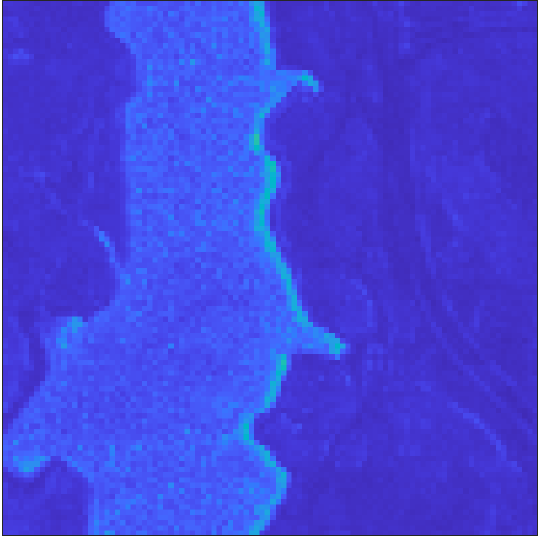}&
\includegraphics[width=0.084\textwidth]{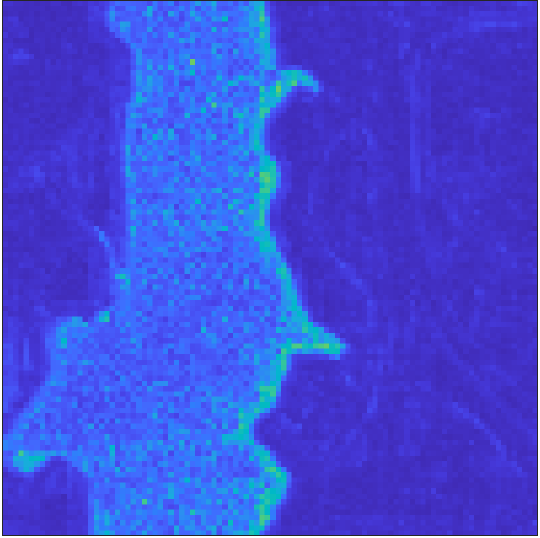}&
\includegraphics[width=0.084\textwidth]{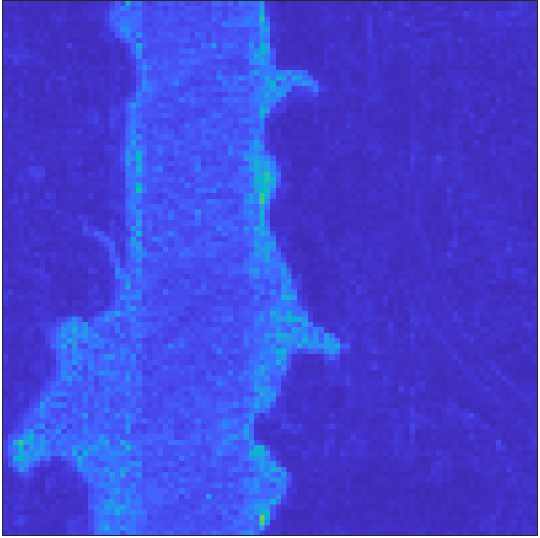}&
\includegraphics[width=0.101\textwidth]{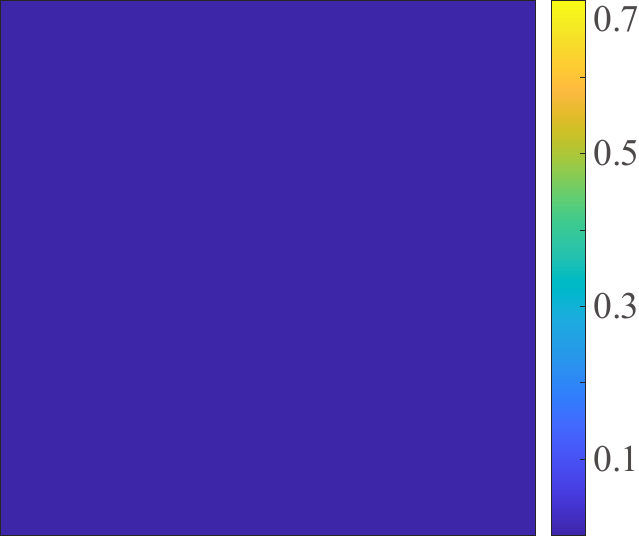}\\
  (a) \texttt{CNMF} & (b) \texttt{FUSE} & (c) \texttt{BSCOTT} & (d) \texttt{BSTEREO} & (e) \texttt{BSCLL1} & (f) \texttt{CBSTAR} & (g) \texttt{BTDvar} & (h) \texttt{NPTSR} & (i) \texttt{NLSTF} & (j) \texttt{BCLIMB} & (k) HSI\\
\end{tabular}
\caption{The recovered results of Jasper Ridge with the spatial degradation unknown.
First row: the recovered SRIs of the 34th band; 
Second row: the corresponding residual images of the 34th band; 
Third row: the SAM maps.}
  \label{fig:Ridge_blind}
\end{center}
\end{figure*}

\begin{table}[!htp]
\renewcommand\arraystretch{1}
\setlength{\tabcolsep}{5pt}%
\renewcommand\arraystretch{1.1}%
  \centering
  \caption{Performance for Jasper Ridge data with the spatial degradation unknown. (The highest and second-highest values are highlighted in bold and underlined, respectively.)}
  \resizebox{0.6\linewidth}{!}{
  \begin{tabular}{c|c|c|c|c|c|c}\hline

    \hline
    Methods & \multicolumn{1}{c|}{RSNR} & \multicolumn{1}{c|}{SSIM} & \multicolumn{1}{c|}{CC} 
    & \multicolumn{1}{c|}{ERGAS} & \multicolumn{1}{c|}{RMSE} 
    & \multicolumn{1}{c}{SAM}\\ \hline
    \texttt{CNMF}   & 25.52    & 0.9639    & 0.9903   & 0.3589    & 0.0155    & 0.0738 \\
    \texttt{FUSE}   & 16.91    & 0.8203    & 0.9520    & 0.8230    & 0.0414    & 0.1312\\
    \texttt{BSCOTT}  & 27.17    & 0.9724    & 0.9926    & \underline{0.3160}    & 0.0127    & 0.0682\\
    \texttt{BSTEREO} & 24.33    & 0.9459    & 0.9855    & 0.4754    & 0.0176    & 0.0883\\
    \texttt{BSCLL1}  & \underline{27.62}    & \textbf{0.9785}    & \underline{0.9928}    & 0.3168    & \underline{0.0121}    & \textbf{0.0638}\\
    \texttt{CBSTAR} & 19.68    & 0.9393    & 0.9794    & 0.5524    & 0.0301    & 0.0899\\
    \texttt{BTDvar} & 24.95    & 0.9602    & 0.9890    & 0.3812    & 0.0164    & 0.0800\\
    \texttt{NPTSR}  & 27.02    & \underline{0.9744}    & 0.9917   & 0.5748    & 0.0129    & 0.0692\\
    \texttt{NLSTF}  & 26.92    & 0.9647    & 0.9890   & 0.5587    & 0.0131    & 0.0809\\
    \texttt{BCLIMB}  & \textbf{28.28}    & 0.9716    & \textbf{0.9929}    & \textbf{0.2973}    & \textbf{0.0112}    & \underline{0.0668}\\
    \hline
    \end{tabular}}%
  \label{table:Ridge_blind}%
\end{table}%

\subsection{Semi-real Experiments with Unknown Spatial Degradation Operators}
In this subsection, we evaluate the proposed semi-blind \texttt{CLIMB} model (i.e., \eqref{eq:model-blind}, referred to as \texttt{BCLIMB}) under the setting where the spatial degradation operator is unknown. The experiments use the Jasper Ridge and Pavia University datasets with the same configurations as before. 

Since CP-, Tucker-, and LL1-based CTD methods can also operate in the semi-blind setting, we include their corresponding variants---\texttt{BSTEREO} \cite{Kanatsoulis2018HSR}, \texttt{BSCOTT} \cite{Prevost2020HSR}, and \texttt{BSCLL1} \cite{Ding2021HSR}---as baselines. Additionally, for \texttt{CBSTAR} and \texttt{BTDvar}, we estimate the spatial degradation matrices $\bm{P}_1$ and $\bm{P}_2$ using the approach in \cite{Dian2023Zero}.

Fig.~\ref{fig:Ridge_blind} presents the recovered SRIs, residual images, and SAM maps for the Jasper Ridge dataset. Compared to the baselines, \texttt{BCLIMB} better preserves edges and textures.

Tables~\ref{table:Ridge_blind} and~\ref{table:Pavia_blind} report the reconstruction metrics for the Jasper Ridge and Pavia University datasets, respectively. In both cases, \texttt{BCLIMB} outperforms the baselines across most evaluation metrics.

\begin{table}[!htp]
\renewcommand\arraystretch{1}
\setlength{\tabcolsep}{5pt}
\renewcommand\arraystretch{1.1}
  \centering
  \caption{Performance for Pavia University data with the spatial degradation unknown. (The highest and second-highest values are highlighted in bold and underlined, respectively.)}
  \resizebox{0.6\linewidth}{!}{
  \begin{tabular}{c|c|c|c|c|c|c}\hline

    \hline
    Methods & \multicolumn{1}{c|}{RSNR} & \multicolumn{1}{c|}{SSIM} & \multicolumn{1}{c|}{CC} 
    & \multicolumn{1}{c|}{ERGAS} & \multicolumn{1}{c|}{RMSE} 
    & \multicolumn{1}{c}{SAM}\\ \hline
    \texttt{CNMF}   & 20.50    & 0.9512    & 0.9814    & 0.5387    & 0.0213    & 0.0771\\
    \texttt{FUSE}   & 18.04    & 0.9109    & 0.9621    & 0.6586    & 0.0282    & 0.0897\\
    \texttt{BSCOTT}  & 24.29    & 0.9588    & 0.9903    & 0.3286    & 0.0137    & 0.0640\\
    \texttt{BSTEREO} & 23.33    & 0.9470    & 0.9886    & 0.3690    & 0.0153    & 0.0680\\
    \texttt{BSCLL1}  & 26.06    & 0.9755    & 0.9936    & 0.2776    & 0.0116    & 0.0565\\
    \texttt{CBSTAR} & 18.23    & 0.9118    & 0.9645    & 0.6568    & 0.0276    & 0.0865\\
    \texttt{BTDvar} & 19.39    & 0.8995    & 0.9726    & 0.6066    & 0.0242    & 0.0932\\
    \texttt{NPTSR}  & 25.98    & \textbf{0.9770}    & \underline{0.9948}    & \underline{0.2626}    & 0.0113    & 0.0520\\
    \texttt{NLSTF}  & \underline{26.13}    & \underline{0.9763}    & 0.9936   & 0.2732    & \underline{0.0111}    & \underline{0.0508}\\
    \texttt{BCLIMB}  & \textbf{27.50}    & 0.9754    & \textbf{0.9949}    & \textbf{0.2333}    & \textbf{0.0095}    & \textbf{0.0474}\\
    \hline
    \end{tabular}}%
  \label{table:Pavia_blind}%
\end{table}%

\begin{figure*}[!t]
\scriptsize\setlength{\tabcolsep}{0.8pt}
\begin{minipage}[b]{1\linewidth}
\centering
\begin{tabular}{ccccccc}
\includegraphics[width=0.15\textwidth]{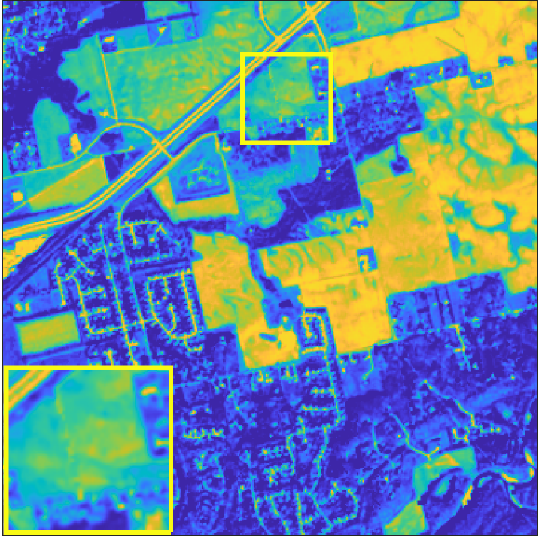}&
\includegraphics[width=0.15\textwidth]{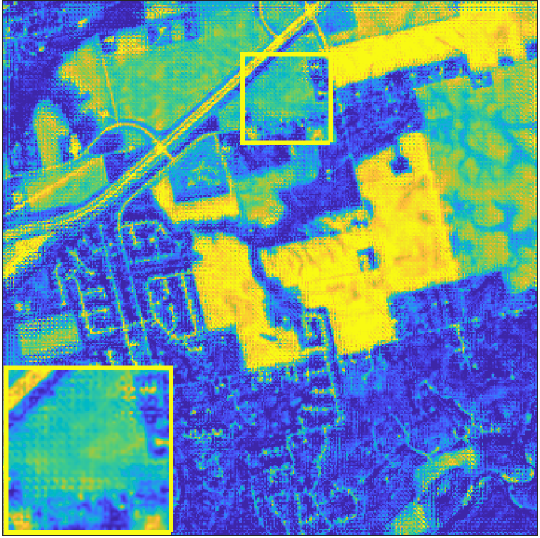}&
\includegraphics[width=0.15\textwidth]{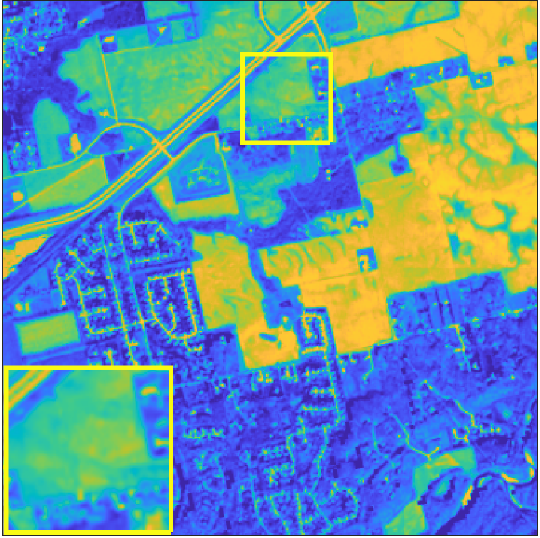}&
\includegraphics[width=0.15\textwidth]{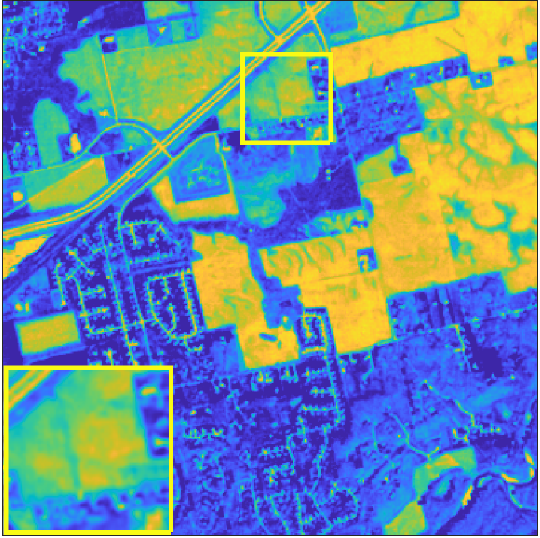}&
\includegraphics[width=0.15\textwidth]{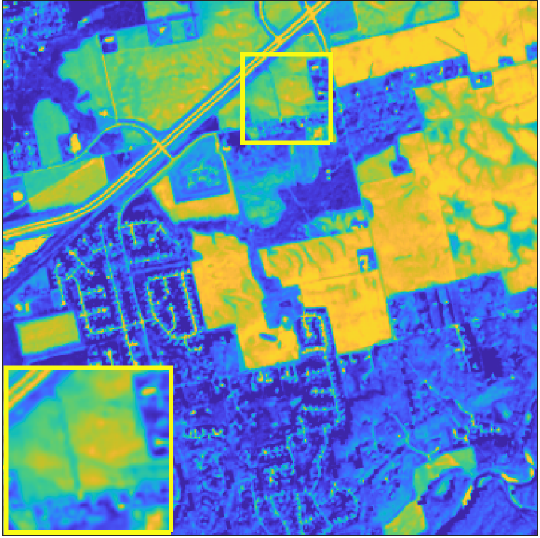}&
\includegraphics[width=0.15\textwidth]{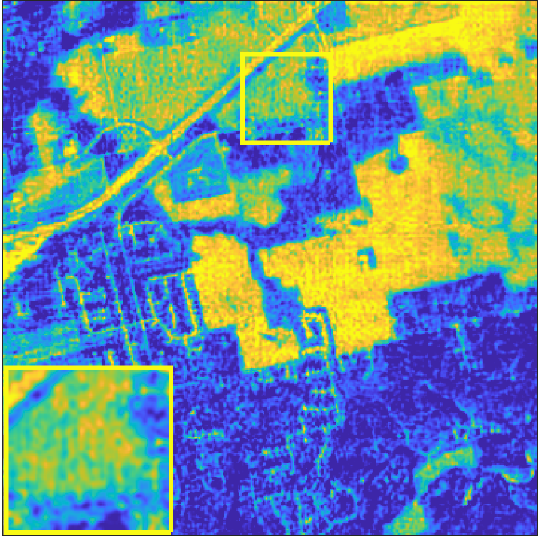}\\
(a) \texttt{CNMF} & (b) \texttt{FUSE} & (c) \texttt{BSCOTT} & (d) \texttt{BSTEREO} & (e) \texttt{BSCLL1} & (f) \texttt{CBSTAR} \\
\end{tabular}
\end{minipage}
\begin{minipage}[b]{1\linewidth}
\centering
\begin{tabular}{ccccccc}
\includegraphics[width=0.15\textwidth]{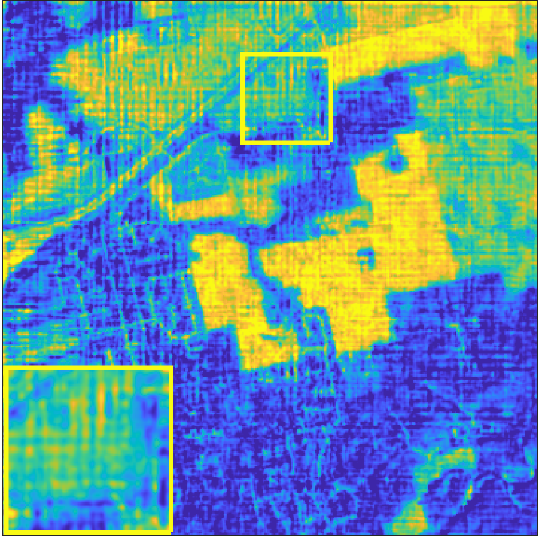}&
\includegraphics[width=0.15\textwidth]{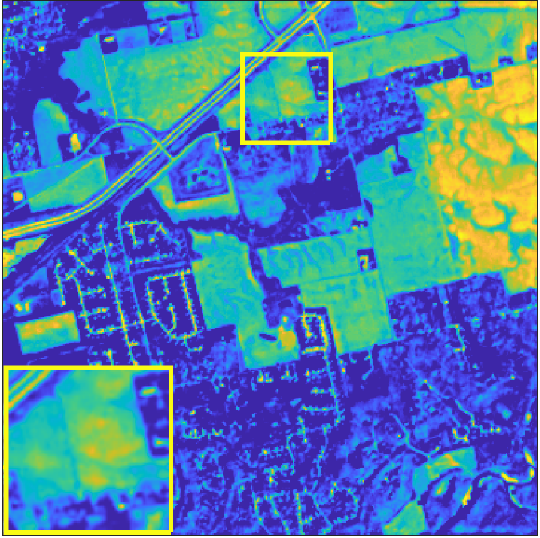}&
\includegraphics[width=0.15\textwidth]{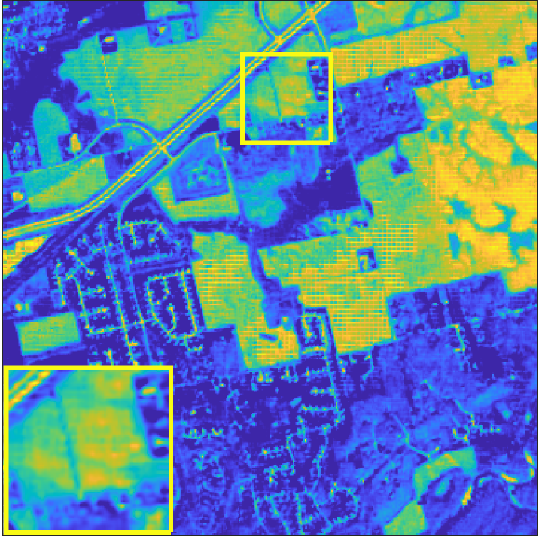}&
\includegraphics[width=0.15\textwidth]{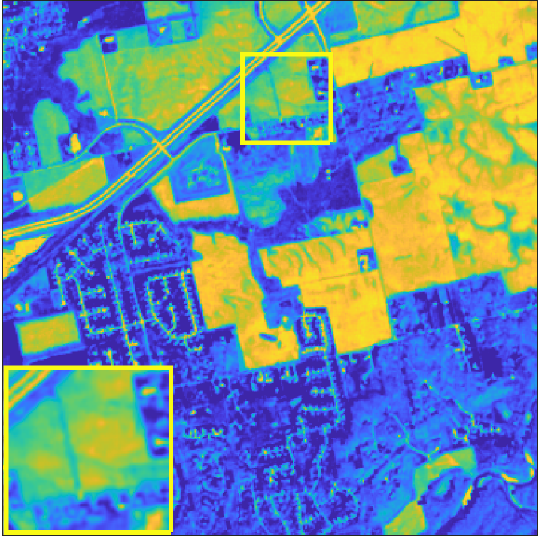}&
\includegraphics[width=0.178\textwidth]{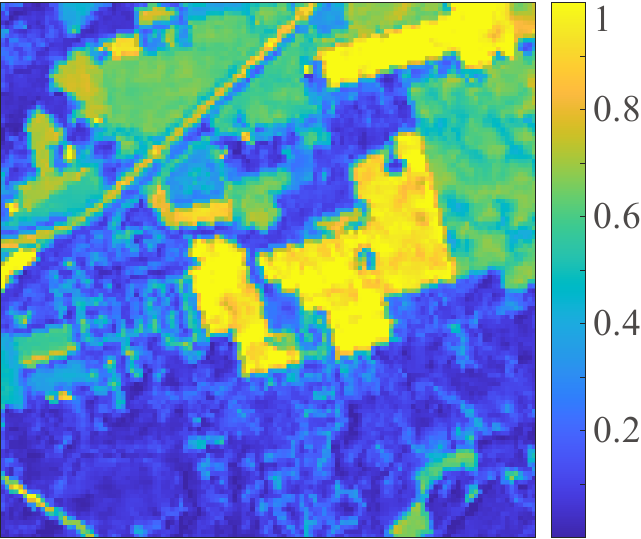}\\
(g) \texttt{BTDvar} & (h) \texttt{NPTSR} & (i) \texttt{NLSTF} & (j) \texttt{BCLIMB} & (k) HSI\\
\end{tabular}
\end{minipage}
\caption{The recovered SRIs of the real-data experiment.}
  \label{fig:Real}
\end{figure*}

\subsection{Real Data Experiment}
To further demonstrate the effectiveness of the proposed method, we evaluate all algorithms on a real HSI-MSI fusion task \cite{Yang2018HSR}. The MSI data is collected by the Sentinel-2A satellite, which contains 13 spectral bands. Following the setup in \cite{Wu2019Hyperspectral}, we select 4 bands, corresponding to central wavelengths of 490\,nm, 560\,nm, 665\,nm, and 842\,nm, with $360\times 360$ spatial resolution.
The HSI data is acquired by the Hyperion sensor, originally with 220 bands. After removing severely corrupted bands, a subscene of size $120\times 120\times 89$ is used in this experiment.

Note that in the real scenario, both spatial and spectral degradation operators are unknown. For \texttt{CBSTAR}, \texttt{BTDvar}, \texttt{NPTSR}, and \texttt{NLSTF}, we estimate the spatial degradation matrices $\bm{P}_1$, $\bm{P}_2$ and the spectral degradation matrix $\bm{P}_{\rm M}$ from the observed HSI and MSI using the method in \cite{Dian2023Zero}. 
For the matrix-based baselines \texttt{CNMF} and \texttt{FUSE}, we apply their built-in procedures to estimate the degradations. For the tensor-based baselines \texttt{SCOTT}, \texttt{STEREO}, and \texttt{SCLL1}, we adopt their blind versions, i.e., \texttt{BSCOTT}, \texttt{BSTEREO}, and \texttt{BSCLL1}. The rank parameter is set to $R=4$ based on visual inspection.

\begin{table}[!t]
\renewcommand\arraystretch{1}
\setlength{\tabcolsep}{3pt}%
\renewcommand\arraystretch{1.1}%
  \centering
  \caption{Performance for Washington DC and Pavia University datasets under different regularization terms with SNR=35dB.}
  \resizebox{0.65\linewidth}{!}{
  \begin{tabular}{c|c|c|c|c|c}\hline

    \hline
   Datasets & Metrics & \multicolumn{1}{c|}{no priors} 
   & \multicolumn{1}{c|}{ {\makecell[c]{only spectral \\ smoothness}} } 
   & \multicolumn{1}{c|}{ {\makecell[c]{only spatial \\ smoothness}} } 
   & \multicolumn{1}{c}{\texttt{CLIMB}} \\
   \hline
   \multirow{6}{*}{ Washington DC } 
   & RSNR  & 29.01  & 29.57   & 29.81  & 30.37 \\
   & SSIM  & 0.9904 & 0.9914  & 0.9917 & 0.9921 \\
   & CC    & 0.9714 & 0.9727  & 0.9678 & 0.9680 \\
   & ERGAS & 0.8377 & 0.7996  & 1.0120 & 0.7011 \\
   & RMSE  & 0.0072 & 0.0068  & 0.0066 & 0.0062 \\
   & SAM   & 0.0384 & 0.0358  & 0.0349 & 0.0326 \\
   \hline
   \multirow{6}{*}{ Pavia University } 
   & RSNR  & 24.81  & 25.21   & 27.78  & 28.19 \\
   & SSIM  & 0.9492 & 0.9537  & 0.9756 & 0.9773 \\
   & CC    & 0.9909 & 0.9916  & 0.9952 & 0.9957 \\
   & ERGAS & 0.3113 & 0.3000  & 0.2262 & 0.2152 \\
   & RMSE  & 0.0129 & 0.0124  & 0.0092 & 0.0088 \\
   & SAM   & 0.0690 & 0.0625  & 0.0465 & 0.0448 \\
    \hline
  \end{tabular}}%
  \label{table:regu}%
\end{table}%

\begin{figure}[!t]
\scriptsize\setlength{\tabcolsep}{0.3pt}
\begin{center}
\begin{tabular}{cccc}
\includegraphics[width=0.245\textwidth]{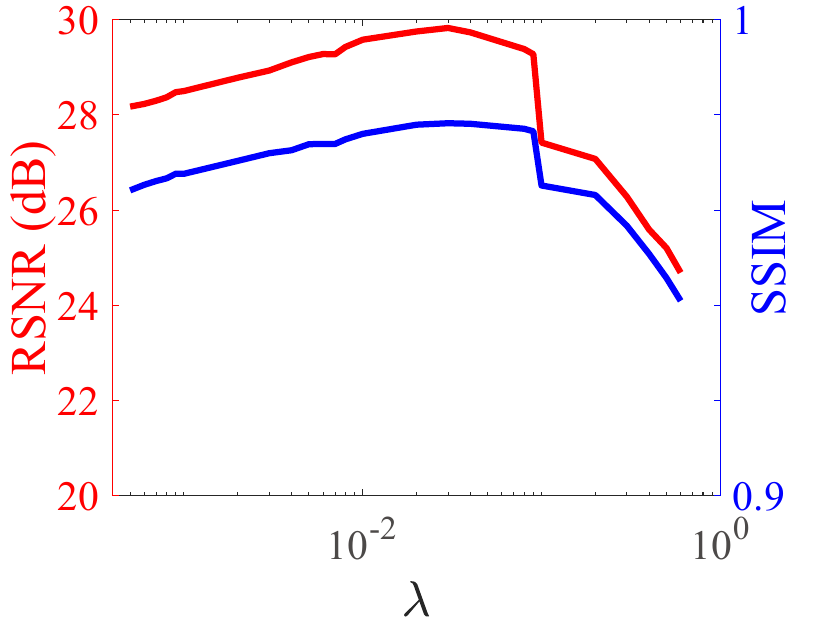}&
\includegraphics[width=0.245\textwidth]{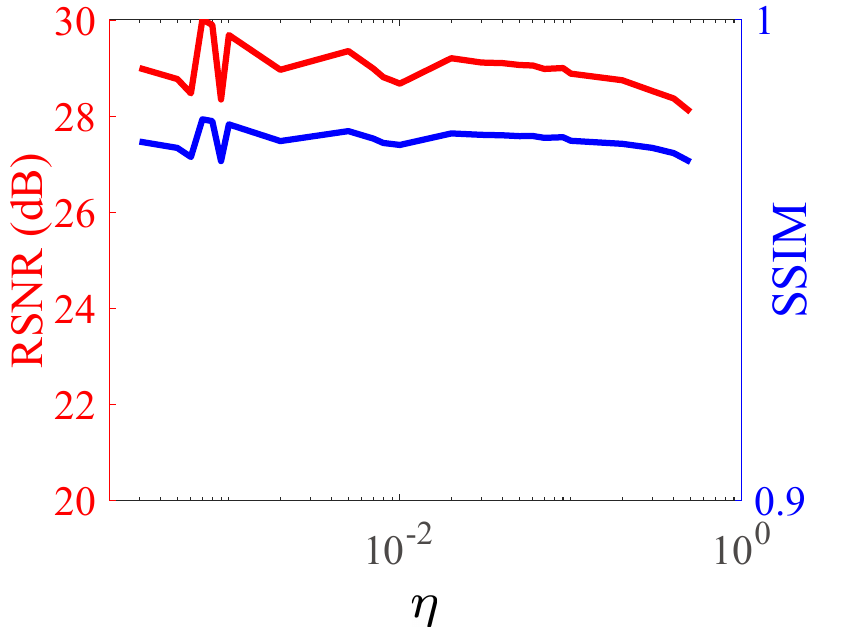}&
\includegraphics[width=0.245\textwidth]{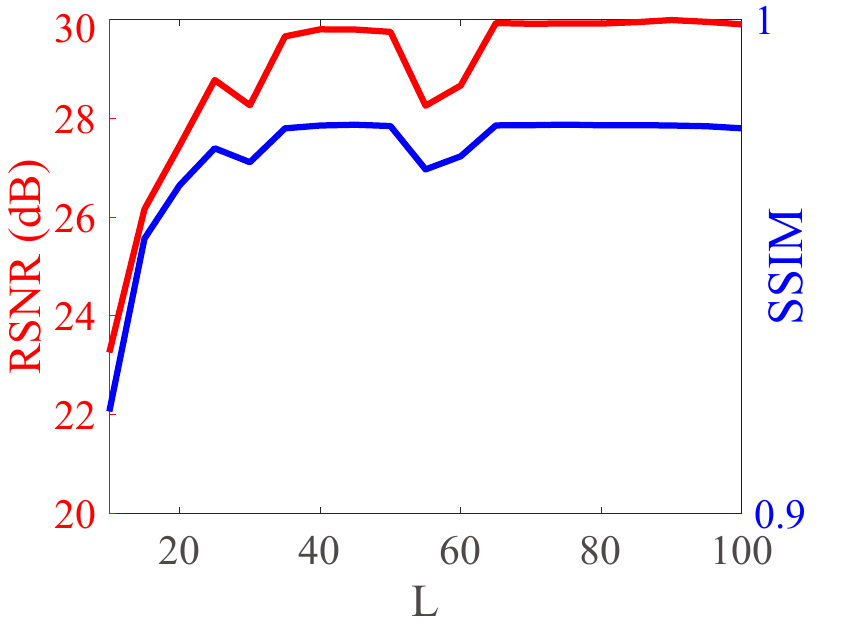}&
\includegraphics[width=0.245\textwidth]{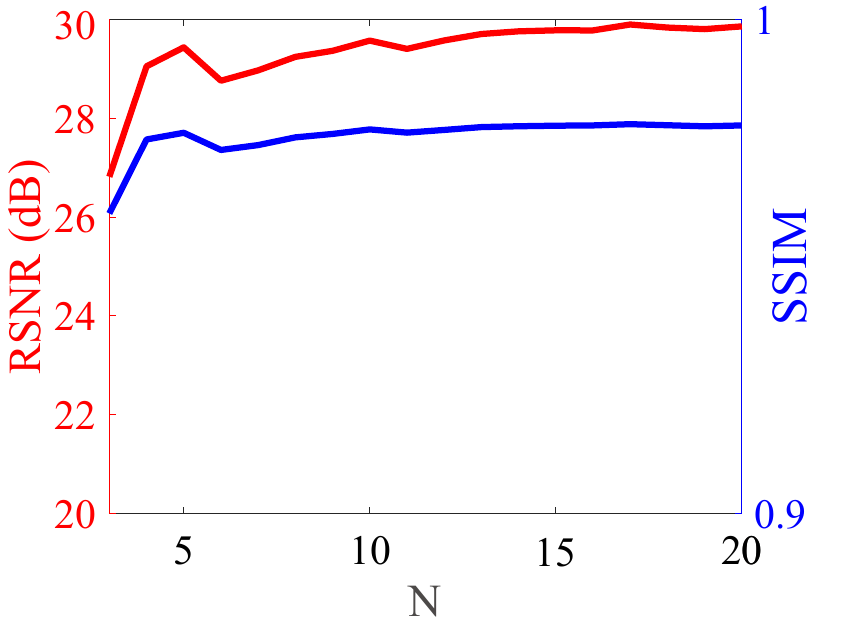}\\
 (a) $\lambda$ & (b) $\eta$ & (c) L & (d) N\\
\end{tabular}
\caption{RSNR (dB) and SSIM values under different parameters on Jasper Ridge dataset. When changing one parameter, the others are fixed to the ``optimal values'' as reported in the figures.}
\label{fig:para}
\end{center}
\end{figure}

Fig.~\ref{fig:Real} shows the reconstructed SRI images produced by all methods. All approaches enhance spatial resolution and yield improved visual quality compared to the input HSI (see subfigure~(k)). 
However, \texttt{FUSE}, \texttt{CBSTAR}, \texttt{BTDvar}, and \texttt{NLSTF} introduce noticeable artifacts such as glitches and striping. While \texttt{CNMF} preserves sharp edges, it appears to miss certain objects and fine details, particularly in the lower-right region. The baselines \texttt{BSCOTT}, \texttt{BSTEREO}, \texttt{BSCLL1}, and \texttt{NPTSR} perform reasonably well in recovering the SRI. 
In comparison, the proposed \texttt{BCLIMB} better preserves both texture details and spatial smoothness, as evident in the highlighted regions.

\subsection{Discussions}

\subsubsection{Effect of Regularization Terms}
We evaluate the contribution of each regularization term in the proposed model. Table~\ref{table:regu} compares the reconstruction performance using spectral smoothness only, spatial smoothness only, and both (i.e., the full \texttt{CLIMB}) on the semi-real Washington DC and Pavia University datasets under $\mathrm{SNR}=35$~dB.

For the Washington DC dataset, even without regularization, the proposed method outperforms the baselines, highlighting its strength in modeling endmember variability. Adding either spectral or spatial smoothness further improves performance, confirming the value of each prior. The best results are achieved when both priors are applied, supporting the design choice of integrating them into the LMN-based HSR framework.

\subsubsection{Effect of Parameter Selection}
\label{sec:para}

Fig.~\ref{fig:para} shows the sensitivity of the proposed algorithm to key parameters, including the regularization weights $\lambda$ and $\eta$, and the rank parameters $L$ (or $M$) and $N$. The figure reports RSNR and SSIM values on the Jasper Ridge dataset with $\mathrm{SNR}=35$~dB under different parameter settings. The results demonstrate that \texttt{CLIMB} maintains strong performance over a wide range of values, indicating its robustness to parameter selection.

\section{Conclusion and Future Works}
\label{sec:Conclusion}

In this work, we revisited CTD-based hyperspectral super-resolution with a particular emphasis on modeling endmember variability. We introduced a tensor formulation based on the LMN decomposition, which generalizes the CPD, Tucker, and LL1 models as special cases. This formulation accommodates EV while preserving physical interpretability.  
We established that, similar to existing CTD approaches, the super-resolution image under the LMN model can be exactly recovered under reasonable conditions. Moreover, we proposed a CTD algorithm based on our model by incorporating prior knowledge of the LMN latent factors---specifically, spatial and spectral smoothness---via nonconvex total variation and Tikhonov regularizations.  
Extensive experiments on semi-real and real datasets, both of which exhibit nontrivial EV, validated the proposed method and consistently demonstrated superior reconstruction quality compared to state-of-the-art baselines. These results confirm the LMN framework as a flexible and interpretable tool for advancing hyperspectral super-resolution under endmember variability. 

Some directions are worth exploring as future works.
First, the LMN model requires specification of more parameters compared to other models such as CPD and LL1 (e.g., the multi-linear rank of each $\underline{\bm T}_r$). Although the model is not very sensitive to these parameters, having a mechanism to automatically determine these parameters would be more ideal. This may be done with Bayesian tensor modeling techniques; see, e.g., \cite{Cheng2023Bayesian}.
Second, the basic assumption of our method (and most of the CTD methods) is that the HSI and MSI are spatially co-registered---i.e., the images cover exactly the same spatial region with the same underlying coordinate system.
In practice, this assumption is nontrivial to satisfy and may take a lot of visual/manual tuning, especially when the acquired raw HSI and MSI are spatially misaligned \cite{Zhou2020Integrated,Qu2022Unsupervised}.
Therefore, it would be meaningful to handle the unregistered HSR problem with recoverability and interpretability.

\section*{Acknowledgement}
The authors would like to thank Dr. Yang Xu and Dr. Wei Wan for providing their source code.

\bibliographystyle{siam}
\bibliography{ref}

\newpage

{\centering{Supplemental Materials for \\
    \emph{``Rethinking Coupled Tensor Analysis for Hyperspectral Super-Resolution: Recoverable Modeling Under Endmember Variability''}\\
    \emph{Meng Ding and Xiao Fu\\[10pt]}}}

\appendices

\section{Low Multilinear Rank of Materials of Spectral Images}
\label{app:low-rankness}
To verify the low multilinear rank of $\underline{\bm T}_r$, in Figs. \ref{fig:Ridge_LR} and \ref{fig:WDC_LR}, we show the singular values of the mode-$n$ (where $n=1,2,3$) unfoldings of $\underline{\bm T}_r$ extracted from Jasper Ridge and Washington DC. 
The terms $\underline{\bm{T}}_r$'s are obtained using the nonnegative BTD implementation in Tensorlab \cite{Vervliet2016Tensorlab}. 
Note that if $\underline{\bm{T}}_r$ admits low multi-linear rank, then its mode-$n$ unfoldings all should exhibit low matrix rank.

\begin{figure*}[!ht]
\scriptsize\setlength{\tabcolsep}{0.3pt}
\begin{center}
\begin{tabular}{ccccc}
\includegraphics[width=0.22\textwidth]{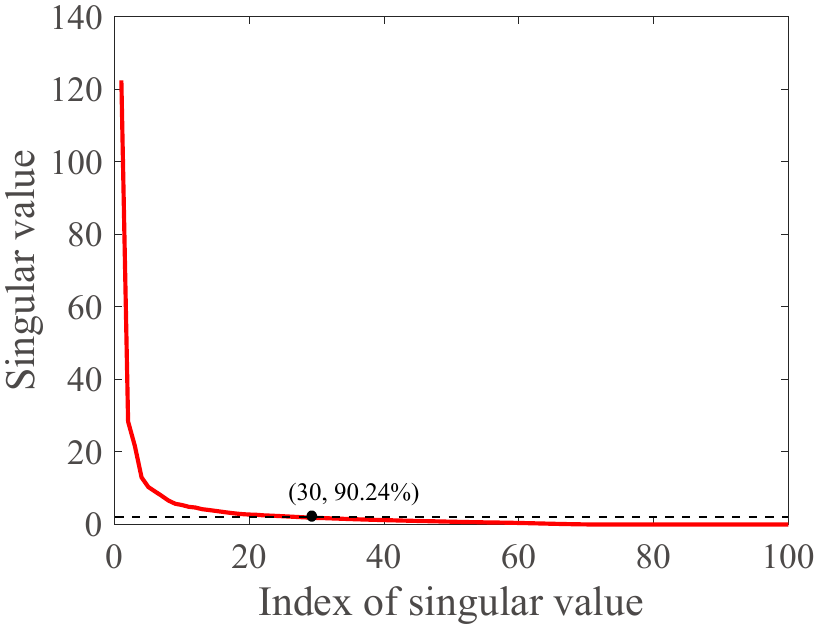}&
\includegraphics[width=0.22\textwidth]{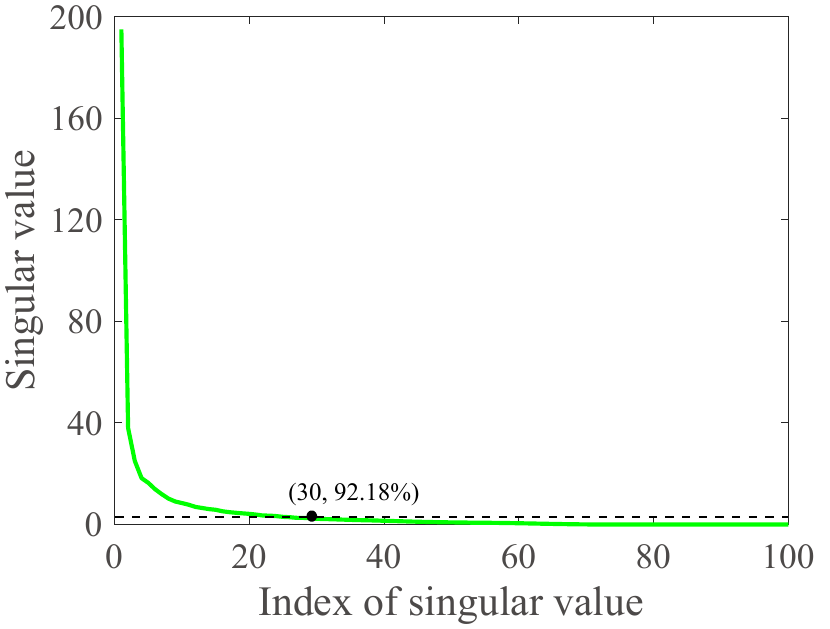}&
\includegraphics[width=0.22\textwidth]{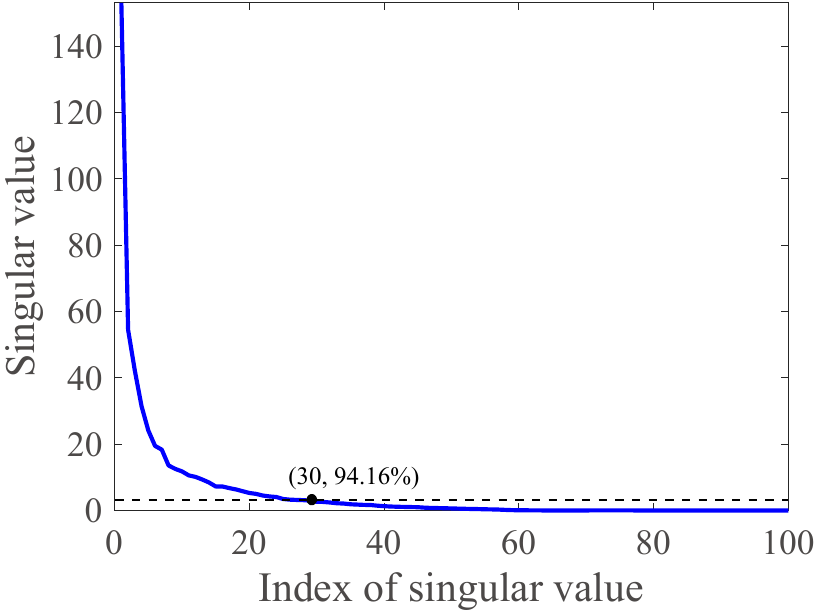}&
\includegraphics[width=0.22\textwidth]{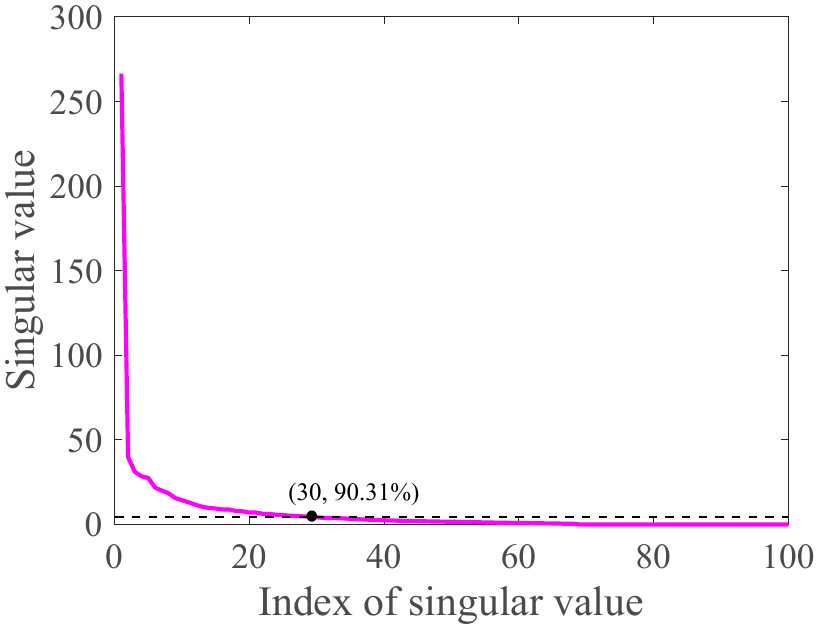}\\
\includegraphics[width=0.22\textwidth]{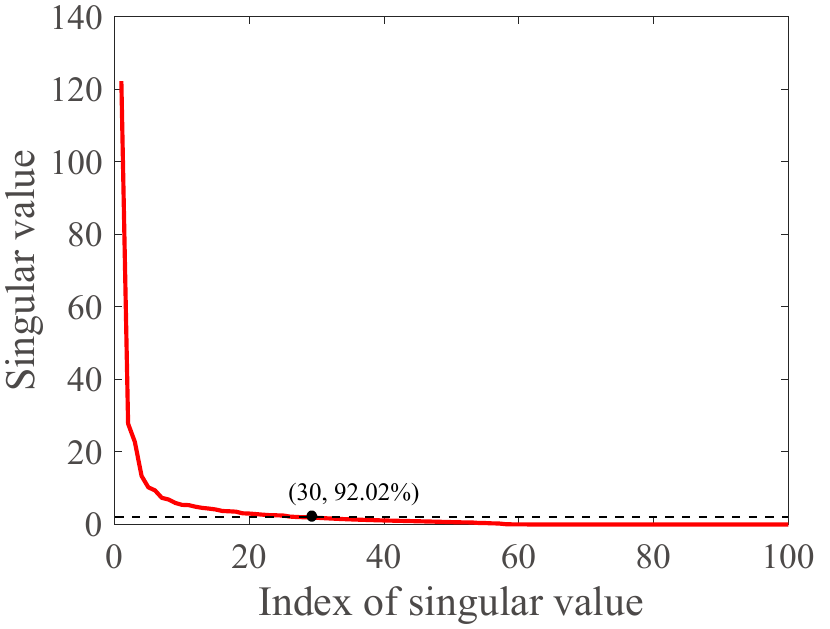}&
\includegraphics[width=0.22\textwidth]{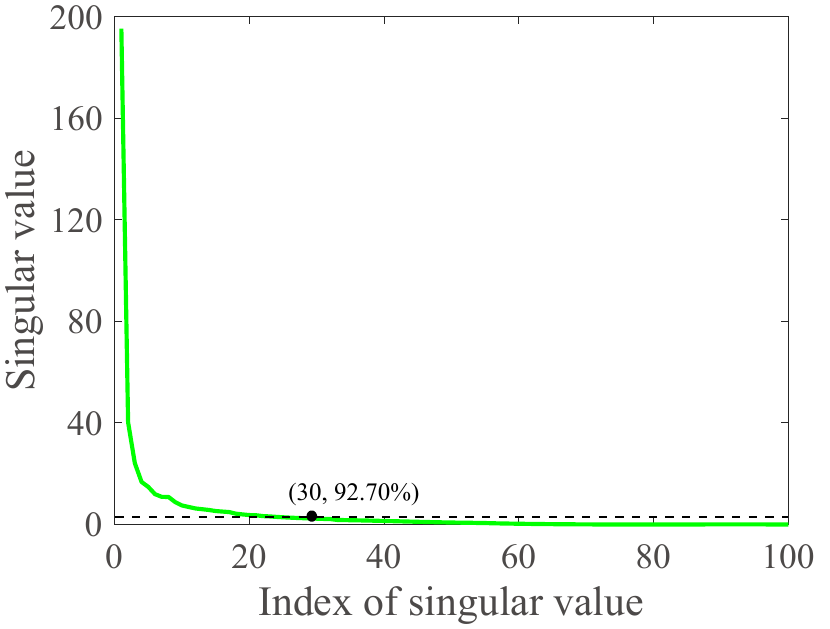}&
\includegraphics[width=0.22\textwidth]{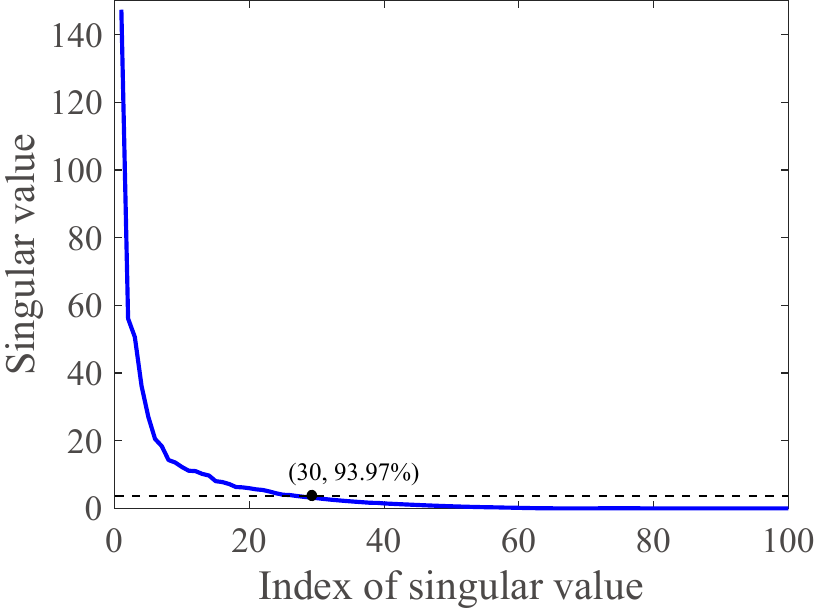}&
\includegraphics[width=0.22\textwidth]{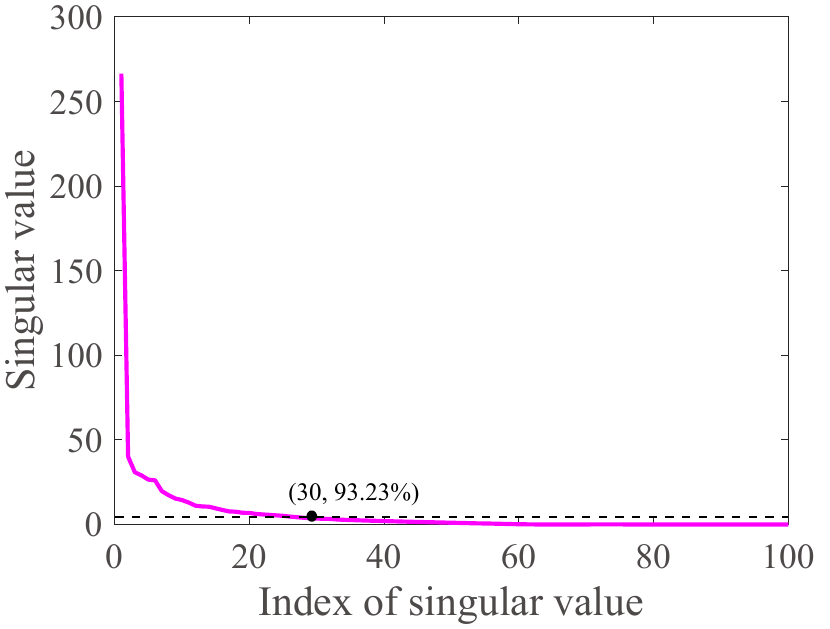}\\
\includegraphics[width=0.22\textwidth]{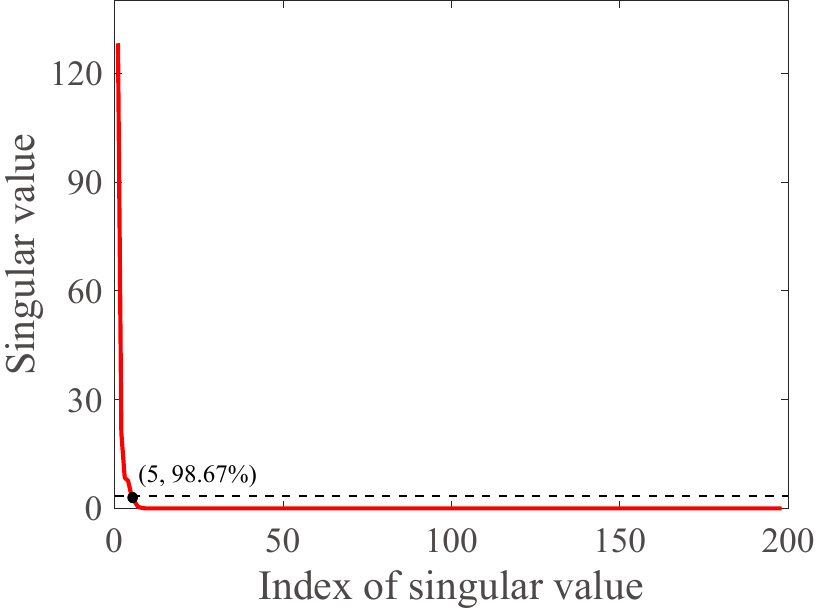}&
\includegraphics[width=0.22\textwidth]{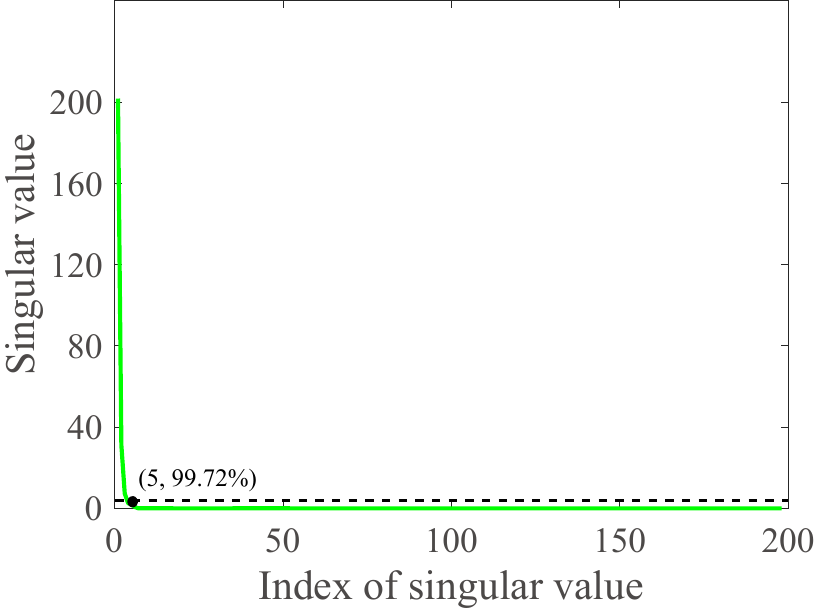}&
\includegraphics[width=0.22\textwidth]{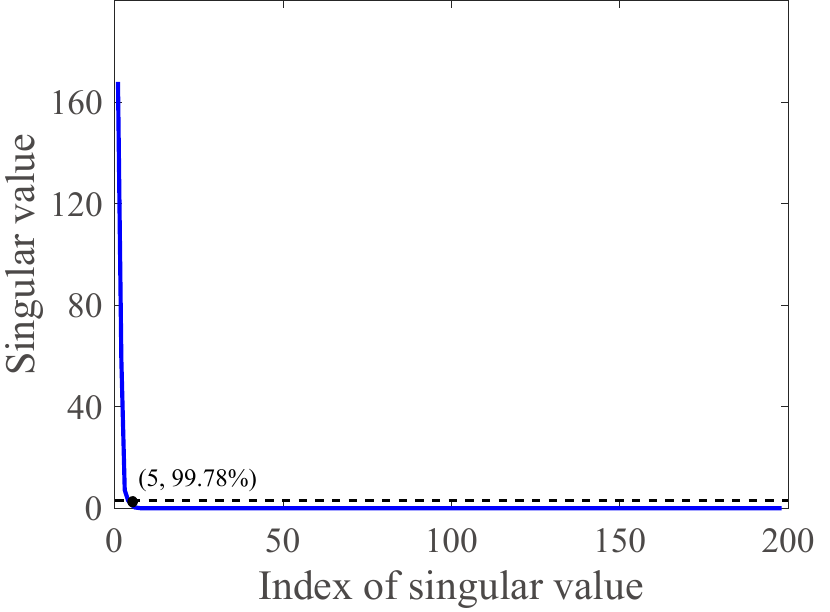}&
\includegraphics[width=0.22\textwidth]{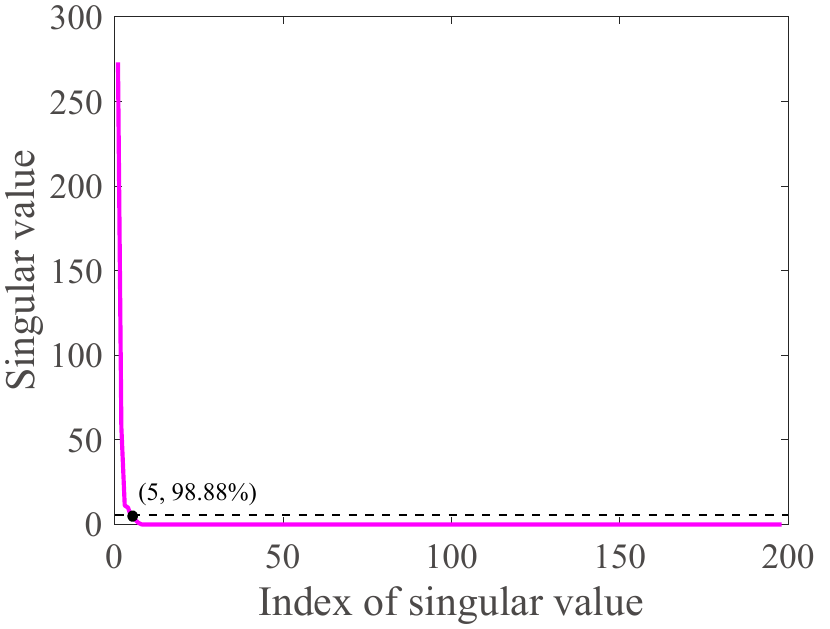}\\
(a)~~$\underline{\bm{T}}_1$ &  (b)~~$\underline{\bm{T}}_2$ & (c)~~$\underline{\bm{T}}_3$ &  (d)~~$\underline{\bm{T}}_4$\\
\end{tabular}
\caption{The singular values of unfolding matrices (top to bottom: mode-1, mode-2, and mode-3, respectively) of each $\underline{\bm{T}}_r$ for Jasper Ridge dataset. } 
  \label{fig:Ridge_LR}
\end{center}
\end{figure*}

From Figs. \ref{fig:Ridge_LR} and \ref{fig:WDC_LR}, one can see that for Jasper Ridge, the first 30 principal components (the first 30 left, right singular vectors and singular values) of the mode-1 unfolding matrices of $\underline{\bm{T}}_r$s contain more than $90\%$ of its energy.
For mode-2 and 3 unfolding matrices of $\underline{\bm{T}}_r$s, the first 30 and 5 principal components contain more than $90\%$ and $98\%$ of the corresponding energy.
For Washington DC, the first 40 (40, 4) principal components of the mode-1 (2, 3) unfolding matrices of $\underline{\bm{T}}_r$s contain about $90\%$ ($90\%$, $98\%$) of its energy.

\begin{figure*}[!t]
\scriptsize\setlength{\tabcolsep}{0.3pt}
\begin{center}
\begin{tabular}{ccccc}
\includegraphics[width=0.22\textwidth]{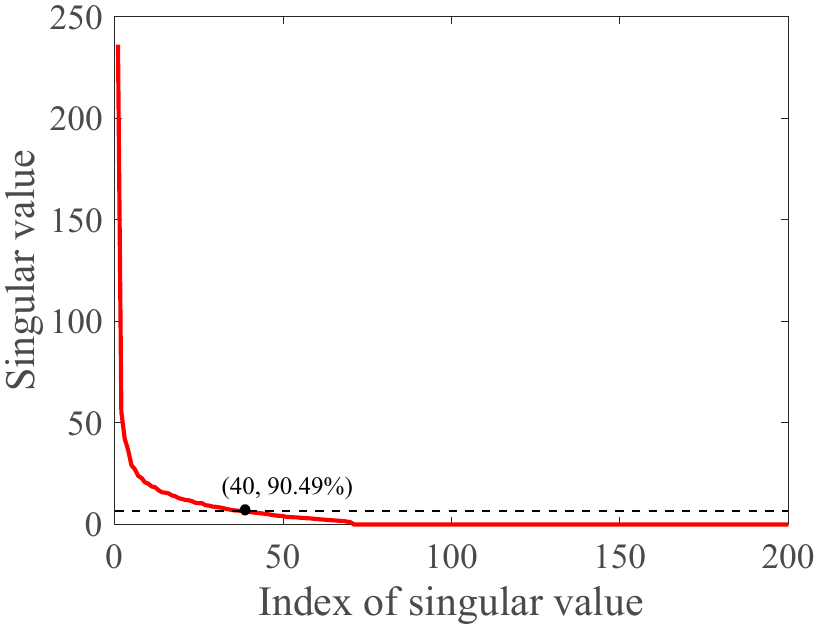}&
\includegraphics[width=0.22\textwidth]{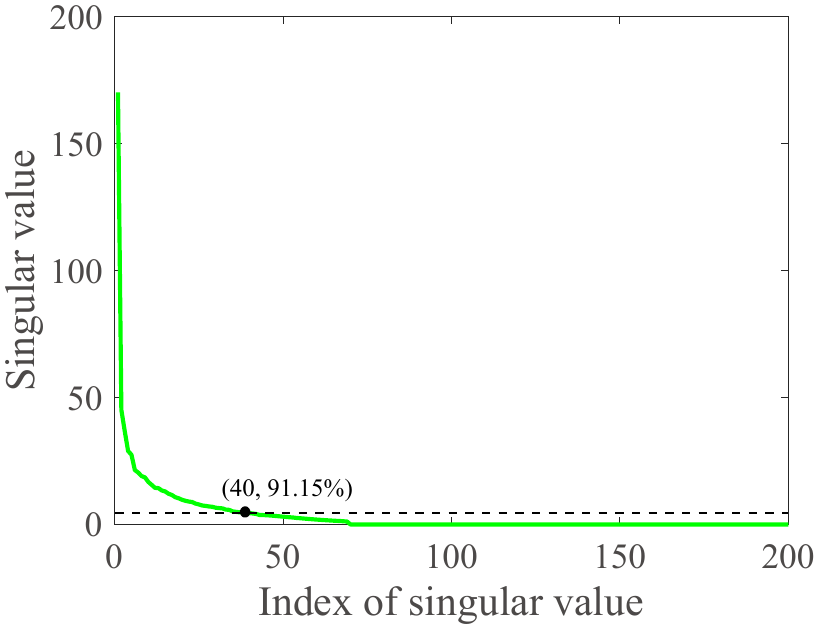}&
\includegraphics[width=0.22\textwidth]{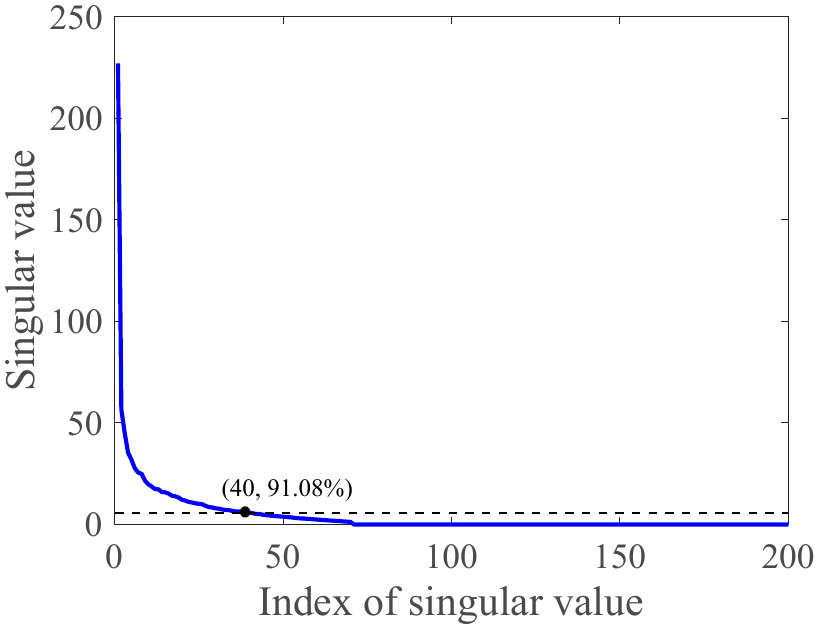}&
\includegraphics[width=0.22\textwidth]{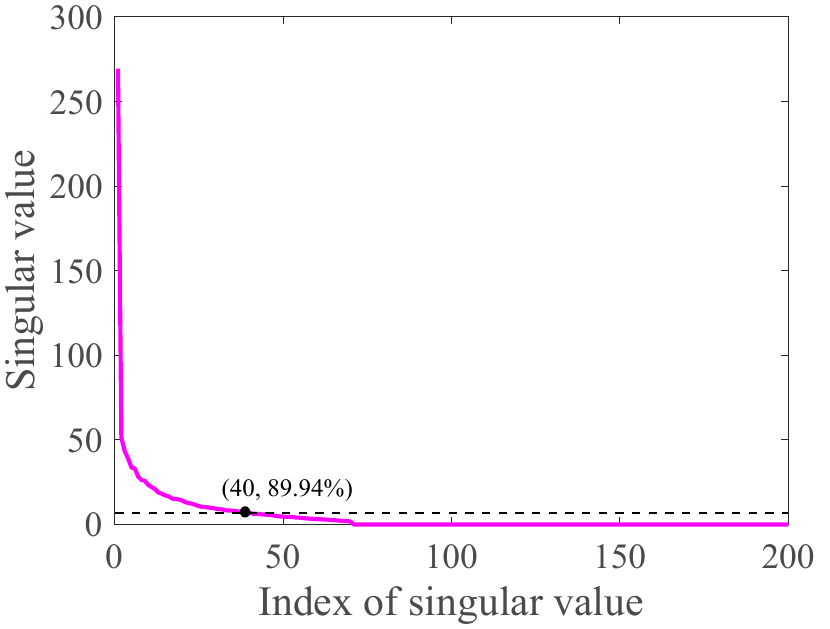}\\
\includegraphics[width=0.22\textwidth]{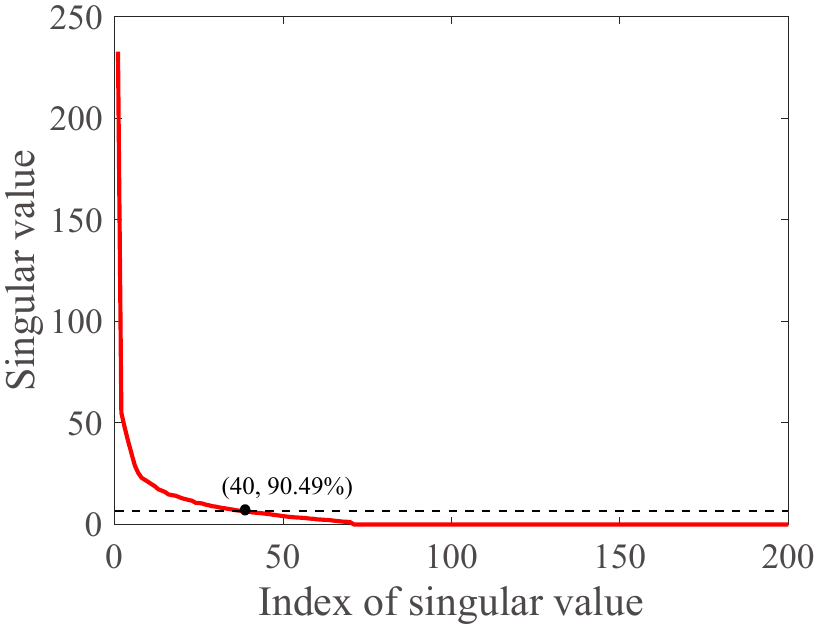}&
\includegraphics[width=0.22\textwidth]{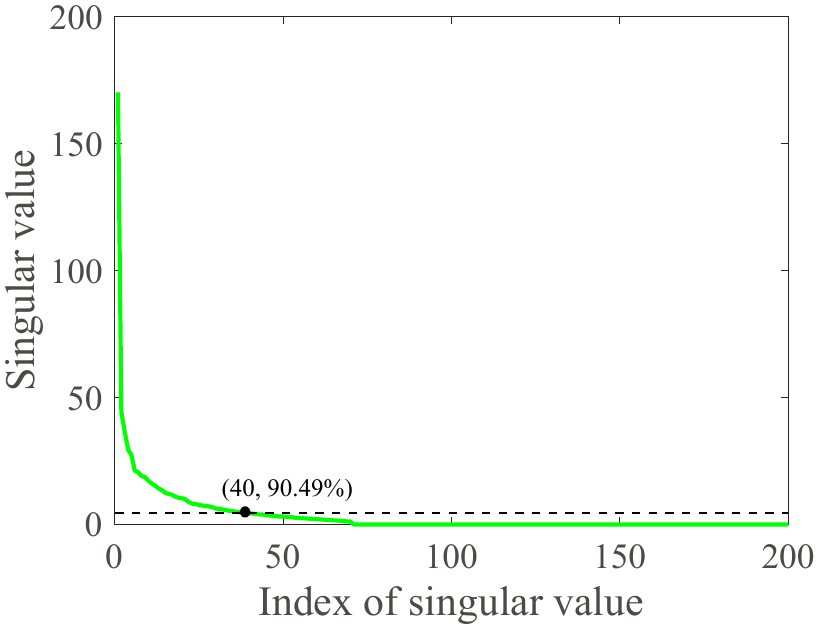}&
\includegraphics[width=0.22\textwidth]{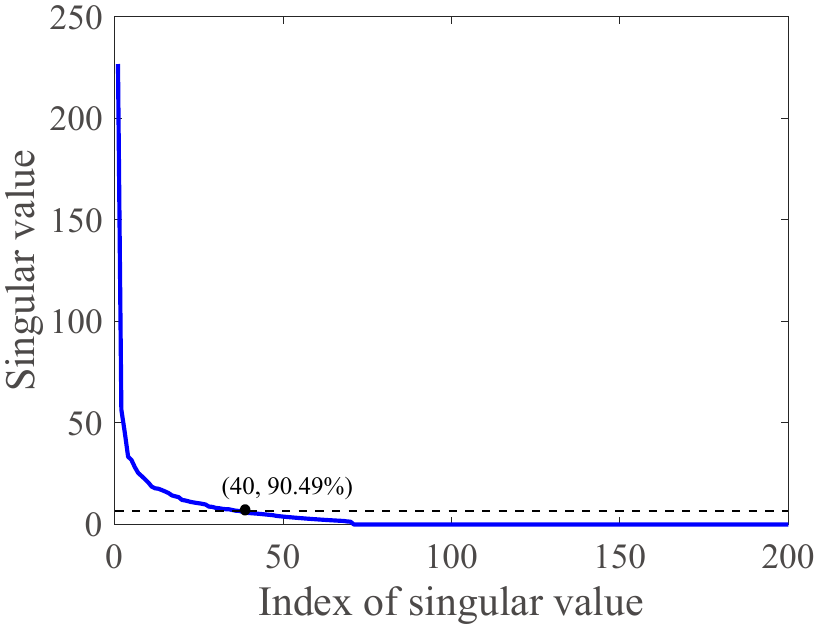}&
\includegraphics[width=0.22\textwidth]{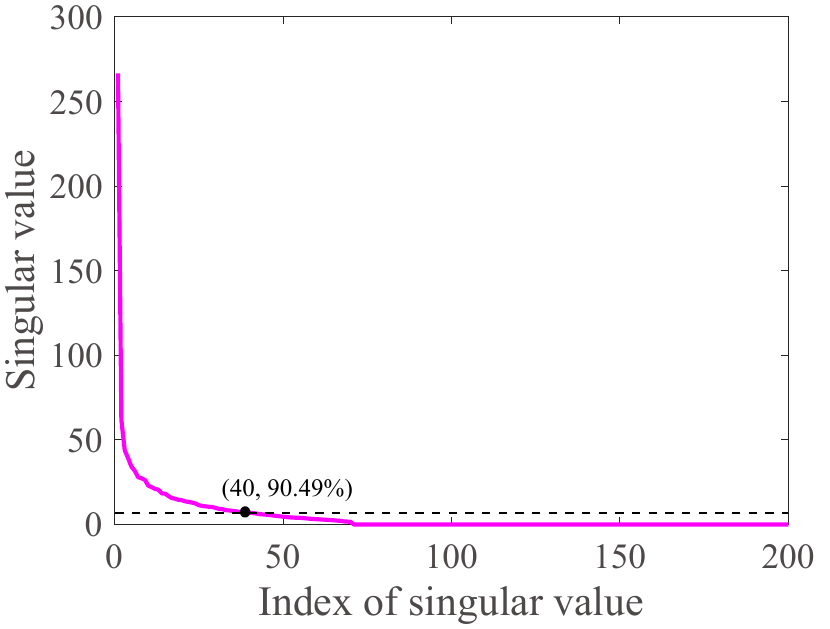}\\
\includegraphics[width=0.22\textwidth]{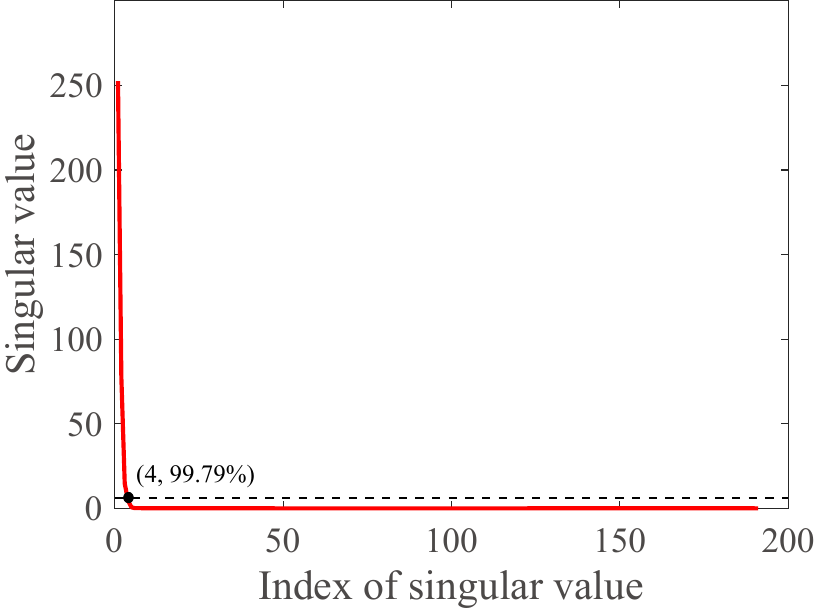}&
\includegraphics[width=0.22\textwidth]{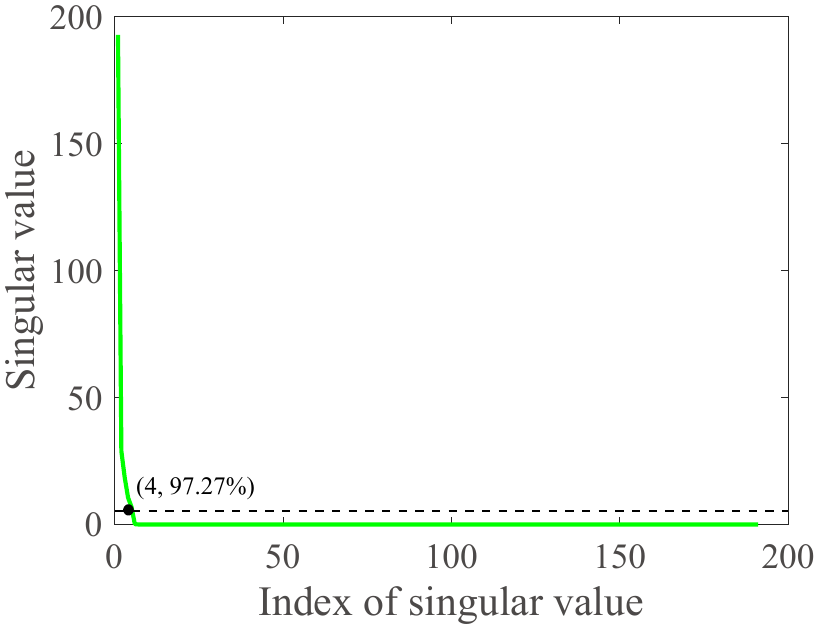}&
\includegraphics[width=0.22\textwidth]{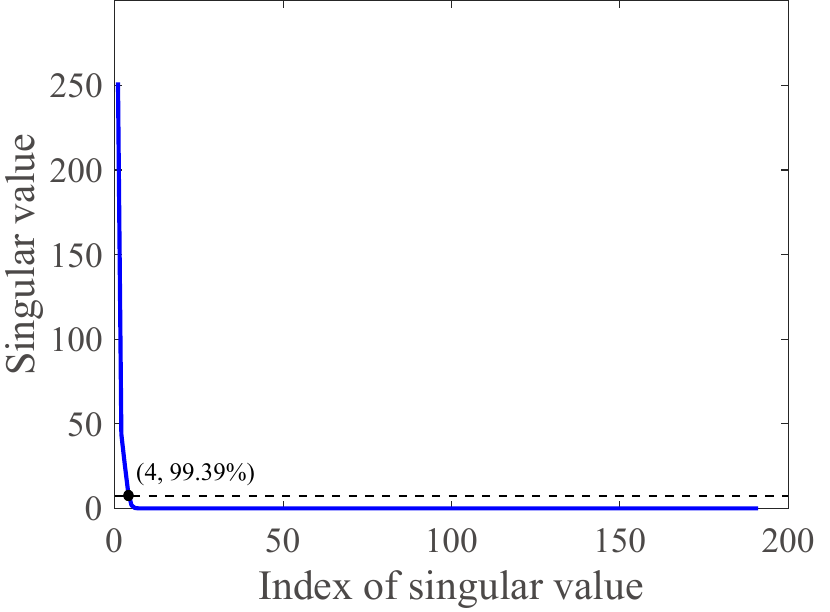}&
\includegraphics[width=0.22\textwidth]{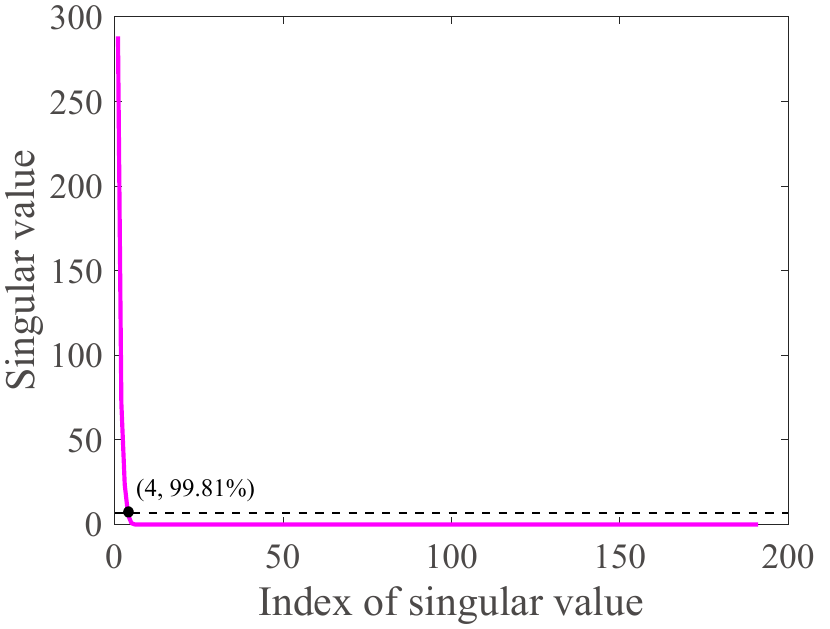}\\
(a)~~$\underline{\bm{T}}_1$ &  (b)~~$\underline{\bm{T}}_2$ & (c)~~$\underline{\bm{T}}_3$ &  (d)~~$\underline{\bm{T}}_4$\\
\end{tabular}
\caption{The singular values of unfolding matrices (top to bottom: mode-1, mode-2, and mode-3, respectively) of each $\underline{\bm{T}}_r$ for Washington DC dataset. } 
  \label{fig:WDC_LR}
\end{center}
\end{figure*}

\section{Proof of Theorem \ref{the:identifiability_LMN}}
\label{app:identifiability_LMN}

To prove the theorem, we first invoke the following lemma:
\begin{lemma}[Corollary 2.6 \cite{Domanov2024BTD}]
\label{lemma:identifiability_LM}
Let $\underline{\bm{Y}}\in \mathbb{R}^{I\times J \times K}$ be a sum of $R$ rank-$(L,M,\cdot)$ terms, i.e.,
$\underline{\bm{Y}} = \sum_{r=1}^R\underline{\bm{D}}_r\times_1 \bm{A}_r\times_2 \bm{B}_r$ with 
$\bm{A}_{r}\in \mathbb{R}^{I \times L}, \bm{B}_{r}\in \mathbb{R}^{J \times M}, \underline{\bm{D}}_r\in \mathbb{R}^{L\times M \times K}, ~r=1,\ldots, R$. Assume that $\bm{A}_{r}$, $\bm{B}_{r}$, and $\underline{\bm{D}}_r$ are drawn from any absolutely continuous distributions. Then, the rank-$(L,M,\cdot)$ decomposition of $\underline{\bm Y}$ is essentially unique almost surely, if 
\begin{align}\label{eq:lmdot}
 I\geq LR, ~J\geq MR, ~K \geq \max\{\lceil L/M \rceil + \lceil M/L \rceil, 3\},~ 1\leq r\leq R.\nonumber
\end{align}
\end{lemma}
Note that any LMN representation has a corresponding $(L,M,\cdot)$ representation. 
Indeed, assume that $\tY = \sum_{r=1}^R \tD_r \times_1 \A_r \times_2 \B_r \times_3 \C_r$. 
One can let
\[          {\tD}_r' = \tD_r\times_3\C_r, \]
which leads to 
\begin{align}\label{eq:lmn_lmdot}
     \tY = \sum_{r=1}^R \tD_r \times_1 \A_r \times_2 \B_r \times_3 \C_r = \sum_{r=1}^R {\tD}_r' \times_1 \A_r \times_2 \B_r. 
\end{align}
Next, let us consider that the LMN tensor on the LHS of \eqref{eq:lmn_lmdot} has dimensions satisfying \eqref{eq:lmdot}.
Then, the corresponding $(L,M,\cdot)$ representation on the RHS in \eqref{eq:lmn_lmdot} is essentially unique. Next, we follow the idea in \cite[Theorem 5.1]{Lathauwer2008BTD2} to show that the LMN decomposition on the LHS is also essentially unique.

By the essential uniqueness in Lemma~\ref{lemma:identifiability_LM}, we have that, for any $\bar{\A}_r\in \mathbb{R}^{I\times L}$ and $\bar{\B}_r\in\mathbb{R}^{J\times M}$ that satisfy $\underline{\bm{Y}} = \sum_{r=1}^R\bar{\tD}_r'\times_1 
\bar{\bm{A}}_r\times_2 \bar{\bm{B}}_r$, the following holds:
\begin{align}\label{eq:YAB}
\bar{\bm{A}}_{\pi(r)}=\bm{A}_r\bm \Theta_{a,r}, ~\bar{\bm{B}}_{\pi(r)}=\bm{B}_r\bm \Theta_{b,r}, 
\end{align}
where $\bm \Theta_{a,r}$ and $\bm \Theta_{b,r}$ are nonsingular matrices. Let $\bar{\tD}_r' = \bar{\tD}_r\times_3\bar{\C}_r$ and ${\tD}_r' = \tD_r\times_3\C_r$, we have
\begin{align}
     \underline{\bar{\bm{D}}}'_{\pi(r)}\times_1 \bm \Theta_{a,r}\times_2 \bm \Theta_{b,r}=\underline{\bm{D}}_r', 
\end{align}
i.e., 
\begin{align}\label{eq:A_DC}
     \underline{\bar{\bm{D}}}_{\pi(r)}\times_1 \bm \Theta_{a,r}\times_2 \bm \Theta_{b,r}\times_3 \bar{\bm{C}}_{\pi(r)}=\underline{\bm{D}}_r\times_3 \bm{C}_r.     
\end{align}

Considering the mode-3 unfolding $\bm{D}_{r}^{(3)}\in \mathbb{R}^{N\times LM}$ and $\bar{\bm{D}}_{\pi(r)}^{(3)}\in \mathbb{R}^{N\times LM}$ of tensors $\underline{\bar{\bm{D}}}_r$ and $\underline{\bar{\bm{D}}}_{\pi(r)}$, we have
\begin{align}
\bar{\bm{C}}_{\pi(r)}  \bar{\bm{D}}_{\pi(r)}^{(3)}
=  \bm{C}_r \bm{D}_r^{(3)}(\bm\Theta_{b,r}^{-1} \otimes \bm\Theta_{a,r}^{-1})^\top,\nonumber
\end{align}
where $\otimes$ denotes the Kronecker product.
Since $\underline{\bm{D}}_{r}$ and $\bm{C}_r$ are drawn from any absolutely continuous distributions, and $N\leq LM$, then $\textrm{rank}(\bm{D}_r^{(3)})=N$, and $\textrm{rank}(\bm{C}_r)=N$ is full column rank almost surely.
We also get that
\begin{align}
N &= \textrm{rank}( \bm{C}_r \bm{D}_r^{(3)}(\bm\Theta_{b,r}^{-1} \otimes \bm\Theta_{a,r}^{-1})^\top)\leq \textrm{rank}(\bar{\bm{C}}_{\pi(r)})\leq N,\nonumber\\
N &= \textrm{rank}( \bm{C}_r \bm{D}_r^{(3)}(\bm\Theta_{b,r}^{-1} \otimes \bm\Theta_{a,r}^{-1})^\top)\leq \textrm{rank}(\bar{\bm{D}}_{\pi(r)}^{(3)})\leq N,\nonumber
\end{align}
i.e., $\textrm{rank}(\bar{\bm{C}}_{\pi(r)})= N$ and $\textrm{rank}(\bar{\bm{D}}_{\pi(r)}^{(3)})= N$ almost surely. 

Let $\bm\Theta_{c,r} = \bm{D}_r^{(3)}(\bm\Theta_{b,r}^{-1} \otimes \bm\Theta_{a,r}^{-1})^\top (\bar{\bm{D}}_{\pi(r)}^{(3)})^\dagger\in \mathbb{R}^{N\times N}$, where $(\bar{\bm{D}}_{\pi(r)}^{(3)})^\dagger$ is the pseudo-inverse of $\bar{\bm{D}}_{\pi(r)}^{(3)}$, then $\bar{\bm{C}}_{\pi(r)}=\bm{C}_r\bm \Theta_{c,r}$. 
Due to $N = \textrm{rank}(\bar{\bm{C}}_{\pi(r)})\leq \textrm{rank}(\bm\Theta_{c,r})\leq N$, we have $\textrm{rank}(\bm\Theta_{c,r})= N$, i.e., $\bm\Theta_{c,r}$ is a nonsingular matrix.
We rewrite \eqref{eq:A_DC} as 
\begin{align}\label{eq:A_DC2}
     \underline{\bar{\bm{D}}}_{\pi(r)}\times_1 \bm \Theta_{a,r}\times_2 \bm \Theta_{b,r}\times_3 (\bm{C}_r\bm \Theta_{c,r})=\underline{\bm{D}}_r\times_3 \bm{C}_r.     
\end{align}
By applying mode-3 multiplication by $(\bm{C}_r^\top \bm{C}_r)^{-1}\bm{C}_r^\top$ to both sides of \eqref{eq:A_DC2}, we have 
\begin{align}\label{eq:YCD}
  \underline{\bar{\bm{D}}}_{\pi(r)}\times_1 \bm \Theta_{a,r}\times_2 \bm \Theta_{b,r}\times_3 \bm \Theta_{c,r}=\underline{\bm{D}}_r.
\end{align}

Therefore, the LMN decomposition of $\underline{\bm{Y}}$ is essentially unique almost surely. This completes the proof.

\section{Proof of Theorem \ref{the:recoverability_nonblind}}
\label{app:recoverability_nonblind}
The proof steps resemble those in \cite{Kanatsoulis2018HSR,Ding2021HSR}. Nonetheless, as the LMN structure and its essential uniqueness definition are different from those of the CPD and LL1 models in \cite{Kanatsoulis2018HSR,Ding2021HSR}, careful modifications are made to accommodate the changes in tensor models.

To proceed, we present two important lemmas.

\begin{lemma} \label{lemma:joint} \cite{Kanatsoulis2018HSR}
Let $\widetilde{\bm{A}}=\bm{P}\bm{A}$, where the elements of $\bm{A}$ are drawn from any absolutely continuous joint distribution in $\mathbb{R}^{I\times L}$, and $\bm{P}\in \mathbb{R}^{ I'\times I}$ is full row rank. Then, the elements of $\widetilde{\bm{A}}$ follow an absolutely continuous joint distribution in $\mathbb{R}^{I'\times L}$.
\end{lemma}

\begin{lemma} \label{lemma:rank} \cite{Lathauwer2008BTD1}
 Assume $\{\bm{A}_r\in \mathbb{R}^{I\times L_r}\}_{r=1}^R$ and $\{\bm{B}_r\in \mathbb{R}^{J\times M_r}\}_{r=1}^R$ are drawn from any absolutely continuous joint distributions, then 
$\textrm{rank}([\bm{B}_1\otimes \bm{A}_1,\ldots,\bm{B}_R\otimes \bm{A}_R])=\min(IJ, \sum_{r=1}^R L_r M_r)$.
\end{lemma}

Assume 
$\{\underline{\bm{D}}_r^{\star},\bm{A}_{r}^{\star},\bm{B}_{r}^{\star},\bm{C}_r^{\star}\}_{r=1}^{R}$ to be an optimal solution to \eqref{Non_blind_BTD_model},
and let $\{\underline{\bm{D}}_r,\bm{A}_{r},\bm{B}_{r},\bm{C}_r\}_{r=1}^{R}$ represent the ground-truth factors of the SRI. Then, we have the following expressions:
\begin{align}\label{eq:MSIBTD}
\underline{\bm{Y}}_{\rm M} 
&=\sum_{r=1}^R\underline{\bm{D}}_r^{\star}\times_1 \bm{A}_r^{\star}\times_2 \bm{B}_r^{\star}\times_3  \bm{P}_{\rm M}\bm{C}_r^{\star},\nonumber\\
&=\sum_{r=1}^R\underline{\bm{D}}_r\times_1 \bm{A}_r\times_2 \bm{B}_r\times_3  \bm{P}_{\rm M}\bm{C}_r,
\end{align}
\begin{align}\label{eq:HSIBTD}
\underline{\bm{Y}}_{\rm H} 
&=\sum_{r=1}^R\underline{\bm{D}}_r^{\star}\times_1 (\bm{P}_1\bm{A}_r^{\star})\times_2 (\bm{P}_2\bm{B}_r^{\star})\times_3  \bm{C}_r^{\star},\nonumber\\
&=\sum_{r=1}^R\underline{\bm{D}}_r\times_1 (\bm{P}_1\bm{A}_r)\times_2 (\bm{P}_2\bm{B}_r)\times_3  \bm{C}_r.
\end{align}

Since that $\bm{P}_1$, $\bm{P}_2$, and $\bm{P}_{\rm M}$ are full row rank,
using Lemma \ref{lemma:joint}, $\bm{P}_1\bm{A}_{r}$, $\bm{P}_2\bm{B}_{r}$, and $\bm{P}_{\rm M}\bm{C}_{r}$ are drawn from certain absolutely continuous distributions.
Under the conditions stated in Theorem \ref{the:recoverability_nonblind}, the LMN representation of MSI $\underline{\bm{Y}}_{\rm M}$ is essentially unique, which means that from \eqref{eq:MSIBTD}, the following expressions hold:
\begin{align}\label{eq:Astar}
\bm{A}^{\star}_{\pi(r)}=\bm{A}_r\bm \Theta_{a,r},~
\bm{B}^{\star}_{\pi(r)}=\bm{B}_r\bm \Theta_{b,r},~
\bm{P}_{\rm M}\bm{C}^{\star}_{\pi(r)} = \bm{P}_{\rm M}\bm{C}_r \bm \Theta_{c,r},
\end{align}
where $\pi$ is permutation of: $\{1,\ldots,R\}$, $\bm \Theta_{a,r}$, $\bm \Theta_{b,r}$, and $\bm \Theta_{c,r}$ are nonsingular matrices.

To proceed, we consider the mode-3 unfolding of $\underline{\bm{Y}}_{\rm H}$, i.e., rearranging the pixels of $\underline{\bm{Y}}_{\rm H}$ as following:
\begin{align*}
\bm{Y}_{H}^{(3)}=\left[\underline{\bm{Y}}_{H}(1,1,:),\underline{\bm{Y}}_{H}(2,1,:),\ldots,\underline{\bm{Y}}_{H}(I_{\rm H},J_{\rm H},:)\right].
\end{align*}
From \eqref{eq:HSIBTD} and \eqref{eq:Astar}, we have
\begin{align}\label{eq:YH3}
&\bm{Y}_{H}^{(3)} \\
=&[\bm{C}_1,\ldots, \bm{C}_R]\textrm{blockdiag}(\bm{D}_1^{(3)},\ldots,\bm{D}_R^{(3)})[\bm{P}_2\bm{B}_1\otimes \bm{P}_1\bm{A}_1,\ldots,\bm{P}_2\bm{B}_R\otimes \bm{P}_1\bm{A}_R]^\top  \nonumber\\
=&[\bm{C}_{\pi(1)}^\star,\ldots, \bm{C}_{\pi(R)}^\star]\textrm{blockdiag}(\bm{D}_{\pi(1)}^{\star(3)},\ldots,\bm{D}_{\pi(R)}^{\star(3)})[\bm{P}_2\bm{B}_{\pi(1)}^\star\otimes \bm{P}_1\bm{A}_{\pi(1)}^\star,\ldots,\bm{P}_2\bm{B}_{\pi(R)}^\star\otimes \bm{P}_1\bm{A}_{\pi(R)}^\star]^\top  \nonumber\\
=&[\bm{C}_{\pi(1)}^\star,\ldots, \bm{C}_{\pi(R)}^\star]
\textrm{blockdiag}\left( \bm{D}_{\pi(1)}^{\star(3)}(\bm \Theta_{b,1} \otimes \bm \Theta_{a,1}),\ldots, \bm{D}_{\pi(R)}^{\star(3)}(\bm \Theta_{b,R} \otimes \bm \Theta_{a,R}) \right)  \nonumber\\
&\phantom{=\;\;} \hspace{8cm} [\bm{P}_2\bm{B}_1\otimes \bm{P}_1\bm{A}_1,\ldots,\bm{P}_2\bm{B}_R\otimes \bm{P}_1\bm{A}_R]^\top, \nonumber
\end{align}
where $\otimes$ denotes the Kronecker product. 
Note that $\bm{P}_1\bm{A}_r\in \mathbb{R}^{I_H\times L}$ and $\bm{P}_2\bm{B}_r\in \mathbb{R}^{J_H\times M}$ are drawn from any absolutely continuous distributions, then $\textrm{rank}([\bm{P}_2\bm{B}_1\otimes \bm{P}_1\bm{A}_1,\ldots,\bm{P}_2\bm{B}_R\otimes \bm{P}_1\bm{A}_R])=\min(I_H J_H, LMR)$ using Lemma \ref{lemma:rank}. Due to $I_H J_H\geq LMR$, $[\bm{P}_2\bm{B}_1\otimes \bm{P}_1\bm{A}_1,\ldots,\bm{P}_2\bm{B}_R\otimes \bm{P}_1\bm{A}_R]$ has full column rank.
Multiplying $\left([\bm{P}_2\bm{B}_1\otimes \bm{P}_1\bm{A}_1,\ldots,\bm{P}_2\bm{B}_R\otimes \bm{P}_1\bm{A}_R]^\top\right)^\dag$ to the right of \eqref{eq:YH3}, we obtain that
\begin{align*}
\bm{C}_{\pi(r)}^\star \bm{D}_{\pi(r)}^{\star(3)} (\bm \Theta_{b,r} \otimes \bm \Theta_{a,r}) 
= \bm{C}_r \bm{D}_r^{(3)} ,~1\leq r\leq R,
\end{align*}
i.e.,
\begin{align}\label{eq:DC}
\underline{\bm{D}}_{\pi(r)}^{\star}\times_1 \bm \Theta_{a,r}\times_2 \bm \Theta_{b,r}\times_3 \bm{C}_{\pi(r)}^{\star}=\underline{\bm{D}}_r\times_3 \bm{C}_r.
\end{align}

Combining \eqref{eq:Astar} and \eqref{eq:DC}, one can get
\begin{align}
    \underline{\bm{D}}_{\pi(r)}^{\star}
    \times_1 \bm{A}_{\pi(r)}^{\star}
    \times_2 \bm{B}_{\pi(r)}^{\star}
    \times_3 \bm{C}_{\pi(r)}^{\star}
    =\underline{\bm{D}}_r\times_1 \bm{A}_r
    \times_2 \bm{B}_r\times_3 \bm{C}_r.\nonumber
\end{align}
Hence, the ground-truth can be recovered by
$\underline{\bm{Y}}_{\rm S} = \sum_{r=1}^R\underline{\bm{D}}_r^{\star}\times_1 \bm{A}_r^{\star}\times_2 \bm{B}_r^{\star}\times_3 \bm{C}_r^{\star}$.
This completes the proof.

\section{Proof of Theorem \ref{the:recoverability_blind}}
\label{app:recoverability_blind}
Assume that $(\{\widetilde{\A}_{r}^{\star},\widetilde{\B}_{r}^{\star}\}_{r=1}^{R},\{\underline{\bm{D}}_r^\star,\bm{A}_{r}^{\star},\bm{B}_{r}^{\star},\bm{C}_{r}^{\star}\}_{r=1}^{R})$ is an optimal solution of \eqref{blind_BTD_model}, then
\begin{align}
    &\tY_{\rm H} = \sum_{r=1}^R\underline{\bm{D}}_r^{\star}\times_1 \widetilde{\bm{A}}_r^{\star}\times_2 \widetilde{\bm{B}}_r^{\star}\times_3  \bm{C}_r^{\star}, \label{eq:blind_hsi_constraint} \\
    & \tY_{\rm M} = \sum_{r=1}^R\underline{\bm{D}}_r^{\star}\times_1 \bm{A}_r^{\star}\times_2 \bm{B}_r^{\star}\times_3  (\bm{P}_{\rm M}\bm{C}_r^{\star}).\label{eq:blind_msi_constraint}
\end{align}
Based on the assumptions and Theorem~\ref{the:identifiability_LMN}, the MSI decomposition is essentially unique almost surely. Thus,
we have the following with probability one:
\begin{align}\label{eq:Astar_MSI}
&  \bm{A}^{\star}_{\pi_1(r)}=\bm{A}_r\bm \Theta_{a,r},~
\bm{B}^{\star}_{\pi_1(r)}=\bm{B}_r\bm \Theta_{b,r},~
\bm{P}_{\rm M}\bm{C}^{\star}_{\pi_1(r)} = \bm{P}_{\rm M}\bm{C}_r \bm \Theta_{c,r},\nonumber \\
&  \underline{\bm{D}}^{\star}_{\pi_1(r)}\times_1 \bm \Theta_{a,r}\times_2 \bm \Theta_{b,r}\times_3 \bm \Theta_{c,r}=\underline{\bm{D}}_r, 
\end{align}
where $\pi_1$ is permutation of $\{1,\ldots,R\}$, and $\bm \Theta_{a,r}$, $\bm \Theta_{b,r}$, and $\bm \Theta_{c,r}$ are nonsingular matrices.
The $\pi_1$, $\bm \Theta_{a,r}$, $\bm \Theta_{b,r}$, and $\bm \Theta_{c,r}$ are associated with the MSI decomposition, i.e., the decomposition model in \eqref{eq:blind_msi_constraint}.

Moreover, applying Theorem \ref{the:identifiability_LMN} to the HSI and the equality constraint in \eqref{eq:blind_hsi_constraint}, the following holds almost surely:
\begin{align}\label{eq:Astar_HSI}
&   \widetilde{\bm{A}}^{\star}_{\pi_2(r)}=\bm{A}_r\widetilde{\bm{\Theta}}_{a,r},~
\widetilde{\bm{B}}^{\star}_{\pi_2(r)}=\bm{B}_r\widetilde{\bm{\Theta}}_{b,r},~
\bm{C}^{\star}_{\pi_2(r)} = \bm{C}_r \widetilde{\bm{\Theta}}_{c,r}, \nonumber \\
&  \underline{\bm{D}}^{\star}_{\pi_2(r)}\times_1 \widetilde{\bm{\Theta}}_{a,r}\times_2 \widetilde{\bm{\Theta}}_{b,r}\times_3 \widetilde{\bm{\Theta}}_{c,r}=\underline{\bm{D}}_r, 
\end{align}
where $\pi_2$ is permutation of $\{1,\ldots,R\}$, and $\widetilde{\bm{\Theta}}_{a,r}$, $\widetilde{\bm{\Theta}}_{b,r}$, and $\widetilde{\bm{\Theta}}_{c,r}$ are nonsingular matrices associated with the decomposition model in \eqref{eq:blind_hsi_constraint}.

Next, we hope to show that $\pi_1=\pi_2$ and $\bm \Theta_{c,r}=\widetilde{\bm \Theta}_{c,r}$. 
From \eqref{eq:Astar_MSI} and \eqref{eq:Astar_HSI}, it is seen that
\begin{align}
    &\bm{P}_{\rm M}\bm{C}^{\star}_{r} = \bm{P}_{\rm M}\bm{C}_{\pi_1(r)} \bm \Theta_{c,\pi_1(r)},\\
    &\bm{C}^{\star}_{r} = \bm{C}_{\pi_2({r})} \widetilde{\bm{\Theta}}_{c,\pi_2({r})}.
\end{align}
Combining the above, we have
\begin{align}\label{eq:keyequal}
    \bm P_{\rm M}[\bm C_1,\ldots,\bm C_R     ] \bm \Pi_1  {\bm \Lambda}_c = \bm P_{\rm M}[\bm C_1,\ldots,\bm C_R     ] \bm \Pi_2 \widetilde{\bm \Lambda}_c,
\end{align}
where $\bm \Pi_1\in \mathbb{R}^{NR\times NR}$ and $\bm \Pi_2\in \mathbb{R}^{NR\times NR}$ are block permutation matrices that permutes the blocks $\bm C_1,\ldots,\bm C_R$. 
These block permutation matrices are entirely determined by $\pi_1$ and $\pi_2$.
For example, if $R=2$ and $\pi_1=\{2,1\}$, we have
\begin{align}
    \bm \Pi_1 = \begin{bmatrix}
        \bm 0_{N\times N},~ \bm I_{N\times N},\\
         \bm I_{N\times N},~\bm 0_{N\times N}.
    \end{bmatrix}.
\end{align}
Hence, our task now boils down to proving $\bm \Pi_1=\bm \Pi_2$.
In addition, we have 
\begin{align}
   {\bm \Lambda}_c &={\rm blockdiag}(\bm \Theta_{\c,\pi_1(1)},\ldots,\bm \Theta_{c,\pi_1(R)}),\\
   \widetilde{\bm \Lambda}_c & = {\rm blockdiag}(\widetilde{\bm \Theta}_{\c,\pi_2(1)},\ldots, \widetilde{\bm \Theta}_{c,\pi_2(R)}).
\end{align}
Next, we look at the $q$-th columns on both sides of \eqref{eq:keyequal}:
\begin{align}\label{eq:Z}
    \bm P_{\rm M} \bm C \bm Z_1(:,q) =   \bm P_{\rm M} \bm C \bm Z_2(:,q),~ q=1,\ldots,NR,
\end{align}
where $\bm Z_1 = \bm \Pi_1 \bm \Lambda_c$ and $ \bm Z_2 = \bm \Pi_2 \widetilde{\bm \Lambda}_c$.

By Lemma \ref{lemma:joint}, $\bm{P}_{M}\bm{C}$ is drawn from an absolutely continuous joint distribution. 
By assuming that $K_M\geq 2N$, any $2N$ columns of $\bm{P}_{M}\bm{C}$ are linearly independent.
Re-writing \eqref{eq:Z} as follows:
\begin{align}\label{eq:Z_rewrite}
\bm{P}_{M}\C (\bm Z_1(:,q) -\bm Z_2(:,q) )=\bm 0.     
\end{align}
Note that the matrix $\bm Z_1(:,q) -\bm Z_2(:,q)$ has at most $2N$ nonzero  elements.
However, any $2N$ columns of $\bm{P}_{M}\C$ are linearly independent.
This means that \eqref{eq:Z_rewrite} holds if and only if $\bm Z_1(:,q) = \bm Z_2(:,q)$ for $q=1,\ldots, NR$, which leads to $\bm \Pi_1 =\bm \Pi_2$ and $\bm \Lambda_c =\widetilde{\bm \Lambda}_c$.
Then $\bm{C}^{\star}_{\pi_1(r)} = \bm{C}_r \bm{\Theta}_{c,r}$.
Combining \eqref{eq:Astar_MSI}, we have
\begin{align}
    \underline{\bm{D}}_{\pi_1(r)}^{\star}
    \times_1 \bm{A}_{\pi_1(r)}^{\star}
    \times_2 \bm{B}_{\pi_1(r)}^{\star}
    \times_3 \bm{C}_{\pi_1(r)}^{\star}
    =\underline{\bm{D}}_r\times_1 \bm{A}_r
    \times_2 \bm{B}_r\times_3 \bm{C}_r.\nonumber
\end{align}
Hence, we can recover $\underline{\bm{Y}}_{\rm S} =\sum_{r=1}^R \underline{\bm{D}}_{\pi_1(r)}^{\star} \times_1 \bm{A}_{\pi_1(r)}^{\star} \times_2 \bm{B}_{\pi_1(r)}^{\star} \times_3 \bm{C}_{\pi_1(r)}^{\star}.$
This completes the proof.

\section{Algorithm Details for Model \eqref{eq:model-non-blind}}
\label{app:algorithm-non-blind}
In this section, we present the details for solving Problem \eqref{eq:model-non-blind}, including the algorithm design, parameter setting, and computation analysis.
The algorithm follows the paradigm of alternating accelerated proximal gradient framework in \cite{xu2017globally} and the design of HSR/HU algorithms in \cite{Ding2021HSR,Ding2023Fast}.
The proposed \texttt{CLIMB} for handling \eqref{eq:model-non-blind} is summarized in Algorithm \ref{algorithm-non-blind}.

\begin{algorithm}
\caption{\texttt{CLIMB} for solving \eqref{eq:model-non-blind}.}\label{algorithm-non-blind}
\small
      \textbf{Input:} 
      HSI $\underline{\bm{Y}}_{\rH}$, MSI $\underline{\bm{Y}}_{\rM}$;
      pre-defined sequences of $\mu_{r,i}^{(t)}$ and $\alpha_{r,i}^{(t)}$, $r=1,\ldots,R,~ i=1,\ldots,4$;
      starting points $\left\{\underline{\bm{D}}_r^{(0)}, \bm{A}_{r}^{(0)},\bm{B}_{r}^{(0)}, \bm{C}_r^{(0)} \right\}_{r=1}^{R}$.

      ~~~~\textbf{Parameters:} $\lambda$, $\eta$, $L$, and $N$.

      ~~~~$\check{\bm{A}}_r^{(0)}=\bm{A}_r^{(0)}$, $\check{\bm{B}}_r^{(0)}=\bm{B}_r^{(0)}$, $\check{\bm{C}}_r^{(0)}=\bm{C}_r^{(0)}$, $\check{\underline{\bm{D}}}_r^{(0)}=\underline{\bm{D}}_r^{(0)}$, $t=0$.

      ~~~~\textbf{repeat} 

      ~~~~\textbf{for} $r=1:R$ 

      ~~~~$\%\%$ update $\bm{A}_r$ $\%\%$ 

      ~~~~$\bm{A}_r^{(t+1)}\leftarrow \check{\bm{A}}_r^{(t)}-\alpha_{r,1}^{(t)}\bm{G}_{\check{\bm{A}}_r}^{(t)}$;
      ~~$\check{\bm{A}}_r^{(t+1)}\leftarrow \bm{A}_r^{(t+1)}+\mu_{r,1}^{(t)}(\bm{A}_r^{(t+1)}-\bm{A}_r^{(t)})$; \vspace{0.1cm}

      ~~~~$\%\%$ update $\bm{B}_r$ $\%\%$ 

      ~~~~$\bm{B}_r^{(t+1)}\leftarrow \check{\bm{B}}_r^{(t)}-\alpha_{r,2}^{(t)}\bm{G}_{\check{\bm{B}}_r}^{(t)}$;
      ~~$\check{\bm{B}}_r^{(t+1)}\leftarrow \bm{B}_r^{(t+1)}+\mu_{r,2}^{(t)}(\bm{B}_r^{(t+1)}-\bm{B}_r^{(t)})$; \vspace{0.1cm}

      ~~~~$\%\%$ update $\bm{C}_r$ $\%\%$ 

      ~~~~$\bm{C}_r^{(t+1)}\leftarrow \check{\bm{C}}_r^{(t)}-\alpha_{r,3}^{(t)}\bm{G}_{\check{\bm{C}}_r}^{(t)}$;
      ~~$\check{\bm{C}}_r^{(t+1)}\leftarrow \bm{C}_r^{(t+1)}+\mu_{r,3}^{(t)}(\bm{C}_r^{(t+1)}-\bm{C}_r^{(t)})$; \vspace{0.1cm}

      ~~~~$\%\%$ update $\underline{\bm{D}}_r$ $\%\%$  

      ~~~~$\underline{\bm{D}}_r^{(t+1)}\leftarrow \check{\underline{\bm{D}}}_r^{(t)}-\alpha_{r,4}^{(t)}\bm{G}_{\check{\underline{\bm{D}}}_r}^{(t)}$;
      ~~$\check{\underline{\bm{D}}}_r^{(t+1)}\leftarrow \underline{\bm{D}}_r^{(t+1)}+\mu_{r,4}^{(t)}(\underline{\bm{D}}_r^{(t+1)}-\underline{\bm{D}}_r^{(t)})$; \vspace{0.1cm}

      ~~~~\textbf{end}

      ~~~~$t=t+1$;

      ~~~~\textbf{until} satisfying the stopping criterion. 
    
    \textbf{Output:} $\widehat{\bm{A}}_r=\bm{A}_r^{(t)}$, $\widehat{\bm{B}}_r=\bm{B}_r^{(t)}$, $\widehat{\bm{C}}_r=\bm{C}_r^{(t)}$, and $\widehat{\underline{\bm{D}}}_r=\underline{\bm{D}}_r^{(t)}$. Recover $\underline{\widehat{\bm{Y}}}_{\rm S}=\sum_{r=1}^R\underline{\widehat{\bm{D}}}_r\times_1 \widehat{\bm{A}}_r\times_2 \widehat{\bm{B}}_r\times_3  \widehat{\bm{C}}_r$.
\end{algorithm}

To ensure convergence, it is crucial to select $\{\mu_{r,i}^{(t)}\}$.
As in previous works \cite{Ding2021HSR,Ding2023Fast}, we find that selecting the sequences $\{\mu_{r,i}^{(t)}\}$ following the Nesterov's extrapolation strategy \cite{Nesterov1983Extrapolation} as a heuristic works fairly well.
Therefore, we simply set
\begin{align}
   \mu_{r,i}^{(t)}=\frac{\gamma_{r,i}^{(t)}-1}{\gamma_{r,i}^{(t+1)}}, ~~\gamma_{r,i}^{(t+1)}=\frac{1+\sqrt{1+4\left(\gamma_{r,i}^{(t)}\right)^{2}}}{2}, \nonumber
\end{align}
with $\gamma_{r,i}^{(0)}=1$ for $r=1,\ldots,R,~ i=1,\ldots,4$. In addition, we choose the step sizes $\alpha_{r,i}^{(t)}$ as $\alpha_{r,1}^{(t)} = 1/L_{\bm A_r}^{(t)}, ~\alpha_{r,2}^{(t)} = 1/L_{\bm B_r}^{(t)}, ~\alpha_{r,3}^{(t)} = 1/L_{\bm C_r}^{(t)}, ~\alpha_{r,4}^{(t)} = 1/L_{\underline{\bm D}_r}^{(t)}$, in which $L_{\bm A_r}^{(t)}$, $L_{\bm B_r}^{(t)}$, $L_{\bm C_r}^{(t)}$, and $L_{\underline{\bm D}_r}^{(t)}$ are given in follows.

Next, we give the expressions of gradients $\bm{G}_{\bm{A}_r}$, $\bm{G}_{\bm{B}_r}$, $\bm{G}_{\bm{C}_r}$, and $\bm{G}_{\underline{\bm D}_r}$, and the parameters $L_{\bm A_r}^{(t)}$, $L_{\bm B_r}^{(t)}$, $L_{\bm C_r}^{(t)}$, and $L_{\underline{\bm D}_r}^{(t)}$.
When we update the $r$-th block $\bm{A}_r$, $\bm{B}_r$, $\bm{C}_r$, and $\underline{\bm{D}}_r$ by fixing other variables, the optimization problem is
\begin{align}
\min_{\underline{\bm{D}}_r, \bm{A}_{r},\bm{B}_{r}, \bm{C}_r}
&\frac{1}{2}\left\|\underline{\widetilde{\bm{Y}}}_{H}-\underline{\bm{D}}_r \times_1 (\bm{P}_1\bm{A}_r)\times_2 (\bm{P}_2\bm{B}_r)\times_3  \bm{C}_r\right\|_{F}^2\nonumber\\
+&\frac{1}{2}\left\|\underline{\widetilde{\bm{Y}}}_{M}-\underline{\bm{D}}_r\times_1 \bm{A}_r\times_2 \bm{B}_r\times_3  (\bm{P}_{\rm M}\bm{C}_r)\right\|_{F}^2
+ \lambda  \phi_r(\bm{A}_r,\bm{B}_r,\bm{C}_r) + \frac{\eta}{2}\|\tD_r\|_{F}^2, \nonumber
\end{align}
where $\underline{\widetilde{\bm{Y}}}_{H} = \underline{\bm{Y}}_{H}-\sum_{\tilde{r}\neq r}\underline{\bm{D}}_{\tilde{r}} \times_1 (\bm{P}_1\bm{A}_{\tilde{r}})\times_2 (\bm{P}_2\bm{B}_{\tilde{r}})\times_3  \bm{C}_{\tilde{r}}$ and $\underline{\widetilde{\bm{Y}}}_{M} = \underline{\bm{Y}}_{M}-\sum_{\tilde{r}\neq r}\underline{\bm{D}}_{\tilde{r}} \times_1 \bm{A}_{\tilde{r}}\times_2 \bm{B}_{\tilde{r}}\times_3  (\bm{P}_{M}\bm{C}_{\tilde{r}})$.

For simplicity, we rewrite the $\bm{A}_r$-subproblem as
\begin{align}\label{eq:A}
\mathcal{J}(\bm{A}_r)
=&{\frac{1}{2}\left\|\underline{\widetilde{\bm{Y}}}_{H}-\underline{\bm{D}}_r^{(t)} \times_1 (\bm{P}_1\bm{A}_r)\times_2 (\bm{P}_2\bm{B}_r^{(t)})\times_3  \bm{C}_r^{(t)}\right\|_{F}^2}\nonumber\\
+&{\frac{1}{2}\left\|\underline{\widetilde{\bm{Y}}}_{M}-\underline{\bm{D}}_r^{(t)}\times_1 \bm{A}_r\times_2 \bm{B}_r^{(t)}\times_3  (\bm{P}_{\rm M}\bm{C}_r^{(t)})\right\|_{F}^2}
+\lambda\phi_{p,\varepsilon}(\bm{H}_1\bm{A}_r).\nonumber
\end{align}
Since the design of $\phi_{p,\varepsilon}(\bm{H}_1\bm{A}_r)$, the gradient $\bm{G}_{\bm{A}_r}$ of $\mathcal{J}(\bm{A}_r)$ exists but is hard to compute. Therefore, we first construct a tight quadratic upper bounded function ${\cal L}(\bm{A}_r, \bm{B}_r^{(t)}, \bm{C}_r^{(t)}, \underline{\bm{D}}_r^{(t)};\bm{A}_r^{(t)})$ such that 
${\cal L}(\bm{A}_r^{(t)}, \bm{B}_r^{(t)}, \bm{C}_r^{(t)}, \underline{\bm{D}}_r^{(t)};\bm{A}_r^{(t)}) \geq {\cal J}(\bm{A}_r^{(t)})$
and
$\nabla_{\bm{A}_r}{\cal L}(\bm{A}_r^{(t)}, \bm{B}_r^{(t)}, \bm{C}_r^{(t)}, \underline{\bm{D}}_r^{(t)};\bm{A}_r^{(t)}) = \nabla_{\bm{A}_r}{\cal J}(\bm{A}_r^{(t)})$.

Then, we compute $\bm G_{\bm{A}_r}^{(t)}=\nabla_{\bm{A}_r}{\cal L}(\bm{A}_r^{(t)}, \bm{B}_r^{(t)}, \bm{C}_r^{(t)}, \underline{\bm{D}}_r^{(t)};\bm{A}_r^{(t)})$.
Note that according to the work in \cite{Fu2015Joint}, $\phi_{p,\varepsilon}(\bm{x})$ ($0<p\leq 1$) admits a majorizer $\widetilde{\phi}(\bm{x},\bm{x}^{(t)})$ as
\begin{align}
\widetilde{\phi}(\bm{x},\bm{x}^{(t)}) & = \sum_{i} [\bm{w}^{(t)}]_{i}[\bm{x}]_{i}^{2}+\frac{2-p}{2}\Big(\frac{2}{p}[\bm{w}^{(t)}]_{i}\Big)^{\frac{p}{p-2}}+\varepsilon [\bm{w}^{(t)}]_{i}\nonumber \\
&=\frac{p}{2}\bm{x}^{\top}\bm{W}^{(t)}\bm{x}+{\rm const},
\end{align}
where $[\bm{w}^{(t)}]_{i}=\frac{p}{2}\big(([\bm{x}^{(t)}]_{i})^{2}+\varepsilon\big)^{\frac{p-2}{2}}$,  $\bm{W}^{(t)}$ is a diagonal matrix with $[\bm{W}^{(t)}]_{i,i}=[\bm{w}^{(t)}]_{i}$ and ${\rm const}$ is a constant. Hence, we get the quadratic majorizer as
\begin{align}\label{eq:Aproblem}
&{\cal L}(\bm{A}_r, \bm{B}_r^{(t)}, \bm{C}_r^{(t)}, \underline{\bm{D}}_r^{(t)};\bm{A}_r^{(t)})\\
=&{\frac{1}{2}\left\|\underline{\widetilde{\bm{Y}}}_{H}-\underline{\bm{D}}_r^{(t)} \times_1 (\bm{P}_1\bm{A}_r)\times_2 (\bm{P}_2\bm{B}_r^{(t)})\times_3  \bm{C}_r^{(t)}\right\|_{F}^2}\nonumber\\
+&{\frac{1}{2}\left\|\underline{\widetilde{\bm{Y}}}_{M}-\underline{\bm{D}}_r^{(t)}\times_1 \bm{A}_r\times_2 \bm{B}_r^{(t)}\times_3  (\bm{P}_{\rm M}\bm{C}_r^{(t)})\right\|_{F}^2} 
+\lambda\widetilde{\phi}(\tilde{\bm{H}}_{1}{\bm a}_r,\tilde{\bm{H}}_{1}{\bm a}_r^{(t)}),\nonumber
\end{align}
where $\tilde{\bm{H}}_1 = \bm{H}_1\otimes \bm{I}_{L}$ ($\bm{I}_L$ is the identify matrix with size $L \times L$), and $\bm{a}_r^{(t)}$ is the vectorization of $\bm{A}_r^{(t)}$.

The gradient $\bm G_{\bm{A}_r}^{(t)}$ can be expressed as follows:
\begin{align}\label{eq:gradient_A}
\bm G_{\bm{A}_r}^{(t)}=&\bm{P}_1^{\top}(\bm{P}_1 \bm{A}_r^{(t)} \bm{V}_{r1}^{(t)} - \widetilde{\bm{Y}}_{H}^{(1)})(\bm{V}_{r1}^{(t)})^\top 
+(\bm{A}_r\bm{U}_{r1}^{(t)} - \widetilde{\bm{Y}}_{M}^{(1)})(\bm{U}_{r1}^{(t)})^\top
+\lambda p\tilde{\bm{H}}_{1}^{\top}\bm{W}_{r1}^{(t)}\tilde{\bm{H}}_{1}\bm{a}_{r}^{(t)}, \nonumber
\end{align}
where $\widetilde{\bm{Y}}_{H}^{(1)} = (\underline{\widetilde{\bm{Y}}}_{H})_{(1)}$, $\widetilde{\bm{Y}}_{M}^{(1)} = (\underline{\widetilde{\bm{Y}}}_{M})_{(1)}$, $\bm{V}_{r1}^{(t)}=(\underline{\bm{D}}_r^{(t)}\times_2(\bm{P}_2\bm{B}_r^{(t)})\times_3 \bm{C}_r^{(t)})_{(1)}$, $\bm{U}_{r1}^{(t)}=(\underline{\bm{D}}_r^{(t)}\times_2\bm{B}_r^{(t)}\times_3 (\bm{P}_{\rm M}\bm{C}_r^{(t)}))_{(1)}$, and
$\bm{W}_{r1}^{(t)}$ is diagonal with
$[\bm{W}_{r1}^{(t)}]_{i,i}=([\tilde{\bm{H}}_1{\bm a}_r^{(t)}]_{i}^{2}+\varepsilon)^{\frac{p-2}{2}}$.

The parameter $L_{\bm A_r}^{(t)}$ is determined by the the Lipschitz constant of the $\bm{A}_r$-subproblem \eqref{eq:Aproblem}. One can see that \eqref{eq:Aproblem} is Lipschitz continuous with the following Lipschitz constant:
\begin{align}
L_{\bm A_r}^{(t)}&
=\sigma_{\rm max}\left(\bm{P}_1^\top\bm{P}_1\right)\sigma_{\rm max}\left(\bm{V}_{r1}^{(t)}(\bm{V}_{r1}^{(t)})^\top\right)
+ \sigma_{\rm max}\left(\bm{U}_{r1}^{(t)}(\bm{U}_{r1}^{(t)})^\top\right)
+ \lambda p\sigma_{\textrm{max}}\left(\tilde{\bm{H}}_1^{\top}\bm{W}_{r1}^{(t)}\tilde{\bm{H}}_1 \right).\nonumber
\end{align}

Similarly, we can obtain the gradients $\bm G_{\bm{B}_r}^{(t)}$, $\bm G_{\bm{C}_r}^{(t)}$, and parameters $L_{\bm B_r}^{(t)}$, $L_{\bm C_r}^{(t)}$ as follows.

\begin{align}
\bm G_{\bm{B}_r}^{(t)}&=\bm{P}_2^{\top}(\bm{P}_2 \bm{B}_r^{(t)} \bm{V}_{r2}^{(t)} - \widetilde{\bm{Y}}_{H}^{(2)})(\bm{V}_{r2}^{(t)})^\top
+ (\bm{B}_r\bm{U}_{r2}^{(t)} - \widetilde{\bm{Y}}_{M}^{(2)})(\bm{U}_{r2}^{(t)})^\top
+\lambda p\tilde{\bm{H}}_{2}^{\top}\bm{W}_{r2}^{(t)}\tilde{\bm{H}}_{2}\bm{b}_{r}^{(t)},\nonumber\\
L_{\bm B_r}^{(t)}&
=\sigma_{\rm max}\left(\bm{P}_2^\top\bm{P}_2\right)\sigma_{\rm max}\left(\bm{V}_{r2}^{(t)}(\bm{V}_{r2}^{(t)})^\top\right)
+ \sigma_{\rm max}\left(\bm{U}_{r2}^{(t)}(\bm{U}_{r2}^{(t)})^\top\right)
+ \lambda p\sigma_{\textrm{max}}\left(\tilde{\bm{H}}_2^{\top}\bm{W}_{r2}^{(t)}\tilde{\bm{H}}_2 \right),\nonumber
\end{align}
where $\bm{V}_{r2}^{(t)}=(\underline{\bm{D}}_r^{(t)}\times_1(\bm{P}_1\bm{A}_r^{(t+1)})\times_3 \bm{C}_r^{(t)})_{(2)}$, $\bm{U}_{r2}^{(t)}=(\underline{\bm{D}}_r^{(t)}\times_1\bm{A}_r^{(t+1)}\times_3 (\bm{P}_{M}\bm{C}_r^{(t)}))_{(2)}$, $\tilde{\bm{H}}_2 = \bm{H}_2\otimes \bm{I}_{M}$ ($\bm{I}_{\rm M}$ is the identify matrix with size $M \times M$), and
$\bm{W}_{r2}^{(t)}$ is a diagonal matrix with
$[\bm{W}_{r2}^{(t)}]_{i,i}=([\tilde{\bm{H}}_2{\bm b}_r^{(t)}]_{i}^{2}+\varepsilon)^{\frac{p-2}{2}}$, where $\bm{b}_r^{(t)}$ is the vectorization of $\bm{B}_r^{(t)}$.

\begin{table}[!t]
\renewcommand\arraystretch{1}
\setlength{\tabcolsep}{2pt}
\renewcommand\arraystretch{1.5}
\centering
\caption{Complexity of each term of \texttt{CLIMB}. }
\resizebox{\linewidth}{!}{
    \begin{tabular}{c|c}\hline

    \hline
    Terms  & \multicolumn{1}{c}{Complexity}  \\ \hline
    $\bm{P}_1^{\top}(\bm{P}_1 \bm{A}_r^{(t)} \bm{V}_{r1}^{(t)} - \widetilde{\bm{Y}}_{H}^{(1)})(\bm{V}_{r1}^{(t)})^\top$   & $\mathcal{O}(J_{\rm H} J_{\rm M} M_r + I_{\rm H} N_r(L_r M_r + M_r J_{\rm H} + J_{\rm H} K_{\rm H})+ J_{\rm H} L_r(N_r M_r + N_r K_{\rm H} + I_{\rm H} K_{\rm H}) + I_{\rm M} J_{\rm H}K_{\rm H}(I_{\rm H}+L_r))$   \\ \hline
    $(\bm{A}_r\bm{U}_{r1}^{(t)} - \widetilde{\bm{Y}}_{M}^{(1)})(\bm{U}_{r1}^{(t)})^\top$    & $\mathcal{O}(K_{\rm H} K_{\rm M} N_r + L_rN_r (J_{\rm M} M_r + J_{\rm M} K_{\rm M} + I_{\rm M} L_r) + I_{\rm M}J_{\rm M} (N_r M_r + N_r K_{\rm M} +  K_{\rm M} L_r))$ \\ \hline
    $\bm{P}_2^{\top}(\bm{P}_2 \bm{B}_r^{(t)} \bm{V}_{r2}^{(t)} - \widetilde{\bm{Y}}_{H}^{(2)})(\bm{V}_{r2}^{(t)})^\top$   & $\mathcal{O}(I_{\rm H} I_{\rm M} L_r + I_{\rm H} N_r(L_r M_r + M_r J_{\rm H} + J_{\rm H} K_{\rm H}) + I_{\rm H} K_{\rm H}(I_{\rm M} M_r + M_r J_{\rm H} + I_{\rm M} K_{\rm H}) + I_{\rm H} M_rN_r(K_{\rm H}+L_r))$\\ \hline
    $(\bm{B}_r\bm{U}_{r2}^{(t)} - \widetilde{\bm{Y}}_{M}^{(2)})(\bm{U}_{r2}^{(t)})^\top$
    & $\mathcal{O}(K_{\rm H} K_{\rm M} N_r + L_rN_r (I_{\rm M} L_r + I_{\rm M} K_{\rm M} + I_{\rm M} M_r) + I_{\rm M}J_{\rm M} (N_r M_r + N_r K_{\rm M} +  K_{\rm M} M_r))$\\ \hline
    $\bm{P}_{M}^{\top}(\bm{P}_{M} \bm{C}_r^{(t)} \bm{V}_{r3}^{(t)} - \widetilde{\bm{Y}}_{M}^{(3)})(\bm{V}_{r3}^{(t)})^\top$   
    & $\mathcal{O}(I_{\rm M} N_r (M_r L_r + J_{\rm M} K_{\rm M} + J_{\rm M} M_r) + I_{\rm M}J_{\rm M} (K_{\rm H}N_r + K_{\rm M}K_{\rm H} +  K_{\rm M} M_r)+ I_{\rm M} M_r N_r(J_{\rm M}+L_r))$\\ \hline
    $(\bm{C}_r\bm{U}_{r3}^{(t)} - \widetilde{\bm{Y}}_{H}^{(3)})(\bm{U}_{r3}^{(t)})^\top$
    & $\mathcal{O}(I_{\rm H} I_{\rm M} L_r + J_{\rm H} J_{\rm M} M_r + I_{\rm H} N_r(M_rJ_{\rm H}+M_rL_r + J_{\rm H} K_{\rm H}))$\\ \hline
    $(\bm{U}^{(t)}_{\bm{D}_r})^{\top}
(\bm{U}^{(t)}_{\bm{D}_r}\bm{d}_r^{(t)} - \bm{y}_{m})$   &
    $\mathcal{O}(K_{\rm M} L_r(I_{\rm M}J_{\rm M}+J_{\rm M} M_r+M_rN_r))$\\ \hline
    $(\bm{V}^{(t)}_{\bm{D}_r})^{\top}
(\bm{V}^{(t)}_{\bm{D}_r}\bm{d}_r^{(t)} - \bm{y}_{h})$
    & $\mathcal{O}(K_{\rm H} L_r(I_{\rm H} J_{\rm H} + J_{\rm H} M_r+M_rN_r))$\\\hline
    $\tilde{\bm{H}}_{1}^{\top}\bm{W}_{r1}^{(t)}\tilde{\bm{H}}_{1}\bm{a}_{r}^{(t)}$
    & $\mathcal{O}(I_{\rm M} L_r)$\\\hline
    $\tilde{\bm{H}}_{2}^{\top}\bm{W}_{r2}^{(t)}\tilde{\bm{H}}_{2}\bm{b}_{r}^{(t)}$
    & $\mathcal{O}(J_{\rm M} M_r)$\\\hline
    $\bm{H}_{3}^{\top}\bm{H}_{3}\bm{C}_{r}^{(t)}$
    & $\mathcal{O}(K_{\rm H}^2 N_r)$\\\hline
    $\alpha_{r,i}$
    & $\mathcal{O}(L_r^3 + M_r^3 + N_r^3)$\\

    \hline
    \end{tabular}}%
  \label{table:complexity}%
\end{table}%

\begin{align}
\bm G_{\bm{C}_r}^{(t)}&=\bm{P}_{M}^{\top}(\bm{P}_{M} \bm{C}_r^{(t)} \bm{V}_{r3}^{(t)} - \widetilde{\bm{Y}}_{M}^{(3)})(\bm{V}_{r3}^{(t)})^\top 
+(\bm{C}_r\bm{U}_{r3}^{(t)} - \widetilde{\bm{Y}}_{H}^{(3)})(\bm{U}_{r3}^{(t)})^\top
+2 \lambda \bm{H}_{3}^{\top}\bm{H}_{3}\bm{C}_{r}^{(t)}, \nonumber\\
L_{\bm C_r}^{(t)}&
=\sigma_{\rm max}\left(\bm{P}_{M}^\top\bm{P}_{M}\right)\sigma_{\rm max}\left(\bm{V}_{r3}^{(t)}(\bm{V}_{r3}^{(t)})^\top\right)+ \sigma_{\rm max}\left(\bm{U}_{r3}^{(t)}(\bm{U}_{r3}^{(t)})^\top\right)
+2 \lambda \sigma_{\textrm{max}}\left( \bm{H}_{3}^{\top}\bm{H}_{3} \right),\nonumber
\end{align}
where $\bm{V}_{r3}^{(t)}=(\underline{\bm{D}}_r^{(t)}\times_1\bm{A}_r^{(t+1)}\times_2 \bm{B}_r^{(t+1)})_{(3)}$ and $\bm{U}_{r3}^{(t)}=(\underline{\bm{D}}_r^{(t)}\times_1(\bm{P}_1\bm{A}_r^{(t+1)})\times_2 (\bm{P}_2\bm{B}_r^{(t+1)})_{(3)}$.

The $\underline{\bm{D}}_r$-subproblem can be rewritten as
\begin{align}\label{eq:D}
\mathcal{J}(\underline{\bm{D}}_r)
=&{\frac{1}{2}\left\|\underline{\widetilde{\bm{Y}}}_{H}-\underline{\bm{D}}_r \times_1 (\bm{P}_1\bm{A}_r^{(t+1)})\times_2 (\bm{P}_2\bm{B}_r^{(t+1)})\times_3  \bm{C}_r^{(t+1)}\right\|_{F}^2}\nonumber\\
+&{\frac{1}{2}\left\|\underline{\widetilde{\bm{Y}}}_{M}-\underline{\bm{D}}_r\times_1 \bm{A}_r^{(t+1)}\times_2 \bm{B}_r^{(t+1)}\times_3  (\bm{P}_{\rm M}\bm{C}_r^{(t+1)})\right\|_{F}^2 }
+\frac{\eta}{2}\|\underline{\bm{D}}_r\|_{F}^2.\nonumber
\end{align}

Then one can easily obtain the gradient $\bm G_{\underline{\bm{D}}_r}^{(t)}$ and parameter $L_{\underline{\bm D}_r}^{(t)}$.
\begin{align}
\bm G_{\underline{\bm{D}}_r}^{(t)}&=(\bm{U}^{(t)}_{\underline{\bm{D}}_r})^{\top}
(\bm{U}^{(t)}_{\underline{\bm{D}}_r}\bm{d}_r^{(t)} - \bm{y}_{m})
+(\bm{V}^{(t)}_{\underline{\bm{D}}_r})^{\top}
(\bm{V}^{(t)}_{\underline{\bm{D}}_r}\bm{d}_r^{(t)} - \bm{y}_{h})+\eta \bm{d}_{r}^{(t)},\nonumber \\
L_{\underline{\bm D}_r}^{(t)}
&=\sigma^2_{\rm max}\left(\bm{P}_{M}\bm{C}^{(t+1)}\right)\sigma^2_{\rm max}\left(\bm{A}_r^{(t+1)}\right)\sigma^2_{\rm max}\left(\bm{B}_r^{(t+1)}\right)\nonumber \\
&+ \sigma^2_{\rm max}\left(\bm{C}_r^{(t+1)}\right)\sigma^2_{\rm max}\left(\bm{P}_1\bm{A}_r^{(t+1)}\right)\sigma^2_{\rm max}\left(\bm{P}_2\bm{B}_r^{(t+1)}\right)+\eta,\nonumber
\end{align}
where $\bm{U}^{(t)}_{\underline{\bm{D}}_r}=(\bm{P}_{M}\bm{C}_r^{(t+1)})\otimes\bm{B}_r^{(t+1)}\otimes\bm{A}_r^{(t+1)}$, $\bm{V}^{(t)}_{\underline{\bm{D}}_r}=\bm{C}_r^{(t+1)}\otimes(\bm{P}_1\bm{A}_r^{(t+1)})\otimes(\bm{P}_2\bm{B}_r^{(t+1)})$, $\bm{d}_r^{(t)}$, $\bm{y}_m$, and $\bm{y}_h$ denote the vectorization of $\underline{\bm{D}}_r^{(t)}$, $\underline{\bm{Y}}_{\rm M}$, and $\underline{\bm{Y}}_{\rm H}$, respectively.

The detailed complexity cost of the proposed algorithm is listed in Table \ref{table:complexity}. Assuming $I_{\rm M} \approx J_{\rm M} \approx K_{\rm H}, ~I_{\rm H} \approx J_{\rm H} \geq K_{\rm M},~ L_r \approx M_r \geq N_r$, the computation of the gradients $\bm{G}^{(t)}_{\bm{A}_r}$, $\bm{G}^{(t)}_{\bm{B}_r}$, $\bm{G}^{(t)}_{\bm{C}_r}$, and $\bm{G}^{(t)}_{\underline{\bm{D}}_r}$ takes
$\mathcal{O}( I_{\rm M}^2(I_{\rm H}(I_{\rm H}+L_r)+L_r(K_{\rm M}+N_r)) )$,
$\mathcal{O}( I_{\rm M}^2(I_{\rm H}(I_{\rm H}+L_r)+L_r(K_{\rm M}+N_r)) )$,
$\mathcal{O}( I_{\rm M}^2 K_{\rm H}(K_{\rm M}+N_r) )$, and
$\mathcal{O}( I_{\rm M}^2 L_r K_{\rm M}+ K_{\rm H}L_r(I_{\rm H}^2 + I_{\rm H} L_r + L_rN_r) )$, respectively.

Another part costs many computation flops is to compute the step sizes $\alpha_{r,i}^{(t)}$. To reduce the computation, instead of computing the exact values, we compute their upper bounds as in \cite{Ding2021HSR,Ding2023Fast}. For example, we have
\begin{equation}
    \sigma_{\textrm{max}}(\tilde{\bm{H}}_1^{\top}\bm{W}_{r1}^{(t)}\tilde{\bm{H}}_1)\leq \sigma_{\textrm{max}}(\tilde{\bm{H}}_1^{\top})\sigma_{\max}(\bm{W}_{r1}^{(t)})\sigma_{\rm max}(\tilde{\bm{H}}_1),\nonumber
\end{equation}
where the symbol $\sigma_{\max}(\bm{X})$ denotes the largest singular value of the matrix $\bm{X}$. Here, $\sigma_{\max}(\bm{W}_{r1}^{(t)})$ is simply chosen to be the largest value of the diagonal matrix $\bm{W}_{r1}^{(t)}$ and the other two terms are pre-computed. Therefore, computing $\alpha_{r,i}^{(t)}$ takes $\mathcal{O}(L_r^3+M_r^3+N_r^3)$ flops. In summary, the per-iteration complexity is 
$\mathcal{O}(I_{\rm M}^3 (I_{\rm H} + L_r))$.

\begin{algorithm}
\caption{\texttt{BCLIMB} for solving \eqref{eq:model-blind}.}\label{algorithm-blind}
\small
      \textbf{Input:} 
      HSI $\underline{\bm{Y}}_{\rH}$, MSI $\underline{\bm{Y}}_{\rM}$;
      pre-defined sequences of $\mu_{r,i}^{(t)}$, $\alpha_{r,i}^{(t)}$, $\widetilde{\mu}_{r,j}^{(t)}$, and $\widetilde{\alpha}_{r,j}^{(t)}$ ($r=1,\ldots,R,~ i=1,\ldots,4,~j=1,2$);
      starting points $\left\{\underline{\bm{D}}_r^{(0)}, \widetilde{\bm{A}}_{r}^{(0)},\widetilde{\bm{B}}_{r}^{(0)}, \bm{A}_{r}^{(0)},\bm{B}_{r}^{(0)}, \bm{C}_r^{(0)} \right\}_{r=1}^{R}$. 
      
      ~~~~\textbf{Parameters:} $\lambda$, $\eta$, $p$, $\varepsilon$, $L$, and $N$. 
      
      ~~~~$\check{\underline{\bm{D}}}_r^{(0)}=\underline{\bm{D}}_r^{(0)}$, $\bar{\bm{A}}_r^{(0)}=\widetilde{\bm{A}}_r^{(0)}$, $\bar{\bm{B}}_r^{(0)}=\widetilde{\bm{B}}_r^{(0)}$, $\check{\bm{A}}_r^{(0)}=\bm{A}_r^{(0)}$, $\check{\bm{B}}_r^{(0)}=\bm{B}_r^{(0)}$, $\check{\bm{C}}_r^{(0)}=\bm{C}_r^{(0)}$, $t=0$.
      
      ~~~~\textbf{repeat} 
      
      ~~~~\textbf{for} $r=1:R$  
      
      ~~~~$\%\%$ update $\widetilde{\bm{A}}_r$ $\%\%$ 
     
      ~~~~$\widetilde{\bm{A}}_r^{(t+1)}\leftarrow \bar{\bm{A}}_r^{(t)}-\widetilde{\alpha}_{r,1}^{(t)}\bm{G}_{\bar{\bm{A}}_r}^{(t)}$;
      ~~$\bar{\bm{A}}_r^{(t+1)}\leftarrow \widetilde{\bm{A}}_r^{(t+1)}+\widetilde{\mu}_{r,1}^{(t)}(\widetilde{\bm{A}}_r^{(t+1)}-\widetilde{\bm{A}}_r^{(t)})$;\vspace{0.1cm}
      
      ~~~~$\%\%$ update $\bm{A}_r$ $\%\%$ \vspace{0.1cm}
      
      ~~~~$\bm{A}_r^{(t+1)}\leftarrow \check{\bm{A}}_r^{(t)}-\alpha_{r,1}^{(t)}\bm{G}_{\check{\bm{A}}_r}^{(t)}$;
      ~~$\check{\bm{A}}_r^{(t+1)}\leftarrow \bm{A}_r^{(t+1)}+\mu_{r,1}^{(t)}(\bm{A}_r^{(t+1)}-\bm{A}_r^{(t)})$; 
      
      ~~~~$\%\%$ update $\widetilde{\bm{B}}_r$ $\%\%$ \vspace{0.1cm}
     
      ~~~~$\widetilde{\bm{B}}_r^{(t+1)}\leftarrow \bar{\bm{B}}_r^{(t)}-\widetilde{\alpha}_{r,2}^{(t)}\bm{G}_{\bar{\bm{B}}_r}^{(t)}$;
      ~~$\bar{\bm{B}}_r^{(t+1)}\leftarrow \widetilde{\bm{B}}_r^{(t+1)}+\widetilde{\mu}_{r,2}^{(t)}(\widetilde{\bm{B}}_r^{(t+1)}-\widetilde{\bm{B}}_r^{(t)})$; 
      
      ~~~~$\%\%$ update $\bm{B}_r$ $\%\%$ \vspace{0.1cm}
      
      ~~~~$\bm{B}_r^{(t+1)}\leftarrow \check{\bm{B}}_r^{(t)}-\alpha_{r,2}^{(t)}\bm{G}_{\check{\bm{B}}_r}^{(t)}$;
      ~~$\check{\bm{B}}_r^{(t+1)}\leftarrow \bm{B}_r^{(t+1)}+\mu_{r,2}^{(t)}(\bm{B}_r^{(t+1)}-\bm{B}_r^{(t)})$; 
      
      ~~~~$\%\%$ update $\bm{C}_r$ $\%\%$ \vspace{0.1cm}
      
      ~~~~$\bm{C}_r^{(t+1)}\leftarrow \check{\bm{C}}_r^{(t)}-\alpha_{r,3}^{(t)}\bm{G}_{\check{\bm{C}}_r}^{(t)}$;
      ~~$\check{\bm{C}}_r^{(t+1)}\leftarrow \bm{C}_r^{(t+1)}+\mu_{r,3}^{(t)}(\bm{C}_r^{(t+1)}-\bm{C}_r^{(t)})$; 
      
      ~~~~$\%\%$ update $\underline{\bm{D}}_r$ $\%\%$ \vspace{0.1cm}
      
      ~~~~$\underline{\bm{D}}_r^{(t+1)}\leftarrow \check{\underline{\bm{D}}}_r^{(t)}-\alpha_{r,4}^{(t)}\bm{G}_{\check{\underline{\bm{D}}}_r}^{(t)}$;
      ~~$\check{\underline{\bm{D}}}_r^{(t+1)}\leftarrow \underline{\bm{D}}_r^{(t+1)}+\mu_{r,4}^{(t)}(\underline{\bm{D}}_r^{(t+1)}-\underline{\bm{D}}_r^{(t)})$; 
       
      ~~~~\textbf{end}
      
      ~~~~$t=t+1$;
	  
	 ~~~~ \textbf{until} satisfying the stopping criterion.
     
      \textbf{Output:}  $\widehat{\bm{A}}_r=\bm{A}_r^{(t)}$, $\widehat{\bm{B}}_r=\bm{B}_r^{(t)}$, $\widehat{\bm{C}}_r=\bm{C}_r^{(t)}$, and $\widehat{\underline{\bm{D}}}_r=\underline{\bm{D}}_r^{(t)}$. Recover $\underline{\widehat{\bm{Y}}}_{\rm S}=\sum_{r=1}^R\underline{\widehat{\bm{D}}}_r\times_1 \widehat{\bm{A}}_r\times_2 \widehat{\bm{B}}_r\times_3  \widehat{\bm{C}}_r$.
\end{algorithm}

\section{Algorithm Details for Model \eqref{eq:model-blind}}
\label{app:algorithm-blind}
In this section, we present the details for our semi-blind \texttt{CLIMB} (\texttt{BCLIMB}) for handling \eqref{eq:model-blind}, which is 
a simple extension of \texttt{CLIMB} and 
is summarized in Algorithm \ref{algorithm-blind}.
When the spatial degradation operators are unknown, one can apply the similar calculations of gradients and parameters as in Appendix \ref{app:algorithm-non-blind}. 
Then we can obtain the following expressions.

\begin{align}
\bm G_{\widetilde{\bm{A}}_r}^{(t)}&=(\widetilde{\bm{A}}_r^{(t)} \bar{\bm{V}}_{r1}^{(t)} - \bm{Y}_{H}^{(1)})(\bar{\bm{V}}_{r1}^{(t)})^\top,\nonumber\\
L_{\widetilde{\bm{A}}_r}^{(t)}&
=\sigma_{\rm max}\left(\bar{\bm{V}}_{r1}^{(t)}(\bar{\bm{V}}_{r1}^{(t)})^\top\right),\nonumber
\end{align}
where $\bar{\bm{V}}_{r1}^{(t)}=(\underline{\bm{D}}^{(t)}\times_2\widetilde{\bm{B}}_r^{(t)}\times_3 \bm{C}_r^{(t)})_{(1)}$ and $\bm{Y}_{H}^{(1)} = (\underline{\bm{Y}}_{H}-\sum_{\bar{r}\neq r}\underline{\bm{D}}^{(t)}_{\bar{r}} \times_1 \widetilde{\bm{A}}^{(t)}_{\bar{r}}\times_2 \widetilde{\bm{B}}^{(t)}_{\bar{r}}\times_3  \bm{C}^{(t)}_{\bar{r}})_{(1)}$.

\begin{table}[!t]
\renewcommand\arraystretch{1}
\setlength{\tabcolsep}{2pt}
\renewcommand\arraystretch{1.5}
\centering
\caption{Complexity of each term of \texttt{BCLIMB}. }
\resizebox{\linewidth}{!}{
    \begin{tabular}{c|c}\hline

    \hline
    Terms  & \multicolumn{1}{c}{Complexity}  \\ \hline
    $(\widetilde{\bm{A}}_r^{(t)} \bar{\bm{V}}_{r1}^{(t)} - \bm{Y}_{H}^{(1)})(\bar{\bm{V}}_{r1}^{(t)})^\top$   & $\mathcal{O}(I_{\rm H} N_r(L_r M_r + M_r J_{\rm H} + J_{\rm H} K_{\rm H})+ J_{\rm H} L_r(N_r M_r + N_r K_{\rm H} + I_{\rm H} K_{\rm H}))$   \\ \hline
    $(\bm{A}_r\bar{\bm{U}}_{r1}^{(t)} - \bm{Y}_{M}^{(1)})(\bar{\bm{U}}_{r1}^{(t)})^\top$    & $\mathcal{O}(K_{\rm H} K_{\rm M} N_r + L_rN_r (J_{\rm M} M_r + J_{\rm M} K_{\rm M} + I_{\rm M} L_r) + I_{\rm M}J_{\rm M} (N_r M_r + N_r K_{\rm M} +  K_{\rm M} L_r))$ \\ \hline
    $(\widetilde{\bm{B}}_r^{(t)} \bar{\bm{V}}_{r2}^{(t)} - \bm{Y}_{H}^{(2)})(\bar{\bm{V}}_{r2}^{(t)})^\top$   & $\mathcal{O}(I_{\rm H} N_r(L_r M_r + M_r J_{\rm H} + J_{\rm H} K_{\rm H}) + I_{\rm H} K_{\rm H}(I_{\rm M} M_r + M_r J_{\rm H} + I_{\rm M} K_{\rm H}))$\\ \hline
    $(\bm{B}_r\bar{\bm{U}}_{r2}^{(t)} - \bm{Y}_{M}^{(2)})(\bar{\bm{U}}_{r2}^{(t)})^\top$
    & $\mathcal{O}(K_{\rm H} K_{\rm M} N_r + L_rN_r (I_{\rm M} L_r + I_{\rm M} K_{\rm M} + I_{\rm M} M_r) + I_{\rm M}J_{\rm M} (N_r M_r + N_r K_{\rm M} +  K_{\rm M} M_r))$\\ \hline
    $\bm{P}_{M}^{\top}(\bm{P}_{M} \bm{C}_r^{(t)} \bar{\bm{V}}_{r3}^{(t)} - \bm{Y}_{M}^{(3)})(\bar{\bm{V}}_{r3}^{(t)})^\top$   
    & $\mathcal{O}(I_{\rm M} N_r (M_r L_r + J_{\rm M} K_{\rm M} + J_{\rm M} M_r) + I_{\rm M}J_{\rm M} (K_{\rm H}N_r + K_{\rm M}K_{\rm H} +  K_{\rm M} M_r)+ I_{\rm M} M_r N_r(J_{\rm M}+L_r))$\\ \hline
    $(\bm{C}_r\bar{\bm{U}}_{r3}^{(t)} - \bm{Y}_{H}^{(3)})(\bar{\bm{U}}_{r3}^{(t)})^\top$
    & $\mathcal{O}(I_{\rm H} I_{\rm M} L_r + J_{\rm H} J_{\rm M} M_r + I_{\rm H} N_r(M_rJ_{\rm H}+M_rL_r + J_{\rm H} K_{\rm H}))$\\ \hline
    $(\bar{\bm{U}}^{(t)}_{\underline{\bm{D}}_r})^{\top}
(\bar{\bm{U}}^{(t)}_{\underline{\bm{D}}_r}\bm{d}_r^{(t)} - \bm{y}_{m})$   &
    $\mathcal{O}(K_{\rm M} L_r(I_{\rm M}J_{\rm M}+J_{\rm M} M_r+M_rN_r))$\\ \hline
    $(\bar{\bm{V}}^{(t)}_{\underline{\bm{D}}_r})^{\top}
(\bar{\bm{V}}^{(t)}_{\underline{\bm{D}}_r}\bm{d}_r^{(t)} - \bm{y}_{h})$
    & $\mathcal{O}(K_{\rm H} L_r(I_{\rm H} J_{\rm H} + J_{\rm H} M_r+M_rN_r))$\\\hline
    $\tilde{\bm{H}}_{1}^{\top}\bm{W}_{r1}^{(t)}\tilde{\bm{H}}_{1}\bm{a}_{r}^{(t)}$
    & $\mathcal{O}(I_{\rm M} L_r)$\\\hline
    $\tilde{\bm{H}}_{2}^{\top}\bm{W}_{r2}^{(t)}\tilde{\bm{H}}_{2}\bm{b}_{r}^{(t)}$
    & $\mathcal{O}(J_{\rm M} M_r)$\\\hline
    $\bm{H}_{3}^{\top}\bm{H}_{3}\bm{C}_{r}^{(t)}$
    & $\mathcal{O}(K_{\rm H}^2 N_r)$\\\hline
    $\alpha_{r,i}$, $\widetilde{\alpha}_{r,j}$
    & $\mathcal{O}(L_r^3 + M_r^3 + N_r^3)$\\

    \hline
    \end{tabular}}%
  \label{table:complexity_blind}%
\end{table}%

\begin{align}
\bm G_{\bm{A}_r}^{(t)}&=(\bm{A}_r\bar{\bm{U}}_{r1}^{(t)} - \bm{Y}_{M}^{(1)})(\bar{\bm{U}}_{r1}^{(t)})^\top
+\lambda p\tilde{\bm{H}}_{1}^{\top}\bm{W}_{r1}^{(t)}\tilde{\bm{H}}_{1}\bm{a}_{r}^{(t)},\nonumber\\
L_{\bm A_r}^{(t)}&
=\sigma_{\rm max}\left(\bar{\bm{U}}_{r1}^{(t)}(\bar{\bm{U}}_{r1}^{(t)})^\top\right)
+\lambda p\sigma_{\textrm{max}}\left(\tilde{\bm{H}}_1^{\top}\bm{W}_{r1}^{(t)}\tilde{\bm{H}}_1 \right),\nonumber
\end{align}
where $\bar{\bm{U}}_{r1}^{(t)}=(\underline{\bm{D}}_r^{(t)}\times_2\bm{B}_r^{(t)}\times_3 (\bm{P}_{\rm M}\bm{C}_r^{(t)}))_{(1)}$ and $\bm{Y}_{M}^{(1)} = (\underline{\bm{Y}}_{M}-\sum_{\bar{r}\neq r}\underline{\bm{D}}^{(t)}_{\bar{r}} \times_1 \bm{A}^{(t)}_{\bar{r}}\times_2 \bm{B}^{(t)}_{\bar{r}}\times_3  (\bm{P}_{M}\bm{C}^{(t)}_{\bar{r}}))_{(1)}$.

\begin{align}
\bm G_{\widetilde{\bm{B}}_r}^{(t)}&=(\widetilde{\bm{B}}_r^{(t)} \bar{\bm{V}}_{r2}^{(t)} - \bm{Y}_{H}^{(2)})(\bar{\bm{V}}_{r2}^{(t)})^\top,\nonumber\\
L_{\widetilde{\bm{B}}_r}^{(t)}&
=\sigma_{\rm max}\left(\bar{\bm{V}}_{r2}^{(t)}(\bar{\bm{V}}_{r2}^{(t)})^\top\right),\nonumber
\end{align}
where $\bar{\bm{V}}_{r2}^{(t)}=(\underline{\bm{D}}_r^{(t)}\times_1 \widetilde{\bm{A}}_r^{(t+1)} \times_3 \bm{C}_r^{(t)})_{(2)}$ and $\bm{Y}_{H}^{(2)} = (\underline{\bm{Y}}_{H}-\sum_{\bar{r}\neq r}\underline{\bm{D}}^{(t)}_{\bar{r}} \times_1 \widetilde{\bm{A}}^{(t)}_{\bar{r}}\times_2 \widetilde{\bm{B}}^{(t)}_{\bar{r}}\times_3  \bm{C}^{(t)}_{\bar{r}})_{(2)}$.

\begin{align}
\bm G_{\bm{B}_r}^{(t)}&=(\bm{B}_r\bar{\bm{U}}_{r2}^{(t)} - \bm{Y}_{M}^{(2)})(\bar{\bm{U}}_{r2}^{(t)})^\top
+\lambda p\tilde{\bm{H}}_{2}^{\top}\bm{W}_{r2}^{(t)}\tilde{\bm{H}}_{2}\bm{b}_{r}^{(t)},\nonumber\\
L_{\bm B_r}^{(t)}&
= \sigma_{\rm max}\left(\bar{\bm{U}}_{r2}^{(t)}(\bar{\bm{U}}_{r2}^{(t)})^\top\right)
+\lambda p\sigma_{\textrm{max}}\left(\tilde{\bm{H}}_2^{\top}\bm{W}_{r2}^{(t)}\tilde{\bm{H}}_2 \right),\nonumber
\end{align}
where $\bar{\bm{U}}_{r2}^{(t)}=(\underline{\bm{D}}_r^{(t)}\times_2\bm{B}_r^{(t)}\times_3 (\bm{P}_{\rm M}\bm{C}_r^{(t)}))_{(2)}$ and $\bm{Y}_{M}^{(2)} = (\underline{\bm{Y}}_{M}-\sum_{\bar{r}\neq r}\underline{\bm{D}}^{(t)}_{\bar{r}} \times_1 \bm{A}^{(t)}_{\bar{r}}\times_2 \bm{B}^{(t)}_{\bar{r}}\times_3  (\bm{P}_{M}\bm{C}^{(t)}_{\bar{r}}))_{(2)}$.

\begin{align}
\bm G_{\bm{C}_r}^{(t)}&=\bm{P}_{M}^{\top}(\bm{P}_{M} \bm{C}_r^{(t)} \bar{\bm{V}}_{r3}^{(t)} - \bm{Y}_{M}^{(3)})(\bar{\bm{V}}_{r3}^{(t)})^\top +(\bm{C}_r\bar{\bm{U}}_{r3}^{(t)} - \bm{Y}_{H}^{(3)})(\bar{\bm{U}}_{r3}^{(t)})^\top
+2 \lambda \bm{H}_{3}^{\top}\bm{H}_{3}\bm{C}_{r}^{(t)}, \nonumber\\
L_{\bm C_r}^{(t)}&
=\sigma_{\rm max}\left(\bm{P}_{M}^\top\bm{P}_{M}\right)\sigma_{\rm max}\left(\bar{\bm{V}}_{r3}^{(t)}(\bar{\bm{V}}_{r3}^{(t)})^\top\right)
+ \sigma_{\rm max}\left(\bar{\bm{U}}_{r3}^{(t)}(\bar{\bm{U}}_{r3}^{(t)})^\top\right)
+ 2\lambda \sigma_{\textrm{max}}\left( \bm{H}_{3}^{\top}\bm{H}_{3} \right),\nonumber
\end{align}
where 
$\bar{\bm{V}}_{r3}^{(t)}=(\underline{\bm{D}}_r^{(t)}\times_1\bm{A}_r^{(t+1)}\times_2 \bm{B}_r^{(t+1)})_{(3)}$, $\bar{\bm{U}}_{r3}^{(t)}=(\underline{\bm{D}}_r^{(t)}\times_1\widetilde{\bm{A}}_r^{(t+1)}\times_2 \widetilde{\bm{B}}_r^{(t+1)})_{(3)}$, 
$\bm{Y}_{M}^{(3)} = (\underline{\bm{Y}}_{M}-\sum_{\bar{r}\neq r}\underline{\bm{D}}^{(t)}_{\bar{r}} \times_1 \bm{A}^{(t)}_{\bar{r}}\times_2 \bm{B}^{(t)}_{\bar{r}}\times_3  (\bm{P}_{M}\bm{C}^{(t)}_{\bar{r}}))_{(3)}$, $\bm{Y}_{H}^{(3)} = (\underline{\bm{Y}}_{H}-\sum_{\bar{r}\neq r}\underline{\bm{D}}^{(t)}_{\bar{r}} \times_1 \widetilde{\bm{A}}^{(t)}_{\bar{r}}\times_2 \widetilde{\bm{B}}^{(t)}_{\bar{r}}\times_3  \bm{C}^{(t)}_{\bar{r}})_{(3)}$.

\begin{align}
\bm G_{\underline{\bm{D}}_r}^{(t)}&=(\bar{\bm{U}}^{(t)}_{\underline{\bm{D}}_r})^{\top}
(\bar{\bm{U}}^{(t)}_{\underline{\bm{D}}_r}\bm{d}_r^{(t)} - \bm{y}_{m})
+(\bar{\bm{V}}^{(t)}_{\underline{\bm{D}}_r})^{\top}
(\bar{\bm{V}}^{(t)}_{\underline{\bm{D}}_r}\bm{d}_r^{(t)} - \bm{y}_{h})+\eta \bm{d}_{r}^{(t)},\nonumber \\
L_{\underline{\bm D}_r}^{(t)}
&=\sigma^2_{\rm max}\left(\bm{P}_{M}\bm{C}^{(t+1)}\right)\sigma^2_{\rm max}\left(\bm{A}_r^{(t+1)}\right)\sigma^2_{\rm max}\left(\bm{B}_r^{(t+1)}\right)\nonumber \\
&+ \sigma^2_{\rm max}\left(\bm{C}_r^{(t+1)}\right)\sigma^2_{\rm max}\left(\bm{P}_1\bm{A}_r^{(t+1)}\right)\sigma^2_{\rm max}\left(\bm{P}_2\bm{B}_r^{(t+1)}\right)+\eta,\nonumber
\end{align}
where $\bar{\bm{U}}^{(t)}_{\underline{\bm{D}}_r}=(\bm{P}_{M}\bm{C}_r^{(t+1)})\otimes\bm{B}_r^{(t+1)}\otimes\bm{A}_r^{(t+1)}$, $\bar{\bm{V}}^{(t)}_{\underline{\bm{D}}_r}=\bm{C}_r^{(t+1)}\otimes\widetilde{\bm{A}}_r^{(t+1)}\otimes\widetilde{\bm{B}}_r^{(t+1)}$.

We list the detailed complexity cost of the proposed \texttt{BCLIMB} in Table \ref{table:complexity_blind}. 
Assuming $I_{\rm M} \approx J_{\rm M} \approx K_{\rm H}, ~I_{\rm H} \approx J_{\rm H} \geq K_{\rm M},~ L_r \approx M_r \geq N_r$, the computation of the gradients $\bm{G}^{(t)}_{\widetilde{\bm{A}}_r}$, $\bm{G}^{(t)}_{\bm{A}_r}$, $\bm{G}^{(t)}_{\widetilde{\bm{B}}_r}$, $\bm{G}^{(t)}_{\bm{B}_r}$, $\bm{G}^{(t)}_{\bm{C}_r}$, and $\bm{G}^{(t)}_{\underline{\bm{D}}_r}$ takes
$\mathcal{O}( I_{\rm M}^2I_{\rm H}(I_{\rm H}+L_r) )$,
$\mathcal{O}( I_{\rm M}^2L_r(K_{\rm M}+N_r) )$,
$\mathcal{O}( I_{\rm M}^2I_{\rm H}(I_{\rm H}+L_r) )$,
$\mathcal{O}( I_{\rm M}^2L_r(K_{\rm M}+N_r) )$,
$\mathcal{O}( I_{\rm M}^2 K_{\rm H}(K_{\rm M}+N_r) + I_{\rm H}N_r(I_{\rm M} + L_r^2 + I_{\rm H}L_r) )$, and
$\mathcal{O}( I_{\rm M}^2 L_r K_{\rm M}+ K_{\rm H}L_r(I_{\rm H}^2 + I_{\rm H} L_r + L_rN_r) )$, respectively.
Computing $\alpha_{r,i}^{(t)}, \widetilde{\alpha}_{r,i}^{(t)}$ takes $\mathcal{O}(L_r^3+M_r^3+N_r^3)$ flops. In summary, the per-iteration complexity is 
$\mathcal{O}(I_{\rm M}^3 (I_{\rm H} + L_r))$.

\end{document}